\documentclass[reqno]{amsart}
\usepackage{amsmath}
\usepackage{amsthm}
\usepackage{amssymb}
\usepackage{mathrsfs}
\usepackage{amsfonts}
\usepackage{newlfont}
\usepackage{dcolumn}
\usepackage{bm}
\usepackage{relsize}
\usepackage[margin=1.0in]{geometry}
\numberwithin{equation}{section}
\newtheorem{thm}{Theorem}[section]
\newtheorem{cor}[thm]{Corollary}
\newtheorem{lem}[thm]{Lemma}
\newtheorem{prop}[thm]{Proposition}
\newtheorem{defn}[thm]{Definition}
\newtheorem{rem}[thm]{Remark}
\begin{document}
{\allowdisplaybreaks
\title[stochastic averaging of the Einstein vacuum equations]{STOCHASTIC AVERAGING OF THE EINSTEIN VACUUM EQUATIONS ON A TOROIDAL MANIFOLD WITH RANDOMLY PERTURBED RADIAL MODULI: STABILITY CRITERIA AND INDUCED 'COSMOLOGICAL CONSTANT' TERMS}
\author{Steven D. Miller}
\address{Rytalix Analytics, Strathclyde, Scotland}
\email{stevendM@ed-alumnus.net}
\date{\today}
\begin{abstract}
The Einstein vacuum equations on an (n+1)-dimensional toroidal manifold $\mathbb{M}^{n+1}=\mathbb{T}^{n}\times\mathbb{R}^{+}$ reduce to a system of n-dimensional nonlinear ODEs in terms of the set of toroidal radii $(a_{i})_{i=1}^{n}$ or the radial moduli fields $(\psi_{i})_{i=1}^{n}=(\log(a_{i}(t))_{i=1}^{n}$ of the n-torus $\mathbb{T}^{n}$. This geometry is also the basis of Kasner-Bianchi-type cosmologies. The equations are trivially satisfied for static solutions $\psi_{i}^{E}=\psi^{E}$ or radii $a_{i}^{E}=a^{E}$, describing an initially static toroidal 'micro-universe' or 'vacuum bubble'. It is Lyapunov stable to short-pulse deterministic perturbations, which have a sharp Gaussian profile: the perturbed radii rapidly converge to new attractors and therefore to new stable equilibria. These perturbations induce transitions between stable states. Introducing classical random fluctuations or perturbations, with a regulated covariance, the radial moduli become Gaussian random fields paramatrizing a 'toroidal random geometry'. The randomly perturbed Einstein equations are then interpreted as a stochastic n-dimensional nonlinear dynamical system. Non-vanishing 'cosmological constant' terms are retained within the averaged equations since they are nonlinear. This is analogous to averaging the Navier-Stokes equations in statistical turbulence theory, which yields an additional non-vanishing Reynolds term, since like the Einstein equations they are also of nonlinear hyperbolic type. The expectations of the randomly perturbed toric radii can be estimated from a cumulant expansion method. The initially static toroidal vacuum bubble undergoes eternal 'noise-induced' stochastic exponential growth or 'inflation'. Random radial moduli fields within this scenario therefore act like a 'dark energy'. Finally, a class of random perturbations is considered for which this Einstein system is stable.
\end{abstract}
\maketitle
\tableofcontents
\section{Introduction and motivation}
This paper promotes the potential applications of stochastic and probabilistic methods within mathematical general relativity, as well as concepts from the theory of nonlinear deterministic and random dynamical systems. Stability within a general-relativistic cosmological context is approached by reducing the vacuum Einstein equations on an n-torus to a multi-dimensional nonlinear dynamical system of ordinary differential equations, with initial Cauchy data, which can then be perturbed by deterministic or random perturbations. A key concern with nonlinear systems is the determination of the steady state or stationary motions and their corresponding stability. The stability of equilibrium points is generally ascertained by linear stability analysis when the system is perfectly deterministic. This is usually the case for macroscopic systems in classical dynamics or dynamical systems theory, celestial mechanics, and for scenarios within gravitation and astrophysics [1-9]. Major technical results in stability analysis within pure general relativity have been the Christodoulou-Klainerman proof of the nonlinear stability of Minkowski space [10,11] and more recently proofs for black holes [12].

However, mesoscopic and microscopic systems in physics, chemistry and biology are not perfectly deterministic and are subject to intrinsic or extrinsic noise, fluctuations or random perturbations; indeed all known physical systems will invariably possess noise on some critical length scale, of either thermal or quantum origin. This requires the utilization of stochastic tools [13-31]. A challenge has been to extend existing techniques for linear systems to deal with nonlinear systems coupled to stochastic noise. Calculating properties of noisy stochastic nonlinear systems is generally fraught with difficulties however; in particular, what were established as stable fixed points via a deterministic stability analysis of a system may actually be unstable when subject to intrinsic stochastic perturbations or an external noise bath, where the coupling of the noise/fluctuations to the nonlinearity becomes a crucial issue. Conversely, an unstable deterministic system, described by either ODEs or PDEs, may actually become stable or even metastable in the presence of random perturbations or noise [32,33,34,35,36] and may also restore uniqueness or well-posedness to ODES and PDEs via "regularisation by noise"[22]; indeed ODEs and PDEs can be well-posed under broader general conditions in the presence of noise, than for the purely deterministic situation. The theory of stochastic PDEs is also a growing area of research [37].

Another interesting property of random perturbations or noises is that they can sometimes dissipate or remove blowups and singularities which exist in the purely deterministic dynamical problem [19,22]. In cosmological applications of general relativity, random perturbations/fluctuations in the very early universe are also of considerable interest in relation to Big Bang singularities and structure formation, especially in relation to stochastic or chaotic inflation where a 'bubble' of vacuum will initially grow exponentially [38]. Early universe cosmology is also a regime where gravitation and general relativity are applied on microscopic scales and the effects of randomness and random perturbations become highly relevant and crucial. Cosmological density fluctuations are also taken to be Gaussian random fields [39]. Furthermore, as there is no complete or applicable theory of quantum gravity, it may still be possible to tentatively (but rigorously)apply methods from classical stochastic functional analysis, and incorporate classical noise, fluctuations and random perturbations into general-relativistic scenarios. This has been outlined in [40,41,42]. Self-gravitating Brownian motions within Newtonian theory have also been studied in [43,44,45,46].

In this paper, we consider the Einstein vacuum equations on an n-dimensional toroidal geometry--which is the basis of Kasner-Bianchi type cosmological models--and then develop, and tentatively apply, methods for studying both deterministic and random perturbations of systems of n-dimensional nonlinear dynamical systems to this problem. In particular, the stochastically perturbed and averaged Einstein vacuum equations on a 'random toroidal geometry' can be derived. The randomness of the geometry arise from intrinsic random perturbations or fluctuations of the radial moduli fields which parametrise the metric of the hypertorus. This interprets the cosmological problem as a random n-dimensional nonlinear dynamical system. The outline of the paper is as follows:\raggedbottom
\begin{enumerate}
\item In Section 2, we consider a specific class of n-dimensional nonlinear ODEs and consider methods to ascertain the effects of both 'short-pulse' deterministic perturbations and also random perturbations or noise on static or equilibrium solutions of these ODEs.
\item In Section 3, the Einstein vacuum equations $\mathbf{Ric}_{AB}=0$ and $\mathbf{Ric}_{AB}=\bm{g}_{AB}\Lambda$ for a cosmological constant $\Lambda$ are formulated on the globally hyperbolic spacetime $\mathbb{M}^{n+1}=\mathbb{T}^{n}\times\mathbb{R}^{+}$, where $\mathbb{T}^{n}$ is an isotropic n-torus with metric $\bm{g}_{ij}=\delta_{ij}|2\pi a_{i}(t)|^{2}dX^{i}\otimes dX^{j}\equiv=2\pi\delta_{ij}\exp(2\psi_{i}(t))dX^{i}\otimes dX^{j}$, parametrized by a set of real modulus functions $(\psi_{i}(t))_{i=1}^{n}$, which span $\mathbb{R}^{n}$. The toroidal radii are then $a_{i}(t)=\exp(\psi_{i}(t))$ for i=1...n. The Einstein vacuum equations reduce to n-dimensional sets of nonlinear autonomous ODEs for $\psi_{i}$(t) and $a_{i}(t)$, and are essentially of the form discussed in Section 2. Static solutions $a_{i}^{E}=\exp(\psi_{i}^{E})$ and dynamic solutions $a_{i}(t)=\exp(\psi_{i}(t))$ can be found. These essentially describe static and expanding 'Kasner universes' or "rolling radii".
\item In Section 4, the methods of Section 2 are applied such that the static or equilibrium solutions $a_{i}^{E}=a^{E}$ are subjected to deterministic 'short-pulse' Gaussian perturbations and also to a continuous perturbation of constant amplitude. The perturbed ($\mathcal{L}_{2}$) norms $\|\overline{\bm{a}(t)}-\bm{a}^{E}\|$ can be estimated and the asymptotic stability studied for
    $\lim_{t\uparrow\infty}\|\overline{\bm{a}(t)}-\bm{a}^{E}\|$, where $\bm{a}(t)=(a_{1}(t),...,a_{n}(t))$
    \item In Section 5, the methods of Section 2 are applied such that the static or equilibrium solutions are subjected to random perturbations or noise which are taken to arise from intrinsic random fluctuations of the moduli fields $( \psi_{i}(t))$. The stochastically averaged Einstein equations then lead to extra non-vanishing terms that can be identified as a 'cosmological constant', which arises solely from the nonlinearity of the equations. This is analogous to a Reynolds number arising within stochastically averaged nonlinear Navier-Stokes equations. The randomly perturbed static and dynamical solutions are also solutions of the stochastically averaged Einstein equations.
    \item In Section 6, the stochastic expectation or average $\bm{\mathrm{I\!E}}\lbrace\|\widehat{a}_{i}(t)-\bm{a}^{E}\|\rbrace$ of the stochastically perturbed norm $\|\widehat{a}_{i}(t)-\bm{a}^{E}\|$ is estimated from a cluster expansion method with truncation at second order for a Gaussian dominance approximation. Using a regulated 2-point function ansatz the norm is estimated. The averaged norm $\bm{\mathrm{I\!E}}\lbrace\|\widehat{\bm{a}}(t)-\bm{a}^{E}\|\rbrace$ then grows exponentially or "inflates" for eternity so that $\lim_{t\uparrow\infty}\bm{\mathrm{I\!E}}\lbrace\|\widehat{\bm{a}}(t)-\bm{a}^{E}\|\rbrace=
        \infty$ and with probability $\bm{\mathrm{I\!P}}(\|\widehat{a}_{i}(t)\|=\infty)=1.$
        \item In the final section, a class of random perturbations are considered for which the Einstein system is stable.
\end{enumerate}
\section{'Short-pulse' deterministic and random perturbations of an n-dimensional nonlinear system}
First the following systems of n-dimensional nonlinear ordinary differential equations are considered. Systems of nonlinear ODEs with this structure or form arise in cosmology when applying the Einstein vacuum equations to a n-dimensional toroidal spacetime, which will be derived in detail in Section 3. However, as a prerequisite the generic properties of these forms of nonlinear ODEs are discussed in relation to both 'short-pulse' deterministic and random perturbations or noise of their solutions. In particular, we are interested in the stochastically averaged differential equations.
\begin{prop}
For all $t\in\mathbb{R}^{+}=[0,\infty)$ let
$(\psi_{i}(t))_{i=1}^{n}\equiv\bm{\psi}(t)=(\psi_{1}(t),...,\psi_{n}(t))$ be a set of smooth real scalar functions  spanning $\mathbb{R}^{n}$, that describe the evolution of a nonlinear n-dimensional dynamical system from some initial data set $(\psi_{i}(0))_{i=1}^{n}$. Then let $(u_{i}(t))_{i=1}^{n}\equiv \bm{u}(t)=(u_{1}(t),...,u_{n}(t))$ be a set spanning $\mathbb{R}^{n}$ such that for all $t\in\mathbb{R}^{+}$ and $i=1...n$.
\begin{equation}
u_{i}(t)=\exp(\psi_{i}(t))
\end{equation}
or $\psi_{i}(t)=\ln(u_{i}(t))$. Then the process described by $u_{i}(t)$ is essentially parametrized by the underlying functions $\psi_{i}(t)$. It will be convenient to use the notations $\partial_{t}\equiv d/dt$ and $\partial_{tt}\equiv d^{2}/dt^{2}$ and retain summations throughout. Consider n-dimensional nonlinear ordinary differential equations of the very general form:
\begin{equation}
\mathlarger{\bm{\mathbf{H}}}_{n}\psi_{i}(t)=\sum_{i=1}^{n}\partial_{tt}{\psi}_{i}(t)+
\beta\sum_{i=1}^{n}\partial_{t}{\psi}_{i}(t)\partial_{t}{\psi}_{i}(t)=0
\end{equation}
with $\beta>0$. If $u_{i}(t)=\exp(\psi_{i}(t))$, then $\partial_{t}u_{i}(t)=\partial_{t}u_{i}(t)/u_{i}(t)$ and the second derivative is $\partial_{tt}\psi_{i}(t)=(\partial_{tt}u_{i}(t))/u_{i}(t))-(\partial_{t}u_{i}(t)
\partial_{t}u_{i}(t)/u_{i}(t)u_{i}(t))$, so that an equivalent set of differential equations is
\begin{equation}
\mathbf{D}_{n}u_{i}(t)\equiv\sum_{i=1}^{n}\frac{\partial_{tt}u_{i}(t)}{u_{i}(t)}+(\beta-1)
\sum_{i=1}^{n}\frac{\partial_{t}u_{i}(t)\partial_{t}u_{i}(t)}{u_{i}(t)u_{j}(t)}=0
\end{equation}
with some initial data $(u_{i}(0))_{i=1}^{n}$. For some $C>0$, the inhomogeneous equations are
\begin{align}
&\mathbf{H}_{n}\psi_{i}(t)\equiv\sum_{i=1}^{n}\partial_{tt}{\psi}_{i}(t)+\beta\sum_{i=1}^{n}
\partial_{t}{\psi}_{i}(t)\partial_{t}{\psi}_{i}(t)=\sum_{i=1}^{n}C_{i}=C\\&\mathbf{D}_{n}u_{i}(t)\equiv\sum_{i=1}^{n}\frac{\partial_{tt}u_{i}(t)}{u_{i}(t)}+(\beta-1)
\sum_{i=1}^{n}\frac{\partial_{t}u_{i}(t)\partial_{t}u_{i}(t)}{u_{i}(t)u_{i}(t)}=\sum_{i=1}^{n}C_{i}=C
\end{align}
Here, $\mathbf{H}_{n}$ and $\mathbf{D}_{n}$ are nonlinear differential operators on $\mathbb{R}^{n}$ such that
\begin{align}
&\mathbf{H}_{n}(...)\equiv\sum_{i=1}^{n}\partial_{tt}(...)+
\beta\sum_{i=1}^{n}\partial_{t}(...)\partial_{t}(...)\\&\mathbf{D}_{n}(...)\equiv\sum_{i=1}^{n}\frac{\partial_{tt}(...)}{(...)}+(\beta-1)
\sum_{i=1}^{n}\frac{\partial_{t}(...)\partial_{t}(...)}{(...)(...)}
\end{align}
\end{prop}
The n-dimensional nonlinear differential equations can also be expressed in terms of $L_{2}$ norms.
\begin{cor}
In terms of $L_{2}$ norms, (2.2) and (2.3) are
\begin{align}
&\mathbf{H}_{n}\psi_{i}(t)= \big\|\sqrt{\partial_{tt}\psi_{i}(t)}\big\|^{2}+\beta\big\|\partial_{t}\psi_{i}(t)\big\|^{2}
\\&
\mathbf{D}_{n}u_{i}(t)=\bigg\|\frac{\partial_{tt}a_{i}(t)}{a_{i}(t)}\bigg\|^{2}
+(\beta-1)\bigg\|\frac{\partial_{t}a_{i}(t)}{a_{i}(t)}\bigg\|^{2}
\end{align}
\end{cor}
\begin{cor}
The solution sets $(\psi_{i}(t))_{i=1}^{n}\equiv \bm{\psi}(t)=(\psi_{1}(t),...,\psi_{n}(t))$ and $(u_{i}(t))_{i=1}^{n}\equiv\bm{u}(t)=(u_{1}(t),...,u_{n}(t))$ for all $t\in\mathbb{R}^{+}$ describe a time-dependent vector in $\mathbb{R}^{n}$. A trivial set of static or equilibrium solutions are $\psi_{i}(t)=\psi_{i}(0)=\psi_{i}^{E}$ for all $t\ge 0$ and $u_{i}^{E}=u_{i}^{E}(t)=\exp(\psi_{i}^{E})$. However, the inhomogeneous equations $\mathbf{H}_{n}\psi_{i}(t)=C$ and $\mathbf{D}_{n}a_{i}(t)=C$ can have no static or equilibrium solutions when $C > 0$.
\end{cor}
\begin{lem}
The equations $\mathbf{H}_{n}\psi_{i}(t)=0$ and $\mathbf{D}_{n}u_{i}(t)=0$ have the solutions
\begin{align}
&\psi_{i}(t)=\psi_{i}(0)+q_{i}\ln|t|\\&
u_{i}(t)=u_{i}(0)|t|^{q_{i}}
\end{align}
where $q_{i}\in\mathbb{R}$, provided that $\sum_{i=1}^{n}q_{i}=\beta\sum_{i=1}^{n}q_{i}^{2}$. If $q=q_{i}$ for $i=1...n$ then $\psi_{i}(t)=\psi_{i}(0)+q\ln|t|$ and $u_{i}(t)=u_{i}(0)|t|^{q}$ are solutions if $\beta=1/q$. The static or equilibrium solutions are $\psi_{i}(t)=\psi^{E}$ and $u_{i}(t)=u_{i}^{E}$. The inhomogeneous equations $\mathbf{H}_{n}\psi_{i}(t)=C$ and $\mathbf{D}_{n}u_{i}(t)=C$ have the solutions
\begin{equation}
\psi_{i}(t)=\psi_{i}(0)+q_{i}t
\end{equation}
\begin{equation}
u_{i}(t)=u_{i}(0)\exp(q_{i}t)
\end{equation}
provided that the $q_{i}$ satisfy the constraints
$\sum_{i=1}^{n}q_{i}=\beta\sum_{i=1}^{n}q_{i}^{2}$. If $q=q_{i}$ for $i=1...n$, then $q=\pm(C\beta n)^{1/2}$ and
\begin{align}
&\psi_{i}^{(\pm)}(t)=\psi_{i}(0)\pm(C/\beta n)^{1/2}t\\&
u_{i}(t)^{(\pm)}=u_{i}(0)\exp(\pm(C/\beta n)^{1/2}t)
\end{align}
\end{lem}
\begin{proof}
Since $\partial_{t}\psi_{i}(t)=q_{i}/t$ and $\partial_{tt}\psi_{i}(t)=-q_{i}/t^{2}$
\begin{align}                                                                             &\mathbf{H}_{n}\psi_{i}(t)=-\sum_{i=1}^{n}\frac{q_{i}}{t^{2}}+\beta\sum_{i=1}^{n}\frac{q_{i}^{2}}{t^{2}}=0\nonumber\\&
\Rightarrow -\sum_{i=1}^{n}q_{i}+\beta\sum_{i=1}^{n}q_{i}^{2}=0
\end{align}
If $q_{i}=q$ for $i=1...n$ then $-\sum_{i=1}^{n}q_{i}+\beta\sum_{i=1}^{n}q_{i}^{2}=-nq+\beta n q^{2}$ so that $\beta=1/q$. Similarly
\begin{align}
&\mathbf{D}_{n}u_{i}(t)=\frac{q_{i}(q_{i}-1)|t|^{q_{i}-2}}{a_{i}(0)|t|^{q_{i}}}+(\beta-1)
\sum_{i=1}^{n}\frac{a_{i}(0)a_{i}(0)|t|^{q_{1}-1}|t|^{q_{i}-1}}{a_{i}(0)|t|^{q_{i}}a_{i}(0)|t|^{q_{i}}}\nonumber\\
&=\sum_{i=1}^{n}(q_{i}q_{i}-q_{i})+(\beta-1)\sum_{i=1}^{n}q_{i}q_{i}
\Rightarrow -\sum_{i=1}^{n}q_{i}+\beta\sum_{i=1}^{n}q_{i}^{2}=0
\end{align}
which again gives $\beta=1/q$ if $q_{i}=q$. Static or equilibrium solutions are simply any
$\psi_{i}(t)=\psi^{E}=const.$ and $ u_{i}(t)=u_{i}^{E}=\exp(\psi^{E})$. Since $\partial_{t}\psi_{i}(t)=q_{i}$ and $\partial_{tt}\psi_{i}(t)=0$, then the inhomogeneous equation $\mathbf{H}_{n}\psi_{i}(t)=C$ becomes $ \beta\sum_{i=1}^{n}q_{i}q_{i}=C $
so that $\beta nq^{2}=C$ if $q=q_{i}$ and $\beta=\pm(C/n\beta)^{1/2}$. Since $\partial_{t}u_{i}(t)=u_{i}(0)q_{i}\exp(q_{i}t)$ and $\partial_{tt}u_{i}(t)=u_{i}(0)q_{i}q_{i}\exp(q_{i}t)$ then for the inhomogeneous equation $\mathbf{D}_{n} u_{i}(t)=C$
\begin{align}
&\sum_{i=1}^{n}\frac{u_{i}(0)q_{i}q_{i}\exp(q_{i}t)}{u_{i}(0)\exp(q_{i}t)}+(\beta-1)\sum_{i=1}^{n}
\frac{u_{i}(0)u_{i}(0)q_{i}q_{i}\exp(2q_{i}t)}{u_{i}(0)u_{i}(0)\exp(2q_{i}t}\nonumber\\&
=\sum_{i=1}^{n}q_{i}^{2}+(\beta-1)\sum_{i=1}^{n}q_{i}q_{i}=C
\end{align}
If $q_{i}=q$ then $nq^{2}+n(\beta-1)q^{2}=C$  so that again $q=\pm(C/n\beta)^{1/2}$.
\end{proof}
\begin{rem}
The nonlinear system described by (2.2) and (2.3) can be considered as special cases of the following general system of n-dimensional ODEs such that
\begin{equation}
\mathbf{H}_{n}\psi_{i}(t)\equiv\alpha\sum_{i=1}^{n}\partial_{tt}\psi_{i}(t)+\beta\sum_{i=1}^{n}(\partial_{t}\psi_{i}(t))^{b}+
\gamma\sum_{i=1}^{n}(\psi_{i}(t))^{a}=0
\end{equation}
with $(\alpha,\beta,\gamma)\in\mathbb{R}^{+}$ and integers $(a,b)\in\mathbb{Z}$. In terms of $u_{i}(t)=\exp(\psi_{i}(t))$ the equation is equivalently
\begin{align}
&\mathbf{D}_{n}a_{i}(t)=\alpha\sum_{i=1}^{n}\frac{\partial_{tt}u_{i}(t)}{u_{i}(t)}-\alpha\sum_{i=1}^{n}
\frac{\partial_{t}u_{i}(t)\partial_{i}(t)}{u_{i}(t)u_{i}(t)}\\&+\beta\sum_{i=1}^{n}
\left(\frac{\partial_{t}u_{i}(t)}{u_{i}(t)}\right)^{b}+\gamma\sum_{i=1}^{n}
(\ln(u_{i}(t))^{a}=0
\end{align}
These are essentially n-dimensional nonlinear polynomial autonomous systems. Equation (2.2) is the case for $\gamma=0, b=2$. For $\alpha=0,\gamma=1, \beta=0, a=2$, we have a system of Riccati equations.
\begin{equation}
\sum_{i=1}^{n}\partial_{t}\psi_{i}(t)+\gamma\sum_{i=1}^{n}(\psi_{i}(t))^{2}=0
\end{equation}
with (non-global) solutions $\psi_{i}(t)=|t-\psi_{i}^{-1}(0)|^{-1}$. For $\alpha=1,\beta=0,\gamma=1,a=1$, equation (2.17) reduces to a linear system of coupled simple harmonic oscillators
\begin{equation}
\sum_{i=1}^{n}\partial_{tt}\psi_{i}(t)+\gamma\sum_{i=1}^{n}\psi_{i}(t)=0
\end{equation}
with the basic solutions $\psi_{i}(t)=\psi_{i}(0)\cos(\sqrt{\gamma}t-\varphi)$, where $\varphi$ is a phase angle.
\end{rem}
\begin{defn}
\begin{enumerate}
The following standard definitions are given
\item Let $\bm{u}(t)=(u_{1}(t),...,u_{n}(t))$. The $\mathcal{L}_{2}$ norm of a vector $u_{i}(t)$ is $\|\bm{u}(t)\|\equiv\|u_{i}(t)\|=\left(\sum_{i=1}^{n}|u_{i}(t)|^{2}\right)^{1/2}$ and the $L_{p}$-norm is $\|\bm{u}(t)\|\equiv \|u_{i}(t)\|=\left(\sum_{i=1}^{n}|u_{i}(t)|^{p}\right)^{1/p}$.
\item A set of $n$ equilibrium stable fixed points is denoted $(u_{i}^{E})_{i=1}^{n}\equiv \bm{u}^{E}=(u_{1}^{E},...,u_{1}^{E})$ with $\mathcal{L}_{2}$ norm $\|\bm{u}^{E}\|$. For an isotropic equilibrium configuration, one can set $u_{i}^{E}=u^{E}$ for $i=1...n$. For initial data $\bm{\psi}^{E}=\bm{\psi}(0)$ or $\bm{u}^{E}=\bm{u}(0)$, the Cauchy developments for $t>0$ are $\bm{\psi}(t)$ and $\bm{u}(t)$.
\item Given an Euclidean ball $\mathbb{B}(L)$ of radius $L$, a set of stable fixed points $u_{i}^{E}$ spanning $\mathbb{R}^{n}$ is Lyapunov stable for all $\mathbb{B}_{1}\subset\mathbb{R}^{n}$ where $\|\bm{u}^{E}\|\in\mathbb{B}_{1}$ if $\exists$ $\mathbb{B}_{2}\in\mathbb{R}^{n}$ with $\|\bm{u}_{E}\|\in\mathbb{B}_{2}$, if for all $t>0$, $\exists \delta>0$ and $\delta < \epsilon$ such that $\|\bm{u}(t)-\bm{u}^{E}\|<\delta $ implies $\|\bm{u}(t)-\bm{u}^{E}\|<\epsilon $. One can choose any Euclidean ball $\mathbb{B}(\epsilon)$ with $\bm{u}^{E}\in\mathbb{B}(\epsilon)$ for any small $\epsilon>0$ such that all future states $\|\bm{\psi}(t)\|\in \mathbb{B}_{\epsilon}$ are trapped within $\mathbb{B}(\epsilon)$ provided that they start out in a smaller ball $\mathbb{B}(\delta)\subset\mathbb{B}(\epsilon)$. So $\exists$ finite ball $\mathbb{B}(\epsilon)$ such that for all $t>0$, $\bm{u}(t)\in\mathbb{B}(\epsilon)$.
\item Lyapunov stability requires a convergent norm such that $\exists$ $K>0$ whereby
    $ 0<\lim_{t\uparrow\infty}\|{u}_{i}(t)-u_{i}\|\le K $. The norm is always taken to the $\mathcal{L}_{2}$-norm although the definitions will still hold for $L_{p}$-norms.
\item  The Lyapunov stability criterion is weaker than that of asymptotic stability. Suppose Lyapunov stability holds, then if $\exists \mathbb{B}(\delta)\subset\mathbb{R}^{n}$ of radius $\delta$ such that $\|\bm{u}(0)-\bm{u}^{E}\|<\delta $ then $\lim_{\delta\rightarrow o}|\bm{u}(t)-\bm{u}^{E}\|=0$. If $u_{i}^{E}=0$ for $i=1...n$ or $\bm{u}^{E}=0$, then the equilibrium points are just the origin of $\mathbb{R}^{n}$ and $\lim_{\delta\uparrow 0}\|\bm{u}(t)\|=0$. For asymptotic stability,the norm converges to zero such that $\lim_{t\uparrow\infty}\|\bm{u}(t)-\bm{u}^{E}\|=0 $
\item The system is essentially unstable if $\lim_{t\uparrow\infty}\|u_{i}(t)-\bm{u}_{i}^{E}\|=
\infty $.
\end{enumerate}
\end{defn}
For example, for the dynamic solution (2.10) with $u_{i}(0)\equiv u_{i}^{E}$ we have the estimate
\begin{align}
&\lim_{t\uparrow\infty}\|\bm{u}(t)-\bm{u}(0)\|\le \lim_{t\uparrow\infty}\|\bm{u}(t)\|-\|\bm{u}(0)\|\nonumber\\
&=\lim_{t\uparrow\infty}\left(\sum_{i=1}^{n}|u_{i}(0)|^{2}|t|^{2q_{i}}\right)^{1/2}-\|\bm{u}(0)\|\nonumber\\
&=\lim_{t\uparrow\infty}n^{1/2}u(0)|t|^{q}-\|\bm{u}(0)\|=\infty
\end{align}
where $u_{i}(0)=u(0)$ for $i=1...n$, so there is no convergence to equilibrium and the system is unstable. Similarly for $u_{i}(t)=u_{i}(0)\exp((C/n\beta)^{1/2}t)$ we have
\begin{align}
&\lim_{t\uparrow\infty}\|\bm{u}^{(+)}(t)-\bm{u}(0)\|\le \lim_{t\uparrow\infty}\|\bm{u}^{(+)}(t)\|-\|\bm{u}(0)\|\nonumber\\&
=\lim_{t\uparrow\infty}\left(\sum_{i=1}^{n}|u_{i}(0)|^{2}\exp(+2(C/n\beta)^{1/2}t)
\right)^{1/2}-\|\bm{u}(0)\|\nonumber\\
&=\lim_{t\uparrow\infty}(n^{1/2}u(0)\exp((C/n\beta)^{1/2}t)-\|\bm{u}(0)\|)=\infty
\end{align}
with $u_{i}(0)=u(0)$ for $i=1...n$, giving an unbounded exponential expansion. However, the solution $u_{i}^{(-)}(t)=u_{i}(0)\exp(-(C/n\beta)^{1/2}t)$ is asymptotically stable in that
\begin{align}
&\lim_{t\uparrow\infty}\|\bm{u}^{(-)}(t)-\bm{u}(0)\|\le
\lim_{t\uparrow\infty}\|\bm{u}^{(-)}(t)\|-\|\bm{u}(0)\|\nonumber\\&
\le \lim_{t\uparrow\infty}\left(\sum_{i=1}^{n}|u_{i}(0)|^{2}\exp(-2(C/n\beta)^{1/2}t\right)^{1/2}\nonumber\\
&=\lim_{t\uparrow\infty}(n^{1/2}u(0)\exp(-(C/n\beta)^{1/2}t))=0.
\end{align}
The standard and practical method for evaluating the stability of fixed points in nonlinear dynamical systems and classical mechanics is linear stability analysis [1]. Consider a first-order NLDE. For small perturbations $\xi_{i}(t)$ around the equilibrium points
so that $\overline{u_{i}(t)}=u_{i}^{E}+\xi_{i}(t)$. NLDEs can linearised by dropping higher-order terms and performing a 'normal-mode analysis'. However, in this paper, the aim is to study specific types of 'short-pulse' and random perturbations while retaining the full nonlinearity of the equations; in particular, if the Einstein equations are reduced to an n-dimensional nonlinear system of ODES within a dynamical systems interpretation, it is the nonlinearity that is of prime interest.

Nonlinear systems can also become chaotic, whereby the future evolution no longer becomes predictable from initial Cauchy data [9,47,48]. A useful'acid test' for chaos is the Lyapunov characteristic exponent (LCE) which gives the rate of exponential divergence from perturbed initial conditions.
\begin{defn}
The LCE of a dynamical system quantifies the rate of change or divergence or separation of initially infinitesimally close trajectories in phase space. For an n-dimensional system, if $\overline{u_{i}(t)}$ and $u_{i}(t)$ span $\mathbb{R}^{n}$ for $i=1...n$, with $t\in[0,\infty)$ then let $\delta {u}_{i}(t)=\overline{u_{i}(t)}-u_{i}(t)$ and $\|\delta \bm{u}_{i}(0)\|=\|\overline{\bm{u}(0)}-\bm{u}(0)\|$. If
\begin{equation}
\frac{\|\delta\bm{u}(t)\|}{\|\delta \bm{u}(0)\|}=\frac{\|\overline{\bm{u}(t)}-\bm{u}(t)\|}{\|\overline{\bm{u}(0)}-\bm{u}(0)\|}= \frac{\left(\sum_{i=1}^{n}|\overline{u_{i}(t)}-u_{i}(t)|^{2}\right)^{1/2}}{
\left(\sum_{i=1}^{n}|\overline{u_{i}(0)}-u_{i}(0)|^{2}\right)^{1/2}}\sim\exp(\mathbf{Ly} t)
\end{equation}
then $\mathbf{Ly}$ is a LCE of the system and $\delta u_{i}(t)\sim\delta u_{i}(0)\exp(\lambda t)$. The maximal LCE is the average deviation from the unperturbed state or orbit at time $t>0$ and is established by the Oseledec Theorem [6] as
\begin{equation}
\mathbf{Ly}=\lim_{t\uparrow\infty}\lim_{\|\delta \bm{u}(0)\|\uparrow 0}\frac{1}{t}\ln\left(\frac{\|\delta \bm{u}(t)\|}{\|\delta \bm{u}(0)\|}\right)
\end{equation}
Then:
\begin{enumerate}
\item In a chaotic region, future evolution is independent of initial conditions. When $\lambda<0$, the orbits attract to a stable fixed points or 'attractors' and the system exhibits asymptotic stability. If $\lambda=-\infty$, the system is superstable.
\item If $\lambda=0$, the system is Lyapunov stable.
\item For all $\lambda >0$, the system is unstable and nearby points or orbits diverge exponentially to arbitrary large separations.
\end{enumerate}
\end{defn}
\subsection{'Short-pulse' deterministic perturbations}
In this paper, "short-pulse" deterministic perturbations with respect to fully nonlinear ODEs of the from (2.2) or (2.3) will initially be considered, as a prerequisite to studying the effect of random perturbations or noise.
\begin{prop}
Let $(\mathcal{U}_{i}(t,\vartheta_{i})_{1\le i\le n}$ be a set of functions $\mathcal{U}_{i}:\mathbb{R}^{+}\rightarrow\mathbb{R}^{+}$ such that:
\begin{enumerate}
\item $\mathcal{U}_{i}(t,\vartheta_{i})\rightarrow 0$ for $t\gg\|\bm{\vartheta}\|$, with $\|\bm{\vartheta}\|\ll 1 $. For example, a set of highly peaked Gaussian functions with widths $\vartheta_{i}$ or 'smeared-out' delta functions. Define a vector
    $\bm{\mathcal{U}}(t,\bm{\vartheta})=\mathcal{U}_{1}(t,\epsilon_{1}),...,
    \mathcal{U}_{n}(t,\vartheta_{n}))$ and $\bm{\vartheta}=(\vartheta_{1},...,\vartheta_{n})$ spanning $\mathbb{R}^{n}$. The $L_{2}$ norm is $\|\mathcal{U}_{i}(t,\vartheta_{i})\|\equiv \mathbf{\mathcal{U}}(t,\bm{\vartheta})$.
\item The functions $\mathcal{U}_{i}(t,\vartheta_{i})$ are sufficiently smooth such that the derivatives $\partial_{t}\mathcal{U}_{i}(t,\vartheta_{i})$ and $\partial_{tt}\mathcal{U}_{i}(t,\vartheta_{i})$ exist and also rapidly decay for
$t\gg\|\bm{\vartheta}\|$ so that $\|\partial_{t}\bm{\mathcal{U}}(t,\vartheta)\|\rightarrow 0$ and $\|\partial_{tt}\bm{\mathcal{U}}(t,\vartheta)\|\rightarrow 0$
\item The integrals $\int_{0}^{t}\mathcal{U}_{i}(\tau,\vartheta_{i})d\tau>0$ exist and are well defined. The integrals
\begin{equation}
\lim_{t\uparrow\infty}\int_{0}^{t}\mathcal{U}_{i}(\tau,\vartheta_{i})d\tau,
\lim_{t\uparrow\infty}\left\|\int_{0}^{t}\bm{\mathcal{U}}_{i}(\tau,\bm{\vartheta})d\tau\right\|
\end{equation}
may or may not not converge, although we primarily are concerned with convergent integrals such that
\begin{equation}
\lim_{t\uparrow\infty}\int_{0}^{t}\mathcal{U}_{i}(\tau,\vartheta_{i})d\tau=Q <\infty
\end{equation}
\end{enumerate}
Then if $\psi_{i}^{E}$ and $u_{i}^{E}$ are static equilibrium fixed points or solutions of (2.2)and (2.3) such that $\mathbf{H}_{n}\psi_{i}^{E}=0$ and $\mathbf{D}_{n}a_{i}^{E}=0$ then the perturbed static solutions are
\begin{align}
&\overline{\psi_{i}(t)}=\psi_{i}^{E}+\int_{0}^{t}\mathcal{U}_{i}(\tau,\vartheta_{i})d\tau\\&
\overline{u_{i}(t)}=u_{i}^{E}\exp\left(\int_{0}^{t}\mathcal{U}_{i}(\tau,\vartheta_{i}\right)d\tau
\equiv u_{i}^{E}\mathcal{B}(t)
\end{align}
Then
\begin{equation}
\overline{u_{i}(t)}-u_{i}^{E}=u_{i}^{E}\exp\left(\int_{0}^{t}
\mathcal{U}_{i}(\tau,\vartheta)d\tau
\right)-u_{i}^{E}< u_{i}^{E}\exp\left(\int_{0}^{t}\mathcal{U}_{i}(\tau,\vartheta_{i})d\tau\right)
\end{equation}
since $u_{i}^{E}>0$.
The perturbed norm $\|\overline{\bm{u}(t)}-\bm{u}^{E}\|$ is then estimated as
\begin{align}
&\|\overline{\bm{u}(t)}-\bm{u}^{E}\|\le \|\overline{\bm{u}(t)}\|-\|\bm{u}^{E}\|
=\left\|u_{i}^{E}\exp\left(\int_{0}^{t}\bm{\mathcal{U}}(t,\bm{\vartheta})d\tau\right)\right\|
-\|\bm{u}^{E}\|\nonumber\\&
=\left(\sum_{i=1}^{n}|u_{i}^{E}\exp(\int_{0}^{t}\mathcal{U}_{i}(\tau,\vartheta_{i})
d\tau)|^{2}\right)^{1/2}-\|\bm{u}^{E}\|\nonumber\\&
\le\left(\sum_{i=1}^{n}\left|u_{i}^{E}\exp\left(\int_{0}^{t}\mathcal{U}_{i}(\tau,\vartheta_{i})
d\tau\right)\right|^{2}\right)^{1/2}\equiv n^{1/2}u^{E}\exp\left(\int_{0}^{t}\mathcal{U}(\tau,\vartheta)d\tau\right)
\end{align}
if $\mathcal{U}_{i}(t,\vartheta_{i})=\mathcal{U}(t,\vartheta)$ and $\vartheta_{i}=\vartheta$ for $i=1...n$, representing an isotropic set of perturbations. Evaluating the norm estimate then enables the asymptotic behavior and stability to be deduced for $t\rightarrow\infty$. Stability and Lyupunov stability requires $\exists K$ such that
\begin{equation}
\lim_{t\uparrow\infty}\|\overline{\bm{u}(t)}-\bm{u}^{E}\| \le \lim_{t\uparrow\infty}n^{1/2}u^{E}\exp\left(\int_{0}^{t}\mathcal{U}(\tau,\vartheta)d\tau\right)<K<\infty
\end{equation}
which will be the case if the integral $\lim_{t\uparrow\infty}\int_{0}^{t}
\mathcal{U}(\tau,\vartheta)d\tau$ converges.
\end{prop}
\begin{lem}
Equations (2.31) and (2.32) are solutions of the perturbed differential equations
\begin{equation}
\mathbf{H}_{n}\overline{\psi_{i}(t)}=\big\|\sqrt{\partial_{t}\mathcal{U}_{i}(t,\vartheta_{i})}\big\|_{L_{2}}^{2}
+\beta\big\|\mathcal{U}_{i}(t,\vartheta_{i})\big\|_{L_{2}}^{2}
\end{equation}
\begin{equation}
\mathbf{D}_{n}\overline{u_{i}(t)}=\big\|\sqrt{\partial_{t}\mathcal{U}_{i}(t,\vartheta_{i})}\big\|_{L_{2}}^{2}
+\beta\big\|\mathcal{U}_{i}(t,\vartheta_{i})\big\|_{L_{2}}^{2}
\end{equation}
If $\mathcal{U}_{i}(t,\vartheta)=\mathcal{U}(t,\vartheta)$ for $i=1...n$, then equations (2.31) and (2.32) are solutions of the equivalent perturbed ODEs
\begin{equation}
\mathbf{H}_{n}\overline{\psi_{i}(t)}=n\partial_{t}\mathcal{U}(t,\vartheta)+
n\beta |\mathcal{U}(t,\vartheta)|^{2}=\mathcal{S}(t,\vartheta)
\end{equation}
\begin{equation}
\mathbf{D}_{n}\overline{u_{i}(t)}=n\partial_{t}\mathcal{U}(t,\vartheta)+
n\beta |\mathcal{U}(t,\vartheta)|^{2}=\mathcal{S}(t,\vartheta)
\end{equation}
\end{lem}
\begin{proof}
The perturbed equation is
\begin{align}
&\mathbf{H}\overline{\psi_{i}(t)}=\sum_{i=1}^{n}\partial_{tt}\overline{\psi_{i}(t)}+\beta\sum_{i=1}^{n}
\partial_{t}\overline{\psi_{i}(t)}\partial_{t}\overline{\psi_{i}(t)}\nonumber\\
&=\sum_{i=1}^{n}\partial_{t}\mathcal{U}_{i}(t,\vartheta)+\beta\sum_{i=1}^{n}
\mathcal{U}_{i}(t,\vartheta)\mathcal{U}_{i}(t,\vartheta_{i})=\mathcal{S}_{i}(t,\varsigma_{i})
\end{align}
If $\mathcal{U}_{i}(t,\vartheta_{i})=\mathcal{U}(t,\vartheta)$ for £$i=1...n$ then (2.36) follows. Using (2.3), the perturbed nonlinear system of ODEs is equivalently
\begin{align}
\mathbf{D}_{n}\overline{u_{i}(t)}&=\sum_{i=1}^{n}\frac{\partial_{tt}\overline{u_{i}(t)}}
{\overline{u_{i}(t)}}+(\beta-1)\sum_{i=1}^{n}\frac{\partial_{t}\overline{u_{i}(t)}\partial_{t}
\overline{u_{i}(t)}}{\overline{u_{i}(t)}\overline{u_{i}(t)}}=\sum_{i=1}^{n}
\frac{u_{i}^{E}\partial_{t}\mathcal{C}_{i}(t,\vartheta_{i})\mathcal{J}_{i}(t)}{u_{i}^{E}
\mathcal{J}_{i}(t)}\nonumber\\&+\sum_{i=1}^{n}\frac{u_{i}^{E}\mathcal{U}_{i}(t,\vartheta_{i})
\mathcal{U}_{i}(t,\vartheta_{i})\mathcal{J}_{i}(t)}{u_{i}^{E}\mathcal{J}_{i}(t)}+(\beta-1)\sum_{i=1}^{n}\frac{u_{i}^{E}u_{i}^{E}
\mathcal{U}_{i}(t,\vartheta_{i})\mathcal{U}_{i}(t,\vartheta_{i})\mathcal{J}_{i}(t)\mathcal{J}_{i}(t)}
{u_{i}^{E}\mathcal{J}_{i}(t)u_{i}^{E}\mathcal{J}_{i}(t)}\nonumber\\&=\sum_{i=1}^{n}\partial_{t}
\mathcal{U}_{i}(t,\vartheta_{i})+\beta\sum_{i=1}^{n}\mathcal{C}_{i}(t,\vartheta)\mathcal{U}_{i}(t,\vartheta)
\nonumber\\&=\sum_{i=1}^{n}\partial_{t}\mathcal{U}_{i}(t,\vartheta_{i})+\beta\sum_{i=1}^{n}\delta_{ii}\mathcal{U}_{i}(t,\vartheta_{i})
\mathcal{U}_{i}(t,\vartheta_{i})\nonumber\\&\equiv
\big\|\sqrt{\partial_{t}\mathcal{C}_{i}(t,\vartheta_{i})}\big\|_{L_{2}}^{2}
+\beta\big\|\mathcal{U}_{i}(t,\vartheta_{i})\big\|_{L_{2}}^{2}
=\sum_{i=1}^{n}\mathcal{S}_{i}(t,\vartheta)
\end{align}
where $\mathcal{J}_{i}(t)=\exp(\int_{o}^{t}\mathcal{U}_{i}(t,\vartheta_{i})d\tau)$. If $\mathcal{U}_{i}(t,\vartheta_{i})=\mathcal{U}(t,\vartheta)$ and $\vartheta_{i}=\vartheta$ for $i=1...n$ then
\begin{equation}
\mathbf{D}_{n}\overline{u_{i}(t)}=n\partial_{t}\mathcal{U}(t,\vartheta)+
n\beta |\mathcal{U}(t,\vartheta)|^{2}=\mathcal{S}(t,\vartheta)
\end{equation}
But the perturbed equations converge back rapidly for $t\gg \|\bm{\vartheta}\|$ so that
$\mathbf{D}_{n}\overline{u_{i}(t)}\rightarrow 0$ very rapidly for $t\gg \|\bm{\vartheta}\|$ if $\mathcal{U}(t,\vartheta)\rightarrow 0$ and
$\partial_{t}\mathcal{U}(t,\vartheta)\rightarrow 0$ for $t\gg \|\bm{\vartheta}\|$
\end{proof}
\begin{cor}
Since $\|\mathcal{U}(t,\vartheta_{i})\|\rightarrow 0$ and $\|\partial_{t}\mathcal{U}(t,\vartheta_{i})\|\rightarrow 0$ for $t\gg\|\vartheta_{i}\|$ then
the perturbed equations decay rapidly to the unperturbed equations for $t\gg\|\vartheta_{i}\|$ so that
\begin{align}
&\lim_{t\uparrow\infty}\mathbf{H}_{n}\overline{\psi_{i}(t)}=\lim_{t\uparrow\infty}\big\|\sqrt{\partial_{t}\mathcal{U}_{i}(t,\vartheta_{i})}\big\|_{L_{2}}^{2}+\lim_{t\uparrow\infty}\big\|\mathcal{U}(t,\vartheta)\big\|^{2}_{L_{2}}=\mathbf{H}_{n}\psi_{i}(t)
\\&
\lim_{t\uparrow\infty}\mathbf{D}_{n}\overline{u_{i}(t)}=\lim_{t\uparrow\infty}\big\|\sqrt{\partial_{t}\mathcal{U}_{i}(t,\vartheta_{i})}\big\|_{L_{2}}^{2}+\lim_{t\uparrow\infty}\big\|\mathcal{U}(t,\vartheta)\big\|^{2}_{L_{2}}=\mathbf{D}_{n}u_{i}(t)
\end{align}
\end{cor}                                                                                    \begin{prop}                                                                                 Suppose $u_{i}^{E}$ are a set of stable solutions of equilibrium fixed points of the system of ODEs $\mathbf{D}_{n}\psi{i}(t)=\mathbf{D}_{n}u_{i}^{E}=0$. Let $\overline{u_{i}(t)}$ denote perturbations of the stable fixed points then $\|\delta\overline{\bm{u}(t)}\|=\|\overline{\bm{u}(t)}-\bm{u}^{E}\|$. If
\begin{equation}
\frac{\|\delta\overline{\bm{u}(t)}-\bm{u}^{E}\|}{\|\bm{u}^{E}\|}\equiv \frac{\|\overline{\bm{u}_{i}(t)}-\bm{u}^{E}\|}{\|\bm{u}^{E}\|}\sim \exp(\mathbf{Ly} t)
\end{equation}
then $\mathbf{Ly}$ is essentially a LCE with stability if $\mathbf{Ly}<0$ and instability if $\mathbf{Ly} >0$. For example, if $\mathcal{U}_{i}(t,\vartheta)=\mathcal{A}_{i}$ is a constant perturbation or amplitude then
\begin{equation}
\overline{u_{i}(t)}=u_{i}^{E}\exp(\left(\int_{0}^{t}\mathcal{A}_{i}d\tau\right)=u_{i}^{E}\exp(\mathcal{A}_{i}t)
\end{equation}
Using (2.33), the norm is estimated as $\|\overline{\bm{u}(t)}-\bm{u}^{E}\|< n^{1/2}u^{E}\exp(\mathcal{A}t)$ if $\mathcal{A}_{i}=\mathcal{A}$ so that
\begin{equation}
\frac{\|\delta\overline{\bm{u}(t)}\|}{\|\bm{u}^{E}\|}\equiv \frac{\|\overline{\bm{u}(t)}-\bm{u}^{E}\|}{\|\bm{u}^{E}\|}\sim \exp(\mathcal{A} t)\equiv
\exp(\mathbf{Ly}t)
\end{equation}
so that $\mathcal{A}$ is essentially a Lyupunov exponent, with stability for $\mathcal{A}<0$ and instability for $\mathcal{A}>0$.
\end{prop}
\begin{cor}
Equation (2.46) is also a solution of the ODEs
\begin{equation}
\mathbf{D}_{n}\overline{u_{i}(t)}=\beta\big\|\mathcal{A}_{i}\big\|_{L_{2}}
=n\beta\lambda^{2}
\end{equation}
if $\mathcal{A}_{i}=\mathcal{A}$ for all $i=1...n$, and $\mathcal{A}^{2}=\lambda^{2}/\beta$.
Since $\partial_{t}\overline{u_{i}(t)}=u_{i}^{E}\mathcal{A}_{i}\exp(\mathcal{A}_{i}t)$ and $\partial_{tt}\overline{u_{i}(t)}=u_{i}^{E}\mathcal{A}_{i}\mathcal{A}_{i}\exp(\mathcal{A}_{i}t)$ then if $\mathcal{A}_{i}=\mathcal{A}$
\begin{align}
\mathbf{D}_{n}\overline{u_{i}(t)}&\equiv\sum_{i=1}^{n}\frac{\partial_{tt}\overline{u_{i}(t)}}{
\overline{u_{i}(t)}}
+(\beta-1)\sum_{i=1}^{n}\frac{\partial_{t}\overline{u_{i}(t)}\partial_{t}\overline{u_{i}(t)}}{
\overline{u_{i}(t)}\overline{u_{j}(t)}}\nonumber\\&
=\sum_{i=1}^{n}\frac{u_{i}^{E}\mathcal{A}_{i}\mathcal{A}_{i}\exp(\mathcal{A}_{i}t)}{u_{i}^{E}\exp(\mathcal{A}_{i}t)}+(\beta-1)
\sum_{i=1}^{n}\frac{u_{i}^{E}u_{i}^{E}\mathcal{A}_{i}\mathcal{A}_{i}\exp(2\mathcal{A}_{i}t)}{u_{i}^{E}u_{i}^{E}\exp(2\mathcal{A}_{i}t)} \\&=\sum_{i=1}^{n}\mathcal{A}_{i}\mathcal{A}_{i}+(\beta-1)\sum_{i=1}^{n}\mathcal{A}_{i}\mathcal{A}_{i}\\&
=\beta\sum_{i=1}^{n}\mathcal{A}_{i}\mathcal{A}_{i}\equiv\beta\big\|A_{i}\big\|^{2}=n\beta \mathcal{A}^{2}\equiv n\mathbf{Ly}^{2}
\end{align}
if $\mathcal{A}^{2}=\lambda^{2}/\beta$.                                                                \end{cor}
\subsection{Random perturbations}
Suppose instead, the nonlinear system described by (2.2) or (2.3) is subject to random perturbations, fluctuations or noise, either externally via an external coupled noise or thermal bath, or intrinsically. A random field can be defined as follows. (Appendix A.)
\begin{defn}
If $(\Omega,\mathfrak{F},\mu,T)$ is a probability space then for all $t\in\mathbb{R}^{+}$ and $\omega\in\Omega$ there is a map $ \mathfrak{M}:\omega\times\mathbb{R}^{+}\rightarrow \widehat{\mathscr{U}}_{i}(t,\omega)\equiv \widehat{\mathscr{U}}_{i}(t)$. The stochastic expectation or average of any stochastic quantity or field $\widehat{\mathscr{S}}(t,\omega)$ is obtained by integration over the measure $\mu(\omega)$ so that $\bm{\mathrm{I\!E}}\lbrace\widehat{\mathscr{S}}(t,\omega)\rbrace
=\int_{\Omega}\widehat{\mathscr{S}}(t,\omega)d\mu(\omega)$. Note that all stochastic quantities will have an overhead tilda. For Gaussian free vector fields
$\widehat{\mathscr{U}}_{i}(t)$ with $Q\ge 0$
\begin{equation}
\bm{\mathrm{I\!E}}\lbrace\widehat{\mathscr{U}}_{i}(t,\omega)\rbrace
=\int_{\Omega}\widehat{\mathscr{U}}_{i}(t,\omega)d\mu(\omega)\le Q
\end{equation}
The map $T:\Omega\rightarrow\Omega$ is a measure-preserving transformation such that $\mu(T^{-1}B)=\mu(B)$ for all $B\in\mathfrak{F}$. For random dynamical systems $(\Omega,\mathfrak{F},\mu,T)$ one introduces the concepts of mixing and ergodicity such that for all $A,B\in\mathfrak{F}$ one has strong 2-mixing $\mu(T^{-1}A\cap B)\rightarrow \mu(A)\mu(B)$. The covariance for the Gaussian field $\widehat{\mathscr{U}}_{i}(t)$ is formally
\begin{align}
&\bm{\mathrm{Cov}}_{ij}(t,s)=\int_{\Omega}\int_{\Omega}\mathscr{U}_{i}(t,\omega)\mathscr{U}_{j}(s,\xi)d\mu(\omega)d\mu(\xi)\nonumber\\&
\le\bm{\mathrm{I\!E}}\bigg
\lbrace\widehat{\mathscr{U}}_{i}(t)\widehat{\mathscr{U}}_{j}(s)\bigg\rbrace-Q^{2}=\delta_{ij}J(\Delta;\varsigma)-Q^{2}
\end{align}
For a set of n correlated Gaussian noises $\widehat{\mathscr{U}}_{i}(t)$, one has $\bm{\mathrm{I\!E}}\lbrace\widehat{\mathscr{U}}_{i}(t)\rbrace=0$ and a regulated 2-point function:
\begin{equation}
\bm{\mathrm{I\!E}}\bigg\lbrace\widehat{\mathscr{U}}_{i}(t)\widehat{\mathscr{U}}_{j}(s)\bigg\rbrace
=\delta_{ij}J(\Delta;\varsigma)
\end{equation}
with $\Delta=|t-s|$ and converges to white noise in the limit as $\varsigma\rightarrow 0$, so that $\bm{\mathrm{I\!E}}\lbrace \widehat{\mathscr{W}}_{i}(t)\widehat{\mathscr{W}}_{j}(s)
\rbrace=\delta_{ij}\alpha\delta(t-s)$ for constant $\alpha$. The standard Brownian motion is $d\widehat{\mathscr{B}}(t)=\widehat{\mathscr{W}}(t)dt$. For a thermal bath of Gaussian white noise $\bm{\mathrm{I\!E}}\lbrace
\widehat{\mathscr{W}}_{i}(t)\widehat{\mathscr{W}}_{j}(s)
\rbrace=\delta_{ij}\alpha k_{B}T\delta(t-s)$ where $k_{B}$ is Boltzmann's constant and $T$ is the temperature. We assume non-white Gaussian random fields with regulated covariance throughout.
\end{defn}
\begin{lem}
The random field $\mathscr{U}(t)$ also has the properties [18]:
\begin{enumerate}
\item Stochastic continuity such that for any $(t,s)$ and $\epsilon>0$ one has the probabilities
\begin{equation}
\bm{\mathrm{I\!P}}(\|\mathscr{U}_{i}(t)-\mathscr{U}_{i}(s)\|>\epsilon)=0
\end{equation}
For any pair $(t_{1},t_{2})$ and $(\alpha,\beta, K)>0$
\begin{equation}
\bm{\mathrm{I\!E}}\bigg\lbrace\bigg \|\mathscr{U}_{i}(t_{2})-\mathscr{U}_{i}(t_{1})\bigg\|^{\alpha}\bigg\rbrace \le K|t_{2}-t_{1}|^{\beta+1}
\end{equation}
\item $\exists T$ such that for any $t>T$ and some $(\epsilon,\delta)>0$
\begin{equation}
\bm{\mathrm{I\!P}}\left(\bigg\|\frac{1}{t}\int_{t_{0}}^{t_{0}+t}\mathscr{U}_{i}(s)ds-\frac{1}{t}
\int_{t_{0}}^{t_{0}+t}\bm{\mathrm{I\!E}}\bigg\lbrace\mathscr{U}(s)\bigg\rbrace ds\bigg|>\delta\right)<\epsilon
\end{equation}
\end{enumerate}
\end{lem}
The effect of such noise or random perturbations on nonlinear classical systems has become a subject of considerable interest. Coupling noise to classical nonlinear ODEs and PDEs is a powerful and useful methodology with applications to turbulence, chaos and pattern formation [14,15,25,26,27,28]. Noise can destabilise a stable system, or a system considered stable to deterministic perturbations, but can also stabilize an unstable system [31-36]. It can also smooth out or dissipate blowups or singularities which exist for the purely deterministic problem. For multiplicative noise or random perturbations $\widehat{\mathscr{U}}(t)$, equation (2.19) can become a nonlinear stochastic differential equation
\begin{equation}
\alpha\sum_{i=1}^{n}\partial_{tt}\psi_{i}(t)+\beta \sum_{i=1}^{n}(\partial_{t}\psi_{i}(t))^{b}+
\sum_{i=1}^{n}(\gamma+\widehat{\mathscr{U}}_{i}(t))(\psi_{i}(t))^{a}=0
\end{equation}
Setting $\alpha=a=1$ and $\beta=0$ for example, gives
\begin{equation}
\sum_{i=1}^{n}\partial_{tt}\psi_{i}(t)+\sum_{i=1}^{n}[\gamma+\widehat{\mathscr{U}}_{i}(t)] \psi_{i}(t)\equiv\sum_{i=1}^{n}\partial_{tt}\psi_{i}(t)+ \sum_{i=1}^{n}[\widehat{\gamma}(t)] \psi_{i}(t)
\end{equation}
which describes a linear set of n noisy harmonic oscillators with random frequencies $\widehat{\gamma}(t)$. One can also consider static or equilibrium  solutions $\psi_{i}^{E}$ and subject these to stochastic perturbations so that $\psi_{i}(t)=\psi_{i}^{E}+\widehat{\mathscr{U}}_{i}(t)$ and with averaged norms $\bm{\mathrm{I\!E}}\lbrace\ \|\bm{\psi}(t)-\bm{\psi}_{i}^{E}\|\rbrace $.

Extending the classical Oseledec Theorem [6] to noisy or random systems presents technical challenges, but the following proposition for Lyapunov characteristic exponents can be considered for the randomly perturbed deterministic or equilibrium solutions [18].
\begin{prop}
Suppose $u_{i}^{E}$ are a set of stable solutions of equilibrium fixed points of the system of ODEs $\mathbf{D}_{n}u_{i}(t)=0$ or $ \mathbf{D}_{n}u_{i}^{E}=0$. Let $\widehat{u}_{i}(t)$ be the stochastic perturbations of the stable fixed points with $\widehat{u}_{i}(t)=u_{i}^{E}\exp(\int_{0}^{t}\widehat{\mathscr{U}}(\tau)d\tau)$ then $\|\delta\widehat{\bm{u}}(t)\|=\|\widehat{\bm{u}}(t)-\bm{u}_{i}\|$. If
\begin{equation}
\|\bm{u}^{E}\|^{-1}\bm{\mathrm{I\!E}}\bigg\lbrace\bigg\|\delta\widehat{\bm{u}}(t)\bigg\|\bigg\rbrace\equiv \|\bm{u}^{E}\|^{-1}\bm{\mathrm{I\!E}}\bigg\lbrace\bigg\|\widehat{\bm{u}}(t)-\bm{u}^{E}\bigg\|\bigg\rbrace
\sim \exp(\mathbf{Ly} t)
\end{equation}
then $\mathbf{Ly}$ is a LCE with stability if $\mathsf{Ly}<0$ and instability if $\mathbf{Ly} >0 $ and so by analogy with (2.28)
\begin{equation}
\mathbf{Ly}=\sim\lim_{t\uparrow\infty}\frac{1}{t}\|\bm{u}^{E}\|^{-1}\ln \bm{\mathrm{I\!E}}\bigg\lbrace\delta\bigg\|\widehat{\bm{u}}(t)\bigg\|\bigg\rbrace
=\lim_{t\uparrow\infty}\frac{1}{t}(\bigg\|\bm{u}^{E}\bigg\|)^{-1}\log \bm{\mathrm{I\!E}}\bigg\lbrace\bigg\|\widehat{\bm{u}}(t)-\bm{u}^{E}\bigg\|\bigg\rbrace
\end{equation}
\end{prop}
The general $\ell$-moment Lyapunov characteristic exponent (LCE) can be defined as follows
\begin{defn}
If $\widehat{u}_{i}(t)$ is a solution of a SDE or a randomly perturbed deterministic solution with initial data $u_{0}=u^{E}$, then the $\ell^{th}$ moment is
\begin{equation}
\bm{\mathbf{Ly}}(\ell)=\lim_{t\uparrow\infty}\frac{1}{t}\ln\bm{\mathrm{I\!E}}\bigg\lbrace \bigg\|\widehat{\bm{u}}(t)-\bm{u}^{E}\bigg\|^{\ell}\bigg\rbrace
\end{equation}
The stability criteria are then
\begin{enumerate}
\item If ${\mathbf{Ly}}(\ell)>0$ then $\bm{\mathrm{I\!E}}\lbrace \|\widehat{\bm{u}}(t)-\bm{u}^{E}\|^{\ell}\rbrace\rightarrow \infty$ as $t\rightarrow\infty$ and the randomly perturbed system cannot reach a new stable state. The system is then unstable.
\item If ${\mathbf{Ly}}(\ell)=-\infty$ then the randomly perturbed system is superstable.
\item If $\bm{\mathbf{Ly}}(p)<0$ then $\bm{\mathrm{I\!E}}\lbrace \|\widehat{\bm{u}}(t)-\bm{u}^{E}\|^{\ell}\rbrace\rightarrow 0$ then the randomly perturbed system is stable.
\item If ${\mathbf{Ly}}(\ell)=0$ then $\exists (B,C)>0 $ such for $t>0$ one has $B\le\bm{\mathrm{I\!E}}\lbrace \|\widehat{\bm{u}}(t)-\bm{u}^{E}\|^{\ell}\le C$.
\end{enumerate}
Using (2.62) then
\begin{align}
&{\mathbf{Ly}}(\ell)=\lim_{t\uparrow\infty}
\frac{1}{t}\ln\bm{\mathrm{I\!E}}\bigg\lbrace\bigg\|u(t)-u^{E}\bigg\|^{\ell}\bigg\rbrace\nonumber\\&
\le\lim_{t\uparrow\infty}\frac{1}{t}\ln\bigg(|u^{E}|^{\ell}n^{\ell/2}
\bm{\mathrm{I\!E}}\bigg\lbrace\exp\bigg(\zeta\ell\int_{0}^{t}\widehat{\mathscr{U}}(s)ds\bigg)\bigg\rbrace\bigg)\nonumber\\&
=\lim_{t\uparrow\infty}\frac{1}{t}\log (|u^{E}|^{\ell}n^{\ell/2})
+\lim_{t\uparrow\infty}\frac{1}{t}\log\bm{\mathrm{I\!E}}\bigg\lbrace\exp\bigg(\zeta\ell\int_{0}^{t}\widehat{\mathscr{U}}(s)ds\bigg)\bigg\rbrace\nonumber\\&=
\lim_{t\uparrow\infty}\frac{1}{t}\log\bm{\mathrm{I\!E}}\bigg\lbrace\exp\bigg(\zeta\ell\int_{0}^{t}\widehat{\mathscr{U}}(s)ds\bigg)\bigg\rbrace\nonumber\\&
\end{align}
\end{defn}
which again depends on the convergence properties of the stochastic integral.
\begin{defn}
The $\ell^{th}$-order mean $\bm{\mathrm{I\!E}}\lbrace\|\bm{u}(t)-\bm{u}^{E}\|^{\ell}\rbrace $ is also associated with the characteristic function $\bm{\Psi}(t)(z)$ where z is complex with $z\in\mathbb{C}$ and $\mathcal{L}=\log\|\widehat{\bm{u}}(t)-\bm{u}^{E}\|$. Then
\begin{equation}
\bm{\Psi}(z)=\bm{\mathrm{I\!E}}\bigg\lbrace\exp(iz {\mathbf{Ly}}(t))\bigg\rbrace=\bm{\mathrm{I\!E}}
\bigg\lbrace iz\log\bigg\|\hat{\bm{u}}(t)-\bm{u}^{E}\bigg\|)
\bigg\rbrace
\end{equation}
If $z=\ell\in\mathbb{Z}$ then $\bm{\Psi}(\ell)=\bm{\mathrm{I\!E}}\lbrace\exp(iz\ln\|\widehat{\bm{u}}(t)-\bm{u}^{E}\|)
\rbrace $. The Lyapunov functional in the complex plane then has the representation
\begin{equation}
{\mathbf{Ly}}(iz)=\lim_{t\uparrow\infty}\frac{1}{t}\log\bm{\Psi}(z;t)
=\lim_{t\uparrow\infty}\frac{1}{t}\bm{\mathrm{I\!E}}
\bigg\lbrace\exp(iz\ln\bigg\|\widehat{\bm{u}}(t)-\bm{u}^{E}\bigg\|)\bigg\rbrace
\end{equation}
\end{defn}
\begin{prop}
Consider an n-dimensional 1st-order linear system of the general form
\begin{equation}
\sum_{i=1}^{n}\frac{\partial_{t}u_{i}(t)}{u_{i}(t)}=
\zeta\sum_{i=1}^{n}\frac{f_{i}(t)}{\beta_{i}}\equiv n\zeta                                                                                                                                          \frac{f(t)}{\beta}
\end{equation}
if $f_{i}(t)=f(t)$ and $\beta_{i}=\beta$ for $i...n$,and where
$f_{i}:\mathbb{R}^{+}\rightarrow\mathbb{R}^{+}$ with $\zeta>0$ and $\beta_{i}$ are constants. If $f(t)=0$ then $ \sum_{i=1}^{n}\frac{\partial_{t}u_{i}(t)}{u_{i}(t)}=0$, which is also equivalent to $\sum_{i=1}^{n}\partial_{t}\psi_{i}(t)=0$. There is then a trivial set of stable or equilibrium solutions $u_{i}(t)=u_{i}^{E}$ for $i=1...n$ for the homogenous equations. For the inhomogeneous equations the solution is
\begin{equation}
u_{i}(t)=u_{i}(0)\exp\left(\frac{\zeta}{\beta_{i}}\int_{0}^{t}f_{i}(\tau)d\tau\right)
\end{equation}
and we can set $\beta_{i}=1$. Let $\widehat{\mathscr{W}}_{i}(t)$ be a set of n independent Gaussian white noises and let $\widehat{\mathscr{U}}_{i}(t)$ be a non-white Gaussian noise with correlation $\varsigma$ so that for any $t,s\in\mathbb{R}^{+}$
$\bm{\mathrm{I\!E}}\lbrace\widehat{\mathscr{U}}_{i}(t)\rbrace=0$ and with a regulated 2-point function $\bm{\mathrm{I\!E}}\lbrace\widehat{\mathscr{U}}_{i}(t)
\widehat{\mathscr{U}}_{i}(s)\rangle=\delta_{ij}J(\Delta;\varsigma)$, where $J(0;\varsigma)<\infty$ such that
\begin{equation}
\lim_{\varsigma\uparrow 0}\bm{\mathrm{I\!E}}\bigg\lbrace\widehat{\mathscr{U}}_{i}(t)
\widehat{\mathscr{U}}_{i}(s)\bigg\rbrace=\bm{\mathrm{I\!E}}\bigg\lbrace\widehat{\mathscr{W}}_{i}(t)
\widehat{\mathscr{W}}_{i}(s)\bigg\rbrace=\alpha \delta_{ij}\delta(t-s)
\end{equation}
Also, unlike for white noise, the derivative $\partial_{t}\widehat{\mathscr{U}}(t)$ exists.(Appendix A.) The following tentative SDEs are then possible:
\begin{equation}
\sum_{i=1}^{n}\frac{\partial_{t}\widehat{u}_{i}(t)}{\widehat{u}_{i}(t)}=
\zeta\sum_{i=1}^{n}f_{i}(t)+\zeta\sum_{i=1}^{n}\widehat{\mathscr{W}}_{i}(t)
\end{equation}
\begin{equation}
\sum_{i=1}^{n}\frac{\partial_{t}\widehat{u}_{i}(t)}{\widehat{u}_{i}(t)}=\zeta\sum_{i=1}^{n}f_{i}(t)+
\zeta\sum_{i=1}^{n}\widehat{\mathscr{U}}_{i}(t)
\end{equation}
where $\zeta>0$. Using the Stratanovich interpretation (Appendix A) the rules of ordinary calculas apply so that the solutions of (2.69) and (2.70) are
\begin{align}
&\widehat{u}_{i}(t)=u_{i}(0)\exp\left(\zeta\int_{0}^{t}{f}_{i}(\tau)d\tau+\zeta\int_{0}^{t}
\widehat{\mathscr{W}}_{i}(\tau)d\tau\right)\nonumber\\&\equiv u_{i}(t)\exp\left(\zeta\int_{0}^{t}\widehat{\mathscr{W}}_{i}(\tau)d\tau\right)
\end{align}
\begin{align}
&\widehat{u}_{i}(t)=u_{i}(0)\exp\left(\zeta\int_{0}^{t}{f}_{i}(\tau)d\tau+\zeta\int_{0}^{t}
\widehat{\mathscr{U}}_{i}(\tau)d\tau\right)\nonumber\\&\equiv u_{i}(t)\exp\left(\zeta\int_{0}^{t}\widehat{\mathscr{U}}_{i}(\tau)d\tau\right)
\end{align}
Then for any $t>0$, $\bm{u}(t)=\lbrace \widehat{u}_{1}(t),...,\widehat{u}_{n}(t)\rbrace$ is essentially a random matrix. Concentrating on (2.72), the expected value or stochastic average $\bm{\mathrm{I\!E}}\lbrace...\rbrace$ is
\begin{equation}
\bm{\mathrm{I\!E}}\bigg\lbrace\widehat{u}_{i}(t)\bigg\rbrace=u_{i}(0)\exp\left(\zeta\int_{0}^{t}f_{i}(\tau)d\tau\right)
\bm{\mathrm{I\!E}}\left\lbrace\exp\left(\zeta\int_{0}^{t}\widehat{\mathscr{U}}_{i}(\tau)
d\tau\right)\right\rbrace
\end{equation}
and if $f_{i}(t)=0$ for $i=1...n$ then
$\bm{\mathrm{I\!E}}\lbrace\widehat{u}_{i}(t)\rbrace
=u_{i}(0)\bm{\mathrm{I\!E}}\lbrace\exp(\zeta\int_{0}^{t}\widehat{\mathscr{U}}_{i}(\tau)d\tau)\rbrace$
\end{prop}
The stochastic integral (2.73) exists and can be shown to be well defined.(Appendix A.)

As an example of stability or instability induced by noise or random perturbations, consider again equation (2.66) which describes a (linear) n-dimensional system subject to white noise. This SDE can be solved exactly and one can then apply the Lyapunov exponent (2.62) to the solution to test stability.
\begin{lem}
(Noise-induced destabilisation and stabilisation). Let $u_{i}(0)$ be initial data for an n-dimensional linear system and let $)\alpha,\beta)\in\mathbb{R}$. Let $\mathscr{W}(t)$ be a white noise and $ d\mathscr{B}(t)=\mathscr{W}(t)dt$, the standard Brownian motion with $\mathscr{B}(0)=0$. Then:
\begin{enumerate}
\item The n-dimensional stable system
\begin{equation}
\sum_{i=1}^{n}\partial_{t}u_{i}(t)=-\alpha\sum_{i=1}^{n}u_{i}(t)
\end{equation}
with solution $u_{i}(t)=u_{i}(0)\exp(-\alpha t)$ which is randomly perturbed as
\begin{equation}
\sum_{i=1}^{n}\partial_{t}u_{i}(t)=-\alpha\sum_{i=1}^{n}u_{i}(t)+\zeta\sum_{i=1}^{n}
u_{i}(t)\mathscr{W}(t)
\end{equation}
and which is equivalent to the n-dimensional Brownian motion
\begin{equation}
\sum_{i=1}^{n}d\widehat{u}_{i}(t)=-\alpha\sum_{i=1}^{n}u_{i}(t)dt+\zeta\sum_{i=1}^{n}
u_{i}(t)d\mathscr{B}(t)
\end{equation}
is destabilised by the noise or random perturbation if $(-\alpha-\frac{1}{2}\zeta^{2}>0) $
but remains stable if$ (-\alpha-\frac{1}{2}\zeta^{2})<0$
\item The n-dimensional unstable system
\begin{equation}
\sum_{i=1}^{n}\partial_{t}u_{i}(t)=\alpha\sum_{i=1}^{n}u_{i}(t)
\end{equation}
with solution $\widehat{u}(t)=u(0)\exp(\alpha t)$ is stabilized by random perturbations of the form
\begin{equation}
\sum_{i=1}^{n}d\widehat{u}_{i}(t)=\alpha\sum_{i=1}^{n}u_{i}(t)dt+\zeta\sum_{i=1}^{n}
u_{i}(t)d\mathscr{B}(t)
\end{equation}
if $\alpha-\frac{1}{2}\zeta^{2}<0$
\end{enumerate}
\end{lem}
\begin{proof}
If $\mathcal{F}(u_{i}(t)$ is a $C^{2}$-differentiable functional of $u_{i}(t)$ and $D=d/du_{i}(t)$ then the Ito Lemma gives
\begin{align}
&d\mathcal{F}(\widehat{u}_{i}(t)=D\mathcal{F}(u_{i})t)d\widehat{u}_{i}(t)
+\frac{1}{2}|D^{2}\mathcal{F}(u_{i}(t))|\zeta^{2}|u_{i}(t)|^{2}dt\nonumber\\&
=D\mathcal{F}(u_{i})t)(-\alpha u(t) dt+\zeta u(t) \mathscr{B}(t))+\frac{1}{2}|D^{2}\mathcal{F}(u_{i}(t))|\zeta^{2}|u_{i}(t)|^{2}dt
\end{align}
so that for $\mathcal{F}(\widehat{u}_{i}(t)=\log u_{i}(t)$
\begin{align}
&d\mathcal{F}(\widehat{u}_{i}(t)=\frac{1}{u_{i}(t)}d\widehat{u}_{i}(t)
-\frac{1}{2}\frac{\zeta^{2}}{|u_{i}(t))|^{2}}|u_{i}(t)|^{2}dt\nonumber\\&
=(-\alpha-\frac{1}{2}\zeta^{2})dt+\zeta d\mathscr{B}(t)
\end{align}
The solution is
\begin{equation}
\log|u_{i}(t)|=\log|u_{i}(0)|+\int_{0}^{t}(-\alpha-\frac{1}{2}\zeta^{2})ds+\zeta\int_{0}^{t}d\mathscr{B}_{i}(s)
\end{equation}
so that
\begin{equation}
\widehat{u}_{i}(t)=u_{i}(0)\exp\big(-(\alpha-\tfrac{1}{2}\zeta^{2})t+\zeta\mathscr{B}(t)\big)
\end{equation}
The LCE is then
\begin{equation}
\mathbf{Ly}=\lim_{t\uparrow\infty}\frac{1}{t}\log(\bm{\mathrm{I\!E}}\bigg\lbrace \widehat{u}_{i}(t))\bigg\rbrace=-\alpha-\frac{1}{2}\zeta^{2}
\end{equation}
If $\mathbf{Ly}=-\alpha-\frac{1}{2}\zeta^{2}<0$ then the system it remains stable but if $\mathbf{Ly}=-\alpha-\frac{1}{2}\zeta^{2}<0$ then it is unstable to the random perturbations. Repeating with $\alpha$ replacing $-\alpha$, shows that noise will stabilise the unstable system (2.77) for $\alpha-\frac{1}{2}\zeta^{2}<0$.
\end{proof}
We consider now only the non-white random perturbations.                                            \begin{lem}
Let $f_{i}(t)=0$. If $\widehat{\mathscr{U}}_{i}(t)=\widehat{\mathscr{U}}(t)$ and $u_{i}^{E}=u^{E}$ for $i=1...n$ then the estimates for the expectations of the norms $\bm{\mathrm{I\!E}}\lbrace\|\widehat{\bm{u}}(t)-\bm{u}^{E}\|\rbrace$ and moments $\bm{\mathrm{I\!E}}\lbrace\|\widehat{\bm{u}}(t)-\bm{u}^{E}\|^{\ell}\rbrace$ for integers $\ell\in\mathbb{Z}$ are
\begin{align}
&\bm{\mathrm{I\!E}}\bigg\lbrace\bigg\|\widehat{\bm{u}}(t)-\bm{u}^{E}\bigg\|\bigg\rbrace\le
u^{E}n^{1/2}\bm{\mathrm{I\!E}}\left\lbrace\exp\left(\zeta\int_{0}^{t}\widehat{\mathscr{U}}(\tau)d\tau
\right)\right\rbrace\equiv u^{E}n^{1/2}\bm{\mathrm{I}}(t)\\&
\bm{\mathrm{I\!E}}\bigg\lbrace\|\widehat{\bm{u}}(t)-\bm{u}^{E}\|^{\ell}\bigg\rbrace\le
|u^{E}|^{\ell}n^{\ell/2}\bm{\mathrm{I\!E}}\left\lbrace\exp\left(\zeta\ell\int_{0}^{t}
\widehat{\mathscr{U}}(\tau)d\tau\right)\right\rbrace\equiv u^{E}n^{1/2}\mathbf{I}(t,\ell)
\end{align}
Then:
\begin{enumerate}
\item The system is asymptotically stable if $\lim_{t\uparrow\infty}\bm{\mathrm{I\!E}}\lbrace \|\widehat{\bm{u}}(t)-\bm{u}^{E}\|\rbrace=0$ and Lyapunov stable if $\exists$ $K>0$ such that $\lim_{t\uparrow\infty}\bm{\mathrm{I\!E}}\lbrace \|\widehat{\bm{u}}(t)-\bm{u}^{E}\|\rbrace<K$
\item Random perturbations then destabilize the system if
$\lim_{t\uparrow\infty}\bm{\mathrm{I\!E}}\lbrace \|\widehat{\bm{u}}(t)-\bm{u}^{E}\|\rbrace=\infty$
\end{enumerate}
\end{lem}
\begin{proof}
The estimate for $\bm{\mathrm{I\!E}}\lbrace\|\widehat{\bm{u}}(t)-\bm{u}^{E}\|\rbrace$ is
\begin{align}
\bm{\mathrm{I\!E}}\bigg\lbrace\bigg\|\widehat{\bm{u}}(t)&-\bm{u}^{E}\bigg\|\bigg\rbrace\le
\bm{\mathrm{I\!E}}\bigg\lbrace\bigg\|\widehat{\bm{u}}(t)\bigg\|\bigg\rbrace-\|\bm{u}^{E}\|\nonumber\\&
=\bm{\mathrm{I\!E}}\left\lbrace\left(\sum_{i=1}^{n}\left|u_{i}^{E}\exp\left(\zeta\int_{0}^{t}
\widehat{\mathscr{U}}_{i}(\tau)d\tau\right)\right|^{2}\right)^{1/2}\right\rbrace-\left(\sum_{i=1}^{n}|u_{i}^{E}|^{2}\right)^{1/2}\nonumber\\&
<\bm{\mathrm{I\!E}}\bigg\lbrace\sqrt{\bigg(\sum_{i=1}^{n}\bigg|u_{i}^{E}\exp\bigg(\zeta\int_{0}^{t}
\widehat{\mathscr{U}}_{i}(\tau)d\tau\bigg)\bigg|^{2}\bigg)}\bigg\rbrace\nonumber\\&
=u^{E}n^{\ell/2}\bm{\mathrm{I\!E}}
\left\lbrace\exp\left(\zeta\int_{0}^{t}\widehat{\mathscr{U}}(\tau)d\tau
\right)\right\rbrace\equiv u^{E}n^{1/2}\mathlarger{\mathbf{I}}(t)
\end{align}
if $\widehat{\mathscr{U}}_{i}(t)=\widehat{\mathscr{U}}(t)$ and $u_{i}^{E}=u^{E}$ for $i=1...n$. For any integer $\ell\in\mathbb{Z}$
\begin{align}
\bm{\mathrm{I\!E}}\bigg\lbrace\bigg\|\widehat{\bm{u}}(t)&-\bm{u}^{E}\bigg\|^{\ell}\bigg\rbrace\le
\bm{\mathrm{I\!E}}\bigg\lbrace\bigg\|\widehat{\bm{u}}(t)\bigg\|^{\ell}\bigg\rbrace-\|\bm{u}^{E}\|^{\ell}\nonumber\\
&=\bm{\mathrm{I\!E}}\left\lbrace\left(\sum_{i=1}^{n}\left|u_{i}^{E}\exp\left(\zeta\int_{0}^{t}
\widehat{\mathcal{U}}_{i}(\tau)d\tau\right)\right|^{2}\right)^{\ell/2}\right\rbrace-\left(\sum_{i=1}^{n}|u_{i}^{E}|^{2}\right)^{\ell/2}\nonumber\\
&<\bm{\mathrm{I\!E}}\left\lbrace\left(\sum_{i=1}^{n}\left|u_{i}^{E}\exp\left(\zeta\int_{0}^{t}
\widehat{\mathscr{U}}_{i}(\tau)d\tau\right)\right|^{2}\right)^{\ell/2}\right\rbrace\nonumber\\
&=\bm{\mathrm{I\!E}}\left\lbrace\left(\sum_{i=1}^{n}|u_{i}^{E}|^{2}\exp\left(2\zeta\int_{0}^{t}
\widehat{\mathscr{U}}_{i}(\tau)d\tau\right)\right)^{\ell/2}\right\rbrace\nonumber\\&
=\bm{\mathrm{I\!E}}\left\lbrace\left(n|u_{i}^{E}|^{2}\exp\left(2\zeta\int_{0}^{t}
\widehat{\mathscr{U}}_{i}(\tau)d\tau\right)\right)^{\ell/2}\right\rbrace\nonumber\\&
=|u^{E}|^{\ell}n^{\ell/2}\bm{\mathrm{I\!E}}\left\lbrace\exp\left(\zeta\ell\int_{0}^{t}\widehat{\mathscr{U}}(\tau)d\tau
\right)\right\rbrace\equiv u^{E}n^{1/2}\mathlarger{\mathbf{I}}(t,\ell)
\end{align}
The evaluation of stability then requires an estimate of the stochastic integrals $\mathlarger{\mathbf{I}}(t)$ or $\mathlarger{\mathbf{I}}(t,\ell)$.
\end{proof}
As before, a key question of interest is often to determine the extent to which such non-white random perturbations or noise can induce transitions between stable states of a system, especially a nonlinear system, or whether noise will actually destabilize the system: \emph{what was established as a stable point via a deterministic linear stability analysis may actually be unstable, or at best 'quasi-stable', in the presence of stochastic noise}. However, in general, most SNLDEs will be impossible to solve. Due to the presence of noise terms it is usually more appropriate to consider the maxima of a probability density distribution function $\mathscr{P}(u(t),t)$ rather than fixed points of the dynamics [21,22]. The $\mathscr{P}(u(t),t)$ would be stationary solutions of a Kolmogorov forward equation or Fokker-Planck (FP) equation. However, this is only possible for first-order equations. Again, such FP equations are often impossible to solve although in the infinite-time relaxation limit, the equilibrium solution can very often be found for nonlinear equations.

One could consider the following candidates for 2nd-order n-dimensional nonlinear SDES
\begin{equation}
\mathbf{D}_{n}\widehat{u}_{i}(t)=\sum_{i=1}^{n}\frac{\partial_{tt}u_{i}(t)}{u_{i}(t)}+(\beta-1)
\sum_{i=1}^{n}\frac{\partial_{t}u_{i}(t)\partial_{t}u_{i}(t)}{u_{i}(t)u_{j}(t)}=\zeta\sum_{i=1}^{n}
\mathscr{W}_{i}(t)
\end{equation}
\begin{equation}
\mathbf{D}_{n}\widehat{u}_{i}(t)=\sum_{i=1}^{n}\frac{\partial_{tt}u_{i}(t)}{u_{i}(t)}+(\beta-1)
\sum_{i=1}^{n}\frac{\partial_{t}u_{i}(t)\partial_{t}u_{i}(t)}{u_{i}(t)u_{j}(t)}=\zeta\sum_{i=1}^{n}
\mathscr{U}_{i}(t)
\end{equation}
or
\begin{equation}
\mathbf{H}_{n}\widehat{\psi}_{i}(t)=\sum_{i=1}^{n}\partial_{tt}{\psi}_{i}(t)+
\beta\sum_{i=1}^{n}\partial_{t}{\psi}_{i}(t)\partial_{t}{\psi}_{i}(t)=\zeta\sum_{i=1}^{n}
\mathscr{U}_{i}(t)
\end{equation}
However, these are impossible to solve. Instead one could substitute randomly perturbed solutions of the original deterministic equations and substitute back into the original deterministic equations, and then take the stochastic expectation or average. Because of the nonlinearity, additional terms can be induced within the stochastically averaged equations.
\begin{prop}
Let $\psi(t)$ and $u_{i}(t)$ be deterministic solutions of (2.2) and (2.3) and let $\lbrace\widehat{\mathscr{U}}_{i}(t)\rbrace$  be a Gaussian non-white regulated noise with $\bm{\mathrm{I\!E}}\lbrace\widehat{\mathscr{U}}_{i}(t)\rbrace=0$,derivative $\partial_{t}\widehat{\mathscr{U}}_{i}(t)$ and
$\bm{\mathrm{I\!E}}\lbrace\widehat{\mathscr{U}}_{i}(t)\widehat{\mathscr{U}}_{j}(t)
\rbrace=\delta_{ij}J(0;\varsigma)<\infty$. Let the randomly perturbed solution be
\begin{equation}
\widehat{\psi}_{i}(t)=\psi_{i}(t)+\zeta \int_{0}^{t}\widehat{\mathscr{U}}_{i}(\tau)d\tau\equiv \psi_{i}(t)+\zeta\int_{0}^{t}d\mathscr{U}_{i}(\tau)
\end{equation}
then since $u_{i}(t)=\exp(\psi_{i}(t))$
\begin{equation}
\widehat{u}_{i}(t)=u_{i}(t)\exp\left(\zeta\int_{0}^{t}\widehat{\mathscr{U}}(\tau)d\tau\right)\equiv u_{i}(t)\mathscr{B}_{i}(t)
\end{equation}
Equation (2.91) is equivalent to the stochastic differential equation
\begin{equation}
d\widehat{\psi}(t)=d\psi(t)+\zeta d\widehat{\mathscr{U}}(t)
\end{equation}
For white noise $\widehat{\mathscr{U}}_{i}(t)=\widehat{\mathscr{W}}_{i}$, this is a simple linear Brownian motion $ d\widehat{\psi}(t)=d\psi(t)+\zeta d\widehat{\mathscr{B}}(t)\equiv d\psi(t)+\zeta\widehat{\mathscr{W}}(t)dt$. Equations (2.91) and (2.92) are then solutions of the stochastically averaged systems of differential equations
\begin{equation}
\bm{\mathrm{I\!E}}\bigg\lbrace \mathbf{H}_{n}\widehat{\psi}_{i}(t)\bigg\rbrace=
\mathbf{H}_{n}\psi_{i}(t)+\zeta^{2}\beta\bm{\mathrm{I\!E}}\bigg\lbrace\bigg\|\mathscr{U}_{i}(t)\bigg\|_{L_{2}}^{2}\bigg\rbrace=
\zeta^{2}\beta\bm{\mathrm{I\!E}}\bigg\lbrace\bigg\|\mathscr{U}_{i}(t)\bigg\|_{L_{2}}^{2}\bigg\rbrace
\end{equation}
\begin{equation}
\widetilde{\bm{\mathrm{I\!E}}}\bigg\lbrace \mathbf{D}_{n}\widehat{u}_{i}(t)\bigg\rbrace=\mathbf{D}_{n}u_{i}(t)+\zeta^{2}\beta
\bm{\mathrm{I\!E}}\bigg\lbrace\bigg\|\mathscr{U}_{i}(t)\bigg\|_{L_{2}}^{2}\bigg\rbrace
=\zeta^{2}\beta\bm{\mathrm{I\!E}}\bigg\lbrace\bigg\|\mathscr{U}_{i}(t)\bigg\|_{L_{2}}^{2}\bigg\rbrace
\end{equation}
or
\begin{equation}
\bm{\mathrm{I\!E}}\bigg\lbrace \mathbf{H}_{n}\widehat{\psi}_{i}(t)\bigg\rbrace=
\mathbf{H}_{n}\psi_{i}(t)+C=C
\end{equation}
\begin{equation}
\widetilde{\bm{\mathrm{I\!E}}}\bigg\lbrace \mathbf{D}_{n}\widehat{u}_{i}(t)\bigg\rbrace=\mathbf{D}_{n}u_{i}(t)+C=C
\end{equation}
when $\mathbf{H}_{n}\psi_{i}(t)=\mathbf{D}_{n}u_{i}(t)=0$ and where $C=\zeta^{2}\beta n J(0;\varsigma)$ when $\mathscr{U}_{i}(t)=\mathscr{U}(t)$.
\end{prop}
\begin{proof}
If $\psi_{i}(t)$ is a solution of the deterministic equations $\mathbf{H}_{n}\psi_{i}(t)=0$ then the randomly perturbed equations are
\begin{equation}
\mathbf{H}_{n}\widehat{\psi}_{n}(t)=\sum_{i=1}^{n}
\partial_{tt}\widehat{\psi}_{i}(t)+\beta\sum_{i=1}^{n}
\partial_{t}\widehat{\psi}(t)\partial_{t}\widehat{\psi}(t)
\end{equation}
The derivatives of (2.91) are $\partial_{t}\widehat{\psi}_{i}(t)=\partial_{t}\psi_{i}(t)+\zeta\widehat{\mathscr{U}}_{i}(t)$ and $\partial_{tt}\widehat{\psi}_{i}(t)=\partial_{tt}\psi_{i}(t)
+\zeta\partial_{t}\widehat{\mathscr{U}}_{u}(t)$ so that (2.98) becomes
\begin{align}
&\mathbf{H}_{n}\widehat{\psi}_{i}(t)=\sum_{i=1}^{n}\partial_{tt}\psi_{i}(t)+
\beta\sum_{i=1}^{n}\partial_{t}\psi_{i}(t)\partial_{t}\psi_{i}(t)+\zeta\sum_{i=1}^{n}\partial_{t}
\widehat{\mathscr{U}}_{i}(t)\nonumber\\&+2\zeta\beta\sum_{i=1}^{n}\mathscr{U}_{i}(t)\partial_{t}
\psi_{i}(t)+\zeta^{2}\beta\sum_{i=1}^{n}\widehat{\mathscr{U}}_{i}(t)\widehat{\mathscr{U}}_{i}(t)
\end{align}
Taking the expectation and using $\bm{\mathrm{I\!E}}\lbrace\widehat{\mathscr{U}}_{i}(t)\rbrace=0$ and $\bm{\mathrm{I\!E}}\lbrace\partial_{t}\widehat{\mathscr{U}}_{i}(t)\rbrace=0$ gives
\begin{align}
\bm{\mathrm{I\!E}}&\bigg\lbrace\mathbf{H}_{n}\widehat{\psi}_{i}(t)\bigg\rbrace=\sum_{i=1}^{n}\partial_{tt}\psi_{i}(t)\nonumber\\&+
\beta\sum_{i=1}^{n}\partial_{t}\psi_{i}(t)\partial_{t}\psi_{i}(t)+\zeta^{2}\beta\sum_{i=1}^{n}
\bm{\mathrm{I\!E}}\bigg\lbrace\widehat{\mathscr{U}}_{i}(t)\widehat{\mathscr{U}}_{i}(t)
\bigg\rbrace\nonumber\\&\equiv\mathbf{H}_{n}\psi_{i}(t)+\zeta^{2}\beta\sum_{i=1}^{n}
\bm{\mathrm{I\!E}}\bigg\lbrace\widehat{\mathscr{U}}_{i}(t)\widehat{\mathscr{U}}_{i}(t)\bigg\rbrace\nonumber\\&
=\mathbf{H}_{n}\psi_{i}(t)+\zeta^{2}\beta\bm{\mathrm{I\!E}}\bigg\lbrace\bigg\|\mathscr{U}_{i}(t)\bigg\|_{L_{2}}^{2}\bigg\rbrace
\nonumber\\&=\mathbf{H}_{n}\psi_{i}(t)+\zeta^{2}\beta\sum_{i=1}^{n}\delta_{ii}J(0,\varsigma)\nonumber\\=
&\mathbf{H}_{n}\psi_{i}(t)+\zeta^{2}\beta n J(0;\varsigma)\equiv \mathbf{H}_{n}\psi_{i}(t)+C=C
\end{align}
The randomly perturbed ODE for $\widehat{u}_{i}(t)$ is
\begin{equation}
\mathbf{D}_{n}\widehat{u}_{i}(t)\equiv\sum_{i=1}^{n}\frac{\partial_{tt}\widehat{u}_{i}(t)}{\widehat{u}_{i}(t)}
+(\beta-1)\sum_{i=1}^{n}\frac{\partial_{t}\widehat{u}_{i}(t)\partial_{t}\widehat{u}_{i}(t)}{\widehat{u}_{i}(t)\widehat{u}_{j}(t)}
\end{equation}
Next, the derivatives of $\widehat{u}_{i}(t)$ are $ \partial_{t}\widehat{u}_{i}(t)=\zeta u_{i}(t)\widehat{\mathscr{U}}_{i}(t)\widehat{\mathscr{J}}_{i}(t)+\widehat{\mathscr{J}}_{i}(t)\partial_{t}u_{i}(t)$ and$\partial_{tt}\widehat{u}_{i}(t) =\zeta^{2}u_{i}
\widehat{\mathscr{U}}_{i}(t)\widehat{\mathscr{J}}_{i}(t)+
\zeta u_{i}(t)\partial_{t}\widehat{\mathscr{U}}_{i}(t)
\widehat{\mathscr{J}}_{i}(t)+(\partial_{t}u_{i}(t))\widehat{\mathscr{U}}_{i}(t)
\widehat{\mathscr{J}}_{i}(t)+\zeta\widehat{\mathscr{U}}_{i}(t)(\partial_{t}u_{i}(t))\widehat{\mathscr{J}}_{i}(t)
+(\partial_{tt}u_{i}(t))\widehat{\mathscr{J}}_{i}(t)$. Then (2.95) becomes
\begin{align}
\mathbf{D}_{n}\widehat{u}_{i}(t)&=\zeta^{2}\sum_{i=1}^{n}\frac{u_{i}(t)\widehat{\mathscr{G}}_{i}(t)\widehat{\mathscr{U}}_{i}(t)
\widehat{\mathscr{J}}_{i}(t)}{u_{i}(t)\widehat{\mathscr{J}}_{i}(t)}+\zeta\sum_{i=1}^{n}\frac{u_{i}(t)(\partial_{t}\widehat{\mathscr{U}}_{i}(t))
\widehat{\mathscr{J}}_{i}(t)}{u_{i}(t)\widehat{\mathscr{J}}_{i}(t)}\nonumber\\&+\zeta\sum_{i=1}^{n}\frac{(\partial_{t}u_{i}(t))
\widehat{\mathscr{G}}_{i}(t)\widehat{\mathscr{J}}_{i}(t)}{u_{i}(t)\mathscr{J}_{i}(t)}+\zeta^{2}\sum_{i=1}^{n}\frac{(\partial_{t}u_{i}(t))
\widehat{\mathscr{J}}_{i}(t)\widehat{\mathscr{G}}_{i}(t)}{u_{i}(t)\widehat{\mathscr{J}}_{i}(t)}\nonumber\\
&+\sum_{i=1}^{n}\frac{(\partial_{tt}u_{i}(t))\widehat{\mathscr{J}}_{i}(t)}{u_{i}(t)
\widehat{\mathscr{J}}_{i}(t)}+\zeta^{2}(\beta-1)\sum_{i=1}^{n}\frac{u_{i}(t)\widehat{\mathscr{U}}_{i}(t)
\widehat{\mathscr{J}}_{i}(t)u_{i}(t)\widehat{\mathscr{G}}_{i}(t)\widehat{\mathscr{J}}_{i}(t)}
{u_{i}(t)u_{i}(t)\widehat{\mathscr{J}}_{i}(t)\widehat{\mathscr{J}}_{i}(t)}\nonumber\\
&+2\zeta(\beta-1)\sum_{i=1}^{n}\frac{u_{i}(t)\widehat{\mathscr{G}}_{i}(t)(\partial_{t}u_{i}(t))
\widehat{\mathscr{J}}_{i}(t)\widehat{\mathscr{J}}_{i}(t)}{u_{i}(t)u_{i}(t)\widehat{\mathscr{J}}_{i}(t)
\widehat{\mathscr{J}}_{i}(t)}\nonumber\\&+(\beta-1)\sum_{i=1}^{n}\frac{(\partial_{t}u_{i}(t))
\widehat{\mathscr{J}}_{i}(t)(\partial_{t}u_{i}(t))
\widehat{\mathscr{J}}_{i}(t)}{u_{i}(t)u_{i}(t)\widehat{\mathscr{J}}_{i}(t)\widehat{\mathscr{J}}_{i}(t)}
\end{align}
Cancelling the $\widehat{\mathscr{J}}_{i}(t)$ terms
\begin{align}
\mathbf{D}_{n}\widehat{u}_{i}(t)&=\zeta^{2}\sum_{i=1}^{n}\widehat{\mathscr{U}}_{i}(t)\widehat{\mathscr{U}}_{i}(t)
+\zeta\sum_{i=1}^{n}(\partial_{t}\widehat{\mathscr{U}}_{i}(t))\nonumber\\&=\zeta\sum_{i=1}^{n}\frac{(\partial_{t}u_{i}(t))
\widehat{\mathscr{U}}_{i}(t)}{u_{i}(t)}+\zeta^{2}\sum_{i=1}^{n}\frac{(\partial_{t}u_{i}(t))
\widehat{\mathscr{U}}_{i}(t)}{u_{i}(t)}+\sum_{i=1}^{n}\frac{(\partial_{tt}u_{i}(t))}{u_{i}(t)}\nonumber\\
&+\zeta^{2}(\beta-1)\sum_{i=1}^{n}\widehat{\mathscr{U}}_{i}(t)\widehat{\mathscr{U}}_{i}(t)+2\zeta(\beta-1)\sum_{i=1}^{n}
\frac{\widehat{\mathscr{U}}_{i}(t)(\partial_{t}u_{i}(t))}{u_{i}(t)}\nonumber\\&+(\beta-1)\sum_{i=1}^{n}\frac{(\partial_{t}u_{i}(t))(\partial_{t}u_{i}(t))
}{u_{i}(t)u_{i}(t)}
\end{align}
and taking the stochastic average
\begin{align}
\bm{\mathrm{I\!E}}\bigg\lbrace\mathbf{D}_{n}\widehat{u}_{i}(t)\bigg\rbrace=\sum_{i=1}^{n}\frac{\partial_{tt}u_{i}(t)}{u_{i}(t)}
+(\beta-1)\sum_{i=1}^{n}\frac{\partial_{t}u_{i}(t)\partial_{t}u_{i}(t)}{u_{i}(t)u_{j}(t)}+
\zeta^{2}\beta\sum_{i=1}^{n}\bm{\mathrm{I\!E}}\bigg\lbrace\widehat{\mathscr{U}}_{i}(t)\widehat{\mathscr{U}}_{i}(t)
\bigg\rbrace
\end{align}
which is
\begin{align}
\bm{\mathrm{I\!E}}\bigg\lbrace \mathbf{D}_{n}\widehat{u}_{i}(t)\bigg\rbrace
&=\mathbf{D}_{n}u_{i}(t)+\zeta^{2}\beta\sum_{i=1}^{n}\bm{\mathrm{I\!E}}
\bigg\lbrace\widehat{\mathscr{U}}_{i}(t)\widehat{\mathscr{U}}_{i}(t)
\bigg\rbrace\nonumber\\&\equiv\mathbf{D}_{n}a_{i}(t)+\zeta^{2}\beta\bm{\mathrm{I\!E}}
\bigg\lbrace\bigg\|\mathscr{U}_{i}(t)\bigg\|_{L_{2}}^{2}\bigg\rbrace=\mathbf{D}_{n}u_{i}(t)+\zeta^{2}\beta\sum_{i=1}^{n}\delta_{ii}J(0;\sigma)\nonumber\\
&=\mathbf{D}_{n}u_{i}(t)+\beta\zeta^{2}n J(0;\sigma)\equiv \mathbf{D}_{n}u_{i}(t)+C=C
\end{align}
Hence (2.94) and (2.95) are established.
\end{proof}
The perturbed norms $\bm{\mathrm{I\!E}}\big\lbrace \|\widehat{\bm{u}}_{i}(t)-\bm{u}^{E}\|\big\rbrace$ are estimated as in (2.86) and (2.87). If $u_{i}(t)=u_{i}^{E}$ are equilibrium fixed points then the random perturbations will destabilize the system. It will either converge to new equilibria or attractors or else diverge to infinity.
\begin{rem}
The non-vanishing terms which arise in the stochastically averaged system of equations are due to the nonlinearity of the equations. For a linear system, the stochastically averaged equations will reduce back to the original deterministic equations. For example, in (2.46), the stochastically perturbed equations are
\begin{align}
&\sum_{i=1}^{n}\frac{\partial_{t}\widehat{u}_{i}(t)}{\widehat{u}_{i}(t)}=\sum_{i=1}^{n}
\frac{\partial_{t}u_{i}(t)\widehat{\mathscr{B}}_{i}(t)}{u_{i}^{E}\widehat{\mathscr{B}}_{i}(t)}+\sum_{i=1}^{n}
\frac{u_{i}(t)\widehat{\mathscr{U}}_{i}(t)\widehat{\mathscr{B}}_{i}(t)}{u_{i}(t)\widehat{\mathscr{B}}_{i}(t)}\nonumber\\
&\equiv\sum_{i=1}^{n}\frac{\partial_{t}u_{i}(t)}{u_{i}(t)}+\zeta\sum_{i=1}^{n}\widehat{\mathscr{U}}_{i}(t)=
\zeta\sum_{i=1}^{n}f_{i}(t)
\end{align}
Taking the stochastic average gives back the original ODE so that
\begin{align}
&\bm{\mathrm{I\!E}}\left\lbrace\sum_{i=1}^{n}\frac{\partial_{t}\widehat{u}_{i}(t)}{\widehat{u}_{i}(t)}
\right\rbrace=\sum_{i=1}^{n}\frac{\partial_{t}u_{i}(t)}{u_{i}(t)}+\zeta\sum_{i=1}^{n}
\bm{\mathrm{I\!E}}\bigg\lbrace\widehat{\mathscr{U}}_{i}(t)\bigg\rbrace\nonumber\\&=\sum_{i=1}^{n}\frac{\partial_{t}u_{i}(t)}{u_{i}(t)}=\zeta\sum_{i=1}^{n}f_{i}(t)
\end{align}
since $\bm{\mathrm{I\!E}}\lbrace\widehat{\mathscr{U}}_{i}(t)\rbrace=0$. If $\bm{\mathrm{I\!E}}\lbrace\widehat{\mathscr{U}}_{i}(t)\rbrace>0$ then an extra term can arise also for averaged linear equations.
\end{rem}
\subsection{Stability criteria}
Given the random perturbations and the random norm $\|\widehat{\bm{u}}(t)-\bm{u}^{E}\|$ the conditions for stability in general probabilistic terms can be tentatively defined as follows:
\begin{prop}
Given the random perturbations $\widehat{\psi}_{i}(t)=\psi_{i}^{E}+\int_{0}^{t}\widehat{\mathscr{U}}(\tau)d\tau $ which
gives $\widehat{u}_{i}(t)=u_{i}^{E}\exp(\int_{0}^{t}\widehat{\mathscr{U}}_{i}(\tau)d\tau$, with initial conditions $u_{i}(0)=u_{i}^{E}$ for a set of stable fixed points $u_{i}^{E}$ then:
\begin{enumerate}
\item The system is stable in probability for all $t>0$ if for any $L>0$ there is a $t>0$ such that
$ \bm{\mathrm{I\!P}}[\|\widehat{\bm{u}}(t)-\bm{u}^{E}\|\le |L|]=1$ or $\bm{\mathrm{I\!P}}[\|\widehat{\bm{u}}(t)-\bm{u}^{E}\|> |L|]=0$. So there is a ball $\mathbb{B}(L)$ of radius $L$ containing $\|\widehat{\bm{u}}(t)-\bm{u}^{E}\|$ for any $t>0$.
\item There is no noise-induced blowup or singularity for any finite $t>0$ if
$ \bm{\mathrm{I\!P}}[\|\widehat{\bm{u}}(t)-\bm{u}^{E}\|=\infty]=0$.
\item The system is unstable if for any $L>0$ and any $t>0$ if $\bm{\mathrm{I\!P}}[\|\widehat{\bm{u}}(t)-\bm{u}^{E}\|>|L|]=1 $. Instability can also be defined asymptotically as $ \lim_{t\uparrow\infty}\bm{\mathrm{I\!P}}[\|\widehat{\bm{u}}(t)-\bm{u}^{E}\|=\infty]=1. $
\item If $\mathbb{B}(L)\subset\mathbb{R}^{n}$ is an Euclidean ball of radius $L$ then if the norm is contained within $\mathbb{B}(L)$ at any $t>0$ then $\|\widehat{\bm{u}}(t)-\bm{u}^{E}\|\in \mathbb{B}(L)$. If this holds asymptotically for  $t>0$ then the randomly perturbed system is stable so that for some $L>0$.
\begin{equation}
\lim_{t\uparrow\infty}\bm{\mathrm{I\!P}}(\|\widehat{\bm{u}}(t)-\bm{u}^{E}\|\in\mathbb{B}(L))\equiv \lim_{t\uparrow\infty}\bm{\mathrm{I\!P}}(\|\widehat{\bm{u}}(t)-\bm{u}^{E}\|\le |L|]=1
\end{equation}or
\begin{equation}
\lim_{t\uparrow\infty}\bm{\mathrm{I\!P}}(\|\widehat{\bm{u}}(t)-\bm{u}^{E}\|\in\mathbb{B}(L))\equiv \lim_{t\uparrow\infty}\bm{\mathrm{I\!P}}(\|\widehat{\bm{u}}(t)-\bm{u}^{E}\|\ge |L|)=0
\end{equation}
or
\begin{equation}
\lim_{t\uparrow\infty}\bm{\mathrm{I\!P}}(\|\widehat{\bm{u}}(t)-\bm{u}^{E}\|\in\mathbb{B}(\infty))\equiv \lim_{t\uparrow\infty}\bm{\mathrm{I\!P}}(\|\widehat{\bm{u}}(t)-\bm{u}^{E}\|=\infty)=0
\end{equation}
\item The randomly perturbed system is unstable if for any ball $\mathbb{B}(L)\subset\mathbb{L}^{n}$
\begin{equation}
\lim_{t\uparrow\infty}\bm{\mathrm{I\!P}}(\|\widehat{\bm{u}}(t)-\bm{u}^{E}\|\in\mathbb{B}(L))\equiv \lim_{t\uparrow\infty}\bm{\mathrm{I\!P}}(\|\widehat{\bm{u}}_{i}(t)-\bm{u}^{E}\|\le |L|)=0
\end{equation}
\begin{equation}
\lim_{t\uparrow\infty}\bm{\mathrm{I\!P}}(\|\widehat{\bm{u}}(t)-\bm{u}^{E}\|\notin\mathbb{B}(L))\equiv \lim_{t\uparrow\infty}\bm{\mathrm{I\!P}}(\|\widehat{\bm{u}}_{i}(t)-\bm{u}^{E}\|> |L|)=1
\end{equation}
or
\begin{equation}
\lim_{t\uparrow\infty}\bm{\mathrm{I\!P}}(\|\widehat{\bm{u}}(t)-\bm{u}^{E}\|\in\mathbb{B}(\infty))\equiv \lim_{t\uparrow\infty}\bm{\mathrm{I\!P}}(\|\widehat{\bm{u}}(t)-\bm{u}^{E}\|=\infty)=1
\end{equation}
Equivalently for all $p\ge 1$
\begin{equation}
\lim_{t\uparrow\infty}\bm{\mathrm{I\!P}}(\|\widehat{\bm{u}}(t)-\bm{u}^{E}\|^{p}\in\mathbb{B}(L))\equiv \lim_{t\uparrow\infty}\bm{\mathrm{I\!P}}(\|\widehat{\bm{u}}_{i}(t)-\bm{u}^{E}\|^{p}\le |L|)=0
\end{equation}
\begin{equation}
\lim_{t\uparrow\infty}\bm{\mathrm{I\!P}}(\|\widehat{\bm{u}}(t)-\bm{u}^{E}\|^{p}\notin\mathbb{B}(L))\equiv \lim_{t\uparrow\infty}\bm{\mathrm{I\!P}}(\|\widehat{\bm{u}}_{i}(t)-\bm{u}^{E}\|^{p}> |L|)=1
\end{equation}
\item The system is p-stable if for all $p\ge 1$ and some $|L|>0$
\begin{equation}
\bm{\mathrm{I\!E}}\bigg\lbrace\bigg\|\widehat{u}(t)-u_{i}^{E}\bigg\|^{p}\bigg\rbrace\le |L|
\end{equation}
or $\bm{\mathrm{I\!E}}\|\widehat{u}(t)-u_{i}^{E}\|^{p}\in\mathbb{B}(L)$. It is asymptotically p -stable if
\begin{equation}
\lim_{t\uparrow\infty}\bm{\mathrm{I\!E}}\bigg\lbrace\bigg\|\widehat{u}(t)-u_{i}^{E}\bigg\|^{p}\bigg\rbrace=0
\end{equation}
\item The system is exponentially p-stable if $\exists$ constants $(\mathcal{A},Q)>0$
such that
\begin{equation}
\bm{\mathrm{I\!E}}\bigg\lbrace\bigg\|\widehat{u}(t)-u_{i}^{E}\bigg\|^{p}\bigg\rbrace\le \mathcal{A}\|u^{E}\|\exp(-Q|t-t_{0}|)
\end{equation}
and exponentially p-unstable if
\begin{equation}
\bm{\mathrm{I\!E}}\bigg\lbrace\bigg\|\widehat{u}(t)-u_{i}^{E}\bigg\|^{p}\bigg\rbrace\le \mathcal{A}\|u^{E}\|\exp(+Q|t-t_{0}|)
\end{equation}
Then $ \lim_{t\uparrow\infty}\bm{\mathrm{I\!E}}\big\lbrace\|\widehat{u}(t)-u_{i}^{E}\|^{p}\big\rbrace\rightarrow 0$ or  $ \lim_{t\uparrow\infty}\bm{\mathrm{I\!E}}\big\lbrace\|\widehat{u}(t)-u_{i}^{E}\|^{p}\big\rbrace\rightarrow \infty $
\end{enumerate}
\end{prop}
\begin{defn}
Given $\gamma\in(0,1]$, and $L>0$, the '$\gamma$-basins of attraction' ($\gamma$-BOA) are the sets
\begin{align}
&\lbrace \bm{u}^{E}\in\mathbb{R}^{n}:\bm{\mathrm{I\!P}}(\|\bm{u}(t)-u^{E}\|=0)\ge\gamma\rbrace
\\& \lbrace \bm{u}^{E}\in\mathbb{R}^{n}:\bm{\mathrm{I\!P}}(\|\bm{u}(t)-u^{E}\|\le |L|)\ge\gamma\rbrace
\end{align}
\end{defn}
Given a set of random variables, it is possible to establish expressions, bounds and estimates for these probabilistic stability criteria.
\begin{defn}
If set of random variables $(\widehat{u}_{i}(t))_{i=1}^{n}
=(\widehat{u}_{1},...,\widehat{u}_{n}(t))$, representing random perturbations of an initially static or equilibrium set $u_{i}^{E}$ are Gaussian, then for any $L>0$ and for some $C>0$.
\begin{align}
\bm{\mathrm{I\!P}}\big(\widehat{S}(t)-\bm{\mathrm{I\!E}}\lbrace&\widehat{ S}(t)\rbrace\ge L\big)\equiv\bm{\mathrm{I\!P}}\bigg(\frac{1}{n}
\sum_{i=1}^{n}\widehat{u}_{i}(t)-\frac{1}{n}\sum_{i=1}^{n}\bm{\mathrm{I\!E}}\bigg\lbrace\widehat{u}_{i}(t)\bigg\rbrace\ge L\bigg)\nonumber\\&\equiv\bm{\mathrm{I\!P}}\bigg(\frac{1}{n}\bigg\|\sqrt{\widehat{u}_{i}(t)}\bigg\|^{2}
-\frac{1}{n}\big\|\sqrt{\bm{\mathrm{I\!E}}\big\lbrace\widehat{u}_{i}(t)\big\rbrace}\big\|^{2}\ge L\bigg)\nonumber\\&\le \frac{1}{\sqrt{2\pi}}\frac{1}{C}\exp\bigg(-\frac{2n^{2}|L|^{2}}{C}\bigg)
\end{align}
The set is sub-Gaussian if
\begin{align}
\bm{\mathrm{I\!P}}\big(\widehat{S}(t)-\bm{\mathrm{I\!E}}\lbrace&\widehat{S}(t)\rbrace\ge L\big)\equiv \bm{\mathrm{I\!P}}\bigg(\frac{1}{n}
\sum_{i=1}^{n}\widehat{u}_{i}(t)-\frac{1}{n}\sum_{i=1}^{n}\bm{\mathrm{I\!E}}\big\lbrace\widehat{u}_{i}(t)\big\rbrace\ge L\bigg)\nonumber\\&\equiv\bm{\mathrm{I\!P}}\bigg(\frac{1}{n}\|\sqrt{\widehat{u}_{i}(t)}\|^{2}
-\frac{1}{n}\big\|\sqrt{\bm{\mathrm{I\!E}}\big\lbrace\widehat{u}_{i}(t)\big\rbrace}\big\|^{2}\ge L\bigg)\nonumber\\&\le \frac{1}{\sqrt{2\pi}}\exp\bigg(-\frac{2n^{2}|L|^{2}}{C}\bigg)
\end{align}
\end{defn}
\begin{lem}
If the set of random variables $\widehat{u}_{i}(t)$ is Gaussian or sub-Gaussian then the stability of condition of (2.113) also holds so that
\begin{align}
\bm{\mathrm{I\!P}}\big(\widehat{S}(t)-\bm{\mathrm{I\!E}}\lbrace\widehat{ S}(t)\rbrace&=\infty\big)\equiv\bm{\mathrm{I\!P}}\bigg(\frac{1}{n}
\sum_{i=1}^{n}\widehat{u}_{i}(t)-\frac{1}{n}\sum_{i=1}^{n}\bm{\mathrm{I\!E}}\big\lbrace\widehat{u}_{i}(t)\big\rbrace= \infty\bigg)\nonumber\\&\equiv\bm{\mathrm{I\!P}}\bigg(\frac{1}{n}\bigg\|
\sqrt{\widehat{u}_{i}(t)}\bigg\|^{2}
-\frac{1}{n}\big\|\sqrt{\bm{\mathrm{I\!E}}\big\lbrace\widehat{u}_{i}(t)\big\rbrace}\big\|^{2}=\infty\bigg)=0
\end{align}
so that there is zero probability that the perturbed system will blow up in any finite time or be asymptotically unstable as $t\rightarrow\infty$.
\end{lem}
In particular, if the set is sub-Gaussian then it is bounded and the Hoeffding inequality and the Chernoff bound inequality apply. This suggests that if a randomly perturbed set
$\widehat{u}_{i}(t)$ is sub-Gaussian then it is bounded and therefore the system is stable to the random perturbations and vice versa.
\begin{prop}
Let $u_{i}^{E}$ be a set of static equilibrium solutions of a nonlinear ODE of the form $\mathbf{D}_{n}u_{i}^{E}=0$. Let the randomly perturbed set of solutions be $\widehat{u}_{i}(t)$. Let $\widehat{u}_{i}^{E*}$ be 'attractors' or new stable equilibrium fixed points such that the perturbed system converges as $\widehat{u}_{i}(t)\rightarrow u_{i}^{E*}$ for some finite $t\gg 0$ or as $t\rightarrow\infty$. Then for all finite $t>0$ the set is bounded in that
\begin{equation}
u_{i}^{E}\le \widehat{u}_{i}(t)\le u_{i}^{E*}\nonumber
\end{equation}
\begin{enumerate}
\item The Hoeffding inequality applies and is then
\begin{align}
\bm{\mathrm{I\!P}}\big(\widehat{S}(t)-&\bm{\mathrm{I\!E}}\lbrace\widehat{ S}(t)\rbrace\ge L\big)\equiv\bm{\mathrm{I\!P}}\bigg(\frac{1}{n}
\sum_{i=1}^{n}\widehat{u}_{i}(t)-\frac{1}{n}\sum_{i=1}^{n}\bm{\mathrm{I\!E}}\big\lbrace\widehat{u}_{i}(t)\big\rbrace\ge L\bigg)\nonumber\\&\equiv\bm{\mathrm{I\!P}}\bigg(\frac{1}{n}\bigg\|\sqrt{\widehat{u}_{i}(t)}\bigg\|^{2}
-\frac{1}{n}\big\|\sqrt{\bm{\mathrm{I\!E}}\big\lbrace\widehat{u}_{i}(t)z\big\rbrace}\big\|^{2}\ge L\bigg)\nonumber\\&\le\exp\bigg(-\frac{2n^{2}|L|^{2}}{\sum_{i=1}^{n}|u_{i}^{E*}-u_{i}^{E}|^{2}}
\bigg)\nonumber\\&\equiv\exp\bigg(-\frac{2n^{2}|L|^{2}}{\big\|u_{i}^{E*}-u_{i}^{E}\big\|^{2}}\bigg)
\end{align}
\item The (left-tail) Chernoff bound is the estimate
\begin{equation}
\bm{\mathrm{I\!P}}\big(\widehat{S}(t)-\bm{\mathrm{I\!E}}\lbrace\widehat{S}(t)\rbrace\le L\big)\le\exp(\beta|L|)\bm{\mathrm{I\!E}}\bigg\lbrace\exp(-\beta (\widehat{S}(t)-\bm{\mathrm{I\!E}}\lbrace\widehat{S}(t)))\bigg\rbrace
\end{equation}
\end{enumerate}
Then Lemma 2.21 holds and the randomly perturbed system is stable in probability, otherwise it is unstable
\begin{align}
\bm{\mathrm{I\!P}}\big(\widehat{S}(t)&-\bm{\mathrm{I\!E}}\lbrace\widehat{ S}(t)\rbrace=\infty\big)\equiv \bm{\mathrm{I\!P}}\bigg(\frac{1}{n}
\sum_{i=1}^{n}\widehat{u}_{i}(t)-\frac{1}{n}\sum_{i=1}^{n}\bm{\mathrm{I\!E}}\big\lbrace\widehat{u}_{i}(t)\big\rbrace= \infty\bigg)\nonumber\\&\equiv\bm{\mathrm{I\!P}}\bigg(\frac{1}{n}\bigg\|\sqrt{\widehat{u}_{i}(t)}\bigg\|^{2}
-\frac{1}{n}\big\|\sqrt{\bm{\mathrm{I\!E}}\big\lbrace\widehat{u}_{i}(t)\big\rbrace}\big\|^{2}=\infty\bigg)=1
\end{align}
The Chernoff bound can also be expressed as
\begin{equation}
\bm{\mathrm{I\!P}}\big(\big\|\widehat{u}_{i}(t)-u_{i}^{E}\big\|\le |L|
\big)\le\exp(\beta|L|)\bm{\mathrm{I\!E}}\bigg\lbrace\exp(-\beta\bigg(\bigg\|\widehat{u}_{i}(t)-u_{i}^{E}\bigg\|\bigg)
\bigg\rbrace
\end{equation}
or asymptotically as
\begin{equation}
\lim_{t\uparrow\infty}\bm{\mathrm{I\!P}}\big(\big\|\widehat{u}_{i}(t)-u_{i}^{E}\big\|\le |L|
\big)\le\lim_{t\uparrow\infty}\exp\big(\beta|L|)\bm{\mathrm{I\!E}}\bigg\lbrace\exp\bigg(-\beta\bigg(\bigg\|\widehat{u}_{i}(t)-u_{i}^{E}\bigg\|\bigg)
\bigg\rbrace
\end{equation}
then
\begin{equation}
\lim_{t\uparrow\infty}\bm{\mathrm{I\!P}}\big(\big\|\widehat{u}_{i}(t)-u_{i}^{E}\big\|\le |L|\big)=0
\end{equation}
if $
\big\lbrace\exp(-\beta(\big\|\widehat{u}_{i}(t)-u_{i}^{E}\big\|)
\big\rbrace\rightarrow 0$ as $t\rightarrow\infty$ and the randomly perturbed variables are not bounded.
\end{prop}
These exponential inequalities are valid for linear combinations of bounded independent random variables, and in particular for the average. But one is often more interested in controlling the maximum or supremum of the set in terms of the maximal estimates.
\begin{lem}
Let $u_{i}^{E}$ be a set of n equilibrium solutions of a nonlinear ODE $D_{n}u_{i}^{E}=0$ and let $\widehat{u}_{i}(t)$ be the set of n randomly perturbed solutions. Let $\sup_{1\le i\le n}\widehat{u}_{i}(t)$ be the supremum or maximum of the set. If the set if bounded it is sub-Gaussian and vice-versa so that for some $(C,L)>0$
\begin{equation}
\bm{\mathrm{I\!P}}\big(\sup_{1\le i\le n} \widehat{u}_{i}(t)\ge |L|\big)\le\exp\bigg(-\frac{L^{2}}{2C^{2}}\bigg)
\end{equation}
Then $\exists(C,B,D)>0$ such that the maximal inequalities hold and the system is stable so that for all $t\in\mathbb{R}^{+}\cup\infty$
\begin{align}
&\bm{\mathrm{I\!E}}\bigg\lbrace\sup_{1\le i\le n}\widehat{u}_{i}(t)
\bigg\rbrace\le C\sqrt{2\log(n)}\le B<\infty
\\&\lim_{t\uparrow\infty}\bm{\mathrm{I\!E}}\bigg\lbrace\sup_{1\le i\le n}\widehat{u}_{i}(t)\bigg\rbrace\le C\sqrt{2\log(n)}\le B < \infty
\end{align}
\begin{align}
&\bm{\mathrm{I\!P}}\big(\sup_{1\le i\le n}\widehat{u}_{i}(t)\ge|L|\big)\le n\exp\bigg(-\frac{L^{2}}{2C^{2}}\bigg)\le D < \infty
\\&\lim_{t\uparrow\infty}\bm{\mathrm{I\!P}}\big(\sup_{1\le i\le n}\widehat{u}_{i}(t)\ge|L|\big)\le n\exp\bigg(-\frac{L^{2}}{2C^{2}}\bigg)\le D < \infty
\end{align}
and
\begin{align}
&\bm{\mathrm{I\!P}}\big(\sup_{1\le i\le n}\widehat{u}_{i}(t)=\infty\big)= 0
\\& \lim_{t\uparrow\infty}\bm{\mathrm{I\!P}}\big(\sup_{1\le i\le n}\widehat{u}_{i}(t)=\infty\big)= 0
\end{align}
\end{lem}
\raggedbottom
\begin{proof}
For any $\xi>0$
\begin{align}
\bm{\mathrm{I\!E}}\bigg\lbrace \sup_{1\le i\le n}&u_{i}(t)\bigg\rbrace
=\frac{1}{\xi}\bm{\mathrm{I\!E}}\bigg\lbrace\log\bigg(\exp(\xi\sup_{1\le i\le n}\widehat{u}_{i}(t)\bigg)\bigg\rbrace\nonumber\\&=\frac{1}{\xi}\underbrace{
\log \bm{\mathrm{I\!E}}\bigg\lbrace\exp(\xi\sup_{1\le i\le n}\widehat{u}_{i}(t)\bigg)\bigg\rbrace}_{by~Jensen~ineq.}\nonumber\\&=\frac{1}{\xi}
\log\bm{\mathrm{I\!E}}\bigg\lbrace\sup_{1\le i\le n}\exp(\xi\widehat{u}_{i}(t)\bigg)\bigg\rbrace\nonumber\\&
=\frac{1}{\xi}\log\sum_{i=1}^{n}\bm{\mathrm{I\!E}}\bigg\lbrace\exp(\xi\widehat{u}_{i}(t)\bigg)\bigg\rbrace\le
\frac{1}{\xi}\log\sum_{i=1}^{n}\exp(\frac{1}{2}C^{2}\xi^{2})\nonumber\\&
=\frac{1}{\xi}\log\big(n\exp(\frac{1}{2}C^{2}\xi^{2}\big)\nonumber\\&
=\frac{1}{\xi}\log(n)+\frac{1}{2}C^{2}\xi)
\end{align}
choosing $\xi=\sqrt{2\log(n)/C^{2}}$ then gives the maximal inequalities (2.132) or (2.133). Next
\begin{align}
\bm{\mathrm{I\!P}}\bigg(\sup_{1\le i\le n}&\widehat{u}_{i}(t)\ge |L|\bigg)=\bm{\mathrm{I\!P}}\bigg(\bigcup_{i=1}^{n}u_{i}(t)\ge |L|\bigg)\nonumber\\&
\le \sum_{i=1}^{n}\bm{\mathrm{I\!P}}(\widehat{u}_{i}(t)\ge |L|)\nonumber\\&\le
n\exp\bigg(-\frac{L^{2}}{2C^{2}}\bigg)<D<\infty
\end{align}
so that (2.134) and (2.135) follow
\end{proof}
We can apply these tools to the random perturbations.
\begin{lem}
Using the basic Markov inequality $\bm{\mathrm{I\!P}}[\|\mathbf{X}\|\ge |L|]\le
|L|^{-1}\bm{\mathrm{I\!E}}\lbrace\|\mathbf{X}\|\rbrace$ for a random variable $\widehat{X}$ and any $|L|>0$, the following estimate can be made for the probability that the stochastic norm $\|\widehat{\bm{u}}(t)-\bm{u}^{E}\|$ is outside a ball $\mathbb{B}(L)$ of any radius $|L|$ at any time $t>0$. Using the estimate (2.86)
\begin{align}
&\bm{\mathrm{I\!P}}[\|\widehat{\bm{u}}(t)-\bm{u}^{E}\|\ge \|L|]\le |L|^{-1}\bm{\mathrm{I\!E}}\bigg\lbrace\bigg\|\widehat{\bm{u}}(t)-\bm{u}^{E}\bigg\|\bigg\rbrace
\nonumber\\&=|L|^{-1}n^{1/2}|u^{E}|\bm{\mathrm{I\!E}}\left\lbrace\exp\left(\zeta\int_{0}^{t}
\widehat{\mathscr{U}}(\tau)d\tau\right)\right\rbrace\nonumber\\&=|L|^{-1}n^{1/2}|u^{E}|
\bm{\mathrm{I}}(t)\le 1
\end{align}
with $\widehat{\mathscr{U}}_{i}(t)=\widehat{\mathscr{U}}(t)$ for $i=1...n$. If $\lim_{t\uparrow\infty}\bm{\mathrm{I\!P}}[\|\widehat{\bm{u}}_{i}(t)-\bm{u}^{E}\|\ge \|L|]=0 $ for all finite $R$ then the system is stable. There is no noise-induced blowup for all finite $t>0$ if
\begin{equation}
\bm{\mathrm{I\!P}}\big(\big\|\widehat{\bm{u}}(t)-\bm{u}^{E}\|
=\infty\big)=\lim_{L\uparrow\infty}u^{E}n^{1/2}\bm{u}^{E}|L|^{-1}
\bm{\mathrm{I\!E}}\left\lbrace\exp\left(\zeta\int_{0}^{t}
\widehat{\mathscr{U}}(\tau)d\tau\right)\right\rbrace=0
\end{equation}
\end{lem}
If $\exists L>0$ such that $0\le \|\widehat{\bm{u}}(t)-\bm{u}^{E}\|\le L$ for all $t>0$ then the system is stable to random perturbations. Given the stochastic norm $\|\widehat{\bm{u}}(t)-\widehat{\bm{u}}^{E}\|$, probabilistic stability criteria can also be established using a stronger Chernoff bound estimate.
\begin{thm}
Let $u_{i}^{E}$ be equilibrium solutions such that $\mathbf{D}_{n}u_{i}^{E}=0$, then for a non-white noise perturbation $\widehat{\mathscr{U}}_{i}(t)$, the randomly perturbed $\mathcal{L}_{2}$ norms are given by $\|\widehat{\bm{u}}(t)-\bm{u}^{E}\|=u^{E}n^{1/2}
\exp(\int_{0}^{t}\widehat{\mathscr{U}}(\tau)d\tau)$. Now let $\mathbb{B}(L)\subset\mathbb{R}^{n}$ be an Euclidean ball of radius $|L|$ and $\ell\in\mathbb{Z}$. Then we can make the estimate
\begin{align}
&\lim_{t\uparrow\infty}\bm{\mathrm{I\!P}}
\big(\|\widehat{\bm{u}}(t)-\widehat{\bm{u}}^{E}\|\in\mathbb{B}(L)\big)
\equiv\lim_{t\uparrow\infty}\bm{\mathrm{I\!P}}
\big\|\widehat{\bm{u}}(t)-\widehat{\bm{u}}^{E}\|\le |L|\big)\nonumber\\ &\le\lim_{t\uparrow\infty}\exp(+\beta|L|)\exp\left(-\beta|u^{E}|n^{1/2}
\left|\bm{\mathrm{I\!E}}\left\lbrace\exp\left(\zeta\ell\int_{0}^{t}
\widehat{\mathscr{U}}(\tau)d\tau\right)\right\rbrace\right|^{1/\ell}\right)
\end{align}
The asymptotic probabilistic stability criteria can then be stated as follows:
\begin{enumerate}
\item The system is asymptotically unstable to the random perturbations for $t\rightarrow\infty$ if the stochastic norm can never be contained within a ball $\mathbb{B}(L)$ of any finite radius $|L|$. The probability is zero such that
\begin{align}
&\lim_{t\uparrow\infty}\bm{\mathrm{I\!P}}
[\|\widehat{\bm{u}}(t)-\widehat{\bm{u}}^{E}\|\in\mathbb{B}(L)]\equiv\lim_{t\uparrow\infty}
\bm{\mathrm{I\!P}}[\|\widehat{u}(t)-\widehat{u}^{E}\|\le |L|]\nonumber\\
&\le\lim_{t\uparrow\infty}\exp(+\beta|L|)\exp\left(-\beta|u^{E}|n^{1/2}
\left|\bm{\mathrm{I\!E}}\left\lbrace\exp\left(\zeta\ell\int_{0}^{t}
\widehat{\mathscr{U}}(\tau)d\tau\right)\right\rbrace\right|^{1/\ell}\right)\nonumber\\&
\equiv\lim_{t\uparrow\infty}\exp(+\beta|L|)\exp(-\beta|u^{E}|n^{1/2}\mathbf{I}(t)|^{1/\ell}=0
\end{align}
which is the case if the stochastic integral diverges
\begin{equation}
\lim_{t\uparrow\infty}\bm{\mathrm{I}}(t)=\lim_{t\uparrow\infty}\bm{\mathrm{I\!E}}\left\lbrace\exp\left(\zeta\ell\int_{0}^{t}
\widehat{\mathscr{U}}(\tau)d\tau\right)\right\rbrace=\infty
\end{equation}
\item The system is asymptotically stable to the random perturbations if the stochastic norm is always contained within a ball of any finite radius $|L|$, with finite or unit probability.
\begin{align}
&\lim_{t\uparrow\infty}\bm{\mathrm{I\!P}}
[\|\widehat{\bm{u}}(t)-\widehat{\bm{u}}^{E}\|\in\mathbb{B}(L)]
\equiv\lim_{t\uparrow\infty}\bm{\mathrm{I\!P}}
[\|\widehat{\bm{u}}(t)-\widehat{\bm{u}}^{E}\|\le |L|]\nonumber\\
&\le\lim_{t\uparrow\infty}\exp(+\beta|L|)\exp\left(-\beta|u^{E}|n^{1/2}
\left|\bm{\mathrm{I\!E}}\left\lbrace\exp\left(\zeta\ell\int_{0}^{t}
\widehat{\mathscr{U}}(\tau)d\tau\right)\right\rbrace\right|^{1/\ell}\right)\nonumber\\
&\equiv\lim_{t\uparrow\infty}\exp(+\beta|L|)\exp(-\beta|u^{E}|n^{1/2}\mathbf{I}(t)|^{1/\ell})\le 1
\end{align}
which is the case if the stochastic integral converges such that
\begin{equation}
\lim_{t\uparrow\infty}\mathbf{I}(t,\ell)=
\lim_{t\uparrow\infty}\bm{\mathrm{I\!E}}\left\lbrace\exp\left(\zeta\ell\int_{0}^{t}
\widehat{\mathscr{U}}(\tau)d\tau\right)\right\rbrace=Q<\infty
\end{equation}
\end{enumerate}
\end{thm}
\begin{proof}
For some $\beta>0$, the generic left-tail Chernoff bound for a stochastic variable $\widehat{X}$ is
\begin{equation}
\bm{\mathrm{I\!P}}[\widehat{X}\le |L|]\le \exp(+\beta|L|)\bm{\mathrm{I\!E}}\bigg\lbrace\bigg(\exp(-\beta\big|\widehat{X}\big|\bigg)\bigg\rbrace
\end{equation}
For a set of independent random variables
$\bm{X}=\lbrace\widehat{X}_{1}...\widehat{X}_{n}\rbrace$ with $\mathcal{L}_{2}$ norm $\|\widehat{\mathbf{X}}\|$
\begin{align}
\bm{\mathrm{I\!P}}\bigg(\sum_{i=1}^{n}|\widehat{X}_{i}|&\le L\bigg)\le \inf_{\beta>0}\exp(\beta L)
\prod_{i=1}^{n}\bm{\mathrm{I\!E}}\bigg\lbrace\exp(-\beta X_{i})\bigg\rbrace\nonumber\\&\equiv
\inf_{\beta>0}\exp(\beta L)\bm{\mathrm{I\!E}}\bigg\lbrace
\exp(-\beta\sum_{i=1}^{n}|\widehat{X}_{i}|)\bigg\rbrace
\end{align}
or
\begin{equation}
\bm{\mathrm{I\!P}}\bigg(\bigg\|\sqrt{\widehat{X}_{i}}\bigg\|^{2}\le L\bigg)\le \inf_{\beta>0}\exp(\beta L)
\equiv \inf_{\beta>0}\exp(\beta L)\bm{\mathrm{I\!E}}\bigg\lbrace
\exp(-\beta\bigg\|\sqrt{\widehat{X}_{i}}\bigg\|^{2})\bigg\rbrace
\end{equation}
or
\begin{equation}
\bm{\mathrm{I\!P}}\big(\big\|\widehat{X}_{i}\big\|\le L\big)\le \inf_{\beta>0}\exp(\beta L)
\equiv \inf_{\beta>0}\exp(\beta L)\bm{\mathrm{I\!E}}\bigg\lbrace
\exp(-\beta\big\|\widehat{X}_{i}\big\|)\bigg\rbrace
\end{equation}
so that for $\|\widehat{\bm{X}}(t)\|\equiv\|\widehat{\bm{u}}(t)-\bm{u}^{E}\|$ and using the estimate (2.86) the Chernoff estimate (2.150) is
\begin{align}
\lim_{t\uparrow\infty}\bm{\mathrm{I\!P}}[\|\widehat{\bm{u}}(t)-\bm{u}^{E}\|& \le |L|]\equiv \lim_{t\uparrow\infty}\bm{\mathrm{I\!P}}[\|\widehat{\bm{u}}(t)-\bm{u}^{E}\|\in \mathbb{B}(L)]\nonumber\\&\le\lim_{t\uparrow\infty}\inf_{\beta>0}\exp(+\beta|L|)
\bm{\mathrm{I\!E}}\bigg\lbrace(\exp(-\beta\bigg\|\widehat{\bm{u}}(t)-
\bm{u}^{E}\bigg\|)\bigg\rbrace\nonumber\\&=
\lim_{t\uparrow\infty}\inf_{\beta>0}\exp(+\beta|L|)\bm{\mathrm{I\!E}}
\left\lbrace\sum_{\ell=0}^{\infty}\frac{(-\beta)^{\ell}}{\ell!}
\bigg\|\widehat{\bm{u}}(t)-\bm{u}^{E}\bigg\|^{\ell}\right\rbrace\nonumber\\&
=\lim_{t\uparrow\infty}\inf_{\beta>0}\exp(+\beta|L|)\sum_{\ell=0}^{\infty}\frac{(-\beta)^{\ell}}{\ell!}
\bm{\mathrm{I\!E}}\bigg\lbrace\bigg\|\widehat{\bm{u}}(t)
-\bm{u}^{E}\bigg\|^{\ell}\bigg\rbrace\nonumber\\&=\lim_{t\uparrow\infty}\inf_{\beta>0}\exp(+\beta|L|)
\sum_{\ell=0}^{\infty}\frac{(-\beta)^{\ell}}{\ell!}
\bm{\mathrm{I\!E}}\bigg\lbrace\bigg\|\bm{u}^{E}
\exp\left(\int_{0}^{t}\widehat{\mathscr{U}}(\tau)d\tau\right)
\bigg\|^{\ell}\bigg\rbrace
\nonumber\\&\le\lim_{t\uparrow\infty}\inf_{\beta>0}\exp(+\beta|L|)\sum_{\ell=0}^{\infty}\frac{(-\beta)^{\ell}}{\ell !}n^{\ell/2}|u^{E}|^{\ell}\bm{\mathrm{I\!E}}\left\lbrace\exp\left(\zeta\ell\int_{0}^{t}
\widehat{\mathscr{U}}(\tau)d\tau\right)\right\rbrace\nonumber\\&
\le \lim_{t\uparrow\infty}\inf_{\beta>0}\exp(+\beta|L|)\sum_{\ell=0}^{\infty}\frac{(-\beta)^{\ell}}{\ell !} n^{\ell/2}|u^{E}|^{\ell}\left(\left|\bm{\mathrm{I\!E}}\left\lbrace\exp\left(\zeta\ell\int_{0}^{t}
\widehat{\mathscr{U}}(\tau)d\tau\right)\right\rbrace\right|^{1/\ell}\right)^{\ell}
\nonumber\\&=\lim_{t\uparrow\infty}\inf_{\beta>0}\exp(+\beta|L|)\exp\left(-\beta|u^{E}|n^{1/2}
\left|\bm{\mathrm{I\!E}}\left\lbrace\exp\left(\zeta\ell\int_{0}^{t}
\widehat{\mathscr{U}}(\tau)d\tau\right)\right\rbrace\right|^{1/\ell}\right)\nonumber\\&
=\lim_{t\uparrow\infty}\inf_{\beta>0}\exp(+\beta|L|)\exp(-\beta|u^{E}|n^{1/2}|\mathbf{I}(t,\ell)|^{1/\ell})
\end{align}
\end{proof}
The stability criteria can then be determined if one can explicitly estimate the stochastic integral
\begin{align}
&\mathbf{I}(t)=\bm{\mathrm{I\!E}}\left\lbrace\exp\left(\zeta\ell\int_{0}^{t}\widehat{\mathscr{U}}(\tau)d\tau\right)
\right\rbrace
\\&\lim_{t\uparrow\infty}\mathbf{I}(t)=\lim_{t\uparrow\infty}\bm{\mathrm{I\!E}}\left\lbrace\exp\left(\zeta\ell\int_{0}^{t}\widehat{\mathscr{U}}(\tau)d\tau\right)
\right\rbrace
\end{align}
This will be done in Section 6. A stability criterion can also be derived from a Hoeffding inequality which provides an upper bound on the probability that the sum of a set of bounded independent (sub-Gaussian) random variables deviates from its expected value by more than a specified amount.
\begin{prop}
As before, let $(\widehat{u}_{i}(t))_{i=1}^{n}=(\widehat{u}_{1}(t)...\widehat{u}_{n}(t))$ be the set of random variables due to random perturbations of the initially static equilibria $u_{i}^{E}$ such that $\widehat{u}_{i}(t)=a_{i}^{E}\exp(\zeta
\int_{0}^{t}\mathscr{U}_{i}(s)ds$. Then $u_{i}(t)$ solves an averaged equation of the form (-) such that $\mathbf{D}_{n}u_{i}(t)=\lambda$. Suppose the perturbed solutions converge to 'attractors' or new equilibrium points within a finite time such that $\widehat{u}_{i}(t)\rightarrow u_{i}^{E*}$. Then $\widehat{u}_{i}^{E}\le \widehat{u}_{i}(t)\le u_{i}^{E*}$ for all finite $t>0$. If
\begin{align}
&\widehat{S}(t)=\frac{1}{n}\sum_{i=1}^{n}\widehat{u}_{i}(t)
=\frac{1}{n}(\widehat{u}_{1}(t)+...+\widehat{u}_{n}(t))\\&
\bm{\mathrm{I\!E}}\big\lbrace\widehat{S}(t)\big\rbrace=\frac{1}{n}\sum_{i=1}^{n}\bm{\mathrm{I\!E}}\big\lbrace\widehat{u}_{i}(t)\big\rbrace
\end{align}
Then
\begin{align}
\bm{\mathrm{I\!P}}(\widehat{S}(t)&-\bm{\mathrm{I\!E}}
\big\lbrace\widehat{S}(t)\big\rbrace)=\bm{\mathrm{I\!P}}\bigg(\frac{1}{n}\sum_{i=1}^{n}|\widehat{u}_{i}(t)|-\frac{1}{n}\sum_{i=1}^{n}
\bm{\mathrm{I\!E}}\big\lbrace\widehat{u}_{i}(t)\rbrace\ge |L|\bigg)\nonumber\\&
=\bm{\mathrm{I\!P}}\bigg(\frac{1}{n}\sum_{i=1}^{n}u_{i}^{E}\exp\bigg(\zeta\int_{0}^{t}\mathscr{U}_{i}(s)ds\bigg)-\frac{1}{n}
\sum_{i=1}^{n}u_{i}^{E}\bm{\mathrm{I\!E}}\bigg\lbrace
\bigg(\exp\bigg(\zeta\int_{0}^{t}\mathscr{U}_{i}(s)ds\bigg)\bigg\rbrace\ge |L|
\bigg)\nonumber\\&=\bm{\mathrm{I\!P}}\bigg(\frac{1}{n}\sum_{i=1}^{n}\bigg|\sqrt{u_{i}^{E}}
\exp\bigg(\frac{\zeta}{2}\int_{0}^{t}\mathscr{U}_{i}(s)ds\bigg)\bigg|
-\frac{1}{n}\sum_{i=1}^{n}\bigg|\sqrt{u_{i}^{E}\bm{\mathrm{I\!E}}\bigg\lbrace
\bigg(\exp\bigg(\zeta\int_{0}^{t}\mathscr{U}_{i}(s)ds\bigg)\bigg\rbrace}\bigg|^{2}\ge |L|
\bigg)\nonumber\\&=\bm{\mathrm{I\!P}}\bigg(\frac{1}{n}\bigg\|\sqrt{u_{i}^{E}}
\exp\bigg(\frac{\zeta}{2}\int_{0}^{t}\mathscr{U}_{i}(s)ds\bigg)\bigg\|_{L_{2}}^{2}
-\frac{1}{n}\bigg\|\sqrt{u_{i}^{E}\bm{\mathrm{I\!E}}\bigg\lbrace
\bigg(\exp\bigg(\zeta\int_{0}^{t}\mathscr{U}_{i}(s)ds\bigg)\bigg\rbrace}\bigg\|_{L_{2}}^{2}\ge |L|\bigg)\nonumber\\&\le\frac{\exp(-2n^{2}|L|^{2}}{
\sum_{i=1}^{n}\big|u_{i}^{E*}-u_{i}^{E*}\big|_{L_{2}}^{2})}
\equiv\frac{\exp(-2n^{2}L^{2}}{\big\|u_{i}^{E*}-u_{i}^{E}\big\|_{L_{2}}^{2}}
\end{align}
Hence if $\big\|u_{i}^{E}-u_{i}^{E*}\big\|_{L_{2}}^{2})<\infty$ for all finite $t\in\mathbb{R}^{+}$ then there is zero                                            probability of blowup or asymptotic instability for any finite $t>0$, so that
\begin{align}
&\bm{\mathrm{I\!P}}\bigg(\frac{1}{n}\bigg\|\sqrt{u_{i}^{E}}
\exp\bigg(\frac{\zeta}{2}\int_{0}^{t}\mathscr{U}_{i}(s)ds\bigg)\bigg\|_{L_{2}}^{2}
-\frac{1}{n}\bigg\|\sqrt{u_{i}^{E}\bm{\mathrm{I\!E}}\bigg\lbrace
\bigg(\exp\bigg(\zeta\int_{0}^{t}\mathscr{U}_{i}(s)ds\bigg)\bigg\rbrace}\bigg\|_{L_{2}}^{2}= \infty\bigg)=0\\&\lim_{t\uparrow\infty}\bm{\mathrm{I\!P}}\bigg(\frac{1}{n}\bigg\|\sqrt{u_{i}^{E}}
\exp\bigg(\frac{\zeta}{2}\int_{0}^{t}\mathscr{U}_{i}(s)ds\bigg)\bigg\|_{L_{2}}^{2}
-\frac{1}{n}\bigg\|\sqrt{u_{i}^{E}\bm{\mathrm{I\!E}}\bigg\lbrace
\bigg(\exp\bigg(\zeta\int_{0}^{t}\mathscr{U}_{i}(s)ds\bigg)\bigg\rbrace}\bigg\|_{L_{2}}^{2}= \infty\bigg)=0
\end{align}
If however, $u_{i}^{E*}\rightarrow\infty$
\begin{align}
\bm{\mathrm{I\!P}}(\widehat{S}(t)&-\bm{\mathrm{I\!E}}
\big\lbrace\widehat{S}(t)\big\rbrace\ge |L|)\nonumber\\&=
\bm{\mathrm{I\!P}}\bigg(\frac{1}{n}\bigg\|\sqrt{u_{i}^{E}}
\exp\bigg(\frac{\zeta}{2}\int_{0}^{t}\mathscr{U}_{i}(s)ds\bigg)\bigg\|_{L_{2}}^{2}
-\frac{1}{n}\bigg\|\sqrt{u_{i}^{E}\bm{\mathrm{I\!E}}\bigg\lbrace
\bigg(\exp\bigg(\zeta\int_{0}^{t}\mathscr{U}_{i}(s)ds\bigg)\bigg\rbrace}\bigg\|_{L_{2}}^{2}\ge |L|\bigg)\nonumber\\&
=\lim_{u_{i}^{E*}\uparrow\infty}\frac{\exp(-2n^{2}L^{2}}{\big\|u_{i}^{E*}-u_{i}^{E}\big\|_{L_{2}}^{2})}=1
\end{align}
and there is then unit probability that the growth of the norm of perturbed solutions cannot be contained within any finite $L>0$. Hence, the system is asymptotically unstable to the random perturbations.
\end{prop}
Note that bounded random variables are always sub-Gaussian. Sub-Gaussianality is then a necessary criteria for stability or convergence of the randomly perturbed system, described by $\widehat{u}_{i}(t)$, to new attractors or equilibrium fixed points.
\section{The Einstein vacuum equations on $\mathbb{T}^{n}\times\mathbb{R}^{+}$ as an n-dimensional system of nonlinear ODES}
Sets of multi-dimensional nonlinear autonomous ODEs with the structure of the form (2.2) and (2.3) arise within general relativity when one reduces the Einstein vacuum equations on a n-dimensional toroidal geometry. Kasner-type cosmological models [47] are interpreted here as nonlinear n-dimensional dynamical systems with both dynamic solutions and static or equilibrium solutions. Stability and instability criteria with respect to sharp or 'short-pulse' deterministic perturbations and then random/stochastic perturbations are developed in detail in subsequent sections, tentatively applying and developing the ideas of Section 2.

In higher-dimensional general relativity, the dynamical quantities are an $(n+1)-$ dimensional manifold $\mathbb{M}^{n+1}$ and the Lorenztian metric $\bm{g}_{AB}$, where $A,B=1...n+1$. If $\mathbf{T}_{AB}$ is an energy-momentum source tensor for fluid, matter and or fields then the Einstein equations are
\begin{equation}
\bm{\mathrm{Ric}}_{AB}-\frac{1}{2}\bm{g}_{AB}\bm{\mathrm{R}}+\Lambda\bm{g}_{AB}=\bm{\mathrm{T}}_{AB}
\end{equation}
and $\bm{\nabla}_{A}T^{AB}=0$ is the energy conservation condition, where $\bm{\nabla}_{A}$ is the covariant derivative and $\Lambda$ is a fixed cosmological constant [48,49,50,51,52,53]. In this paper, the spacetime is cosmological with metric $ds^{2}=-dt^{2}+\delta_{ij}|a_{i}^{2}(t)|dX^{i}\otimes dX^{j}$ and with scale factors $a_{i}(t)$. It is globally hyperbolic, that is, foliated with compact spacelike Cauchy hypersurfaces $\Sigma_{t}$. In particular, the nonlinear ODEs that we will consider arise from reduction of the Einstein vacuum equations on $\mathbb{M}^{n+1}=\mathbb{T}^{n}\times \mathbb{R}^{+}$, where $\mathbb{T}^{n}$ is an n-torus and $\mathbb{R}^{+}=[0,\infty)$, and containing no matter so that $\bm{\mathrm{T}}_{AB}=0$. The dominant cosmological model with matter and/or radiation is the Friedmann-Lemaitre-Robertson-Walker (FLRW) model [48,49,50,51,52,53]. Closed-universe FLRW solutions for example, exhibit uniform curvature blowup along a pair of spacelike hypersurfaces,($\Sigma_{T^{Bang}}, \Sigma_{T^{Crunch}})$ signifying the Big Bang and the Big Crunch. In particular, this is exemplified by the blowup of the invariant $\bm{K}=\bm{\mathrm{Riem}}_{ABCD}\bm{\mathrm{Riem}}^{ABCD}$ along the hypersurfaces ($\Sigma_{T^{Bang}}, \Sigma_{T^{Crunch}})$; so the spacetime is geodesically incomplete to the past and future. In the absence of matter $\bm{\mathrm{T}}_{AB}=0$ the Einstein equations reduce to $\bm{\mathrm{Ric}}_{AB}-\frac{1}{2}\bm{g}\bm{\mathrm{R}}+\Lambda\bm{g}_{AB}=0$ or $\bm{\mathrm{Ric}}_{AB}=\bm{g}_{AB}\Lambda$, describing deSitter or anti-deSitter space if $\Lambda>0$ or $\Lambda<0$. If $\Lambda=0$ then the vacuum equations are $\mathbf{Ric}_{AB}=0$ which still admit dynamical anisotropic cosmological solutions of the Kasner-Bianchi .
\begin{rem}
General relativity admits a well-posed Cauchy initial-value formulation [11,53,54]. The initial data is given by the set $\mathfrak{D}=[\Sigma_{o},\bm{g}_{ij}(0),\bm{k}_{ij}(0)]$, where $\bm{g}_{ij}(0)$ is the Riemannian 3-metric on $\Sigma_{o}$,with $\bm{g}_{00}=-1$ and $\bm{g}_{io}=0$. $\bm{k}_{ij}(0)$ is the covariant symmetric tensor, and the constraints on the initial data are $\mathbf{Ric}_{00}=0$ and $\mathbf{R}_{io}=0$, equivalent to the Codazzi and Gauss constraint conditions. For the Cauchy evolution, the 1st variation equation is
\begin{equation}
\partial_{t}\bm{g}_{ij}(t)\equiv -2\bm{k}_{ij}(t)
\end{equation}
A spacetime manifold $\mathbb{M}^{n+1}$ is then a development of the initial data $\mathfrak{D}$ and there is an imbedding
\begin{equation}
I:\Sigma\rightarrow\mathbb{M}^{3+1}\nonumber
\end{equation}
The equivalence class of all maximal Cauchy Einstein developments of $\mathfrak{D}$ are related by diffeomorphisms.
\end{rem}
In theories of ordinary linear and nonlinear dynamical systems at the Newtonian level, both gravitational and non-gravitational, time and space are absolute and the notion of mechanical phase space is clear: the system evolves against a fixed background reference geometry, essentially $(\mathbb{R}^{4},\bm{\eta}_{\alpha\beta})$. For example, a 'cloud' of N classical particles of mass $m$ interacting gravitationally--for example, a globular star cluster-- with coordinates $\bm{r}_{1}(t),...,\bm{r}_{N}(t)$ and velocities $\bm{v}_{1}(t),...,\bm{v}_{N}(t)$ is an N-body Newtonian dynamical system described by N coupled equations $ \frac{d\bm{r}_{i}(t)}{dt}=\bm{v}_{i}(t) $ and
\begin{align}
&\frac{d\bm{v}_{i}(t)}{dt}=-Gm\sum_{i\ne j}\frac{\bm{r}_{i}(t)-\bm{r}_{j}(t)}{|\bm{r}_{i}(t)-\bm{r}_{j}(t)|^{3}}\equiv
-m\nabla U(\bm{r}_{1}...\bm{r}_{n})
\\&U(\bm{r}_{1}(t)...\bm{r}_{n}(t))=\sum_{i<j}\Phi(\bm{r}_{i}(t)-\bm{r}_{j}(t))
\end{align}
and where $\Phi(\bm{r}_{i}(t)-\bm{r}_{j}(t))=-G/|\bm{r}_{i}(t)-\bm{r}_{j}(t)|$ is the Newtonian potential between pairs of particles. The Hamiltonian is $
H=\sum_{i=1}^{N}\frac{1}{2}|m\bm{v}_{i}(t)|^{2}+m^{2}U(\bm{r}_{1}...\bm{r}_{n})$. The dynamical evolution of the system can be considered from initial data, although the problem even for N=3 can become chaotic and the future evolution cannot be predicted from the initial data. But in general relativity space and time themselves assume a dynamical role with a space-time $(\mathbb{M}^{3+1},\bm{g})$ that is a solution of the Einstein equations. General relativistic systems therefore do not appear to be dynamical systems in the usual sense in that they do not provide an obvious set of parameters 'evolving in time'.

One way to deal with general relativity on $\mathbb{M}^{3+1}$ and on a higher-dimensional manifold $\mathbb{M}^{n+1}$  is to break the space-time covariance of the formulation and use an ADM split [49,50] exploiting the foliation of the  manifold $\mathbb{M}^{n+1}$ with the space-like hypersurfaces ${\Sigma}$. The metric $\bm{g}_{AB}$ on $\mathbb{M}^{n+1}$ induces a an n-metric metric $\bm{g}_{ij}$ on the hypersurface $\Sigma$ so that
\begin{equation}
ds^{2}=-\phi^{2}dt^{2}+g_{ij}(dX^{i}+{\phi}^{i})\otimes(dX^{j}+{\phi}^{j}dt)
\end{equation}
where ${\phi}$ and ${\phi}^{i}$ are the lapse function and shift vectors. Setting ${\phi}^{i}=0$ and ${\phi}=1$ gives the typical cosmological metric form
\begin{equation}
ds^{2}=-dt^{2}+\bm{g}_{ij}dX^{i}\otimes dX^{j}=-dt^{2}+\delta_{ij}a^{2}(t)dX^{i}\otimes dX^{j}
\end{equation}
Then $\mathbb{M}^{n+1}=\Sigma_{t}\times\mathbb{R}^{+}$. Using the ADM split, and the well-defined Cauchy formulation, the evolution of the n-metric $\bm{g}_{ij}$ as
$\partial_{t}\bm{g}_{ij}(t)$ via the 1st variation equation and can then be interpreted as a 'nonlinear dynamical system', on somewhat equal terms as conventional dynamical systems which possess degrees of freedom evolving in time.

The Einstein vacuum equations $\bm{\mathrm{Ric}}_{AB}=0$ are applied to a hypertoroidal cosmological-type metric (3.6) and the static and dynamical solutions are considered. The Einstein vacuum equations will then reduce to a system of nonlinear ODEs in terms of the first and second derivatives of the scale factors $\partial_{t}a(t)$ and $\partial_{tt}a(t)$ which can be interpreted as a nonlinear dynamical system, structurally of the from (2.2) and (2.3).
\begin{defn}
An (n+1)-dimensional toroidal space-time has the following properties
\begin{enumerate}
\item $\mathbb{M}^{n+1}$ is the product
$\mathbb{M}^{n+1}=\mathbb{T}^{n}\times \mathbb{R}^{+}$, where
$\mathbb{T}^{n}$ is an isotropic n-torus.
\item The Einstein vacuum equations are $\bm{\mathrm{Ric}}_{AB}=0$ and $(\mathbb{M}^{n+1},\bm{g}_{AB})$ is a solution.
\item Topologically,$\mathbb{T}^{n}$ is an isotropic or anisotropic n-torus, for which the constant time slices $\Sigma_{t}$ are n-dimensional tori with 'rolling radii' $a_{i}(t)$; that is, radii that depend only the time parameter $t\in\mathbb{R}^{+}$.
\item Topologically, $\mathbb{T}^{n}$ is a Cartesian product of n circles so that $\mathbb{T}^{n}=S^{1}\times S^{2}\times...\times S^{n}$.
\item Retaining summations, the metric is a solution of the Einstein vacuum equations of
the form
\begin{align}
ds^{2}&=\sum_{A=0}^{n+1}\sum_{B=0}^{n+1}\bm{g}_{AB}dX^{A}\otimes dX^{B}=-dt^{2}+
\sum_{i=1}^{n}\sum_{j=1}^{n}\bm{g}_{ij}dX^{i}\otimes dX^{j}\nonumber\\
&= -dt^{2} + \sum_{i=1}^{n}\sum_{j=1}^{n}(|2\pi
a_{i}(t)|^{2}dX^{i}\otimes dX^{j}
\end{align}
with $\bm{g}_{00}=-1$ and $\bm{g}_{ij}=(2\pi a_{i}(t))^{2}$, with all other
off-diagonal components vanishing.
\item Each $X^{i}=X^{i}+2\pi a_{i}$ takes values in a circle of radius $a_{i}(t)$, where $i=1$ to $n$. The constant time slices or Cauchy spacelike surfaces $\Sigma_{t}$ are $n$-dimensional tori with 'rolling radii' $a_{i}(t)$.
\item The radii $a_{i}(t)$ can be parametrized by a set of n scalar modulus functions (moduli) $(\psi_{i}(t))_{i=1}^{n}$ so that $\psi_{i}:\mathbb{R}^{+}\rightarrow\mathbb{R}^{+}$ for $i=1...n$, and $a_{i}(t)=\exp(\psi_{i}(t))$, for $i=1...n$. The metric (3.7) then becomes
\begin{equation}
ds^{2}=-dt^{2}+\sum_{i=1}^{n}\sum_{j=1}^{n}\bm{g}_{ij}dX^{i}\otimes dX^{j}=-dt^{2}+ \sum_{i=1}^{n}\sum_{j=1}^{n}(2\pi)^{2}\delta_{ij}\exp(2\psi_{i}(t))dX^{i}\otimes dX^{j}
\end{equation}
\item If $\psi_{i}(t)=\psi_{i}^{E}=\psi^{E}$ then $a_{i}(t)=a_{i}^{E}
=\exp(\psi_{i}^{E})$ and the dimensions are stable or constant as $t\rightarrow\infty$. This describes a 'static hypertoroidal universe'.
\end{enumerate}
\end{defn}
In general, a metric on an n-torus $\mathbb{T}^{n}$ has $\frac{1}{2}n(n+1)$ moduli, namely $n$ radii and $\frac{1}{2}(n-2)(n-3)$ angles. However, one usually chooses $g_{ij}(t)=0$ thus freezing the "rolling angles" for $i\ne j$ and considering only the 'rolling radii'[52,53].
\begin{defn}
The spatial volume $\mathbf{V}_{\mathbf{g}}(t)$ of the toroidal n-metric (3.7) or (3.8) is defined as
\begin{equation}
\mathbf{V}_{\mathbf{g}}(t)=\prod_{i=1}^{n}\exp(\psi_{i}(t))\equiv\exp\left(\sum_{i=1}^{n}\psi_{i}(t))\right)
\equiv\prod_{i=1}^{n}a_{i}(t)
\end{equation}
and the $\mathcal{L}_{(2,1)}$ norm of the diagonal n-metric is
\begin{equation}
\|\bm{g}(t)\|_{(2,1)}=\sum_{j=1}^{n}\left(\sum_{i=1}^{n}|g_{ij}(t)|^{2}\right)^{1/2}\equiv
\sum_{i=1}^{n}\left(\sum_{i=1}^{n}|g_{ii}(t)|^{2}\right)^{1/2}
\end{equation}
The Frobenius norm can also be used such that
\begin{equation}
\|\bm{g}(t)\|_{F}=\left(\sum_{j=1}^{n}\sum_{i=1}^{n}|g_{ij}(t)|^{2}\right)^{1/2}\equiv
\left(\sum_{i=1}^{n}\sum_{i=1}^{n}|g_{ii}(t)|^{2}\right)^{1/2}
\end{equation}
\end{defn}
\begin{rem}
Such metrics also arise in toroidal compactification of Kaluza-Klein and superstring theories [54,55,56,57,58,59,60], whereby a theory on a manifold $\mathbb{N}^{m+n+1}$ where can be compactified as $\mathbb{N}^{m+n+1}\rightarrow\mathbb{M}^{m+1}\times\mathbb{T}^{n}$ such that
\begin{align}
ds^{2}&=-dt^{2}+\sum_{a,b}g_{ab}dX^{a}dX^{b}+\sum_{i,j}\delta_{ij}|2\pi a_{i}(t)|^{2}dX^{i}\otimes dX^{j}\nonumber\\
&\equiv -dt^{2}+\sum_{a,b}g_{ab}dX^{a}dX^{b}+\sum_{i,j}\delta_{ij}(2\pi)^{2}\exp(2\psi(t))dX^{i}\otimes dX^{j}
\end{align}
with $a,b=1...m$ and $i,j=1...n$. For example, $n=6$ and $m=3$ for a toroidal compactification of a superstring theory or $n=1$ and $m=3$ for a basic 5-dimensional Kaluza-Klein compactification on a circle. For M-theory one has $n=10$
\end{rem}
\begin{rem}
The metric (3.7) or (3.8) also represents the higher-dimensional generalization of the Kasner solutions found in 4-dimensional Bianchi-Type I cosmological models [47,48]
\begin{equation}
ds^{2}=-dt^{2}+\sum_{i=1}^{n}t^{2 {p}_{i}}(dX_{i})^{2}
\end{equation}
with the Kasner constraints
$\sum_{i=1}^{n}{p}_{i}=\sum_{i=1}^{n}{p}_{i}{p}_{i}=1$. For n=4, this
is the Bianchi Type-I Universe
\begin{equation}
ds^{2}=-dt^{2}+t^{2{p}_{1}}dx^{2}+t^{2{p}_{2}}dy^{2}+t^{2{p}_{3}}dz^{2}
\end{equation}
where ${p}_{1}+{p}_{2}+{p}_{3}={p}_{1}^{2}+{p}_{2}^{2}+{p}_{3}^{2}=1$.
\end{rem}
The space is anisotropic if at least two of the three ${p}_{i}$ are different. The Kasner constraints only hold for a pure vacuum and are lost in the presence of a fluid matter source. Kasner-like solutions are also the building blocks of string cosmology [57], and for $(n+1)=11$ they represent vacuum cosmological solutions of the low-energy effective limit of M-theory or 11-dimensional supergravity [58].
\begin{thm}
For empty 'toroidal universes' with no matter and with $\Lambda=0$ The Einstein vacuum field equations are
\begin{equation}
\bm{\mathrm{G}}_{AB}=\bm{\mathrm{Ric}}_{AB}-\frac{1}{2}\bm{g}_{AB}\bm{\mathrm{R}} =0
\end{equation}
or $\bm{\mathrm{Ric}}_{AB}=0$. Using the metric ansatz (3.8) with $\bm{g}_{00}+-1$ the Einstein equations can be reduced to a system of ordinary nonlinear differential equations.
\begin{equation}
\mathbf{H}_{n}\psi_{i}(t)\equiv\sum_{i=1}^{n}\partial_{tt}{\psi}_{i}(t)+\frac{1}{2}
\sum_{i=1}^{n}\partial_{t}{\psi}_{i}(t)\partial_{t}{\psi}_{i}(t)+\frac{1}{2}
\sum_{i-1}^{n}\sum_{j=1}^{n}\partial_{t}{M}_{i}(t)\partial_{t}{\psi}_{i}(t)=0
\end{equation}
Since $a_{i}(t)=\exp(\psi(t))$, an equivalent set of differential equations is
\begin{equation}
\mathbf{D}_{n}a_{i}(t)=\sum_{i=1}^{n}\frac{\partial_{tt}a_{i}(t)}{a_{i}(t)}                                                                                                                                                                                                                                                                                                                                                              -\frac{1}{2}\sum_{i=1}^{n}\frac{\partial_{t}a_{i}(t)\partial_{t}a_{i}(t)}{a_{i}(t)a_{i}(t)}
+\frac{1}{2}\sum_{i=1}^{n}\sum_{j=1}^{n}
\frac{\partial_{t}{a}_{i}(t)\partial_{t}{a}_{j}(t)}{a_{i}(t)a_{j}(t)}=0
\end{equation}
where $\mathbf{H}_{n}$ and $\mathbf{D}_{n}$ are now the nonlinear differential operators on $\mathbb{R}^{n}$ such that
\begin{align}
&\mathbf{H}_{n}(...)\equiv\sum_{i=1}^{n}\partial_{tt}(...)+
\frac{1}{2}\sum_{i=1}^{n}\partial_{t}(...)\partial_{t}(...)+\frac{1}{2}\sum_{i=1}\sum_{j=1}^{n}\partial_{t}(...)\partial_{t}(...)\\&
\mathbf{D}_{n}(...)\equiv\sum_{i=1}^{n}\frac{\partial_{tt}(...)}{a_{i}(t)}-\frac{1}{2}
\sum_{i=1}^{n}\frac{\partial_{t}(...)\partial_{t}(...)}{a_{i}(t)a_{i}(t)}+
\frac{1}{2}\sum_{i=1}^{n}\sum_{j=1}^{n}\frac{\partial_{t}(...)\partial_{t}(...)}
{a_{i}(t)a_{j}(t)}=0
\end{align}
\end{thm}
\begin{proof}
Following [56], we introduce an orthonormal basis of one-forms $\mathbf{e}^{0}=dt$ and $\mathbf{e}^{\hat{i}}=2\pi a_{i}(t)dX^{i}$. The Cartan structure equations
\begin{equation}
\bm{\mathrm{R}}^{i}_{j}=d\bm{\omega}^{i}_{j}+\bm{\omega}^{i}_{k}\wedge \bm{\omega}^{k}_{j}
\end{equation}
are then solved for the spin connection 1-forms so that $\bm{\omega}^{0}_{i}=(a_{i}(t)/a_{i}(t))\mathbf{e}^{\hat{i}}$ and $\bm{\omega}^{\hat{i}}_{\hat{j}}=0$, and the structure equations for curvature give the curvature 2-forms
\begin{equation}
\bm{\mathrm{R}}^{\hat{0}_{\hat{i}}}=d\omega^{0}_{i}=\frac{\partial_{tt}{a}_{i}(t)}{a_{i}(t)} \bm{e}^{\hat{0}}\bm{\wedge} \bm{e}^{\hat{i}}, a^{\hat{i}}_{\hat{j}}=\bm{\omega}^{i}_{0}\wedge\bm{\omega}^{0}_{j}
=\frac{\partial_{t}a_{i}(t)\partial_{t}a_{j}(t)}
{a_{i}(t)a_{j}(t)}
\end{equation}
The only nonvanishing components of the curvature tensor are
\begin{align}
&\bm{\mathrm{Riem}}^{0}_{i0i}=\sum_{i}^{n}\frac{a_{i}(t)}{a_{i}(t)}\\& \bm{\mathrm{Riem}}^{i}_{jij}
=\sum_{i=1}^{n}\sum_{j=1}^{n}\frac{\partial_{t}{a}_{i}(t)\partial_{t}{a}_{j}(t)}{a_{i}(t)a_{j}(t)}
\end{align}
The Ricci tensor components and curvature scalar are then
\begin{align}
&\bm{\mathrm{Ric}}_{00}=-\sum_{i=1}^{n}\frac{\partial_{tt}{a}_{i}(t)}{a_{i}(t)};\bm{\mathrm{R}}_{ii}=\sum_{i=1}^{n}
\frac{\partial_{tt}{a}_{i}(t)}{a_{i}(t)}+\sum_{i=1}^{n}\sum_{j}^{n}\frac{a_{i}(t)a_{j}(t)}{a_{i}(t)a_{j}(t)}
\\&\bm{\mathrm{R}}=2\sum_{i=1}^{n}\frac{\partial_{tt}{a}_{i}(t)}{a_{i}(t)}-
\sum_{i=1}^{n}\frac{\partial_{t}{a}_{i}(t)\partial_{t}{a}_{i}(t)}{a_{i}(t)a_{i}(t)}+
\sum_{i=1}^{n}\sum_{j=1}^{n}\frac{\partial_{t}{a}_{i}(t)\partial_{t}{a}_{j}(t)}{a_{i}(t)a_{j}(t)}
\end{align}
In terms of the moduli $\psi_{i}(t)$ where $a_{i}(t)=\exp(\psi_{i}(t))$, the Ricci curvature scalar for the toroidal metric (3.8) is
\begin{equation}
\bm{\mathrm{R}}=-2\hat{g}^{00}\sum_{i=1}^{n}\partial_{tt}\psi_{i}(t)
+\sum_{i=1}^{n}\partial_{t}{\psi}_{i}(t)\partial_{t}{\psi}_{i}(t)+\sum_{i=1}^{n}\sum_{j=1}^{n}
\partial_{t}{\psi}_{i}(t)\partial_{t}{\psi}_{j}(t)
\end{equation}
with $\bm{\mathrm{R}}_{i0}=0$ and $\bm{\mathrm{R}}_{ij}=0$ for $i\ne j$. We set $\hat{g}_{00}=-1$ as a 'gauge choice' then (3.26) becomes
\begin{equation}
\bm{\mathrm{R}}=2\sum_{n=1}^{n}\partial_{tt}{\psi}_{i}(t)+\sum_{i=1}^{n}\partial_{t}{\psi}_{i}(t)
\partial_{t}{\psi}_{i}(t)+\sum_{i=1}^{n}\sum_{j=1}^{n}\partial_{t}{\psi}_{i}(t)\partial_{t}{\psi}_{j}(t)
\end{equation}
The vacuum Einstein equations are $\bm{\mathrm{Ric}}_{AB}=\bm{g}_{AB}\bm{\mathrm{R}}=0$ or simply $\bm{\mathrm{R}}=0$, giving a set of nonlinear differential equations in terms of the radial moduli functions
\begin{equation}
\mathbf{H}_{n}\psi_{i}(t)\equiv\sum_{i=1}^{n}\partial_{tt}{\psi}_{i}(t)+\frac{1}{2}
\sum_{i=1}^{n}\partial_{t}{\psi}_{i}(t)\partial_{t}{\psi}_{i}(t)+\frac{1}{2}
\sum_{i=1}^{n}\sum_{j=1}^{n}\partial_{t}{\psi}_{i}(t)\partial_{t}{\psi}_{i}(t)=0
\end{equation}
Since $a_{i}(t)=\exp(\psi(t))$, an equivalent set of differential equations in terms of the toroidal radii is
\begin{equation}
\mathbf{D}_{n}a_{i}(t)\equiv\sum_{i=1}^{n}\frac{\partial_{tt}a_{i}(t)}{a_{i}(t)}-\frac{1}{2}
\sum_{i=1}^{n}\frac{\partial_{t}a_{i}(t)\partial_{t}a_{i}(t)}{a_{i}(t)a_{j}(t)}+\frac{1}{2}
\sum_{i=1}^{n}\sum_{j=1}^{n}\frac{\partial_{t}{a}_{i}(t)\partial_{t}{a}_{j}(t)}{a_{i}(t)a_{j}(t)}=0
\end{equation}
\end{proof}
\begin{rem}
In keeping with a dynamical systems interpretation, the action for pure Einstein gravity in n dimensions is
\begin{equation}
S=\int_{\mathbb{M}^{n+1}}d^{n+1}x\sqrt{-\det\bm{g}_{n+1}}\bm{\mathrm{R}}
\end{equation}
and $\delta S=0$ gives the vacuum field equations. Using (3.27) this can be written as
\begin{align}
&S=\int d\bar{t}\int d^{n}x\left(\prod_{k=1}^{n}\exp(\psi_{k}(t)\right)\nonumber\\&\times \left(\sum_{i=1}^{n}
2\partial_{tt}\psi_{i}(t)+\sum_{i=1}^{n}\partial_{t}\psi_{i}(t)\partial_{t}\psi_{i}(t)+
\sum_{i=1}^{n}\partial_{t}\psi_{i}(t)\partial_{t}\psi_{i}(t)\right)\nonumber\\
&=\int d\tau\left(\prod_{k=1}^{n}\exp(\psi_{k}(t)\right)\left(\sum_{i=1}^{n}|\partial_{t}
\psi_{i}(t)|^{2}-\sum_{i=1}^{n}\sum_{j=1}^{n}\partial_{t}\psi_{i}(t)
\partial_{t}\psi_{j}(t)\right)\nonumber\\&=\int d\tau\mathcal{L}_{eff}(\psi_{i}(t),\partial_{t}\psi_{i}(t))
\end{align}
Then the differential equations will follow from the corresponding Euler-Lagrange equations.
\end{rem}```
\begin{cor}
In terms of $L_{2}$ norms, the Einstein systems of nonlinear ODEs are
\begin{align}
&\mathbf{H}_{n}\psi_{i}(t)= \big\|\sqrt{\partial_{tt}\psi_{i}(t)}\big\|^{2}+\frac{1}{2}\big\|\partial_{t}\psi_{i}(t)\big\|^{2}
+\frac{1}{2}\big\|\sqrt{\partial_{t}\psi_{i}(t)}\big\|^{2}\big\|\sqrt{\partial_{t}\psi_{j}(t)}\big\|^{2}=0
\\&
\mathbf{D}_{n}a_{i}(t)=\bigg\|\frac{\partial_{tt}a_{i}(t)}{a_{i}(t)}\bigg\|^{2}
-\frac{1}{2}\bigg\|\frac{\partial_{t}a_{i}(t)}{a_{i}(t)}\bigg\|^{2}
+\frac{1}{2}\bigg\|\sqrt{\frac{\partial_{t}a_{i}(t)}{a_{i}(t)}}\bigg\|^{2}
\bigg\|\sqrt{\frac{\partial_{t}a_{j}(t)}{a_{j}(t)}}\bigg\|^{2}=0
\end{align}
\end{cor}

Equations with this form also arise from the low-energy effective string or supergravity actions [56,57]. For example, the Type-II superstring effective action on a manifold $\mathbb{M}^{9+1}$ is
\begin{equation}
S=\int d^{10}x(-g)^{1/2}\exp(-2\phi)(\mathbf{R}+4\bm{\nabla}_{A}\phi\bm{\nabla}^{A}\phi)+S_{M}
\end{equation}
where $\phi$ is the dilaton, which plays an important role in T-duality symmetry. The Einstein and string frames are related by $\bm{g}_{AB}^{(S)}
=\bm{g}_{AB}^{(E)}\exp(\tfrac{1}{2}\phi)$ so in terms of the Einstein metric the action is
\begin{equation}
S=\int d^{10}x(-g^{(E)})^{1/2}\left(\mathbf{R}-\frac{1}{2}\bm{\nabla}_{A}\phi\bm{\nabla}^{A}\phi\right)+S_{M}
\end{equation}
Dropping the superscript on the metric the equations of motion are then
\begin{align}
&\mathbf{G}_{AB}=\mathbf{R}_{AB}-\frac{1}{2}\mathbf{g}_{AB}\mathbf{R}=\frac{1}{2}\bm{\nabla}_{A}
\phi\bm{\nabla}_{B}\phi-\frac{1}{4}\mathbf{g}_{AB}\bm{\nabla}_{A}\phi\bm{\nabla}^{A}\phi-\frac{1}{(-g)^{1/2}}\frac{\delta S_{M}}{\delta \mathbf{g}^{AB}}\\&
\bm{\nabla}^{2}\phi=\frac{1}{(-g)^{1/2}}\frac{\delta S_{M}}{\delta \mathbf{g}^{AB}}
\end{align}
Since the string coupling is assumed small with $g=\exp(\phi)\ll 1$ for large radii then $\phi=const.$ with no running dilaton so that
\begin{equation}
\mathbf{G}_{AB}=\mathbf{R}_{AB}-\frac{1}{2}\mathbf{g}_{AB}\mathbf{R}
-\frac{1}{(-g)^{1/2}}\frac{\delta S_{M}}{\delta \mathbf{g}^{AB}}
\end{equation}
On a 9-torus $\mathbb{T}^{9}$
\begin{equation}
ds^{2}=-dt^{2}+\sum_{i=1}^{9}(2\pi a_{i}(t))^{2}dX^{i}\otimes dX^{i}
\end{equation}
the non-vanishing components of the Einstein tensor are then
\begin{align}
&\mathbf{G}^{t}_{t}=\frac{1}{2}\sum_{k\ne\ell}^{n}\frac{\partial_{t}a_{k}(t)\partial_{t}a_{\ell}(t)}
{a_{k}(t)a_{\ell}(t)}\\&
\mathbf{G}^{i}_{i}=\sum_{k\ne i}^{n}\frac{\partial_{tt}a_{k}(t)}{a_{k}(t)}+\frac{1}{2}
\sum_{k\ne \ell}^{n}\frac{\partial_{t}a_{k}(t)\partial_{t}a_{l}(t)}{a_{k}(t)a_{l}(t)}-\frac{1}{2}
\sum_{k\ne \ell}^{n}\frac{\partial_{t}{a}_{k}(t)\partial_{t}{a}_{i}(t)}{a_{k}(t)a_{i}(t)}
\end{align}
These can also be written  terms of the radial moduli fields $\psi_{i}(t)$. For the matter contribution within a cosmological scenario, a gas of massless supergravity particles has been utilised [59,60,61]. However, the salient point is that these differential equations have the same basic structure. For this paper, we consider only higher-dimensional pure Einstein gravity.
\subsection{Dynamic and static solutions}
\begin{lem}
The Einstein vacuum equations $\bm{\mathrm{Ric}}_{AB}=0$ on $\mathbb{T}^{n}\times\mathbb{R}^{+}$ in the form $\mathbf{H}_{n}\psi_{i}(t)=0$ and $\mathbf{D}_{n}a_{i}(t)=0$, and for some initial data $\mathfrak{D}=[t=0,\Sigma_{o},\psi_{i}(0)=\psi_{i}^{E},a_{i}(0)=a_{i}^{E}]$
have the power-law solutions for $t>0$.
\begin{align}
&\psi_{i}(t)=\psi_{i}^{E}+\ln|t|^{p_{i}}\equiv \psi_{i}^{E}+p_{i}\ln|t|
\\& a_{i}(t)=a_{i}^{E}|t|^{p_{i}}
\end{align}
provided the Kasner constraints are satisfied
\begin{equation}
\sum_{i=1}^{n}p_{i}^{2}=\sum_{i=1}^{n}p_{i}
\end{equation}
\end{lem}
\begin{proof}
The derivatives of (3.42) are $\partial_{t}\psi_{i}(t)=p_{i}t^{-1}$ and $\partial_{tt}\psi_{i}(t)=-p_{i}t^{-2}$ so that
\begin{equation}
\mathbf{H}_{n}\psi_{i}(t)=-\sum_{i=1}^{n}p_{i}t^{-2}+\frac{1}{2}\sum_{i=1}^{n}p_{i}p_{i}t^{-2}
+\frac{1}{2}\sum_{i=1}\sum_{j=1}p_{i}p_{j}t^{-2}=0
\end{equation}
which holds if $\sum_{i=1}^{n}p_{i}^{2}=\sum_{i=1}^{n}p_{i}$. Given the set of modulus functions $\psi_{i}(t)$ then the radii are $a_{i}(t)=\exp(\psi_{i}(t))$ so that
\begin{equation}
a_{i}(t)=\exp(\psi_{i}^{E})|t|^{p_{i}}\equiv a_{i}^{E}|t|^{p_{i}}
\end{equation}
The derivatives are $\partial_{t}a_{i}(t)=a_{i}^{E}p_{i}|t|^{p_{i}-1}$ and $\partial_{tt}a_{i}(t)=a_{i}^{E}p_{i}(p_{i}-1)|t|^{p_{i}-2}$ and (3.29) becomes
\begin{align}
&\mathbf{D}_{n}a_{i}(t)=\sum_{i=1}^{n}\frac{a_{i}^{E}p_{i}(p_{i}-1)|t|^{p_{i}-2}}{a_{i}^{E}|t|^{p_{i}}}-\frac{1}{2}\sum_{i=1}^{n} \frac{a_{i}^{E}a_{i}
^{E}p_{i}p_{i}|t|^{p_{i}-1}|t|^{p_{i}-1}}{a_{i}|t|^{p_{i}}a_{i}^{E}|t|^{p_{i}}}\nonumber\\&
+\frac{1}{2}\sum_{i=1}^{n}\sum_{j=1}^{n}\frac{a_{i}^{E}a_{j}
^{E}p_{i}p_{j}|t|^{p_{i}-1}|t|^{p_{j}-1}}{a_{i}|t|^{p_{i}}a_{j}^{E}|t|^{p_{j}}}=0
\end{align}
which is
\begin{equation}
\sum_{i=1}^{n}\frac{p_{i}(p_{i}-1)}{t^{2}}-\frac{1}{2}\sum_{i=1}^{n}\frac{p_{i}p_{i}}{t^{2}}+
\frac{1}{2}\sum_{i=1}^{n}\sum_{j=1}^{n}
\frac{p_{i}p_{j}}{t^{2}} =0
\end{equation}
Canceling the $|t|^{p_{i}-2}$ term gives
\begin{equation}
-\sum_{i=1}^{n}p_{i}+\frac{1}{2}\sum_{i=1}^{n}p_{i}^{2}+\frac{1}{2}
\sum_{i=1}^{n}\sum_{j=1}^{n}p_{i}p_{j}
\end{equation}
If $m_{ij}=p_{i}p_{j}$ is a diagonal matrix with $p_{i}=p_{j}$ for $i,j=1...n$ then the Kasner constraints follow so that
\begin{align}
&-\sum_{i=1}^{n}p_{i}+\frac{1}{2}\sum_{i=1}^{n}m_{ii}
+\frac{1}{2}\sum_{i=1}^{n}\sum_{j=1}^{n}m_{ij}\nonumber\\&
\equiv -\sum_{i=1}^{n}p_{i}+\frac{1}{2}\sum_{i=1}^{n}m_{ii}+\frac{1}{2}\sum_{i=1}^{n}m_{ii}
=-\sum_{i=1}^{n}p_{i}+\sum_{i=1}^{n}m_{ii}
\end{align}
which gives $\sum_{i=1}^{n}p_{i}=\sum_{i=1}^{n}p_{i}^{2}$.
\end{proof}
When $p_{i}\ne 0$ the vacuum Einstein equations on the n-torus $\mathbb{T}^{n}$ are a solvable system provided the ${p}_{i}$ obey the Kasner constraints as in the Bianchi model cosmologies. For the Bianchi-I model, for example
\begin{equation}
p_{1}^{2}+p_{2}^{2}+p_{3}^{2}=p_{1}+p_{2}+p_{3}=1
\end{equation}
Choosing the ordering $p_{1}<p_{2}<p_{3}$, these can also be parametrized by the Khalatnikov-Lifshitz parameter $\mathfrak{u}$ which was introduced for cosmological applications studying oscillatory and chaotic behavior within Kasner epochs [61,62,63,64,65].
\begin{equation}
p_{1}=\frac{-\mathfrak{u}}{1+\mathfrak{u}+\mathfrak{u}^{2}}; p_{2}=\frac{1+\mathfrak{u}}{1+\mathfrak{u}+\mathfrak{u}^{2}};p_{3}=\frac{\mathfrak{u}(1+\mathfrak{u})}
{1+\mathfrak{u}+\mathfrak{u}^{2}}
\end{equation}
As $\mathfrak{u}$ varies over $\mathfrak{u}\ge 1$ then $p_{1},p_{2},p_{3}$ can take on all permissible values such that $\frac{1}{3}\le p_{1}\le 0$ with $0\le p_{2}\le \frac{2}{3}$ and $ \frac{2}{3}\le p_{3}\le 1 $. In general, the $p_{i}$ can be chosen in any arbitrary way provided that the Kasner constraints are always satisfied. The cases where two of the $p_{i}$ are equal are given by the Bianchi-I triplets $\mathfrak{B}=(0,0,1)$ and $\mathfrak{B}=(-1/3,2/3,2/3)$. If the ${p}_{i}$ are considered coordinates on $\mathbb{T}^{n}$ then these constraints can be regarded as spatially flat and homogenous solutions of the Einstein equations that reside on the intersection
$\mathbb{I}=\mathbb{HP}_{n}\bigcap\mathbb{HS}_{n}$ of a hypersphere $\mathbb{HS}_{n}$ with a fixed hyper plane $\mathbb{HP}_{n}$. The Ricci curvature scalar $\bm{\mathrm{R}}$ then vanishes at all times along $\mathbb{I}$ so that $\bm{\mathrm{R}}=0$ or $\bm{\mathrm{R}}_{AB}=0$, as expected for vacuum solutions of the Einstein equations.

For ${p}_{i}>0$ for some $i$ then $a_{i}(t)\rightarrow \infty$ as $t\rightarrow\infty$ so the solution grows or 'rolls out'. For ${p}_{i}<0$ then $a_{i}(t)\rightarrow 0$ and the solution collapses or becomes singular in a finite time. These are the rolling radii solutions and it is possible for regions of this universe to expand while other regions collapse. For example, for the Bianchi-I triplet $\mathfrak{B}=(-1/3,2/3,2/3)$, the rolling radii are
\begin{equation}
a_{1}(t)=a_{1}^{E}|t|^{-1/3}\equiv a_{1}(0)|t|^{-1/3}
\end{equation}
\begin{equation}
a_{2}(t)=a_{1}^{E}|t|^{2/3}\equiv a_{2}(0)|t|^{2/3}
\end{equation}
\begin{equation}
a_{3}(t)=a_{1}^{E}|t|^{2/3}\equiv a_{3}(0)|t|^{2/3}
\end{equation}
\begin{defn}
The Kretchsmann scalar invariant, and the expansion and shear associated with the Cauchy hypersurfaces $\Sigma_{t}=const.$ can be defined as follows [58]:
\begin{enumerate}
\item The Kretchsmann scalar for this cosmology is $\bm{\mathrm{K}}=\bm{\mathrm{Riem}}_{ABCD}\bm{\mathrm{Riem}}^{ABCD} \sim t^{-4}$ so that the dynamic Kasner solutions have Big-Bang singularities along the past boundary $\Sigma_{t}=0$. In terms of the moduli
\begin{align}
&\bm{\mathrm{K}}(t)=4\sum_{i=1}^{n}\partial_{tt}\psi_{i}(t)+
4\sum_{i=1}^{n}\partial_{t}\psi_{i}(t)\partial_{t}\psi_{i}+2\sum_{i=1}^{n}\sum_{j}^{n}\left(
\partial_{t}\psi_{i}(t)\partial_{t}\psi_{j}(t)\right)^{2}
\nonumber\\&=\sum_{i=1}^{n}-p_{i}^{2}t^{-2}+\sum_{i=1}^{n }p_{i}^{2}t^{-2}+\sum_{i}^{n}\sum_{j}^{n}p_{i}p_{j}t^{-4}\sim t^{-4}
\end{align}
so that $\bm{\mathrm{K}}(0)=\bm{\mathrm{Riem}}_{ABCD}\bm{\mathrm{Riem}}^{ABCD}=\infty$.
\item The expansion $\bm{\chi}(t)$ associated with the Cauchy surface $\Sigma_{t}$ is
\begin{equation}
\bm{\chi}(t)=\sum_{i=1}^{n}\bigg|\partial_{t}\psi_{i}(t)\partial_{t}\psi_{i}(t)\bigg|
\end{equation}
\item The shear is
\begin{align}
&\bm{\mathfrak{S}}^{2}(t)=\sum_{i}\sum_{j}\bigg|\partial_{t}\psi_{i}(t)-\partial_{t}\psi_{j}(t)|^{2}\bigg|\nonumber\\&
=\sum_{i}\sum_{j}\bigg|\partial_{t}\psi_{i}(t)\partial_{t}\psi_{i}(t)-2\partial_{t}\psi_{i}(t)
\partial_{t}\psi_{j}(t)+\partial_{t}\psi_{j}(t)\partial_{t}\psi_{j}(t)\bigg|
\end{align}
\end{enumerate}
\end{defn}
The following defines an 'eternal' static Kasner torus universe for all $t>0$. This is essentially the (trivial) equilibrium or static solution defined by a set of fixed points $a_{i}(0)=a_{i}^{E}$ for some initial time $T>0$ since $\bm{\mathrm{K}}=\infty$.
\begin{defn}
The toroidal spacetime $\mathbb{M}^{n+1}=\mathbb{T}^{n}\times \mathbb{R}^{+}$ is an eternal static 'Kasner universe' if the following hold:
\begin{enumerate}
\item  The initial data $\mathfrak{D}=[t=0,\Sigma_{o}=0,g_{ij}(0),k_{ij}(0),\psi_{i}(0)=\psi_{i}^{E},
    a_{i}(t)=a_{i}^{E}$ with $g_{oo}=-1,g_{io}=0$ with constraints $\bm{\mathrm{R}}_{oo}=0$ and $\bm{\mathrm{R}}_{io}=0$, and there is no development of the data for any $t>0.$
\item The Einstein vacuum equations $\bm{\mathrm{Ric}}_{AB}=0$ exist on $\mathbb{M}^{n+1}$ in the form
$\mathbf{H}_{n}\psi_{i}(t)=0$ for all $t\in\mathbb{R}^{+}$, and $\mathbf{D}_{n}\psi_{i}^{E}=0$.
\item For all $t\in\mathbb{R}^{+}$, the first variation equation vanishes so that $\partial_{t}\bm{g}_{ij}(t)=2k_{ij}(t)=0$
\item For all $t>T$, the toroidal radii are $a_{i}(t)=a_{i}^{E}=\exp(\psi_{i}^{E})$ corresponding to the set or static moduli $\psi_{i}^{E}$ so that
\begin{align}
&\mathbf{H}_{n}\psi_{i}^{E}=0
\\&
\mathbf{D}_{n}a_{i}^{E}=0
\end{align}
\item For all $t\ge T$ the initially (static) n-metric is
\begin{equation}
ds^{2}=\sum_{i=1}\sum_{j=1}g_{ij}dX^{i}\otimes dX^{j}=\sum_{i=1}^{n}\sum_{j=1}^{n}\delta_{ij}
\exp(2\psi_{i}^{E})dX^{i}\otimes dX^{j}\nonumber
\end{equation}
\item The spatial volume of the static hyper-toroidal Kasner universe is
\begin{equation}
\mathbf{V}_{\mathbf{g}}(t)=\prod_{i=1}^{n}\exp(\psi_{i}^{E})=\exp\left(\sum_{i=1}^{n}\psi_{i}^{E}\right)\equiv\prod_{i=1}^{n}a_{i}^{E}\nonumber
\end{equation}
If $\psi_{i}^{E}=\psi^{E}$ and $a_{i}^{E}=a^{E}$ for $i=1...n$ then the static spatial volume is
\begin{equation}
\mathbf{V}_{\mathbf{g}}(t)=\prod_{i=1}^{n}\exp(\psi_{i}^{E})=n\exp(\psi^{E})\nonumber
\end{equation}
\item The $\mathcal{L}_{(2,1)}$ norm of the static n-metric is
\begin{equation}
\|\bm{g}^{E}\|_{(2,1)}=\sum_{j=1}^{n}\left(\sum_{i=1}^{n}|g_{ij}^{E}|^{2}\right)^{1/2}\equiv
\sum_{i=1}^{n}\left(\sum_{i=1}^{n}|g_{ii}^{E}|^{2}\right)^{1/2}
\end{equation}
\end{enumerate}
\end{defn}
Note that the static radii $a_{i}^{E}$ can be arbitrarily small. One could have $a_{i}^{E}\sim L_{p}$, where $L_{p}$ is the Planck length. In subsequent sections, rigorous stability criteria are developed for both deterministic and stochastic perturbations of this initially static 'micro-universe'.
\begin{lem}
Suppose now, the cosmological constant $\Lambda$ is not zero, then the classical Einstein-Hilbert action is
\begin{equation}
S=\frac{1}{\kappa} \int_{\mathbb{M}^{n+1}}d^{n+1}x(-\bm{g}_{n+1})^{1/2}(\bm{\mathrm{R}}-2\Lambda)
\end{equation}
where $\kappa=1/16\pi G_{N}$ and $G_{N}$ is the Newton constant. On $\mathbb{M}^{n+1}=\mathbb{T}^{n}\times\mathbb{R}^{+}$, this gives the Einstein equations as the sets of n-dimensional inhomogeneous nonlinear ordinary differential equations for $\psi_{i}(t)$ and $a_{i}(t)$ as
\begin{equation}
\frac{1}{2}\bm{\mathrm{R}}\equiv  \mathbf{H}_{n}\psi_{i}(t)=\sum_{i=1}^{n}\partial_{tt}{\psi}_{i}(t)+
\frac{1}{2}\sum_{i=1}^{n}\partial_{t}{\psi}_{i}(t)\partial_{t}{\psi}_{i}(t)+\frac{1}{2}
\sum_{i=1}^{n}\sum_{j=1}^{n}\partial_{t}\psi_{i}(t)\partial_{t}{\psi}_{i}(t)=\frac{\Lambda(1+n)}{(1-n)}\equiv \lambda
\end{equation}
\begin{equation}
\frac{1}{2}\bm{\mathrm{R}}\equiv \mathbf{D}_{n}a_{i}(t)=\sum_{i=1}^{n}\frac{\partial_{tt}a_{i}(t)}{a_{i}(t)}-
\frac{1}{2}\sum_{i=1}^{n}\frac{\partial_{t}a_{i}(t)\partial_{t}a_{i}(t)}{a_{i}(t)a_{i}(t)}+
\frac{1}{2}\sum_{i=1}^{n}\sum_{j=1}^{n}\frac{\partial_{t}{a}_{i}(t)\partial_{t}{a}_{j}(t)}{a_{i}(t)a_{j}(t)}=\frac{\Lambda(1+n)}{(1-n)}\equiv\lambda
\end{equation}
These are essentially of the form (2.4) and (2.5).
\end{lem}
\begin{proof}
The variation of the action $\delta S=0$ gives the Einstein field equations coupled to the cosmological constant
\begin{equation}
\bm{\mathrm{Ric}}_{AB}-\frac{1}{2}\bm{g}_{AB}\bm{\mathrm{R}}+\bm{g}_{AB}\Lambda=0
\end{equation}
Multiplying by $\bm{g}^{AB}$ and using $\bm{g}_{AB}\bm{g}^{AB}=(n+1)$ then
\begin{align}
&\bm{g}^{AB}\bm{\mathrm{Ric}}_{AB}-\frac{1}{2}\bm{g}^{AB}\bm{g}_{AB}\bm{\mathrm{R}}+\bm{g}^{AB}\bm{g}_{AB}
\Lambda\nonumber\\&=\bm{\mathrm{R}}-\frac{1}{2}(n+1)\bm{\mathrm{R}}-\bm{\Lambda}(n+1)=0
\end{align}
so that
\begin{equation}
\frac{1}{2}\bm{\mathrm{R}}=\Lambda\frac{(1+n)}{(1-n)}
\end{equation}
Using (3.28) and (3.29), then gives the equivalent sets of inhomogeneous n-dimensional ordinary nonlinear differential equations
\begin{equation}
\frac{1}{2}\bm{\mathrm{R}}\equiv \mathbf{H}_{n}\psi_{i}(t)=\Lambda\frac{(1+n)}{(1-n)}\equiv \frac{1}{2}\lambda
\end{equation}
\begin{equation}
\frac{1}{2}\bm{\mathrm{R}}\equiv \mathbf{D}_{n}a_{i}(t)=\Lambda\frac{(1+n)}{(1-n)}\equiv
\frac{1}{2}\lambda
\end{equation}
\end{proof}
For these equations, there are no equilibrium or static solutions of the form $\psi_{i}(t)=\psi_{i}^{E}$ and $a_{i}(t)=a_{i}^{E}=\exp(\psi_{i}^{E}) $ and so any solutions are necessarily dynamical. A cosmological constant term is then expected to drive an expansion or collapse of the toroidal Kasner universe.
\begin{lem}
Let $t\in\mathbb{R}^{+}$ with initial data $\mathfrak{D}=(\psi_{i}(0)\equiv\psi_{i}^{E},a_{i}(0)\equiv a_{i}^{E}) $ and the conditions of Definition 3.10. Then the solutions of the Einstein equations with a cosmological constant term are
\begin{equation}
\mathbf{H}_{n}\psi_{i}(t)=\lambda
\end{equation}
\begin{equation}
\mathbf{D}_{n}a_{i}(t)=\lambda
\end{equation}
are of the general form
\begin{equation}
\psi_{i}^{(+)}(t)=\psi_{i}^{E}+q_{i}(n,\lambda)_{i})t)=\psi_{i}(0)+q_{i}(n,\lambda)t
\end{equation}
\begin{equation}
a_{i}^{(+)}(t)=a_{i}^{E}\exp(q_{i}(n,\lambda_{i})t)=a_{i}^{E}\exp(q_{i}(n,\lambda)t)
\end{equation}
where $q_{i}(n,\lambda)$ are some constant functions of $n$ and $\lambda$, where for each $i=1...n$, $q_{i}\in\mathbb{R}^{+}$
provided that
\begin{equation}
\frac{1}{2}\sum_{i=1}^{n}q_{i}(n,\lambda)q_{j}(n,\lambda)+\frac{1}{2}\sum_{i=1}^{n}\sum_{j=1}^{n}q_{i}(n,\lambda)q_{j}(n,\lambda)=\lambda
\end{equation}
If $q_{i}=q_{j}=q=const.$ for all $i=1...n$ then $q=\pm(\lambda/n)^{1/2}$. The full solutions are then
\begin{align}
&\psi_{i}^{(+)}(t)=\psi_{i}^{E}+(\lambda/n)^{1/2}t\equiv\psi_{i}^{E}+(\lambda/n)^{1/2}t
\\& a_{i}^{(+)}(t)=a_{i}^{E}\exp((\lambda/n)^{1/2}t)\equiv a_{i}(0)\exp((\lambda/n)^{1/2}t)
\end{align}
for an expanding universe driven by a cosmological constant so that $\mathbf{H}_{n}\psi_{i}^{(+)}(t)=\lambda$ and $\mathbf{D}_{n}a_{i}^{(+)}(t)=\lambda$
\begin{align}
&\psi_{i}^{(-)}(t)=\psi_{i}^{E}-(\lambda/n)^{1/2}t\equiv\psi_{i}(0)-(\lambda/n)^{1/2}t
\\&a_{i}^{(-)}(t)=a_{i}^{E}\exp(-(\lambda/n)^{1/2} t)\equiv a_{i}(0)\exp(-(\lambda/n)^{1/2} t)\end{align}
for a collapsing universe, such that $\mathbf{H}_{n}\psi_{i}^{(-)}(t)=\lambda$ and $\mathbf{D}_{n}a_{i}^{(-)}(t)=\lambda$.
\end{lem}
\begin{proof}
The derivatives of $\psi_{i}(t)$ are simply $\partial_{t}\psi_{i}(t)=q_{i}$ and $\partial_{tt}\psi_{i}(t)=0$  so that (3.70) becomes
\begin{align}
&\mathbf{H}_{n}\psi_{i}(t)=\frac{1}{2}\sum_{i=1}^{n}q_{i}(n,\lambda)q_{i}(n,\lambda)+\frac{1}{2}\sum_{i=1}^{n}\sum_{j=1}^{n}
q_{i}(n,\lambda)q_{j}(n,\lambda)\nonumber\\&
\equiv\frac{1}{2}\sum_{i=1}^{n}m_{ii}(n,\lambda)+\frac{1}{2}\sum_{i=1}^{n}\sum_{j=1}^{n}m_{ij}(n,\lambda)=\lambda
\end{align}
If $q_{i}=q_{j}=q$ for $i,j=1...n$ then one can choose $M_{ij}=\delta_{ij}q^{2}$ and $M_{ii}=\delta_{ii}q^{2}$, each having n nonzero terms giving
\begin{equation}
\mathbf{H}_{n}\psi_{i}(t)=\frac{1}{2}nq^{2}+\frac{1}{2}nq^{2}=nq^{2}=\lambda
\end{equation}
so that $q=\pm(\lambda/n)^{1/2}$. The same result also follows from the Einstein equations $\mathbf{D}_{n}a_{i}(t)=\lambda$ so that (3.71) becomes
\begin{align}
&\mathbf{D}_{n}a_{i}(t)=\sum_{i=1}^{n}\frac{|a_{i}(0)|^{2}|q_{i}(n,\lambda)|^{2}|X_{i}(t)|^{2}}
{|a_{i}(0)|^{2}|X_{i}(t)|^{2}}-\frac{1}{2}\sum_{i=1}^{n}\frac{|a_{i}(0)|^{2}|q_{i}(n,\lambda)|^{2}|
X_{i}(t)|^{2}}{|a_{i}(0)|^{2}||X_{i}(t)|^{2}}
\nonumber\\&+\frac{1}{2}\sum_{i=1}^{n}\sum_{j=1}^{n}\frac{|a_{i}(0)a_{j}(0)q_{i}(n,\lambda)q_{j}(n,\lambda)
X_{i}(t)X_{j}(t)}{|a_{i}(0)a_{j}(0)||X_{i}(t)|
X_{j}(t)}=\lambda
\end{align}
where $X_{i}(t)=\exp(q_{i}(n,\lambda)t)$. Cancelling terms
\begin{align}
\mathbf{D}_{n}a_{i}(t)&=\sum_{i=1}^{n}|q_{i}(n,\lambda)|^{2}-
\frac{1}{2}\sum_{i=1}^{n}|q_{i}(n,\lambda)|^{2}+\frac{1}{2}\sum_{i=1}^{n}\sum_{j=1}^{n}
q_{i}(n,\lambda)q_{j}(n,\lambda)\nonumber\\
&\equiv \sum_{i=1}^{n}m_{ij}(n,\lambda)-\frac{1}{2}\sum_{i=1}^{n}\sum_{j=1}^{n}m_{ij}(n,\lambda)=\lambda
\end{align}
Again, if $q_{i}=q_{j}=q$ for $i=1...n$ then the matrices $m_{ij}$ and $m_{ii}$ each have n terms and one can choose $m_{ij}=\delta_{ij}q^{2}$ and $m_{ii}=\delta_{ii}q^{2}$. Equation (3.82) reduces to  $n|q(n,\lambda)|^{2}=\lambda$ so that $
q(n,\lambda)=\pm\left(\frac{\lambda}{n}\right)^{1/2} $ as before.
\end{proof}
\section{'Short-pulse' and continuous perturbations of static solutions: stability criteria}
In this section, stability criteria of the static Kasner-type hypertorus universe or 'vacuum bubble' are considered in relation to 'short-pulse' Gaussian deterministic perturbations and also continuous 'step' perturbations of the static or constant moduli fields $\psi_{i}^{E}$. In Section 5, stability criteria for random perturbations or 'noise' are developed and compared. But in each case, the Einstein equations will now be interpreted as a nonlinear multi-dimensional dynamical system of ordinary differential equations with initial data, which are then subject to such perturbations; indeed, from a purely mathematical perspective, the form of these nonlinear ODEs could be considered independently of any general relativistic considerations. On $\mathbb{M}^{n+1}=\mathbb{T}^{n}\times\mathbb{R}^{+}$, the Einstein vacuum equations reduce to the general ODES are $ \mathbf{H}_{n}\psi_{i}^{E}=0$  and $\mathbf{D}_{n}a_{i}(t)=0$. If the static solutions are subject to small perturbation $f_{i}(t)$ for some $f_{i}:\mathbb{R}^{+}\rightarrow\mathbb{R}^{+}$, where $|f_{i}(t)|\ll 1$ about these static equilibrium points then
\begin{equation}
\overline{\psi_{i}(t)}=\psi_{E}+f_{i}(t)
\end{equation}
In general one can proceed to study stability via a nonlinear stability analysis. Rather than linearize the equations, the effect of deterministic 'short-pulse' perturbations on the fully nonlinear equations will be considered. The nonlinearity should be retained since it is a crucial feature of general relativity and gravitational systems.

Consider first, a set of delta-function 'impulse' perturbations of the form
\begin{equation}
\overline{\psi_{i}(t)}=\psi_{i}^{E}+\zeta\int_{0}^{t}\delta_{i}(\tau)d\tau
\end{equation}
or $ \overline{\psi_{i}(t)}=\psi_{E}^{i}+\zeta\delta_{i}(t)$, with $\zeta>0$. This is convergent since $ \lim_{t\rightarrow\infty}
\|\int_{0}^{t}\bm{\delta}_{i}(\tau)d\tau\|=1$ but the derivative $\partial_{t}\delta_{i}(t)$ does not exist. However, the delta functions can be 'smeared out' into very narrow sharply peaked functions such as Gaussians or power-law distributions represented as
$\mathcal{U}_{i}(t,\vartheta_{i})$ with finite widths $\vartheta_{i}$ such that $\lim_{\|\bm{\vartheta}\|\rightarrow 0}\mathcal{U}_{i}(t,\vartheta_{i})=\delta_{i}(t)$. The following perturbations can be considered
\begin{equation}
\overline{\psi_{i}(t)}=\psi_{E}^{i}+\zeta\int_{0}^{t}\mathcal{U}_{i}(\tau,\vartheta_{i})d\tau
=\psi_{i}^{E}+\zeta\mathcal{H}(\tau,\vartheta_{i})
\end{equation}
or
\begin{equation}
\overline{\psi_{i}(t)}=\psi_{E}^{i}+\zeta\mathcal{U}_{i}(t,\vartheta_{i})
\end{equation}
where $\zeta>0$. In this paper, we will use the integral form (4.3).
\subsection{General stability criteria}
In this subsection it will be shown that initially static toroidal universes or "vacuum bubbles" are stable to these types of narrow and sharply peaked deterministic perturbations.
\begin{prop}
Let $\mathbb{M}^{n+1}=\mathbb{T}^{n+1}\times\mathbb{R}^{+}$ be a globally hyperbolic spacetime where $\mathbb{T}^{n}$ is the n-torus. The following hold:
\begin{enumerate}
\item  The initial data $\mathfrak{D}=[t=0,\Sigma_{o}=0,\bm{g}_{ij}(0),k_{ij}(0),\psi_{i}(0)=\psi_{i}^{E},a_{i}(0)=
    a_{i}^{E}$ with $\bm{g}_{oo}=-1,\bm{g}_{io}]=0$ with constraints $\mathbf{R}_{oo}=0$ and $\mathbf{Ric}_{io}=0$, and conditions (2),(3) and (4) of Definition 3.8.
\item A set of smooth functions $(\mathcal{U}_{i}(t,\vartheta_{i}))_{i=1}$ span $\mathbb{R}^{m}$ where $\vartheta_{i}$ are the widths such that $\mathcal{U}_{i}(t,\vartheta_{i})=0$ for $t\gg |\bm{\vartheta}|$ and $\lim_{t\uparrow \infty}|\!|\bm{\mathcal{U}}(t,\bm{\vartheta})|\!|=0$.(E.g., a set of Gaussians with widths $\lbrace\bm{\vartheta}\rbrace $). These can be represented as a vector $\bm{\mathcal{U}}(t,\bm{\vartheta})=\mathcal{U}_{1}(t,\vartheta_{1},...,\mathcal{U}_{n}(t,\vartheta_{n})$ with norms $\|\bm{\mathcal{U}}(t,\bm{\vartheta})\|\equiv\|\mathcal{U}_{i}(\vartheta_{i})\|$
\item The derivative $\partial_{t}\mathcal{U}_{i}(t,\vartheta_{i})$ exists and $\partial_{t}\mathcal{U}_{i}(t,\vartheta_{i})\rightarrow 0 $ for $t\gg|\bm{\vartheta}_{i}|$.
\end{enumerate}
The initially static torus is isotropic if $ \psi_{1}^{E}=\psi_{2}^{E}=...=\psi_{n}^{E}$.
General deterministic perturbations of the modulus functions $\psi_{i}^{E}=\psi^{E}$ can be of the form
\begin{equation}
\overline{\psi_{i}(t)}=\psi_{i}^{E}+\int_{0}^{t}\mathcal{U}_{i}(\tau,\vartheta_{i})d\tau=
\psi_{i}^{E}+\mathcal{H}_{i}(t,\vartheta_{i})
\end{equation}
If at least any two of the set $\widehat{\mathcal{U}}_{i}(t,\vartheta_{i})$ are different, then the perturbations are anisotropic so that $\mathcal{U}_{i}(t,\vartheta_{i})\ne
\mathcal{U}_{j}(t,\vartheta_{j})$ for any $i\ne j$. For example, if the $ \mathcal{U}_{i}(t,\vartheta_{i})$ are a set of Gaussians such that $\mathcal{U}_{i}(t,\vartheta_{i})=\mathcal{A}_{i}
\exp(-t^{2}/2\vartheta_{i})$ then $\mathcal{U}_{i}(t,\vartheta_{i})\ne \mathcal{U}_{j}(t,\vartheta_{j})$ if $\mathcal{A}_{i}\ne \mathcal{A}_{j}$ and/or $\vartheta_{i}\ne\vartheta_{j}$. If the perturbations are isotropic then $\mathcal{U}_{i}(t,\vartheta_{i})=\mathcal{U}(t,\vartheta)$ for $i=1...n$.
\end{prop}
For a 3-torus $\mathbb{T}^{3}$ for example, the anisotropically perturbed components are
\begin{equation}
\overline{\psi_{1}(t)}=\psi_{2}^{E}+\int_{0}^{t}{\mathcal{U}}_{1}(\tau,\vartheta_{1})d\tau\equiv
\psi_{1}^{E}+\mathcal{H}_{1}(t,\vartheta_{1})
\end{equation}
\begin{equation}
\overline{\psi_{2}(t)}=\psi_{1}^{E}+\int_{0}^{t}{\mathcal{U}}_{2}(\tau,\vartheta_{2})d\tau\equiv
\psi_{2}^{E}+\mathcal{H}_{2}(t,\vartheta_{2})
\end{equation}
\begin{equation}
\overline{\psi_{3}(t)}=\psi_{1}^{E}+\int_{0}^{t}{\mathcal{U}}_{1}(\tau,\vartheta_{3})d\tau\equiv
\psi_{3}^{E}+\mathcal{H}_{3}(\tau,\vartheta_{3})
\end{equation}
The initially static radii $a_{i}^{E}=a_{E}$ of the n-torus are then anisotropically perturbed as
\begin{equation}
\overline{a_{i}(t)}=a_{E}^{i}\exp\left(\int_{0}^{t}
\mathcal{U}_{i}(\tau,\vartheta_{i})d\tau\right)=a_{i}^{E}
\exp(\mathcal{H}_{i}(t,\vartheta_{i}))=a_{i}^{E}\mathcal{X}_{i}(t)
\end{equation}
And for $\mathbb{T}^{3}$, the anisotropically perturbed radii are
\begin{equation}
\overline{a_{1}(t)}=a_{1}^{E}\exp\left(\int_{0}^{t}{\mathcal{U}}_{1}(\tau,\vartheta_{1})d\tau\right)\equiv a_{1}^{E}\exp(\mathcal{H}_{1}(t,\vartheta_{1}))\equiv a_{1}^{E}\mathcal{X}_{1}(t)
\end{equation}
\begin{equation}
\overline{a_{2}(t)}=a_{2}^{E}\exp\left(\int_{0}^{t}{\mathcal{U}}_{2}(\tau,\vartheta_{2})d\tau\right)\equiv a_{2}^{E}\exp(\mathcal{H}_{2}(t,\vartheta_{2}))\equiv a_{2}^{E}\mathcal{X}_{2}(t)
\end{equation}
\begin{equation}
\overline{a_{3}(t)}=a_{3}^{E}\exp\left(\int_{0}^{t}{\mathcal{U}}_{3}(\tau,\vartheta_{3})d\tau\right)
\equiv a_{3}^{E}\exp(\mathcal{H}_{3}(t,\vartheta_{3}))\equiv a_{3}^{E}\mathcal{X}_{3}(t)
\end{equation}
For isotropic perturbations with ${\mathcal{U}}_{i}(t,\vartheta_{i})={\mathcal{U}}(t,\vartheta)$ for $i=1...n$ then
\begin{equation}
\overline{\psi_{i}(t)}=\psi_{i}^{E}+\int_{0}^{t}{\mathcal{U}}(\tau,\vartheta)d\tau\equiv\psi_{i}^{E}+{\mathcal{H}}(t,\vartheta)
\end{equation}
so the isotropically perturbed modulus functions are $
\overline{\psi_{1}(t)}=\psi^{E}+\int_{0}^{t}\mathcal{U}(\tau,\vartheta)d\tau
\equiv\psi_{1}^{E}+\mathcal{H}(t,\vartheta)$ and so on.
\begin{prop}
Given the perturbations of the static solutions
\begin{align}
&\overline{\psi_{i}(t)}=\psi_{i}^{E}+\int_{0}^{t}\mathcal{U}_{i}(\tau,\vartheta_{i})d\tau
\\&\overline{a_{i}(t)}=a_{i}^{E}\exp\left(\int_{0}^{t}\mathcal{U}_{i}(\tau,\vartheta_{i})d\tau\right)
\end{align}
with $\|\bm{\vartheta}\|\ll 1$ and the convergence
\begin{align}
&\mathcal{Y}_{i}=\lim_{t\rightarrow\infty}\int_{0}^{t}\mathcal{U}_{i}(\tau,\vartheta_{i})d\tau=
\lim_{t\ge\|\bm{\vartheta}\|}\int_{0}^{t}\mathcal{U}_{i}(\tau,\vartheta_{i})d\tau<\infty
\\&\|\bm{\mathcal{Y}}\|=\lim_{t\rightarrow\infty}\left\|\int_{0}^{t}\bm{\mathcal{U}}(\tau,\bm{\vartheta})d\tau\right\|=
\lim_{t\ge\|\bm{\vartheta}\|}\left\|\int_{0}^{t}\bm{\mathcal{U}}(\tau,\bm{\vartheta})
d\tau\right\|<\infty
\end{align}
Then for $t\gg|\bm{\vartheta}\|$
\begin{equation}
\overline{\psi_{i}(t)}\longrightarrow \psi_{i}^{E}+\mathcal{Y}_{i}\equiv \psi_{i}^{E*}
\end{equation}
\begin{equation}
\overline{a_{i}(t)}\longrightarrow a_{i}^{E}\exp(\mathcal{Y}_{i})\equiv {a}_{i}^{E*}
\end{equation}
where the points $\psi_{i}^{E*}$ and $a_{i}^{E*}$ are 'attractors'. For the perturbed radii of the 3-torus $\mathbb{T}^{3}$ for example
\begin{equation}
\overline{a_{1}(t)}\longrightarrow a_{i}^{E}\exp(\mathcal{Y}_{1})\equiv {a}_{1}^{E*}\nonumber
\end{equation}
\begin{equation}
\overline{a_{2}(t)}\longrightarrow a_{i}^{E}\exp(\mathcal{Y}_{2})\equiv {a}_{2}^{E*}\nonumber
\end{equation}
\begin{equation}
\overline{a_{3}(t)}\longrightarrow a_{i}^{E}\exp(\mathcal{Y}_{3})\equiv {a}_{3}^{E*}\nonumber
\end{equation}
Given the deterministic perturbations $\overline{\psi(t)}=\psi_{i}^{E}+\int_{0}^{t}\mathcal{U}_{i}(\tau;\vartheta_{i})d\tau$ and $\overline{a_{i}(t)}=a_{i}^{E}\exp(\int_{0}^{t}\mathcal{U}_{i}(\tau)d\tau)$, for initial data $\psi_{i}(0)=\psi_{i}^{E}$ the $\mathcal{L}_{2}$ norms are estimated as
\begin{align}
&\|\overline{\bm{\psi}(t)}-\bm{\psi}^{E}\|=\left(\sum_{i=1}^{n}\left |{\psi}_{i}(t)-\psi^{E}\right|^{2}\right)^{1/2}\nonumber\\&
=\left(\sum_{i=1}^{n}\left|\int_{0}^{t}\mathcal{U}_{i}(\tau,\vartheta_{i})d\tau\right|^{2}\right)^{1/2}
=n^{1/2}\left|\int_{0}^{t}\mathcal{U}(\tau,\vartheta)d\tau\right|
\end{align}
and
\begin{align}
\|\overline{\bm{a}(t)}-\bm{a}^{E}\|&\le \|\overline{\bm{a}_{i}(t)}\|-\|\bm{a}^{E}\|=\left\|a_{i}^{E}\exp\bigg(\int_{0}^{t}
\bm{\mathcal{U}}(t,\vartheta)d\tau\bigg)\right\|
-\|\bm{a}^{E}\|\nonumber\\&
=\left(\sum_{i=1}^{n}\left|a_{i}^{E}\exp\left(\int_{0}^{t}
\mathcal{U}_{i}(\tau,\vartheta_{i})d\tau\right)\right|^{2}\right)^{1/2}
-\left(\sum_{i=1}^{n}|a_{i}^{E}|^{2}\right)^{1/2}\nonumber\\&<\left(\sum_{i=1}^{n}
\left|a_{i}^{E}\exp\left(\int_{0}^{t}\mathcal{U}_{i}(\tau,\vartheta_{i})d\tau\right)\right|^{2}\right)^{1/2}\nonumber\\&
=\left(\sum_{i=1}^{n}|a_{i}^{E}|^{2}\exp\left(2\int_{0}^{t}
\mathcal{U}_{i}(\tau,\vartheta_{i})d\tau\right)\right)^{1/2}
\nonumber\\&=\left(n(|a_{i}^{E}|^{2}\exp\left(2\int_{0}^{t}\mathcal{U}(\tau,\vartheta)d\tau\right)\right)^{1/2}
\nonumber\\&=n^{1/2} a_{i}^{E}\exp\left(\int_{0}^{t}\widehat{\mathcal{U}}(\tau;\vartheta)d\tau\right)
\end{align}
if $\mathcal{U}_{i}(t,\vartheta_{i})<\mathcal{U}(t,\vartheta)$ for $i=1...n$. The asymptotic behavior is then estimated as
\begin{equation}
\lim_{t\uparrow\infty}\|\overline{\bm{a}(t)}-\bm{a}^{E}\|<n^{1/2} a_{i}^{E}\lim_{t\uparrow
\infty}\exp\left(\int_{0}^{t}\mathcal{U}(\tau,\vartheta)d\tau\right)
\end{equation}
\end{prop}
The perturbed norms $\|\overline{\bm{a}(t)}-\bm{a}^{E}\|$ can also be expressed as a perturbation series.
\begin{prop}
\begin{align}
\|\overline{\bm{a}_{i}(t)}-\bm{a}^{E}\|& \le \|\overline{\bm{a}_{i}(t)}\|-\|\bm{a}^{E}\|
=\bigg\|a_{i}^{E}\exp\left(\int_{0}^{t}\bm{\mathcal{U}}(t,\vartheta)d\tau
\right)\bigg\|-\|\bm{a}^{E}\|\nonumber\\&<\left(\sum_{i=1}^{n}|a_{i}^{E}|^{2}\exp\left(2\int_{0}^{t}
\mathcal{U}_{i}(\tau,\vartheta_{i})d\tau\right)\right)^{1/2}\nonumber\\&=
\left(\sum_{i=1}^{n}|a_{i}^{E}|^{2}\sum_{m=0}^{\infty}\frac{1}{m!}\left|2\int_{0}^{t}
\mathcal{U}_{i}(\tau,\vartheta_{i})d\tau\right|^{m}\right)^{1/2}\nonumber\\
&< \left(\sum_{i=1}^{n}|a_{i}^{E}|^{2}\sum_{m=0}^{\infty}\frac{1}{m!}\left|\int_{0}^{t}
\mathcal{U}_{i}(\tau,\vartheta_{i})d\tau\right|^{m}\right)\nonumber\\&
=n|a_{i}^{E}|^{2}\sum_{m=0}^{\infty}\frac{1}{m!}\left|\int_{0}^{t}
\mathcal{U}(\tau,\vartheta)d\tau\right|^{m}\nonumber\\&
=n|a_{i}^{E}|^{2}\bigg(1+\bigg|\int_{0}^{t}\mathcal{U}(\tau)d\tau\bigg|+
\frac{1}{2}\bigg|\int_{0}^{t}\mathcal{U}(\tau,\vartheta)d\tau\bigg|^{2}+...\nonumber\\&
+\bigg|\int_{0}^{t}\mathcal{U}(\tau,\vartheta)d\tau+...\bigg|^{m}+...\bigg)\nonumber\\
&=n|a_{i}^{E}|^{2}(1+|\mathcal{H}(t,\vartheta)|+\frac{1}{2}|\mathcal{H}(t,\vartheta)|^{2}+...+
|\mathcal{H}(t,\vartheta)|^{m}+...)
\end{align}
If the integrals $|\mathcal{H}(t,\vartheta)|^{m}=|\int_{0}^{t}\mathcal{U}(\tau,\vartheta)d\tau|^{m}$ converge for all $m\in\mathbb{Z}$, then the perturbation series converges.
\end{prop}
For very sharply peaked functions with $\|\bm{\vartheta}\|\ll 1$ the perturbed radii should converge very quickly to the attractor points. The existence of attractors is also linked with the Lyapunov stability of the system.
\begin{defn}
Given the equilibrium points $a_{i}^{E}$, let $\|\bm{a}^{E}\|$ be contained within an Euclidean ball $\mathbb{B}(L)\subset\mathbb{R}^{n}$ of radius $|L|$, then the system is asymptotically stable if $\|\overline{\bm{a}(t)}-\bm{a}^{E}\|\in\mathbb{B}(L)$ for all $t>0$ where or $\lim_{t\uparrow\infty}\|\overline{\bm{a}(t)}-\bm{a}^{E}\|\in\mathbb{B}(L)$. The system if Lyupunov stable if $\exists \mathbb{B}(L')$ of radius $L'$ with $L'>L$ such that $\lim_{t\uparrow\infty}\|\overline{\bm{a}(t)}-\bm{a}^{E}\|\in\mathbb{B}(L')$. Hence, all future states are trapped in $\mathbb{B}(L')$.
Then:
\begin{enumerate}
\item The static system is stable if $\exists$ any finite ball $\mathbb{B}(R)$ that can contain $\lim_{t\uparrow\infty}\|\overline{\bm{a}(t)}-\bm{a}^{E}\|$ so that $\lim_{t\uparrow\infty}\|\overline{\bm{a}(t)}-\bm{a}^{E}\|=\infty$
\begin{equation}
0<\lim_{t\uparrow\infty}\|\overline{\bm{a}(t)}-\bm{a}^{E}\|=n^{1/2} a_{i}^{E}\lim_{t\uparrow
\infty}\exp\left(\int_{0}^{t}\mathcal{U}(\tau,\vartheta)d\tau\right)<\infty
\end{equation}
and $a_{i}(t)\longrightarrow a_{i}^{E}\exp[\mathcal{A}_{i}]\equiv a_{i}^{*}$ for $i=1...n$. \item The static system is asymptotically perturbatively stable if
\begin{equation}
\lim_{t\uparrow\infty}\|\overline{\bm{a}(t)}-\bm{a}^{E}\|=n^{1/2} a_{i}^{E}\lim_{t\uparrow
\infty}\exp\left(\int_{0}^{t}\mathcal{U}(\tau,\vartheta)d\tau\right)=0
\end{equation}
\item The static system is unstable to the perturbations if
\begin{equation}
\lim_{t\uparrow\infty}\|\overline{\bm{a}_{i}(t)}-\bm{a}^{E}\|=n^{1/2} a_{i}^{E}\lim_{t\uparrow\infty}\exp\left(\int_{0}^{t}\mathcal{U}(\tau,\vartheta)d\tau\right)=\infty
\end{equation}
\end{enumerate}
\end{defn}
For stability, the attractors will satisfy the perturbed differential equations in order to be new equilibrium stable fixed points of the system so that $\mathbf{H}_{n}\psi^{*}=0$ and $\mathbf{D}_{n}a^{*}=0$.
\begin{prop}
Let $\mathbb{M}^{n+1}=\mathbb{T}^{n+1}\times\mathbb{R}^{+}$ be a globally hyperbolic space-time where $\mathbb{T}^{n}$ is the n-torus. The following hold: the initial data $\mathfrak{D}=[t=0,\Sigma_{0}=0,\bm{g}_{ij}(0),k_{ij}(0),\psi_{i}(0)=\psi_{i}^{E},a_{i}(0)=R_{i}^{E}]$ with $\bm{g}_{oo}=-1,\bm{g}_{io}=0$ with constraints $\mathbf{Ric}_{oo}=0$ and $\mathbf{Ric}_{io}=0$. Then given the modulus perturbations $
\overline{\psi_{i}(t)}=\psi_{i}^{E}+\int_{0}^{t}\mathcal{U}_{i}(\tau,\vartheta_{i})d\tau $ or $\overline{a_{i}(t)}=a_{i}^{E}\exp\left(\int_{0}^{t}
\mathcal{U}_{i}(\tau,\vartheta_{i})d\tau\right) $, the perturbed Einstein equations are
\begin{equation}
\mathbf{H}_{n}\overline{\psi_{i}(t)}=\mathcal{S}(t,\vartheta_{i})
\end{equation}
\begin{equation}
\mathbf{D}_{n}\overline{a_{i}(t)}=\mathcal{S}(t,\vartheta_{i})
\end{equation}
where
\begin{align}
S(t,\vartheta_{i})&=\sum_{i=1}^{n}\partial_{t}\mathcal{U}_{i}(t,\vartheta)+
\frac{1}{2}\sum_{i=1}^{n}\mathcal{U}_{i}(t,\vartheta_{i})
\mathcal{U}_{i}(t,\vartheta_{j})+\frac{1}{2}\sum_{i=1}^{n}\sum_{j=1}^{n}
\mathcal{U}_{i}(t,\vartheta_{i})
\mathcal{U}_{j}(t,\vartheta_{j})\nonumber\\&
=\sum_{i=1}^{n}\partial_{t}\mathcal{U}_{i}(t,\vartheta)+
\frac{1}{2}\sum_{i=1}^{n}\mathcal{U}_{i}(t,\vartheta_{j})+\frac{1}{2}\sum_{i=1}^{n}
|\sqrt{\mathcal{U}_{i}(t,\vartheta_{i})}|^{2}\sum_{j=1}^{n}\big|\sqrt{\mathcal{U}_{j}(t,\vartheta_{j})}\big|^{2}
\nonumber\\&\equiv \big\|\sqrt{\partial_{t}\mathcal{U}_{i}(t,\vartheta)}\big\|_{L_{2}}^{2}
+\frac{1}{2}\big\|\mathcal{U}_{i}(t,\vartheta_{i})\big\|_{L_{2}}^{2}
+\frac{1}{2}\left\|\sqrt{\mathcal{U}_{i}(t,\vartheta)}\right\|_{L_{2}}^{2}
\left\|\sqrt{\mathcal{U}_{j}(t,\vartheta)}\right\|_{L_{2}}^{2}
\end{align}
If $\mathcal{U}_{i}(t,\vartheta_{i})=\mathcal{U}(t,\vartheta)$ then the perturbed Einstein equations are
\begin{equation}
\mathbf{H}_{n}\overline{\psi_{i}(t)}=n \partial_{t}\mathcal{U}(t,\vartheta)+\frac{1}{2}n
\mathcal{U}(t,\vartheta)=\mathcal{S}(t,\vartheta)
\end{equation}
\begin{equation}
\mathbf{D}_{n}\overline{a_{i}(t)}=n\partial_{t}\mathcal{U}(t,\vartheta)+\frac{1}{2}n
\mathcal{U}(t,\vartheta)=\mathcal{S}(t,\vartheta)
\end{equation}
and $\lim_{t\uparrow\infty}\mathbf{D}_{n}\overline{\psi_{i}(t)}=\lim_{t\uparrow\infty}\mathcal{S}(t,\vartheta)=0$ and $\lim_{t\uparrow\infty}\mathbf{D}_{n}\overline{a_{i}(t)}
=\lim_{t\uparrow\infty}\mathcal{S}(t,\vartheta)=0$. Also $[\mathbf{H}_{i}\psi_{i}(t)]_{t>\vartheta}=
\mathbf{D}_{n}\overline{a_{i}(t)}]_{t>\vartheta}=0$ for very sharply peaked 'short-pulse' perturbations. This is equivalent to
\begin{equation}
\sum_{i=1}^{n}\frac{\partial_{tt}a_{i}^{E*}}{a_{i}^{E*}}=\frac{1}{2}
\sum_{i=1}^{n}\frac{\partial_{t}a_{i}^{E*}\partial_{t}a_{i}^{E*}}{a_{i}^{E*}a_{j}^{E*}}+
\frac{1}{2}\sum_{i\ne j}^{n}
\frac{\partial_{t}{a}_{i}^{E*}\partial_{t}{a}_{i}^{E*}}{a_{i}^{E*}a_{i}^{*}}=0
\end{equation}
for the attractors $a_{i}^{E*}$
\begin{equation}
a_{i}^{E*}=a_{i}^{E}\lim_{t\uparrow\infty}\exp\left(\int_{0}^{t}\mathcal{U}_{i}(t,\vartheta_{i})d\tau\right)\equiv a_{i}^{E}\exp(\mathcal{A}_{i})=a_{i}^{E}\exp(\mathcal{H}_{i}(t,\vartheta_{i})
\end{equation}
\end{prop}
\begin{proof}
The perturbed Einstein vacuum equations are
\begin{align}
&\mathbf{H}_{n}\overline{\psi_{i}(t)}
=\mathbf{H}_{n}\psi_{i}^{E}+\sum_{i=1}^{n}\partial_{tt}\mathcal{H}_{i}(t,\vartheta_{i})
+\nonumber\\&\frac{1}{2}\sum_{i=1}^{n}\sum_{j=1}^{n}\partial_{t}\mathcal{H}_{i}(t,\vartheta_{i})
{\mathcal{H}}_{j}(t,\vartheta_{j})+\frac{1}{2}\sum_{i=1}^{n}
\partial_{t}{\mathcal{H}}_{i}(t,\vartheta_{i})
\partial_{t}{\mathcal{H}}_{i}(t,\vartheta_{j})
\end{align}
Since $\partial_{t}\mathcal{H}_{i}(t,\vartheta_{i})=\mathcal{U}_{i}(t,\vartheta_{i})$ and $\partial_{tt}\mathcal{H}_{i}(t,\vartheta_{i})=\partial_{t}\mathcal{U}_{i}(t,\vartheta_{i})$
\begin{align}
&\mathbf{H}_{n}\overline{\psi_{i}(t)}=\mathbf{H}_{n}\psi_{i}^{E}+\sum_{i=1}^{n}
\partial_{t}\mathcal{U}_{i}(t,\vartheta_{i})\nonumber\\&+\frac{1}{2}\sum_{i=1}^{n}\sum_{j=1}^{n}
\mathcal{U}_{i}(t,\vartheta_{i})\mathcal{U}_{j}(t,\vartheta_{j})
+\frac{1}{2}\sum_{i=1}^{n}\mathcal{U}_{i}(t,\vartheta_{i})\mathcal{U}_{i
}(t,\vartheta_{j})=\mathcal{S}(t,\vartheta_{i})
\end{align}
Similarly,
\begin{equation}
\mathbf{D}_{n}\overline{a(t)}=\sum_{i=1}^{n}\frac{\partial_{tt}\overline{a_{i}(t)}}{\overline{a_{i}(t)}}
-\frac{1}{2}\sum_{i=1}^{n}\frac{\partial_{t}\overline{a_{i}(t)}\partial_{t}\overline{a_{i}(t)}}
{\overline{a_{i}(t)}\overline{a_{j}(t)}}+
\frac{1}{2}\sum_{i=1}^{n}\sum_{j=1}^{n}\frac{\partial_{t}\overline{a_{i}(t)}\partial_{t}
\overline{a_{i}(t)}}{\overline{a_{i}(t)}\overline{a_{j}(t)}}
\end{equation}
should give the same result.
\begin{align}
\mathbf{D}_{n}\overline{a_{i}(t)}&=\sum_{i=1}^{n}\frac{a_{i}^{E}\partial_{t}\mathcal{U}(t,\vartheta_{i})
\exp[\mathcal{U}_{i}(t,\vartheta_{i})]}{a_{i}^{E}
\exp(\mathcal{U}_{i}(t,\vartheta_{i}))}+\sum_{i=1}^{n}\frac{a_{i}^{E}a_{i}^{E}
\mathcal{U}(t,\vartheta_{i})\mathcal{U}(t,\vartheta_{i})\exp(2\mathcal{H}_{i}(t,\vartheta_{i}))}{a_{i}^{
E}a_{i}^{E}\exp(2\mathcal{H}_{i}(t,\vartheta_{i})}\nonumber\\&
-\frac{1}{2}\sum_{i=1}^{n}\frac{a_{i}^{E}a_{i}^{E}\mathcal{U}(t,\vartheta_{i})
\mathcal{U}(t,\vartheta_{i})\exp(2\mathcal{H}_{i}(t,\vartheta_{i}))}
{a_{i}^{E}a_{i}^{E}\exp(2\mathcal{U}_{i}(t,\vartheta_{i}))}\nonumber\\&+\frac{1}{2}\sum_{i=1}^{n}
\frac{a_{i}^{E}a_{i}^{E}\mathcal{U}(t,\vartheta_{i})\mathcal{U}(t,\vartheta_
{i})\exp(\mathcal{H}_{i}(t,\vartheta_{i}))\exp(\mathcal{H}_{i}(t,\vartheta_{i})}
{a_{i}^{E}a_{j}^{E}\exp(\mathcal{H}_{i}(t,\vartheta_{i}))\exp(\mathcal{H}_{i}(t,\vartheta_{j}))}
\end{align}
Cancelling terms gives
\begin{align}
\mathbf{D}_{n}\overline{a_{i}(t)}&=\mathbf{D}_{n}a_{i}^{E}+\sum_{i=1}^{n}
\partial_{tt}\mathcal{U}_{i}(t,\vartheta_{i})+\frac{1}{2}\sum_{i=1}^{n}\sum_{j=1}^{n}\mathcal{U}_{i}(t,\vartheta_{i})
\mathcal{U}_{j}(t,\vartheta_{j})\nonumber\\&+\frac{1}{2}\sum_{i=1}^{n}\mathcal{U}_{i}(t,\vartheta_{i})
\mathcal{U}_{i}(t,\vartheta_{j})=\mathcal{S}(t,\vartheta_{i})
\end{align}
or
\begin{equation}
\mathbf{D}_{n}\overline{a_{i}(t)}= \sum_{i=1}^{n}\partial_{t}\mathcal{U}_{i}(t,\vartheta)
+\frac{1}{2}\big\|\mathcal{U}_{i}(t,\vartheta_{i})\big\|^{2})
+\frac{1}{2}\left\|\sqrt{\mathcal{U}_{i}(t,\vartheta)}\right\|_{L_{2}}+
\frac{1}{2}\left\|\sqrt{\mathcal{U}_{j}(t,\vartheta)}\right\|_{L_{2}}
\end{equation}
Setting $\mathcal{U}_{i}=\mathcal{U}_{j}=\mathcal{U}$ then gives (4.30) and (4.31).
\end{proof}
\begin{lem}
Given the initial data $\mathfrak{D}$, the evolution of the corresponding perturbed n-metric $\overline{\bm{g}_{ij}(t)}$ is
\begin{equation}
\overline{\bm{g}_{ij}(t)}=2\delta_{ij}\pi\exp(2\overline{\psi_{i}(t)})\exp\left(\int_{0}^{t}\mathcal{U}_{i}
(\tau,\vartheta)d\tau\right)\equiv \bm{g}_{ij}(0)\exp\left(\int_{0}^{t}\mathcal{U}_{i}(\tau,\vartheta_{i})d\tau
\right)
\end{equation}
For a 3-torus $\mathbb{T}^{3}$, for example, the perturbed 3-metric components are
\begin{equation}
\overline{\bm{g}_{11}(t)}=\bm{g}_{11}(0)\exp\left(\int_{0}^{t}\mathcal{U}_{1}(\tau,\vartheta_{1})d\tau\right)
\equiv \bm{g}_{11}(0)\exp(\mathcal{H}_{1}(t,\vartheta_{1}))
\end{equation}
\begin{equation}
\overline{\bm{g}_{22}(t)}=\bm{g}_{22}(0)\exp\left(\int_{0}^{t}\mathcal{U}_{2}(\tau,\vartheta_{1})d\tau\right)\equiv \bm{g}_{22}(0)\exp(\mathcal{H}_{2}(t,\vartheta_{2}))
\end{equation}
\begin{equation}
\overline{\bm{g}_{33}(t)}=\bm{g}_{33}(0)\exp\left(\int_{0}^{t}\mathcal{U}_{3}(\tau,\vartheta_{1})d\tau\right )\equiv
\bm{g}_{33}(0)\exp(\mathcal{H}_{3}(t,\vartheta_{3}))
\end{equation}
Since the integrals converge, the perturbed metric components converge to the 'attractors' and a new stable metric. $\mathfrak{g}_{11},\mathfrak{g}_{22},\mathfrak{g}_{33}$ for $t\rightarrow\infty$ or $t>\|\vartheta_{i}\|$ since $\|\bm{\vartheta}\|\ll 1$
\begin{equation}
\overline{\bm{g}_{11}(t)}\longrightarrow \bm{g}_{11}(0)\exp(\mathcal{Y}_{1})=\bm{\mathfrak{g}}_{11} \nonumber
\end{equation}
\begin{equation}
\overline{\bm{g}_{22}(t)}\longrightarrow \bm{g}_{22}(0)\exp(\mathcal{Y}_{2})=\bm{\mathfrak{g}}_{22}\nonumber
\end{equation}
\begin{equation}
\overline{\bm{g}_{33}(t)}\longrightarrow \bm{g}_{33}(0)\exp(\mathcal{Y}_{3})=\bm{\mathfrak{g}}_{33}\nonumber
\end{equation}
\end{lem}
\begin{lem}
The evolution of the norms and volume of the perturbed n-metric can be defined:
\begin{enumerate}
\item The norm of the perturbed n-metric is
\begin{align}
&\|\overline{\bm{g}(t)}_{(2,1)}\|=\sum_{j=1}^{n}\left(\sum_{i=1}^{n}|\bm{g}_{ij}(t)|^{2}\right)^{1/2}\equiv
\sum_{i=1}^{n}\left(\sum_{i=1}^{n}|\overline{\bm{g}_{ii}(t)}|^{2}\right)^{1/2}\nonumber\\&
=\sum_{i=1}^{n}\left(\sum_{i=1}^{n}\left|\delta^{ii}\exp\left(2\psi_{i}^{E}+2\int_{0}^{t}
\mathcal{U}_{i}(\tau,\vartheta_{i})d\tau \right)\right|^{2}\right)^{1/2}\nonumber\\&
=\sum_{i=1}^{n}\left(\sum_{i=1}^{n}\left|\delta^{ii}\exp(2\psi_{i}^{E})\exp\left(2\int_{0}^{t}
\mathcal{U}_{i}(\tau,\vartheta_{i})d\tau\right)\right|^{2}\right)^{1/2}\nonumber\\&
=\sum_{i=1}^{n}\left(\sum_{i=1}^{n}\bigg|\delta^{ii}\exp(2\psi_{i}^{E})
\exp\left(2\mathcal{H}_{i}(t,\vartheta_{i})d\tau\right)\bigg|^{2}\right)^{1/2}\nonumber\\&
=\sum_{i=1}^{n}\left(\sum_{i=1}^{n}\bigg|\bm{g}_{ii}^{E}
\exp\left(2\mathcal{H}_{i}(t,\vartheta_{i})d\tau\right)\bigg|^{2}\right)^{1/2}
\end{align}
then
\begin{equation}
\lim_{t\uparrow\infty}\|\overline{\bm{g}(t)}_{(2,1)}\|=\lim_{t\uparrow\infty}\sum_{i=1}^{n}\left(\sum_{i=1}^{n}\bigg|g_{ii}^{E}
\exp\left(2\zeta\mathcal{H}_{i}(t,\vartheta_{i})d\tau\right)\bigg|^{2}\right)^{1/2}
\end{equation}
which converges if the integral converges.
\item Given the initially static spatial volume of the hyper-toroidal geometry $\mathbb{T}^{n}$
\begin{equation}
|\mathbf{V}_{\mathbf{g}}^{E}|=\prod_{i=1}^{n}\exp(\psi_{i}^{E})=\exp\left(\sum_{i=1}^{n}\psi_{i}^{E}\right)
\equiv\prod_{i=1}^{n}a_{i}^{E}
\end{equation}
the perturbed volume is
\begin{align}
&\overline{\mathbf{V}_{\mathbf{g}}(t)}=\prod_{i=1}^{n}\exp\left(\psi^{E}+\int_{0}^{t}\mathcal{U}_{i}(\tau)d\tau\right)
\nonumber\\&\equiv \exp\left(\sum_{i=1}^{n}\psi_{i}^{E}\right)\exp\left(\sum_{i=1}^{n}\int_{0}^{t}
\mathcal{U}_{i}(\tau,\vartheta_{i})d\tau\right)\nonumber\\&
\equiv|\mathbf{V}_{\mathbf{g}}^{E}|\exp\left(\sum_{i=1}^{n}\int_{0}^{t}\mathcal{U}_{i}(\tau,\vartheta)d\tau\right)\nonumber\\&\equiv
|\mathbf{V}_{\mathbf{g}}^{E}|\exp\left(n\int_{0}^{t}\mathcal{U}(\tau,\vartheta)d\tau\right)
\end{align}
if $\psi_{i}^{E}=\psi^{E}$ and $a_{i}^{E}=a^{E}$ so that instability occurs if $\lim_{t\uparrow\infty}\overline{\mathbf{V}_{\mathbf{g}}(t)}=\infty$ and stability occurs for
\begin{equation}
\lim_{t\uparrow\infty}\overline{\mathbf{V}_{\mathbf{g}}(t)}=\lim_{t\uparrow\infty}|\mathbf{V}_{\mathbf{g}}^{E}|
\exp\left(n\int_{0}^{t}\mathcal{U}(\tau,\vartheta))d\tau\right)=\mathbf{V}_{\mathbf{g}}^{E*}<\infty
\end{equation}
\end{enumerate}
\end{lem}
\subsection{Short-pulse Gaussian profile perturbations}
The previous propositions for stability criteria are now demonstrated for very sharp Gaussian 'short-pulse' perturbations of the modulus functions. Various functions can satisfy the conditions of Proposition 4.1 but the Gaussian is the most convenient. It is shown that a Kasner-type static universe or 'toroidal vacuum bubble' is stable to this form of perturbation.
\begin{prop}
Setting
\begin{equation}
\mathcal{U}_{i}(t,\vartheta_{i})\equiv\mathcal{G}_{i}(t,\vartheta)=
\mathcal{A}_{i}\exp(-t^{2}/2\vartheta_{i}^{2})
\end{equation}
with amplitudes $\mathcal{A}_{i}$ and widths $\vartheta_{i}$ with $\|\bm{\vartheta}\|\ll 1$, the perturbed modulus functions are
\begin{equation}
\overline{\psi_{i}(t)}=\psi_{i}^{E}+\int_{0}^{t}\mathcal{G}_{i}(\tau,\vartheta_{i})d\tau
=\mathcal{A}_{i}\int_{0}^{t}\exp(-\tau^{2}/2\vartheta_{i}^{2})d\tau<\infty
\end{equation}
Then the integral converges if
\begin{equation}
\mathcal{Y}_{i}=\lim_{t\uparrow\infty}\int_{0}^{t}\mathcal{G}_{i}(\tau,\vartheta_{i})d\tau=
\lim_{t\gg\|\bm{\vartheta}\|}\int_{0}^{t}\mathcal{G}_{i}(\tau,\vartheta_{i})d\tau<\infty
\end{equation}
or
\begin{equation}
\|\bm{\mathcal{Y}}\|=\lim_{t\uparrow\infty}\left\|\int_{0}^{t}\bm{\mathcal{G}}(\tau,\bm{\vartheta})d\tau
\right\|=\lim_{t\gg\|\bm{\vartheta}\|}\left\|\int_{0}^{t}\bm{\mathcal{G}}(\tau,\bm{\vartheta})d\tau
\right\|<\infty
\end{equation}
The incomplete Gaussian integral is the error function 'erf' so that
\begin{equation}
\overline{\psi_{i}(t)}=\psi_{i}^{E}+\int_{0}^{t}\mathcal{G}_{i}(\tau,\vartheta_{i})d\tau\equiv
\psi_{i}^{E}+\mathcal{A}_{i}(\pi/2)^{1/2}\vartheta_{i}^{-1}erf(t/\sqrt{2}\vartheta_{i})
\end{equation}
The estimate is then
\begin{align}
&\left\|\int_{0}^{t}\bm{\mathcal{G}}(\tau,\bm{\vartheta})d\tau\right\|=
\left(\sum_{i=1}^{n}\left|\int_{0}^{t}\mathcal{G}_{i}(\tau,\vartheta_{i})d\tau\
\right|^{2}\right)^{1/2}\nonumber\\&=\left(\sum_{i=1}^{n}\left|\mathcal{A}_{i}\int_{0}^{t}
\exp(-\tau^{2}/2\vartheta_{i}^{2})d\tau)\right|^{2}\right)^{1/2}\nonumber\\&
=\left(\sum_{i=1}^{n}|\mathcal{A}_{i}|^{2}\int_{0}^{t}
\exp(-\tau^{2}/\vartheta_{i}^{2})d\tau)d\tau\right)^{1/2}\nonumber\\&
=\left(\left({\pi}/{2}\right)^{1/2}\sum_{i=1}^{n}A_{i}^{2}
\vartheta_{i}^{-1}erf(t/\sqrt{2}\vartheta_{i})\right)^{1/2}
\end{align}
Since $erf(t/\sqrt{2}\vartheta_{i})\rightarrow 1$ as $t\rightarrow\infty$ or for $t\gg\|\bm{\vartheta}\|$ if $\|\bm{\vartheta}\|\ll 1$ then
\begin{align}
&\|\bm{\mathcal{Y}}\|=\lim_{t\uparrow\infty}\left\|\int_{0}^{t}\bm{\mathcal{G}}(\tau,\bm{\vartheta})
d\tau\right\|=\lim_{t\ge\|\vartheta_{i}\|}\left\|\int_{0}^{t}\bm{\mathcal{G}}(\tau,\bm{\vartheta})d\tau
\right\|\nonumber\\&=\lim_{t\uparrow\infty}\left(\left({\pi}/{2}\right)^{1/2}\sum_{i=1}^{n}
\mathcal{A}_{i}^{2}\vartheta_{i}^{-1}erf(t/\sqrt{2}\vartheta_{i})\right)^{1/2}=
\left({\pi}/{2}\right)^{1/4}\left(\sum_{i=1}^{n}(\mathcal{A}_{i}\vartheta_{i}^{-1})\right)^{1/2}
\end{align}
If $\mathcal{A}_{i}=\mathcal{A},\vartheta_{i}=\vartheta$ for $i=1...n$ then
\begin{equation}
\|\bm{\mathcal{Y}}\|=\left({\pi}/{2}\right)^{1/4}
\left(\sum_{i=1}^{n}(|\mathcal{A}_{i}|^{2}\vartheta_{i}^{-1})\right)^{1/2}
=\left({\pi}/{2}\right)^{1/4}(n\mathcal{A}^{2}\vartheta^{-1})^{1/2}
\end{equation}
The perturbed toroidal radii are
\begin{align}
&\overline{a_{i}(t)}=\exp[\overline{\psi_{i}(t)})=\exp(\psi_{i}^{E})\exp\left(\int_{0}^{t}
\mathcal{G}_{i}(\tau,\vartheta_{i})d\tau\right)\nonumber\\&\equiv
a_{i}^{E}\exp(\mathcal{A}(\pi/2)^{1/2}\vartheta_{i}^{-1}erf(t/\sqrt{2}\vartheta_{i}))
\end{align}
so that the attractors for $t\rightarrow\infty$ or $t\gg\|\bm{\vartheta}\|$ are
\begin{equation}
\overline{a_{i}(t)}\longrightarrow a_{i}^{E}\exp(\mathcal{A}(\pi/2)^{1/2}\vartheta_{i}^{-1})
=a_{i}^{E}\exp(\mathcal{Y}_{i})\equiv a_{i}^{E*}
\end{equation}
and so the perturbed radii converge to new stable values or attractors $a_{i}^{*}$.
\end{prop}
\begin{lem}
The $\mathcal{L}_{2}$-norms are
\begin{equation}
\|\overline{\bm{\psi}(t)}-\bm{\psi^{E}}\|=\left\|\int_{0}^{t}\mathcal{G}(t,\vartheta)d\tau
\right\|=\left(\left({\pi}/{2}\right)^{1/2}\sum_{i=1}^{n}\mathcal{A}_{i}^{2}
\vartheta_{i}^{-1}erf(t/\sqrt{2}\vartheta_{i})\right)^{1/2}
\end{equation}
which follows from (4.21). Then
\begin{align}
&\lim_{t\uparrow\infty}\|\overline{\bm{\psi(t)}}-\bm{\psi^{E}}\|=\lim_{t\uparrow\infty}
\left(\left(\pi/2\right)^{1/2}\sum_{i=1}^{n}\mathcal{A}_{i}^{2}
\vartheta_{i}^{-1}erf(t/\sqrt{2}\vartheta_{i})\right)^{1/2}\nonumber\\&
=\left(\pi/2\right)^{1/4}\left(\sum_{i=1}^{n}(\mathcal{A}_{i}\vartheta_{i}^{-1})\right)^{1/2}
\end{align}
If $\mathcal{A}_{i}=\mathcal{A}$ and $\vartheta_{i}=\vartheta $.
\begin{equation}
\lim_{t\uparrow\infty}\|\overline{\bm{\psi}(t)}-\bm{\psi^{E}}\|
=\left({\pi}/{2}\right)^{1/4}\left(\sum_{i=1}^{n}(|\mathcal{A}_{i}|^{2}\vartheta_{i}^{-1})\right)^{1/2}
=\left({\pi}/{2}\right)^{1/4}(n\mathcal{A}^{2}\vartheta^{-1})^{1/2}
\end{equation}
The $\mathcal{L}_{2}$-norm for the perturbed radii is
\begin{align}
&\|\overline{\bm{a}(t)}-\bm{a}^{E}\|\le \|\overline{\bm{a}(t)}\|-\|\bm{a}^{E}\|\nonumber\\&
=\left\|\bm{a}^{E}\exp\left(\int_{0}^{t}\bm{\mathcal{G}}(t,\vartheta)d\tau\right)\right\|-\|\bm{a}^{E}\|
\nonumber\\&=\left(\sum_{i=1}^{n}\left|a_{i}^{E}\exp\left(\mathcal{A}_{i}\int_{0}^{t}
\exp(-\tau^{2}/\sqrt{2}\vartheta_{i}^{2})d\tau\right)\right|^{2}\right)^{1/2}-\left(\sum_{i=1}^{n}|a_{i}^{E}|^{2}\right)^{1/2}\nonumber\\
&=\left(\sum_{i=1}^{n}|a_{i}^{E}|^{2}\exp(2\mathcal{A}_{i}\left({pi}/{2}\right)^{1/2}\vartheta_{i}^{-1}
erf(t/\sqrt{2}\vartheta_{i})\right)^{1/2}-\left(\sum_{i=1}^{n}|a_{i}^{E}|^{2}\right)^{1/2}
\end{align}
Then asymptotically
\begin{equation}
\lim_{t\uparrow\infty}\|\overline{\bm{a}(t)}-\bm{a}^{E}\|\le \left(\sum_{i=1}^{n}|a_{i}|^{2}
\exp\left(2\mathcal{A}_{i}\left({\pi}/{2}\right)^{1/4}\vartheta_{i}^{-1}\right)\right)^{1/2}
-\left(\sum_{i=1}^{n}|a_{i}^{E}|^{2}\right)^{1/2}
\end{equation}
If $\mathcal{A}_{i}=\mathcal{A},\vartheta_{i}=\vartheta, a_{i}^{E}=a_{i}$ then
\begin{equation}
\lim_{t\uparrow\infty}\|\overline{\bm{a}(t)}-\bm{a}^{E}\|\le n^{1/2}a^{E}\exp\left(\mathcal{A}\left({\pi/2})\right)\vartheta^{-1}\right)-n^{1/2}a^{E}
\end{equation}
\end{lem}
\begin{lem}
Given the initially static spatial volume of the hyper-toroidal geometry $\mathbb{T}^{n}$
\begin{equation}
\mathbf{V}_{\mathbf{g}}^{E}=\prod_{i=1}^{n}\exp(\psi_{i}^{E})=\exp\left(\sum_{i=1}^{n}\psi_{i}^{E}\right)
\equiv\prod_{i=1}^{n}a_{i}^{E}.
\end{equation}
The perturbed spatial volume is
\begin{align}
\overline{\mathbf{V}_{\mathbf{g}}(t)}&=\prod_{i=1}^{n}\exp\left(\psi^{E}+\int_{0}^{t}\mathcal{G}_{i}(\tau,\vartheta_{i})d\tau\right)
\nonumber\\&\equiv\exp\left(\sum_{i=1}^{n}\psi_{i}^{E}\right)\exp\left(\sum_{i=1}^{n}\int_{0}^{t}
\mathcal{G}_{i}(\tau,\vartheta)d\tau\right)\nonumber\\&
\equiv|\mathbf{V}_{\mathbf{g}}^{E}|\exp\left(\sum_{i=1}^{n}\int_{0}^{t}\mathcal{G}_{i}(\tau,\vartheta)d\tau\right)\nonumber\\&\equiv
|\mathbf{V}_{\mathbf{g}}^{E}|\exp\left(\sum_{i=1}^{n}\mathcal{A}_{i}\int_{0}^{t}\exp(-\tau^{2}/2\vartheta^{2}d\tau\right)
\nonumber\\& \equiv |\mathbf{V}_{\mathbf{g}}^{E}|\exp\left(n
\mathcal{A}\int_{0}^{t}\exp(-\tau^{2}/2\vartheta^{2})d\tau\right)\nonumber\\&
=|\mathbf{V}_{\mathbf{g}}^{E}|\exp(n\mathcal{A}(\pi/2)^{1/2}erf(\tau/2\vartheta))
\end{align}
If $\psi_{i}^{E}=\psi^{E}$, $a_{i}^{E}=a^{E}$ and $\mathcal{A}_{i}=\mathcal{A}$, $\vartheta_{i}=\vartheta$. Asymptotic stability then occurs since as the volume of the space evolves to a new stable but larger volume
\begin{align}
&\lim_{t\uparrow\infty}\overline{\mathbf{V}_{\mathbf{g}}(t)}=\lim_{t\uparrow\infty}|\mathbf{V}_{\mathbf{g}}^{E}|
\exp\left(n\int_{0}^{t}\mathcal{G}(\tau,\vartheta)d\tau\right)\nonumber\\
&=\lim_{t\uparrow\infty}|\mathbf{V}_{\mathbf{g}}^{E}|\exp(n\mathcal{A}(\pi/2)^{1/2}erf(\tau/2\vartheta))
=|\mathbf{V}_{\mathbf{g}}^{E}|\exp(n\mathcal{A}(\pi/2)^{1/2})=V^{*}<\infty
\end{align}
with $\mathbf{V}_{\mathbf{g}}^{E*}>\mathbf{V}_{\mathbf{g}}^{E}$.
The norm of the perturbed metric converges or relaxed back to a new value or 'attractor'
\begin{align}
&\lim_{t\uparrow\infty}\|\overline{\bm{g}(t)}\|_{(2,1)}
=\lim_{t\uparrow\infty}\sum_{i=1}^{n}\left(\sum_{i=1}^{n}|g_{ii}^{E}+2\delta_{ii}\mathcal{A}_{i}(\pi/2)^{1/2}
erf(t/2\vartheta_{i})|\right)^{1/2}\nonumber\\&
=\lim_{t\uparrow\infty}\sum_{i=1}^{n}\left(\sum_{i=1}^{n}|g_{ii}^{E}+2\delta_{ii}A_{i}(\pi/2)^{1/2}|\right)^{1/2}
\equiv\|\bm{\mathfrak{g}}\|
\end{align}
\end{lem}
The following theorem for short-pulse Gaussian perturbations of the nonlinear ODE system can now be established
\begin{thm}
Given the following:
\begin{enumerate}
\item The initial data $\mathfrak{D}=[t=0,\Sigma_{o}=0,\bm{g}_{ij}(0),k_{ij}(0),\psi_{i}(0)=\psi_{i}^{E},a_{i}(0)
    =a_{i}^{E}]$ with $\bm{g}_{oo}=-1,\bm{g}_{io}=0$ with constraints $\mathbf{Ric}_{oo}=0$ and $\mathbf{Ric}_{io}=0$, and there is no development of the data for $t>0$ in the absence of perturbations.
\item The Einstein vacuum equations $\mathbf{Ric}_{AB}=0$ exist on $\mathbb{M}^{n+1}$ in the form $\mathbf{H}_{n}\psi_{i}(t)=0$ for all $t\in\mathbb{R}^{+}$, and $\mathbf{D}_{n}\psi_{i}^{E}=0$ and the toroidal radii are $a_{i}(t)=a_{i}^{E}=\exp(\psi_{i}^{E})$ corresponding to the set or static moduli $\psi_{i}^{E}$.
\item Short-pulse Gaussian functions $(\mathcal{G}_{i}(t,\vartheta_{i}))_{i=1}$ spanning $\mathbb{R}^{m}$ where $\vartheta_{i}$ are the widths such that $\mathcal{G}_{i}(t,\vartheta)=0$ for $t\gg \|\bm{\vartheta}\|$ and
    $\lim_{t\uparrow \infty}\|\bm{\mathcal{G}}(t,\bm{\vartheta})\|=0$.(E.g., a set of Gaussians with widths $\vartheta_{i}$). The derivatives $\partial_{t}\mathcal{G}_{i}(t,\vartheta_{i})$ exist and $\partial_{t}\mathcal{G}_{i}(t,\vartheta_{i})\rightarrow 0$ for
    $t\gg \|\bm{\vartheta} \|$
\item The perturbed moduli are
\begin{align}
&\overline{\psi_{i}(t)}=\psi_{i}^{E}+\int_{0}^{t}\mathcal{G}_{i}(\tau,\vartheta_{i})d\tau
\nonumber\\&=\mathcal{A}_{i}\int_{0}^{t}\exp(-\tau^{2}/2\vartheta_{i}^{2})d\tau=\psi_{i}^{E}+
\mathcal{A}(\pi/2)^{1/2}\vartheta_{i}^{-1}erf(t/\sqrt{2}\vartheta_{i})
\end{align}
The perturbed Einstein equations are
\begin{align}
&\mathbf{H}_{n}\overline{\psi_{i}(t)}=\mathcal{S}(t,\vartheta_{i})
\\&
\mathbf{D}_{n}\overline{a_{i}(t)}=\mathcal{S}(t,\vartheta_{i})
\end{align}
where
\begin{align}
&\mathcal{S}(n, t,\vartheta_{i})=\sum_{i=1}^{n}\partial_{t}\mathcal{G}(t,\vartheta_{i})
+\sum_{i=1}^{n}\sum_{j=1}^{n}\mathcal{G}_{i}(t,\vartheta_{i})\mathcal{G}_{i}(t,\vartheta_{j})+
\sum_{i=1}^{n}\mathcal{G}_{i}(t,\vartheta_{i})\mathcal{G}_{i}(t,\vartheta_{i})
\nonumber\\&=-\sum_{i=1}^{n}\frac{1}{\vartheta\sqrt{2}}\mathscr{H}_{1}(t/\vartheta_{i}\sqrt{2})
\mathcal{A}_{i}E_{i}(t)+\sum_{i=1}^{n}\sum_{j=1}^{n}\mathcal{A}_{i}\mathcal{A}_{j}
E_{i}(t)E_{j}(t)+\sum_{i=1}^{n}
\mathcal{A}_{i}\mathcal{A}_{i}E_{i}(t)E_{i}(t)
\end{align}
where $E_{i}(t)\equiv\exp(-t^{2}/2\vartheta_{i})$ and $\mathscr{H}_{1}$ is the first-order Hermite polynomial. Then as $t\uparrow\infty$ or for $t>\|\bm{\vartheta}\|$
\begin{equation}
\lim_{t\uparrow\infty}\mathbf{H}_{n}\overline{\psi_{i}(t)}=\lim_{t\uparrow\infty}
\mathcal{S}(n,t,\vartheta_{i})=0
\end{equation}
\begin{equation}
\lim_{t\uparrow\infty}\mathbf{D}_{n}\overline{a_{i}(t)}=
\lim_{t\uparrow\infty}\mathcal{S}(n,t,\vartheta_{i})=0
\end{equation}
or $[\mathbf{H}_{n}\overline{\psi_{i}(t)}]_{t>\|\vartheta_{i}\|}=
\mathcal{S}(t,\vartheta_{i})_{t>\|\vartheta_{i}\|} $ and $
[\mathbf{D}_{n}\overline{a_{i}(t)}]_{t>\|\vartheta_{i}\|}=J(t,\vartheta_{i})_{t>\|
\vartheta_{i}\|}$.
\item If $\mathcal{A}_{i}=\mathcal{A}_{j}=\mathcal{A}$ and $\vartheta_{i}=\vartheta$ for $i,j=1...n$ then
\begin{equation}
\mathbf{H}_{n}\overline{\psi_{i}(t)}=-\frac{n}{\sqrt{2}\vartheta}\mathscr{H}_{i}(t/\sqrt{2}\vartheta)\mathcal{A}E(t)
+n\mathcal{A}^{2}E(t)=\mathcal{S}(n,t,\vartheta)
\end{equation}
 \begin{equation}
\mathbf{D}_{n}\overline{a_{i}(t)}=-\frac{n}{\sqrt{2}\vartheta}\mathscr{H}_{i}(t/\sqrt{2}\vartheta)\mathcal{A}E(t)
+n\mathcal{A}^{2}E(t)=\mathcal{S}(n,t,\vartheta)
\end{equation}
\end{enumerate}
\end{thm}
\begin{proof}
For Gaussian perturbations
\begin{equation}
\lim_{t\uparrow\infty}\|\overline{\psi_{i}(t)}-\psi_{i}^{E}\|\equiv \lim_{t\uparrow\infty}\|erf(t/2\vartheta_{i}))\|<\infty
\end{equation}
since $\lim_{t\uparrow\infty}erf(t)\rightarrow 1$. If $\mathcal{G}_{i}(t,\vartheta_{i})=\mathcal{G}_{i}(t,\vartheta)$ for all $i=1...n$.
Similarly,
\begin{equation}
\lim_{t\rightarrow\infty}\|\overline{\bm{a}(t)}-\overline{\bm{a}^{E}}\|\le \lim_{t\rightarrow\infty}\|a_{i}^{E}\|\|\exp(2 erf(t/2\bm{\vartheta}))\|-1]<\infty\nonumber
\end{equation}
From (3.25)
\begin{align}
&\mathbf{H}_{n}\overline{\psi_{i}(t)}= \mathbf{H}_{n}\psi_{i}^{E}+\sum_{i=1}^{n}\partial_{tt}
{\mathcal{H}}_{i}(t)\nonumber\\&+\frac{1}{2}\sum_{i=1}^{n}\sum_{j=1}^{n}\partial_{t}
{\mathcal{H}}_{i}(t)\partial_{t}\mathcal{H}_{j}(t)+\frac{1}{2}\sum_{i=1}^{n}
\partial_{t}\mathcal{H}_{i}(t)\partial_{t}\mathcal{H}_{i}(t)
\nonumber\\&=\sum_{i=1}^{n}\partial_{t}\mathcal{G}(t,\vartheta_{i})
+\frac{1}{2}\sum_{i=1}^{n}\sum_{j=1}^{n}\mathcal{G}_{i}(t,\vartheta_{i})
\mathcal{G}_{i}(t,\vartheta_{j})
+\frac{1}{2}\sum_{i=1}^{n}\mathcal{G}_{i}(t,\vartheta_{i})\mathcal{G}_{i}(t,\vartheta_{i})
\end{align}
so for the set of Gaussian functions(4.49)
\begin{equation}
\mathbf{H}_{n}\overline{\psi_{i}(t)}= \sum_{i=1}^{n}\partial_{t}\mathcal{H}_{i}(t,\vartheta_{i})
+\frac{1}{2}\sum_{i=1}^{n}\sum_{j=1}^{n}\mathcal{A}_{i}\mathcal{A}_{j}E_{i}(t)
E_{j}(t)+\frac{1}{2}\sum_{i=1}^{n}\mathcal{A}_{i}\mathcal{A}_{i}E_{i}(t)E_{i}(t)
\end{equation}
which is
\begin{align}
&\mathbf{H}_{n}\overline{\psi_{i}(t)}=-\sum_{i=1}^{n}\frac{1}{\vartheta\sqrt{2}}\mathscr{H}_{1}(t/\vartheta_{i}\sqrt{2})A_{i}
E_{i}(t)\nonumber\\&+\frac{1}{2}\sum_{i=1}^{n}\sum_{j=1}^{n}\mathcal{A}_{i}
\mathcal{A}_{j}E_{i}(t)E_{j}(t))+\frac{1}{2}\sum_{i=1}^{n}
\mathcal{A}_{i}\mathcal{A}_{i}E_{i}(t)E_{i}(t)
\end{align}
where the derivative $\partial_{t}\mathcal{G}_{i}(t,\vartheta_{i})$ is given by the Hermite polynomials $\mathscr{H}_{m}(x)$ with $\mathscr{H}_{1}(x)=2x $ such that
\begin{equation}
\partial_{t}\mathcal{G}_{i}(t,\vartheta)=-\mathscr{H}_{1}(t/\sqrt{2}\vartheta_{i})\exp(-t^{2}/2\vartheta_{i}^{2})
\equiv \mathscr{H}_{1}(t/\sqrt{2}\vartheta_{i})E_{i}(t)
\end{equation}
Then $ \mathbf{H}_{n}\overline{\psi_{i}(t)}\rightarrow 0$ very rapidly for $t>|\bm{\vartheta}|$. Indeed, the nonlinearity accelerates the rate of the convergence to zero. The perturbed equation $\mathbf{D}_{n}\overline{a_{i}(t)}$ gives the same result. The perturbed radii are
\begin{equation}
\overline{a_{i}(t)}=a_{i}^{E}\exp\left(\mathcal{A}_{i}\int_{0}^{t}\exp(\tau^{2}/2\vartheta_{i}^{2})d\tau\right)
\equiv a_{i}^{E}\exp\left(\mathcal{A}_{i}\int_{0}^{t}E_{i}(\tau)d\tau\right)\equiv a_{i}^{E}\mathcal{X}_{i}(t)
\end{equation}
so that
\begin{align}
&\mathbf{D}_{n}\overline{a_{i}(t)}=-\sum_{i=1}^{n}\frac{a_{i}^{E}\mathcal{A}_{i}\mathscr{H}_{i}(t/\sqrt{2}\vartheta_{i})
E_{i}(t)\mathcal{X}_{i}(t)}{\sqrt{2}\vartheta_{i}a_{i}^{E}\mathcal{X}_{i}(t)}
+\sum_{i=1}^{n}\frac{a_{i}^{E}\mathcal{A}_{i}\mathcal{A}_{i}E_{i}(t)
E_{i}(t)\mathcal{X}_{i}(t)}{a_{i}^{E}\mathcal{X}_{i}(t)}\nonumber\\&-\frac{1}{2}\sum_{i=1}^{n}\frac{a_{i}^{E}
a_{i}^{E}\mathcal{A}_{i}\mathcal{A}_{i}E_{i}(t)E_{i}(t)}{\mathcal{X}_{i}(t)\mathcal{X}_{i}(t)a_{i}^{E}a_{i}^{E}}
-\frac{1}{2}\sum_{i=1}^{n}\sum_{j=1}^{n}\frac{a_{i}^{E}a_{j}^{E}\mathcal{A}_{i}\mathcal{A}_{j}
E_{i}(t)E_{j}(t))}{\mathcal{X}_{i}(t)\mathcal{X}_{j}(t)a_{i}^{E}a_{j}^{E}}
\end{align}
Cancelling terms then gives (4.71). Equations (4.75) and (4.76) then follow from setting $\mathcal{A}_{i}=\mathcal{A}$ and $\vartheta_{i}=\vartheta$ for $i=1...n$.
\end{proof}
\begin{lem}
Given a dynamic solution $ \psi=p_{i}\log|t|$, the perturbed quantities defined in Definition 3.9 rapidly converge back to their unperturbed forms.
\begin{enumerate}
\item The perturbed Kretchsmann scalar $\mathbf{K}(t)$ for the dynamical solutions
converges or 'relaxes' back to its original form for the short-pulse Gaussian perturbations so that
\begin{equation}
\lim_{t\uparrow\infty}\overline{\mathbf{K}(t)}\equiv\lim_{t\gg |\bm{\vartheta}|}\overline{\mathbf{K}(t)}=\mathbf{K}(t)
=\sum_{i=1}^{n}\sum_{j=1}^{n}p_{i}p^{i}p_{j}p^{j}t^{-4}
\end{equation}
\item The perturbed expansion $\bm{\chi}(t)$ converges or relaxes back such that
\begin{equation}
\overline{\bm{\chi}(t)}\rightarrow\bm{\chi}(t)
\end{equation}
\item The perturbed shear converges as
\begin{equation}
\overline{\bm{\mathfrak{S}}^{2}(t)}\rightarrow\bm{\mathfrak{S}}^{2}(t)
=\sum_{i=1}^{n}\sum_{j=1}^{n}\left(\bigg|\frac{p_{i}}{t}\bigg|+\bigg|\frac{p_{i}}{t}\bigg|
-2\bigg |\frac{p_{i}p_{j}}{t^{2}}\bigg|\right)
\end{equation}
\end{enumerate}
\end{lem}
\begin{proof}
The perturbed moduli are
\begin{equation}
\overline{\psi_{i}(t)}=p_{i}\ln|t|+\mathcal{A}_{i}\int_{0}^{t}\exp(-\tau^{2}/2\vartheta_{i}^{2})d\tau)
\end{equation}
with derivatives $\partial_{t}\overline{\psi_{i}(t)}=(p_{i}/t)+\mathcal{A}_{i}\exp(-t^{2}/2\vartheta_{i})=(p_{i}/t)+\mathcal{A}_{i}E_{i}(t)$
and $\partial_{tt}\overline{\psi_{i}(t)}=(-p_{i}/t^{2})+(\mathcal{A}_{i}\partial_{t}
\exp(-t^{2}/2\vartheta_{i})=(-p_{i}/t^{2})+(\mathcal{A}_{i}\sqrt{2}t/\vartheta_{i})\exp(-t^{2}/2\vartheta_{i})=-p_{i}/t^{2}
+(\mathcal{A}_{i}\sqrt{2}t/\vartheta_{i})E_{i}(t)$ where $E_{i}(t)=\exp(-t^{2}/2\vartheta_{i})$. The perturbed scalar invariant becomes
\begin{align}
&\overline{\mathbf{K}(t)}=-4\sum_{i=1}^{n}\frac{p_{i}}{t^{2}}+\sqrt{2}t\sum_{i=1}^{n}\frac{\mathcal{A}_{i}}{\vartheta_{i}}E_{i}(t)
+4\sum_{i=1}^{n}\bigg|\frac{p_{i}^{2}}{t^{2}}+2\frac{p_{i}}{t}\mathcal{A}_{i}E_{i}(t)+\mathcal{A}_{i}\mathcal{A}_{i}E_{i}(t)E_{i}(t)\bigg|\nonumber\\&
=2\sum_{i=1}^{n}\sum_{j=1}^{n}\bigg|\frac{p_{i}p^{i}p_{j}p^{j}}{t^{4}}+\frac{2p_{i}^{2}p_{j}\mathcal{A}_{j}}{t^{3}}E_{j}(t)+
\frac{p_{i}^{2}}{t^{2}}\mathcal{A}_{j}^{2}E_{i}(t)E_{j}(t)+\frac{2p_{i}\mathcal{A}_{i}\mathcal{A}_{j}\mathcal{A}_{j}}{t}E_{i}(t)E_{j}(t)E_{j}(t)\bigg|\nonumber\\&
+s\sum_{i=1}^{n}\sum_{j=1}^{n}\bigg|\mathcal{A}_{i}\mathcal{A}_{i}E_{i}(t)E_{i}(t)+ \frac{p_{j}}{t}\mathcal{A}_{i}\mathcal{A}_{j}\mathcal{A}_{j}E_{i}(t)E_{i}(t)E_{j}(t)
+\mathcal{A}_{i}\mathcal{A}_{i}E_{i}(t)E_{i}(t)E_{j}(t)E_{j}(t)
\bigg|
\end{align}
All terms containing $E_{i}(t)=\exp(-t^{2}/2\vartheta_{i})$ vanish very rapidly for $t\gg |\bm{\vartheta}|$ since $|\bm{\vartheta}|\ll 0$ so that
\begin{equation}
\overline{\mathbf{K}(t)}\rightarrow \mathbf{K}(t)=-4\sum_{i=1}^{n}p_{i}^{2}t^{-2}+4\sum_{i=1}^{n}p_{i}^{2}t^{-2}+2\sum_{i=1}^{n}
\sum_{j=1}^{n}p_{i}p_{i}p^{j}p^{j}t^{-4}
\end{equation}
The perturbed expansion is
\begin{equation}
\overline{\bm{\chi}(t)}=\sum_{i=1}^{n}|\partial_{t}\overline{\psi_{i}(t)}|
=\sum_{i=1}^{n}\bigg|\frac{p_{i}}{t}+\mathcal{A}_{i}E_{i}(t)\bigg|
\end{equation}
which rapidly converges back to $\bm{\chi}(t)$ as $t\rightarrow\infty$ or $t\gg |\bm{\vartheta}|$.
Using (3.58), the perturbed shear is given as
\begin{equation}
\overline{\bm{\mathfrak{S}}^{2}(t)}=\sum_{i=1}^{n}\sum_{j=1}^{n}\left(\bigg|\frac{p_{i}}{t}+\mathcal{A}_{i}E_{i}(t)
\bigg|^{2}+\bigg|\frac{p_{j}}{t}+\mathcal{A}_{j}E_{j}(t)\bigg|-2\bigg|\frac{p_{i}}{t}+\mathcal{A}_{i}E_{i}(t)
\frac{p_{i}}{t}+\mathcal{A}_{i}E_{i}(t)\bigg|\right)
\end{equation}
so that for $t\rightarrow\infty$ or $t\gg |\bm{\vartheta}|$
\begin{equation}
\overline{\bm{\mathfrak{S}}^{2}(t)}\rightarrow\bm{\mathfrak{S}}^{2}(t)
=\sum_{i=1}^{n}\sum_{j=1}^{n}\left(\bigg|\frac{p_{i}}{t}\bigg|+\bigg|\frac{p_{i}}{t}\bigg|
-2\bigg |\frac{p_{i}p_{j}}{t^{2}}\bigg|\right)
\end{equation}
\end{proof}
\subsection{Continuous amplitude perturbation as a 'cosmological constant'}
Given the static Kasner micro-universe, it is now shown that it is unstable to a continuous 'step' perturbation of the modulus functions which has constant amplitude and which enters the Einstein equations as a 'cosmological constant-like' term.
\begin{thm}
Given $\mathbb{M}^{n+1}=\mathbb{T}^{n}\times\mathbb{R}^{(+)}$ and the following:
\begin{enumerate}
\item  The initial data $\mathfrak{D}=[t=0,\Sigma_{o}=0,\bm{g}_{ij}(0),l_{ij}(0),\psi_{i}(0)=\psi_{i}^{E},
    a_{i}(0)=a_{i}^{E}]$ with $\bm{g}_{oo}=-1,\bm{g}_{io}=0$ with constraints $\mathbf{Ric}_{oo}=0$ and $\mathbf{Ric}_{io}=0$.
\item The Einstein vacuum equations $\mathbf{Ric}_{AB}=0$ exist on $\mathbb{M}^{n+1}$ in the form $\mathbf{H}_{n}\psi_{i}^{E}=0$.
\item The initially static toroidal radii are $a_{i}^{E}=\exp(\psi_{i}^{E})$
\item There is a set of functions $\mathcal{A}_{i}(t)$ such that $\mathcal{A}(0)=0$ and $\mathcal{A}_{i}(t)=\mathcal{A}=const.$ for all $t>0$.
\end{enumerate}
The initially static modulus functions $\psi_{i}^{E}$ and radii $a_{i}^{E}$ are perturbed so that
\begin{align}
&\overline{\psi_{i}(t)}=\psi_{i}^{E}+\int_{0}^{t}\mathcal{A}_{i}(\tau)d\tau\equiv \overline{\psi_{i}^{E}}+\mathcal{A}_{i}t
\\&\overline{a_{i}(t)}=a_{i}^{E}\exp\left(\int_{0}^{t}\mathcal{A}_{i}d\tau\right)\equiv a_{i}(t)\exp(\mathcal{A}_{i}t)=a_{i}(t)\mathcal{X}_{i}(t)
\end{align}
The derivatives are then $\partial_{t}\overline{\psi_{i}(t)}=\mathcal{A}_{i}$ and
$\partial_{tt}\overline{\psi(t)}=0$ and$\partial_{t}\overline{a_{i}(t)}=a_{i}^{E}\mathcal{A}$ and $\partial_{tt}\overline{a_{i}(t)}=a_{i}^{E}\mathcal{A}^{2}\mathcal{X}_{i}(t)$. The norms are then
\begin{align}
&\|\overline{\bm{\psi}}(t)-\bm{\psi}^{E}\|=\left\|\int_{0}^{t}\bm{\mathcal{A}}d\tau\right\|=
\left(\sum_{i-1}^{n}\left|\int_{0}^{t}\mathcal{A}_{i}d\tau\right|^{2}\right)^{1/2}=
n^{1/2}\left|\int_{0}^{t}\mathcal{A} d\tau\right|=n^{1/2}\mathcal{A} t\\&
\|\overline{\bm{a}(t)}-\bm{a}^{E}\|=\left\|\bm{a}^{E}
\exp\left(\int_{0}^{t}\bm{\mathcal{A}}d\tau\right)\right\|
=n^{1/2}|a_{i}^{E}|\exp(\mathcal{A} t)
\end{align}
so that $\lim_{t\uparrow\infty}\|\overline{\bm{\psi}}(t)-\bm{\psi}_{i}^{E}\|=\infty$ and $\lim_{t\uparrow\infty}\|\overline{\bm{a}(t)}-\overline{\bm{a}}^{E}\|=\infty$. Equations (4.93) and (4.94) are solutions of the perturbed Einstein equations  so that
\begin{equation}
\mathbf{H}_{n}\overline{\bm{\psi}(t)}=\mathbf{H}_{n}\psi_{i}^{E}+n \mathcal{A}^{2}=n \mathcal{A}^{2}\equiv\lambda
\end{equation}
\begin{equation}
\mathbf{D}_{n}\overline{a_{i}(t)}=\mathbf{D}_{n}a_{i}^{E}+n\mathcal{A}^{2}=n\mathcal{A}^{2}\equiv\lambda
\end{equation}
if $\mathcal{A}_{i}=\mathcal{A}$ for $i=1...n$ so that $\lambda=n \mathcal{A}^{2}$ is an induced cosmological constant and we recover the Einstein equations (3.70) and (3.71). This is equivalent to the result of Lemma 3.12, which gives the solutions for the Einstein equations (3.70) and (3.71) with a cosmological constant term.
\end{thm}
\begin{proof}
The perturbed Einstein equations are
\begin{align}
&\mathbf{H}_{n}\overline{\psi_{i}(t)}
=\sum_{i=1}^{n}\partial_{tt}\overline{\psi_{i}(t)}
+\frac{1}{2}\sum_{i=1}\partial_{t}\overline{\psi_{i}(t)}\partial_{t}\overline{\psi_{j}(t)}
+\frac{1}{2}\sum_{i=1}^{n}\sum_{j=1}^{n}\partial_{t}\overline{\psi_{i}(t)}\partial_{t}
\overline{\psi_{j}(t)}\nonumber\\&=\frac{1}{2}\sum_{i=1}^{n}\mathcal{A}_{i}\mathcal{A}_{i}
+\frac{1}{2}\sum_{i=1}^{n}\mathcal{A}_{i}\mathcal{A}_{i}=n\mathcal{A}^{2}\equiv\lambda
\end{align}
if $\mathcal{A}_{i}=\mathcal{A}$ for $i=1...n$. The perturbed equations for the toroidal radii are
\begin{align}
&\mathbf{D}_{n}\overline{a_{i}(t)}=\sum_{i=1}^{n}\frac{a_{i}^{E}|\mathcal{A}_{i}|^{2}
\mathcal{X}_{i}}{a_{i}^{E}X_{i}}-\frac{1}{2}\sum_{i=1}^{n} \frac{a_{i}^{E}a_{i}^{E}\mathcal{X}_{i}\mathcal{X}_{i}|\mathcal{A}_{i}|^{2}}{a_{i}^{E}a_{i}^{E}
\mathcal{X}_{i}\mathcal{X}_{i}}+\frac{1}{2}\sum_{i=1}^{n} \frac{a_{i}^{E}a_{i}^{E}\mathcal{X}_{i}\mathcal{X}_{i}(t)|\mathcal{A}_{i}|^{2}}{a_{i}^{E}a_{i}^{E}
\mathcal{X}_{i}\mathcal{X}_{i}}\nonumber\\&
=\sum_{i=1}^{n}\mathcal{A}_{i}\mathcal{A}_{i}-\frac{1}{2}\sum_{i=1}^{n}\mathcal{A}_{i}
\mathcal{A}_{i}+\frac{1}{2}\sum_{i=1}^{n}\sum_{j=1}^{n}\mathcal{A}_{i}\mathcal{A}_{j}\nonumber\\&
=n \mathcal{A}^{2}-\frac{1}{2}n \mathcal{A}^{2}+\frac{1}{2}n \mathcal{A}^{2}=n \mathcal{A}^{2}\equiv\lambda
\end{align}
\end{proof}
\begin{cor}
The spatial volume "inflates" or grows exponentially
\begin{equation}
\lim_{t\uparrow\infty}\overline{\mathbf{V}_{\mathbf{g}}(t)}=\lim_{t\uparrow\infty}|\mathbf{vol}^{E}|\exp(n \mathcal{A}t)\equiv \exp(\mathbf{Ly}t)=\infty
\end{equation}
\end{cor}
A second corollary is that $\mathcal{A}$ is also a Lyupunov exponent.
\begin{cor}
From Proposition (2.11), it follows that
\begin{equation}
\frac{\|\overline{\bm{a}(t)}-\bm{a}^{E}\|}{\bm{a}^{E}}\sim \exp(\mathcal{A} t)
\end{equation}
so that $\mathcal{A}$ is essentially a LCE and the system is unstable for $\mathcal{A}     >0$ and can never reach equilibrium or stability.
\end{cor}
\section{Random field radial moduli and stochastically averaged Einstein vacuum equations}
The previous analysis established the perturbative stability to a deterministic Gaussian impulse perturbation of the radial modulus functions. But what were established as stable points or static equilibrium solutions, via the deterministic stability analysis, may be unstable in the presence of stochasticity, random fluctuations or 'noise'. The coupling of random fluctuations or noise to the inherent nonlinearity of the problem now becomes a crucial issue. Random perturbations of the radial modulus functions $\lbrace\psi_{i}(t)\rbrace$ are now considered, and the Einstein equations are then interpreted as an n-dimensional nonlinear system of differential equations coupled to Gaussian perturbations or classical random fields. The stochastic average of the systems of differential equations can also be computed and this leads to a non-vanishing extra terms due to nonlinearity--these then enter as induced 'cosmological constant' terms.

Applying the methods of Section 2, the most general random perturbations of the initially static moduli $\psi_{i}^{E}$ are of the (Stratanovitch) stochastic integral form
\begin{equation}
\widehat{\psi}_{i}(t)=
m_{i}^{E}+\zeta\int_{0}^{t}f(\tau)\widehat{\mathscr{U}}_{i}(\tau)d\tau
\end{equation}
or
\begin{equation}
\widehat{\psi}_{i}(t)=
\psi_{i}(t)+\zeta\int_{0}^{t}f(\tau)\widehat{\mathscr{U}}_{i}(\tau)d\tau
\end{equation}
for initially dynamical solutions, where $f:\mathbb{R}^{+}\rightarrow\mathbb{R}^{+}$ is some smooth continuous function and $\widehat{\mathscr{U}}_{i}(t)$ is an n-dimensional Gaussian random vector field, and $\zeta>0$ is a constant. Then $\widehat{a}_{i}(t)
=\exp(\widehat{\psi}_{i}(t))$ are the randomly perturbed radii. The randomly perturbed Einstein system of nonlinear ODEs is $\mathbf{H}_{n}\widehat{\psi}_{i}(t)$  or $\mathbf{H}_{n}\widehat{\psi}_{i}(t)$. The stochastically averaged Einstein system is then
\begin{align}
&\bm{\mathrm{I\!E}}\bigg\lbrace\mathbf{H}_{n}\widehat{\psi}_{i}(t)\bigg\rbrace
=\mathbf{H}_{n}\psi_{i}^{E}+terms=terms
\\&\bm{\mathrm{I\!E}}\bigg\lbrace\mathbf{D}_{n}\widehat{a}_{i}(t)\bigg\rbrace
=\mathbf{D}_{n}{a}_{i}^{E}+terms=terms
\end{align}
for initial static moduli or radii, and
\begin{align}
&\bm{\mathrm{I\!E}}\bigg\lbrace\mathbf{H}_{n}\widehat{\psi}_{i}(t)\bigg\rbrace=
\mathbf{H}_{n}\psi_{i}(t)+terms=terms
\\&\bm{\mathrm{I\!E}}\bigg\lbrace\mathbf{D}_{n}\widehat{a}_{i}(t)\bigg\rbrace
=\mathbf{D}_{n}{a}_{i}(t)+terms=terms
\end{align}
for randomly perturbed dynamical solutions. In each case, new terms are induced in the stochastically averaged equations. We first make the following preliminary definitions.
\begin{defn}
Given a set of Gaussian random fields $\mathscr{U}_{i}(t)\equiv\mathscr{U}_{i}(t)$ for $i=1...n$ then $\bm{\mathrm{I\!E}}\lbrace\mathscr{U}_{i}(t)\rbrace=0$ and the regulated covariance is $\mathsf{Cov}(t,s)
=\bm{\mathrm{I\!E}}\lbrace\mathscr{U}_{i}(t)\mathscr{U}_{j}(s)\rbrace
\rbrace=\delta_{ij}J(\Delta;\varsigma)$ with $|\Delta=|t-s|$ and $\varsigma$ a correlation length with $|\varsigma|\ll 1$. It is regulated so that $\mathsf{Cov}(t,t) =\bm{\mathrm{I\!E}}\lbrace\mathscr{U}_{i}(t)\mathscr{U}_{j}(t)\rbrace
\rbrace=N_{ij}J(0;\varsigma),\infty$, where $N_{ij}$ is an $n\times n$ matrix which can the Kronecker delta $\delta_{ij}$. Let $\mathscr{C}_{ij}(t)=\mathscr{U}_{i}(t)\mathscr{U}_{j}(t)$ then $\mathsf{Cov}(t,t)=\bm{\mathrm{I\!E}}\lbrace\mathscr{C}_{ij}(t)\rbrace$. Then
\begin{equation}
\sqrt{\mathlarger{\mathsf{Cov}}(t,t)}=\sqrt{\bm{\mathrm{I\!E}}\big\lbrace\mathscr{C}_{ij}(t)
\big\rbrace}=\sqrt{\delta_{ij}}\sqrt{J(0;\varsigma)}\equiv \delta_{ij}\sqrt{J(0;\varsigma)}<\infty
\end{equation}
We can define the following $L_{2}$ and Frobenius norms
\begin{align}
&\bm{\mathrm{I\!E}}\bigg\lbrace\sum_{i=1}^{n}\mathscr{U}_{i}(t)\mathscr{U}_{i}(t)\bigg\rbrace\equiv \sum_{i=1}^{n}\bm{\mathrm{I\!E}}\bigg\lbrace\mathscr{U}_{i}(t)\mathscr{U}_{i}(t)\bigg\rbrace
\equiv\underbrace{\sum_{i=1}^{n}\bm{\mathrm{I\!E}}\bigg\lbrace|\mathscr{U}_{i}(t)|^{2}\bigg\rbrace}_{via Fubini Thm} \nonumber\\&
=\sum_{i=1}^{n}\bigg|\sqrt{\bm{\mathrm{I\!E}}\big\lbrace|\mathscr{U}_{i}(t)|^{2}\big\rbrace}\bigg|^{2}\equiv
\bigg\|\sqrt{\bm{\mathrm{I\!E}}\big\lbrace|\mathscr{U}_{i}(t)|^{2}\big\rbrace}\bigg\|_{L_{2}}^{2}
\end{align}
and
\begin{align}
&\bm{\mathrm{I\!E}}\bigg\lbrace\sum_{i=1}^{n}\sum_{j=1}^{n}
\mathscr{U}_{i}(t)\mathscr{U}_{j}(t)\bigg\rbrace\equiv\nonumber\\&\equiv
\sum_{i=1}^{n}\sum_{j=1}^{n}\bm{\mathrm{I\!E}}\bigg\lbrace\mathscr{C}_{ij}(t)\bigg\rbrace\equiv\sum_{i=1}^{n}\sum_{j=1}^{n}\bigg|\sqrt{\bm{\mathrm{I\!E}}
\big\lbrace
\mathscr{C}_{ij}(t)}\big\rbrace\bigg|^{2}\equiv\bigg\|\sqrt{\bm{\mathrm{I\!E}}\big\lbrace
\mathscr{C}_{ij}(t)}\big\rbrace\bigg\|_{F}^{2}
\end{align}
\end{defn}
The stochastic averaging of the randomly perturbed Einstein system is established in the following theorem
\begin{thm}
Let $\mathbb{M}^{n+1}=\mathbb{T}^{n}\times\mathbb{R}^{(+)}$ with the conditions for an initially static torus in equilibrium.
\begin{enumerate}
\item  The initial data
    $\mathfrak{D}=[t=0,\Sigma_{o}=0,g_{ij}(0),\psi_{i}(0)=\psi_{i}^{E},a_{i}(0)=a_{i}^{E}]$ with $g_{oo}=-1,g_{io}=0$ with constraints $\mathbf{R}_{oo}=0$ and $\mathbf{R}_{io}=0$.
\item The Einstein vacuum equations $\mathbf{R}_{AB}=0$ are defined on $\mathbb{M}^{n+1}$ in the form $\mathbf{H}_{n}\psi_{E}=0$ for all $t\in\mathbb{R}^{+}$, and
    $\mathbf{D}_{n}a_{i}^{E}=0$.
\item The static radii are equal so the torus is static in that $a_{i}^{E}=a^{E}\equiv\exp(\psi_{i}^{E})$
\item The Gaussian random fields
$\widehat{\mathscr{U}}_{i}(t,\varsigma_{i})$ are colored or Gaussian-correlated such that $\bm{\mathrm{I\!E}}\lbrace\widehat{\mathscr{U}}_{i}(t)\rbrace=0$ and with a regulated 2-point function
\begin{align}
&\bm{\mathrm{I\!E}}\bigg\lbrace\widehat{\mathscr{U}}_{i}(t,\varsigma))\widehat{\mathscr{U}}_{i}(s,\varsigma)\bigg\rbrace=
N_{ij}\exp\left(-\frac{|\Delta|^{q}}{\varsigma^{q}}\right)\nonumber\\&=N_{ij}C\exp\left(-
\frac{|\Delta|^{q}}{\varsigma^{q}}\right)=N_{ij}J(\Delta;\varsigma)=\delta_{ij}J(\Delta;\varsigma)
\end{align}
where $\Delta=|t-s|$ with $N_{ij}$ an $n\times n$ matrix such as $\delta_{ij}$, $C>0$, and $q=1,2$. The equal-time correlation is finite or regulated so that $\bm{\mathrm{I\!E}}\lbrace\widehat{\mathscr{U}}_{i}(t))\widehat{\mathscr{U}}_{i}(t))\rbrace
=N_{ij}J(0;\varsigma)<\infty $ and the derivative $\partial_{t}\widehat{\mathscr{U}}(t)$ exists.(Appendix A.) Also $N_{ij}J(\Delta;\varsigma)\rightarrow 0$ for $t>|\varsigma|.$
\item The stochastically perturbed modulus functions and radii are (with $f(t)=1$)
\begin{align}
&\widehat{\psi}_{i}(t)=\psi_{i}^{E}+\zeta\int_{0}^{t}\widehat{\mathscr{U}}_{i}(\tau)d\tau
\equiv\psi{i}^{E}+\zeta\mathscr{X}_{i}(t)\\&
\widehat{a}_{i}(t)=a^{E}\exp\left(\zeta\int_{0}^{t}\widehat{\mathscr{U}}_{i}(\tau)d\tau\right)
\equiv a^{E}\mathscr{J}_{i}(t)
\end{align}
where $\zeta>0$ is a real parameter. The "toroidal random geometry" $\widehat{\mathbb{T}}^{n}$ is then defined by the stochastic (n+1)-metric with the stochastic average
\begin{align}
&\bm{\mathrm{I\!E}}\bigg\lbrace d\widehat{s}^{2}\bigg\rbrace=
-dt^{2}+\sum_{i=1}^{n}\sum_{j=1}^{n}\delta^{ij}|a_{i}^{E}|^{2}
\bm{\mathrm{I\!E}}\bigg\lbrace\mathscr{J}_{i}(t)\mathscr{J}_{i}(t)\bigg\rbrace
dX^{i}\otimes dX^{j}
\end{align}
\end{enumerate}
Equations (5.11)and (5.12) are then solutions of the stochastically averaged Einstein equations for this random geometry such that
\begin{align}
&\bm{\mathrm{I\!E}}\bigg\lbrace \mathbf{H}_{n}\widehat{\psi}_{i}(t)\bigg\rbrace
=\bm{\mathrm{I\!E}}\left\lbrace\sum_{i=1}^{n}\partial_{tt}\widehat{\psi}_{i}(t)+
+\frac{1}{2}\sum_{i=1}^{n}\partial_{t}\widehat{\psi}_{i}(t)\partial_{t}\widehat{\psi}_{i}(t)
+\frac{1}{2}\sum_{i=1}^{n}\sum_{j=1}^{n}\partial_{t}\widehat{\psi}_{i}(t)\partial_{t}\widehat{\psi}_{j}(t)
\right\rbrace=\lambda\\&\bm{\mathrm{I\!E}}\bigg\lbrace \mathbf{D}_{n}\widehat{a}_{i}(t)\bigg\rbrace=\bm{\mathrm{I\!E}}\left\lbrace\sum_{i=1}^{n}\frac{\partial_{tt}
\widehat{a}_{i}(t)}{\widehat{a}_{i}(t)}-\frac{1}{2}\sum_{i=1}^{n}\frac{\partial_{t}\widehat{a}_{i}(t)\partial_{t}\widehat{a}_{i}(t)}{\widehat{a}_{i}(t)
\widehat{a}_{j}(t)}+\frac{1}{2}\sum_{i=1}^{n}\sum_{j=1}^{n}\frac{\partial_{t}\widehat{a}_{i}(t)\partial_{t}
\widehat{a}_{i}(t)}{\widehat{a}_{i}(t)\widehat{a}_{j}(t)}\right\rbrace=\lambda
\end{align}
or
\begin{equation}
\bm{\mathrm{I\!E}}\bigg\lbrace \mathbf{H}_{n}\widehat{\psi}_{i}(t)\bigg\rbrace=
\frac{1}{2}\zeta^{2}\bigg\|\sqrt{\bm{\mathrm{I\!E}}\big\lbrace|\mathscr{U}_{i}(t)|^{2}\big\rbrace}
\bigg\|_{L_{2}}^{2}+\frac{1}{2}\zeta^{2}\bigg\|\sqrt{\bm{\mathrm{I\!E}}
\big\lbrace\mathscr{C}_{ij}(t)\big\rbrace}\bigg\|_{F}^{2}=\lambda_{1}+\lambda_{2}=\lambda
\end{equation}
\begin{equation}
\bm{\mathrm{I\!E}}\bigg\lbrace \mathbf{D}_{n}\widehat{a}_{i}(t)\bigg\rbrace=
\frac{1}{2}\zeta^{2}
\bigg\|\bm{\mathrm{I\!E}}\big\lbrace|\mathscr{U}_{i}(t)|^{2}\big\rbrace
\bigg\|_{L_{2}}^{2}+\frac{1}{2}\zeta^{2}\bigg\|\sqrt{\bm{\mathrm{I\!E}}
\big\lbrace\mathscr{C}_{ij}(t)\big\rbrace}\bigg\|_{F}^{2}=\lambda_{1}+\lambda_{2}=\lambda
\end{equation}
where $\lambda$ is an induced positive 'cosmological constant' term given by
\begin{align}
&\lambda=\frac{1}{2}\zeta^{2}
\bigg\|\sqrt{\bm{\mathrm{I\!E}}\big\lbrace|\mathscr{U}_{i}(t)|^{2}\big\rbrace}
\bigg\|_{L_{2}}^{2}+\frac{1}{2}\zeta^{2}\bigg\|\sqrt{\bm{\mathrm{I\!E}}
\big\lbrace\mathscr{C}_{ij}(t)
\big\rbrace}\bigg\|_{F}^{2}\nonumber\\&=\frac{1}{2}\zeta^{2}\sum_{i=1}^{n}\alpha_{ii}J(0;\varsigma)
+\frac{1}{2}\zeta^{2}\sum_{i=1}^{n} \sum_{j=1}^{n}\alpha_{ij}J(0;\varsigma)=\lambda_{1}+\lambda_{2}=\lambda
\end{align}
If $\widehat{\mathscr{U}}_{i}(t)=\widehat{\mathscr{U}}(t)$ for all $i=1...n$ and $\alpha_{ij}=\delta_{ij}$ then $\lambda=\tfrac{1}{2}\zeta^{2}n^{2}J(0;\varsigma)$
\end{thm}
\begin{proof}
The derivatives are $ \partial_{t}\widehat{\psi}_{i}(t)=\widehat{\mathscr{U}}_{i}(t)$ and $ \partial_{tt}\widehat{\psi}_{i}(t)=\partial_{t}\widehat{\mathscr{U}}_{i}(t)$ since $\widehat{\mathscr{U}}(t)=\int_{0}^{t}\widehat{\mathscr{U}}(\tau)d\tau$. The stochastically perturbed Einstein equations are
\begin{align}
&\mathbf{H}_{n}\widehat{\psi}_{i}(t)=\sum_{i=1}^{n}\partial_{tt}\widehat{\psi}_{i}(t)+
\frac{1}{2}\sum_{i=1}^{n}\sum_{j=1}^{n}\partial_{t}\widehat{\psi}_{i}(t)\partial_{t}\widehat{\psi}_{j}(t)+\frac{1}{2}
\sum_{i=1}^{n}\partial_{t}\widehat{\psi}_{i}(t)\partial_{t}\widehat{\psi}_{i}(t)\nonumber \\&=\zeta\sum_{i=1}^{n}\partial_{t}\widehat{\mathscr{U}}_{i}(t)+\frac{1}{2}\zeta^{2}\sum_{i=1}^{n}
\sum_{j=1}^{n}\widehat{\mathscr{U}}_{i}(t)\widehat{\mathscr{U}}_{j}(t)+\frac{1}{2}\zeta^{2}
\sum_{i=1}^{n}\widehat{\mathscr{U}}_{i}(t)\widehat{\mathscr{U}}_{i}(t)
\end{align}
Taking the stochastic average or mean $\bm{\mathrm{I\!E}}\lbrace...\rbrace$ produces new non-vanishing finite terms such that
\begin{align}
&\bm{\mathrm{I\!E}}\bigg\lbrace \mathbf{H}_{n}\widehat{\psi}_{i}(t)\bigg\rbrace=\bm{\mathrm{I\!E}}\left
\lbrace\zeta\sum_{i=1}^{n}\partial_{t}\widehat{\mathscr{U}}_{i}(t)+\frac{1}{2}\zeta^{2}
\sum_{i=1}^{n}\sum_{j=1}^{n}\widehat{\mathscr{U}}_{i}(t)
\widehat{\mathscr{U}}_{j}(t)+\frac{1}{2}\zeta^{2}\sum_{i=1}^{n}\widehat{\mathscr{U}}_{i}(t) \widehat{\mathscr{U}}_{i}(t)\right\rbrace\nonumber\\& \equiv \zeta
\sum_{i=1}^{n}\underbrace{\bm{\mathrm{I\!E}}\left\lbrace\partial_{t}\widehat{\mathscr{U}}_{i}(t)\right\rbrace}_{linear}+
\frac{1}{2}\zeta^{2}\underbrace{\sum_{i=1}^{n}\sum_{j=1}^{n}\bm{\mathrm{I\!E}}\left\lbrace
\widehat{\mathscr{U}}_{i}(t)\widehat{\mathscr{U}}_{j}(t)\right\rbrace}_{nonlinear}+
\frac{1}{2}\zeta^{2}\underbrace{\sum_{i=1}^{n}\bm{\mathrm{I\!E}}\left\lbrace\widehat{\mathscr{U}}_{i}(t)\widehat{\mathscr{U}}_{i}(t)\right\rbrace}_{nonlinear}
\end{align}
The linear term vanishes since $\bm{\mathrm{I\!E}}\lbrace\partial_{t}\widehat{\mathscr{U}}(t)\rbrace=0$ where $\bm{\mathrm{I\!E}}\lbrace\partial_{t}\widehat{\mathscr{U}}_{i}(t)\rbrace\equiv \partial_{t}\bm{\mathrm{I\!E}}\lbrace\widehat{\mathscr{U}}_{i}(t)\rbrace=0$ so that
\begin{align}
&\bm{\mathrm{I\!E}}\bigg\lbrace \mathbf{H}_{n}\widehat{\psi}_{i}(t)\bigg\rbrace=\frac{1}{2}\zeta^{2}\sum_{i=1}^{n} \bm{\mathrm{I\!E}}\bigg\lbrace\widehat{\mathscr{U}}_{i}(t)\widehat{\mathscr{U}}_{j}(t)
\bigg\rbrace+\frac{1}{2}\zeta^{2}\sum_{i=1}^{n}\sum_{j=1}^{n}\bm{\mathrm{I\!E}}
\bigg\lbrace\widehat{\mathscr{U}}_{j}(t)
\widehat{\mathscr{U}}_{i}(t)\bigg\rbrace\nonumber\\&
=\frac{1}{2}\zeta^{2}\sum_{i=1}^{n}\delta_{ii} \bm{\mathrm{I\!E}}\bigg\lbrace\widehat{\mathscr{U}}(t)_{i}\widehat{\mathscr{U}}_{i}(t)\bigg\rbrace
+\frac{1}{2}\zeta^{2}\sum_{i=1}^{n}\sum_{j=1}^{n}\delta_{ij}\bm{\mathrm{I\!E}}\bigg\lbrace
\widehat{\mathscr{U}}(t)_{i}\widehat{\mathscr{U}}_{j}(t)\bigg\rbrace\nonumber\\&
=\frac{1}{2}\sum_{i=1}^{n}\zeta^{2}\delta_{ii}J(0;\varsigma)
+\frac{1}{2}\zeta^{2}\sum_{i=1}^{n}\sum_{j=1}^{n}\delta_{ij}J(0;\varsigma)\\&
=\frac{1}{2}\zeta^{2}n J(0;\varsigma)+\frac{1}{2}\zeta^{2}nJ(0;\varsigma)=
\zeta^{2}n J(0;\varsigma)\equiv\lambda_{1}+\lambda_{2}=\lambda
\end{align}
Or equivalently, to get (5.16)
\begin{align}
&\bm{\mathrm{I\!E}}\bigg\lbrace \mathbf{H}_{n}\widehat{\psi}_{i}(t)\bigg\rbrace=\frac{1}{2}\zeta^{2}\sum_{i=1}^{n} \bm{\mathrm{I\!E}}\bigg\lbrace\widehat{\mathscr{U}}_{i}(t)\widehat{\mathscr{U}}_{j}(t)
\bigg\rbrace+\frac{1}{2}\zeta^{2}\sum_{i=1}^{n}\sum_{j=1}^{n}\bm{\mathrm{I\!E}}
\bigg\lbrace\widehat{\mathscr{U}}_{j}(t)
\widehat{\mathscr{U}}_{i}(t)\bigg\rbrace\nonumber\\&
=\frac{1}{2}\zeta^{2}\sum_{i=1}^{n}\delta_{ii} \bm{\mathrm{I\!E}}\bigg\lbrace\widehat{\mathscr{U}}_{i}(t)\widehat{\mathscr{U}}_{i}(t)\bigg\rbrace
+\frac{1}{2}\zeta^{2}\sum_{i=1}^{n}\sum_{j=1}^{n}\delta_{ij}\bm{\mathrm{I\!E}}\bigg\lbrace
\widehat{\mathscr{U}}_{i}(t)\widehat{\mathscr{U}}_{j}(t)\bigg\rbrace\nonumber\\&
=\frac{1}{2}\zeta^{2}\sum_{i=1}^{n}\bm{\mathrm{I\!E}}\bigg\lbrace\big|\mathscr{U}_{i}|^{2}\bigg\rbrace+
\frac{1}{2}\zeta^{2}\sum_{i=1}^{n}\sum_{j=1}^{n}\bm{\mathrm{I\!E}}\bigg\lbrace
\mathscr{C}_{ij}(t)\bigg\rbrace\nonumber\\&
=\frac{1}{2}\zeta^{2}\sum_{i=1}^{n}\bigg|\sqrt{\bm{\mathrm{I\!E}}\big\lbrace\big|\mathscr{U}_{i}|^{2}\big\rbrace}\bigg|^{2}+
\frac{1}{2}\zeta^{2}\sum_{i=1}^{n}\sum_{j=1}^{n}\bigg|\sqrt{\bm{\mathrm{I\!E}}\big\lbrace
\mathscr{C}_{ij}(t)\big\rbrace}\bigg|^{2}\nonumber\\&=\frac{1}{2}\zeta^{2}
\bigg\|\sqrt{\bm{\mathrm{I\!E}}\big\lbrace|\mathscr{U}_{i}(t)|^{2}\big\rbrace}
\bigg\|_{L_{2}}^{2}+\frac{1}{2}\zeta^{2}\bigg\|\sqrt{\bm{\mathrm{I\!E}}
\big\lbrace\mathscr{C}_{ij}(t)
\big\rbrace}\bigg\|_{F}^{2}=\lambda_{1}+\lambda_{2}=\lambda
\end{align}
The same result must also follow from the nonlinear Einstein system of ODES for the radii. The derivatives of $\widehat{a}_{i}(t)$ are
\begin{align}
&\partial_{t}\widehat{a}_{i}(t)=a_{i}^{E}\zeta\widehat{\mathscr{U}}_{i}(t)
\exp\left(\zeta\int_{0}^{t}\widehat{\mathscr{U}}(\tau)d\tau\right)\equiv a_{i}^{E}\zeta\widehat{\mathscr{U}}_{i}(t)\widehat{\mathscr{J}}_{i}(t)
\end{align}
and $ \partial_{tt}\widehat{a}_{i}(t)=a_{i}^{E}\zeta\widehat{\mathscr{U}}_{i}(t)
\widehat{\mathscr{U}}_{i}(t)\widehat{\mathscr{J}}_{i}(t)+a_{i}^{E}\zeta\partial_{t}
\widehat{\mathscr{U}}_{i}(t)\widehat{\mathscr{J}}_{i}(t)$. The stochastically perturbed Einstein equations for the toroidal radii are
\begin{align}
\mathbf{D}_{n}\widehat{a}_{i}(t)&=\zeta\sum_{i=1}^{n}\frac{\partial_{tt}\widehat{a}_{i}(t)}{\widehat{a}_{i}(t)}-\frac{1}{2}\zeta^{2}
\sum_{i=1}^{n}\frac{\partial_{t}\widehat{a}_{i}(t)\partial_{t}\widehat{a}_{i}(t)}{\widehat{a}_{i}(t)
\widehat{a}_{j}(t)}+\frac{1}{2}\zeta^{2}\sum_{i=1}^{n}\sum_{j=1}^{n}\frac{\partial_{t}\widehat{a}_{i}(t)\partial_{t}
\widehat{a}_{i}(t)}{\widehat{a}_{i}(t)\widehat{a}_{j}(t)}\nonumber\\&
=\zeta\sum_{i=1}^{n}\frac{a_{i}^{E}\partial_{t}\widehat{\mathscr{U}}(t)\widehat{\mathscr{J}}_{i}(t)}
{a_{i}^{E}\widehat{\mathscr{J}}_{i}(t)}+\zeta^{2}\sum_{i=1}^{n}\frac{a_{i}^{E}\widehat{\mathscr{U}}_{i}(t)\widehat{\mathscr{U}}_{i}(t)
\widehat{\mathscr{J}}_{i}(t)}{a_{i}^{E}\widehat{\mathscr{J}}_{i}(t)}\nonumber\\&-\frac{1}{2}\zeta^{2}\sum_{i=1}^{n}\frac{a_{i}^{E}a_{i}^{E}
\widehat{\mathscr{U}}_{i}(t)\widehat{\mathscr{U}}_{i}(t)\widehat{\mathscr{J}}_{i}(t)\widehat{\mathscr{J}}_{i}(t)}
{a_{i}\widehat{\mathscr{J}}_{i}(t)a_{i}\widehat{\mathscr{J}}_{i}(t)} +\frac{1}{2}\zeta^{2}\sum_{i=1}^{n}\sum_{j=1}^{n}\frac{a_{i}^{E}a_{j}^{E}\widehat{\mathscr{U}}_{i}(t)
\widehat{\mathscr{J}}_{j}(t)\widehat{\mathscr{J}}_{i}(t)\widehat{\mathscr{J}}_{j}(t)}{a_{i}
\widehat{\mathscr{J}}_{i}(t)a_{i}\widehat{\mathscr{J}}_{j}(t)}\nonumber\\&=\sum_{i=1}^{n}\zeta\partial_{t}\widehat{\mathscr{U}}(t)+\frac{1}{2}\zeta^{2}\sum_{i=1}^{n}\widehat{\mathscr{U}}_{i}(t)\widehat{\mathscr{U}}_{i}(t)+
\frac{1}{2}\zeta^{2}\sum_{i=1}^{n}\sum_{j=1}^{n}\widehat{\mathscr{U}}_{i}(t)\widehat{\mathscr{U}}_{j}(t)
\end{align}
Taking the stochastic average $\bm{\mathrm{I\!E}}\lbrace...\rbrace$ gives
\begin{align}
&\bm{\mathrm{I\!E}}\bigg\lbrace \mathbf{D}_{n}\widehat{a}_{i}(t)\bigg\rbrace=
\frac{1}{2}\zeta^{2}\sum_{i=1}^{n}\bm{\mathrm{I\!E}}\left\lbrace\widehat{\mathscr{U}}_{i}(t)\widehat{\mathscr{U}}_{i}(t)\right\rbrace+
\frac{1}{2}\zeta^{2}\sum_{i=1}^{n}\sum_{j=1}^{n}\bm{\mathrm{I\!E}}\left\lbrace\widehat{\mathscr{U}}_{i}(t)\widehat{\mathscr{U}}_{j}(t)
\right\rbrace\nonumber\\&=\frac{1}{2}\zeta^{2}\sum_{i=1}^{n}\delta_{ii}
\bm{\mathrm{I\!E}}\left\lbrace\widehat{\mathscr{U}}_{i}(t)
\widehat{\mathscr{U}}_{i}(t)\right\rbrace+\frac{1}{2}\zeta^{2}\sum_{i=1}^{n}\sum_{j=1}^{n}\delta_{ij}
\bm{\mathrm{I\!E}}\left\lbrace\widehat{\mathscr{U}}_{i}(t)\widehat{\mathscr{U}}_{j}(t)\right\rbrace\nonumber\\&=\zeta^{2}\sum_{i=1}^{n}
\delta_{ii}J(0,\varsigma)+\zeta^{2}\sum_{i=1}^{n}\sum_{j=1}^{n}\delta_{ij}J(0;\varsigma)\nonumber\\&
=\frac{1}{2}\zeta^{2}n J(0;\varsigma)+\frac{1}{2}\zeta^{2}nJ(0;\varsigma)=\zeta^{2}
nJ(0;\varsigma)\equiv \lambda
\end{align}
as required.
\end{proof}
\begin{cor}
If the form (5.1) is used for the random perturbations with $\xi(t)=1$ then the stochastic averaging $\bm{\mathrm{I\!E}}\lbrace...\rbrace$ gives
\begin{align}
\bm{\mathrm{I\!E}}\bigg\lbrace \mathbf{H}_{n}\widehat{\psi}_{i}(t)\bigg\rbrace&=\bm{\mathrm{I\!E}}\left
\lbrace\zeta {\xi}(t)\sum_{i=1}^{n}\partial_{t}\widehat{\mathscr{U}}_{i}(t)+
+\zeta\sum_{i=1}^{n}\partial_{t}F(t)\widehat{\mathscr{U}}_{i}(t)\right\rbrace\nonumber\\&
+\frac{1}{2}\bm{\mathrm{I\!E}}\left\lbrace\zeta^{2}\sum_{i=1}^{n}\sum_{j=1}^{n}{F}(t){F}(t)
\widehat{\mathscr{U}}_{i}(t)\widehat{\mathscr{U}}_{j}(t)+\frac{1}{2}\zeta^{2}{F}(t){F}(t)\sum_{i=1}^{n}
\widehat{\mathscr{U}}_{i}(t)\widehat{\mathscr{U}}_{i}(t)\right\rbrace\nonumber\\&\equiv \zeta\sum_{i=1}^{n}\bm{\mathrm{I\!E}}\left\lbrace\partial_{t}\widehat{\mathscr{U}}_{i}(t)
\right\rbrace+\frac{1}{2}\zeta^{2}{F}(t){F}(t)\sum_{i=1}^{n}\sum_{j=1}^{n}\bm{\mathrm{I\!E}}\left\lbrace
\widehat{\mathscr{U}}_{i}(t)\widehat{\mathscr{U}}_{j}(t)\right\rbrace\nonumber\\&+
\frac{1}{2}\zeta^{2}{F}(t){F}(t)\sum_{i=1}^{n}\bm{\mathrm{I\!E}}
\left\lbrace\widehat{\mathscr{U}}_{i}(t)\widehat{\mathscr{U}}_{i}(t)\right\rbrace
\end{align}
The first term vanishes since $\bm{\mathrm{I\!E}}\lbrace\partial_{t}\widehat{\mathscr{U}}(t)\rbrace=0$ where $\bm{\mathrm{I\!E}}\lbrace\partial_{t}\widehat{\mathscr{U}}_{i}(t)\rbrace\equiv \partial_{t}\bm{\mathrm{I\!E}}\lbrace\widehat{\mathscr{U}}_{i}(t)\rbrace=0$ so that
\begin{align}
\bm{\mathrm{I\!E}}\bigg\lbrace\mathbf{H}_{n}\widehat{\psi}_{i}(t)\bigg\rbrace&=\frac{1}{2}
\zeta^{2}f(t)f(t)\sum_{i=1}^{n}\bm{\mathrm{I\!E}}\bigg\lbrace\widehat{\mathscr{U}}_{i}(t)
\widehat{\mathscr{U}}_{j}(t)\bigg\rbrace+\frac{1}{2}\zeta^{2}f(t)f(t)\sum_{i=1}^{n}\sum_{j=1}^{n}
\bm{\mathrm{I\!E}}\bigg\lbrace\widehat{\mathscr{U}}_{j}(t)\widehat{\mathscr{U}}_{i}(t)\bigg\rbrace\nonumber\\&
=\frac{1}{2}\zeta^{2}{f}(t){f}(t)\sum_{i=1}^{n}\delta_{ii}
\bm{\mathrm{I\!E}}\bigg\lbrace\widehat{\mathscr{U}}(t)_{i}\widehat{\mathscr{U}}_{i}(t)\bigg\rbrace
+\frac{1}{2}\zeta^{2}{f}(t){f}(t)\sum_{i=1}^{n}\sum_{j=1}^{n}\delta_{ij}\bm{\mathrm{I\!E}}
\bigg\lbrace
\widehat{\mathscr{U}}(t)_{i}\widehat{\mathscr{U}}_{j}(t)\bigg\rbrace\nonumber\\&
=\frac{1}{2}\sum_{i=1}^{n}\zeta^{2}{f}(t){f}(t)\delta_{ii}J(0;\varsigma)
+\frac{1}{2}\zeta^{2}{f}(t){f}(t)\sum_{i=1}^{n}\sum_{j=1}^{n}\delta_{ij}J(0;\varsigma)\nonumber\\&
=\frac{1}{2}\zeta^{2}n{f}(t){f}(t)J(0;\varsigma)+\frac{1}{2}{f}(t){f}(t)\zeta^{2}n J(0;\varsigma)=\zeta^{2}n{f}(t){f}(t)J(0;\varsigma)\equiv\lambda
\end{align}
Then
\begin{equation}
\bm{\mathrm{I\!E}}\bigg\lbrace \mathbf{H}_{n}\widehat{\psi}_{i}(t)\bigg\rbrace=\zeta^{2}n{f}(t){f}(t)J(0;\varsigma)\equiv\lambda(t)
\end{equation}
and the induced cosmological constant term is now time dependent. Similarly,
\begin{equation}
\bm{\mathrm{I\!E}}\bigg\lbrace \mathbf{D}_{n}\widehat{a}_{i}(t)\bigg\rbrace=\zeta^{2}n{f}(t){f}(t)J(0;\varsigma)\equiv\lambda(t)
\end{equation}
\end{cor}
\begin{rem}
The induced cosmological constant term in the stochastically averaged vacuum equations arises purely from the nonlinearity of the Einstein equations. Intrinsic stochastic fluctuations of the of the moduli therefore act like a "dark energy". The cosmological constant $\lambda$ can be made as small as required by fine tuning the parameter $\zeta$ or reducing or 'diluting' the intensity of the fluctuations so that $J(0;\varsigma)>0$ but with $J(0;\varsigma)\sim 0$.
\end{rem}
\begin{rem}
Reprising Remark $1.1$, there is no analog of this for a purely linear theory. The non-vanishing terms which arise in the stochastically averaged system of equations are due to the nonlinearity of the equations. For example, given the linear ODEs
\begin{equation}
\mathbf{L}_{n}\psi_{i}(t)\equiv \sum_{i=1}^{n}\frac{\partial_{t}\widehat{a}_{i}(t)}{\widehat{a}_{i}(t)}=0
\end{equation}
then there are trivial equilibrium solutions $\psi_{i}(t)=\psi^{E}$ such that $\mathbf{L}_{n}\psi^{E}=0$. If $\widehat{\psi}_{i}(t)
=\psi^{E}+\zeta\int_{0}^{t}\widehat{\mathscr{U}}_{i}(\tau)d\tau$. For a linear system, the stochastically averaged equations will reduce back to the original deterministic equations. The randomly perturbed equations are
\begin{equation}
\mathbf{L}_{n}\widehat{a}_{i}(t) = \sum_{i=1}^{n}\frac{\partial_{t}
\widehat{a}_{i}(t)}{\widehat{a}_{i}(t)}=\sum_{i=1}^{n}\frac{a_{i}(t)\widehat{\mathscr{U}}_{i}(t)
\widehat{\mathscr{J}}_{i}(t)+\widehat{\mathscr{J}}_{i}(t)
\partial_{t}a_{i}(t)}{a_{i}(t)\widehat{\mathscr{J}}_{i}(t)}
\end{equation}
The stochastically averaged equations are then
\begin{align}
&\bm{\mathrm{I\!E}}\lbrace\mathbf{L}_{n}\widehat{a}_{i}(t)\rbrace=
\bm{\mathrm{I\!E}}\left\lbrace\sum_{i=1}^{n}\frac{\partial_{t}
\widehat{a}_{i}(t)}{\widehat{a}_{i}(t)}\right\rbrace\nonumber\\&=\sum_{i=1}^{n}\frac{a_{i}(t)
\bm{\mathrm{I\!E}}\left\lbrace\widehat{\mathscr{U}}_{i}(t)
\right\rbrace+\partial_{t}a_{i}(t)}{a_{i}(t)}
=\sum_{i=1}^{n}\frac{\partial_{t}a_{i}(t)}{a_{i}(t)}=0
\end{align}
since $\bm{\mathrm{I\!E}}\lbrace\widehat{\mathscr{N}}_{i}(t)\rbrace=0$ for a Gaussian random field, and so no new constant terms are induced for a stochastically averaged linear set of ODEs.
\end{rem}
The random perturbations of the radial moduli must have a regulated covariance to induce cosmological constant terms in the stochastically averaged Einstein system, that are both finite and positive. White noise perturbations will induce delta-function singularities in the averaged system.
\begin{lem}
let $\psi_{i}^{E}$ be a set of equilibrium moduli solutions of the Einstein system
$\mathbf{H}_{n}\psi_{i}^{E}=0$ and let $\psi_{i}(t)$ be a dynamical set of solutions such that $\mathbf{H}_{n}\psi_{i}(t)=0$. Let $(\mathscr{N}_{i}(t))_{i=1}^{n}$ be a set of Gaussian white noises with $\bm{\mathrm{I\!E}}\lbrace\mathscr{N}_{i}(t)$=0 and $\bm{\mathrm{I\!E}}\lbrace\mathscr{N}_{i}(t)\mathscr{N}_{j}(t)\rbrace=\alpha\delta){ij}\delta(t-s)$
such that the randomly perturbed moduli are
\begin{align}
&\widehat{\psi}_{i}(t)=\psi_{i}(t)+\zeta\int_{0}^{t}f(\psi(s))d\mathscr{W}(s)
\equiv\psi_{i}(t)+\zeta\int_{0}^{t}f(\psi(s))\mathscr{N}(s)ds
\\&\widehat{\psi}_{i}(t)=\psi_{i}^{E}+\zeta\int_{0}^{t}f(\psi(s))d\mathscr{W}(s)
\equiv\psi_{i}^{E}+\zeta\int_{0}^{t}f(\psi_{i}^{E})\mathscr{N}(s)ds
\end{align}
where $f(\psi(t))$ is some smooth $C^{2}$-differentiable functional and $d\mathscr{W}(t)=\mathscr{n}(t)dt$ is the standard Brownian motion or Weiner process. This is equivalent to the stochastic DE or 'Langevin equation'
\begin{equation}
d\widehat{\psi}_{i}(t)=d\psi_{i}(t)+\zeta f(\psi(t))d\mathscr{W}_{i}(t)
\end{equation}
The stochastically averaged Einstein system then has a delta-function singularity such that
\begin{equation}
\bm{\mathrm{I\!E}}\bigg\lbrace\mathbf{H}_{n}\widehat{\psi}_{i}(t)\bigg\rbrace
=\mathbf{H}_{n}\psi_{i}(t)+n\alpha\delta(0)=\infty
\end{equation}
\end{lem}
\begin{proof}
The derivatives are
\begin{align}
&\partial_{t}\widehat{\psi}_{i}(t)=\partial_{t}\psi_{i}(t)+\zeta f(\psi_{i}(t))\mathscr{N}_{i}(t)\nonumber\\&
\partial_{tt}\widehat{\psi}_{i}(t)=\partial_{tt}\psi_{i}(t)
+\zeta\partial_{t}f(\psi_{i}(t))\mathscr{N}_{i}(t))+\zeta f(\psi_{i}(t))\partial_{t}\mathscr{N}_{i}(t)
\end{align}
while the derivative $\partial_{t}\mathscr{N}_{i}(t)$  does not formally exist, we can still take the expectation since $\bm{\mathrm{I\!E}}\lbrace\partial_{t}\mathscr{N}(t)\rbrace\equiv \partial_{t}\bm{\mathrm{I\!E}}\lbrace\mathscr{N}(t)\rbrace$
\begin{align}
&\bm{\mathrm{I\!E}}\bigg\lbrace\partial_{tt}\widehat{\psi}_{i}(t)\bigg\rbrace=\partial_{tt}\psi_{i}(t)
+\zeta\partial_{t}f(\psi_{i}(t))\bm{\mathrm{I\!E}}\bigg\lbrace\mathscr{S}_{i}(t))\bigg\rbrace\nonumber\\&
+\zeta f(\psi_{i}(t))\partial_{t}\bm{\mathrm{I\!E}}\bigg\lbrace \mathscr{S}_{i}(t)\bigg\rbrace=\partial_{tt}\psi_{i}(t)
\end{align}
The averaged Einstein system is then
\begin{align}
\bm{\mathrm{I\!E}}\bigg\lbrace\mathbf{H}_{n}\widehat{\psi}_{i}(t)\bigg\rbrace&=\mathbf{H}_{n}\psi_{i}(t)
+\frac{1}{2}\sum_{i=1}^{n}|f(\psi_{i}(t)|^{2}\bm{\mathrm{I\!E}}\bigg\lbrace\mathscr{N}_{i}(t)\mathscr{N}_{i}(t)
\bigg\rbrace\nonumber\\&+\frac{1}{2}\sum_{i=1}^{n}\sum_{j=1}^{n}f(\psi_{i}(t)f(\psi_{j}(t)\bm{\mathrm{I\!E}}
\bigg\lbrace\mathscr{N}_{i}(t)\mathscr{N}_{j}(t)\bigg\rbrace\nonumber\\&
=\mathbf{H}_{n}\psi_{i}(t)+\frac{1}{2}\alpha|f(\psi_{i}(t)|^{2}\delta_{ij}\delta(0)
+\frac{1}{2}\sum_{i=1}^{n}\sum_{j=1}^{n}\alpha\delta_{ij}|f(\psi_{i}(t))f(\psi_{j}(t))|\delta(0)
\nonumber\\&= \mathbf{H}_{n}\psi_{i}(t)+\frac{1}{2}\alpha|f(\psi_{i}(t)|^{2}n\delta(0)
+\frac{1}{2}\alpha|f(\psi(t)|^{2}n\delta(0)\nonumber\\&=\alpha n|f(\psi(t))|^{2}\delta(0)=\infty
\end{align}
\end{proof}
\begin{lem}
Given white-noise perturbations of the static radial moduli set $(\psi_{i}(t))_{i=1}^{n}$ of the form
\begin{equation}
\widehat{\psi}(t)=\psi_{i}^{E}+\zeta\int_{0}^{t}f(\psi_{i}(s))d\mathscr{W}(s)
\end{equation}
or
\begin{equation}
\widehat{\psi}(t)=\psi_{i}(t)+\zeta\int_{0}^{t}f(\psi_{i}(s))d\mathscr{W}(s)
\end{equation}
with the conditions
\begin{equation}
\|f_{i}(\psi_{i}(t))\|^{2} \le K\|\psi_{i}(t)\|^{2}
\end{equation}
and
\begin{equation}
\int_{t_{0}}^{t}\|f(\psi_{i}(s)\|^{2}ds < \infty
\end{equation}
then the $l^{th}$-order moments are finite and bounded for all finite $t>t_{o}$ and grow exponentially,with the estimates
\begin{align}
&\bm{\mathrm{I\!E}}\bigg\lbrace\|\sup_{t\le T}\widehat{\psi}_{i}(t)\|^{\ell}\bigg\rbrace\le \|\psi(t)\|^{\ell}\exp(\tfrac{1}{2}K\ell(\ell-1)|T-t_{0}|)
\\&
\bm{\mathrm{I\!E}}\bigg\lbrace\|\sup_{t\le T}\widehat{\psi}_{i}(t)\|^{\ell}\bigg\rbrace\le \|\psi_{\epsilon}\|^{\ell}\exp(\tfrac{1}{2}K\ell(\ell-1)|T-t_{0}|)
\end{align}
\end{lem}
The proof is given in Appendix B.

As an example of specific regulated random modulus perturbations, one can apply an Ornstein-Uhlenbeck process, which is well defined [21,22].
\begin{lem}
Let the radial modulus perturbations be of the form
\begin{equation}
\widehat{\psi}_{i}(t)=\psi^{E}+\zeta\int_{0}^{t}\mathscr{U}_{i}(s)ds
\equiv\psi^{E}+\zeta\int_{0}^{t}\mathscr{O}_{i}(s)ds
\end{equation}
where $\mathscr{O}_{i}(s)$ is the OU process. Then
\begin{align}
\partial_{t}\widehat{\psi}_{i}(t)&=\mathscr{O}_{i}(t)=\exp(-At)\mathscr{O}_{i}(0)+\sigma\exp(-At)\int_{0}^{\infty}
\exp(As)d\mathscr{W}(s)\nonumber\\&
=\exp(-At)\mathscr{O}_{i}(0)+\sigma\int_{0}^{t}\exp(-A|t-s|)d\mathscr{W}(s)
\end{align}
if $\mathscr{O}_{i}(0)=0$. This is a solution of the linear stochastic DE
\begin{equation}
d\mathscr{O}_{i}(t)=-A\mathscr{O}_{i}(t)dt+\sigma d\mathscr{W}(t)
\end{equation}
Then the covariance is
\begin{align}
&\bm{\mathrm{I\!E}}\bigg\lbrace\partial_{t}\widehat{\psi}_{i}(t)\partial_{t}\widehat{\psi}_{i}(s)\bigg\rbrace
=\bm{\mathrm{I\!E}}\bigg\lbrace\mathscr{O}_{i}(t)\mathscr{O}_{j}(s)\bigg\rbrace\nonumber\\&
=\bm{\mathrm{I\!E}}\bigg\lbrace\left|\sigma\int_{0}^{t}\exp(-A|t-s|)d\mathscr{W}(s)\right|^{2}\bigg\rbrace\nonumber\\&
=A\delta_{ij}\exp(-A|t-s|)\equiv \delta_{ij}J(\Delta;\sigma)
\end{align}
which is regulated such that $ \bm{\mathrm{I\!E}}\lbrace\partial_{t}\widehat{\psi}_{i}(t)\partial_{t}\widehat{\psi}_{j}(t)\rbrace
=A\delta_{ij}$. The averaged Einstein system is then
\begin{equation}
\bm{\mathrm{I\!E}}\bigg\lbrace\mathbf{H}_{n}\widehat{\psi}_{i}(t)\bigg\rbrace
=\mathbf{H}_{n}\psi_{i}(t)+nA\sigma^{2}=nA\sigma^{2}
\equiv \lambda>0
\end{equation}
\end{lem}
\begin{proof}
From (5.48), the average of the 2nd derivative is zero, so that $\bm{\mathrm{I\!E}}\lbrace\partial_{tt}\widehat{\psi}(t)\rbrace=0$. The averaged randomly perturbed Einstein system is then
\begin{align}
\bm{\mathrm{I\!E}}\bigg\lbrace\mathbf{H}_{n}\widehat{\psi}_{i}(t)\bigg\rbrace &
=\frac{1}{2}\sum_{i=1}^{n}\bm{\mathrm{I\!E}}\bigg\lbrace\left|\sigma\int_{0}^{t}\exp(-A|t-s|)d\mathscr{W}_{i}(s)
\right|^{2}\bigg\rbrace\nonumber\\&+\frac{1}{2}\sum_{i=1}^{n}\sum_{j=1}^{n}
\bm{\mathrm{I\!E}}\bigg\lbrace\left|\sigma\int_{0}^{t}\exp(-A|t-s|)d\mathscr{W}_{i}(s)
\right|\left|\sigma\int_{0}^{t}\exp(-A|t-s|)d\mathscr{W}_{j}(s)\right|\bigg\rbrace\nonumber\\&
\equiv\frac{1}{2}\bm{\mathrm{I\!E}}\bigg\lbrace\bigg\|\sigma\int_{0}^{t}\exp(-A|t-s|)d\mathscr{W}_{i}(s)
\bigg\|^{2}\bigg\rbrace\nonumber\\&+\bm{\mathrm{I\!E}}\bigg\lbrace\bigg\|\sigma^{2}\int_{0}^{t}\exp(-A|t-s|)d\mathscr{W}_{i}(s)
\int_{0}^{t}\exp(-A|t-s|)d\mathscr{W}_{j}(s)\bigg\|^{2}\bigg\rbrace\nonumber\\&
\equiv\frac{1}{2}\sigma^{2}\sum_{i=1}^{n}\delta_{ii}A+\frac{1}{2}\sigma^{2}\sum_{i=1}^{n}\sum_{j=1}^{n}\delta_{ij}A
\nonumber\\&=\frac{1}{2}\sigma^{2}nA+\frac{1}{2}\sigma^{2}nA=\sigma^{2}nA\equiv\lambda
\end{align}
\end{proof}
Similarly, one could average the equations in terms of the perturbed radii to obtain the same result.
\subsection{Analogy with statistical hydrodynamical turbulence and averaged Navier-Stokes equations}
\begin{rem}
The stochastically averaged Einstein equations giving induced cosmological constant terms are also strongly analogous to what occurs when random fields are coupled to the Navier-Stokes equations in order to incorporate or describe the randomness of turbulence. Because the Navier-Stokes equations are also nonlinear a non-vanishing term arises or is induced when the stochastic Navier-Stokes equations are stochastically averaged or 'Reynolds averaged'. This is essentially an induced Reynolds stress tensor or Reynolds number [26,27,28,29,30]. The Einstein and Navier-Stokes equations are also similar in that they are nonlinear PDEs of hyperbolic type and both describe the dynamical evolution of a continuum approximation'.
\end{rem}
Suppose non-white random vector fields $\widehat{\mathscr{U}}_{i}(x,t)$ are coupled to a fluid flow $u_{i}(x,t)\subset\mathbb{D}\subset\mathbb{R}^{n}$ in a domain $\mathbb{D}$ and described by the Navier-Stokes equations with some boundary conditions and initial data. In order to incorporate or describe turbulence in the fluid, one can couple a random field $\widehat{\mathscr{U}}_{i}(x,t)$, usually Gaussian, that varies randomly in both space and time. The following theorem for stochastically averaged Navier-Stokes equations for a randomly perturbed fluid flow, is then an analog of Theorem (5.1).
\begin{thm}
The random field has the following properties
\begin{enumerate}
\item The fluid flow satisfies the laminar nonlinear Navier-Stokes PDEs
\begin{align}
\mathbf{N}_{n}\widehat{u}_{i}(x,t)&=\sum_{i=1}^{n}\partial_{t}\widehat{u}_{i}(x,t)+
\underbrace{\sum_{i=1}^{n}\sum_{j=1}^{n}\partial^{j}[\widehat{u}_{j}(x,t)\widehat{u}_{i}(x,t)]}_{nonlinear~ convective~term}\nonumber\\&
-\nu\sum_{i=1}^{n}\triangle\widehat{u}_{i}(x,t)+\sum_{i=1}^{n}\partial_{i}p(x,t)-\sum_{i=1}^{n}f_{i}(x,t)
\end{align}
with some suitable initial data and boundary conditions.
\item The smooth/laminar Navier-Stokes flow $u_{i}(x,t)$ is randomly perturbed as
\begin{equation}
\widehat{u}_{i}(x,t)=u_{i}(x,t)+\widehat{\mathscr{U}}_{i}(x,t)
\end{equation}
\item The spatio-temporal random field $\widehat{\mathscr{U}}_{i}(x,t)$ is Gaussian with expectation $\bm{\mathrm{I\!E}}\lbrace\widehat{\mathscr{U}_{i}}(x,t)\rbrace=0$
\item The 2-point function is of the form
\begin{align}
&\mathcal{R}_{ij}(x,y,t,s)=\bm{\mathrm{I\!E}}\bigg\lbrace\widehat{\mathscr{U}}_{i}(x,t)
\widehat{\mathscr{U}}_{j}(y,s)\bigg\rbrace
=\Xi_{ij}(|x-y|;\xi)J(|t-s|;\varsigma)\nonumber\\&\equiv\delta_{ij}\Xi(|x-y|;\xi)J(|t-s|;\varsigma)
\end{align}
where $\xi$ is a spatial correlation length and $\varsigma$ a temporal correlation. Then $\mathcal{R}_{ij}(x,y,t,s)\rightarrow 0$ for $|x-y|\gg\xi$ and/or $|t-s|\gg\varsigma$.
\item It is regulated such that
\begin{align}
\mathcal{R}_{ij}(x,x,t,t)&=\lim_{x\uparrow y}\lim_{t\uparrow s}\bm{\mathrm{I\!E}}\bigg\lbrace\widehat{\mathscr{U}}_{i}(x,t)
\widehat{\mathscr{U}}_{j}(y,s)\bigg\rbrace\nonumber\\&=\lim_{x\uparrow y}\lim_{t\uparrow
s}\delta_{ij}\Xi(|x-y|;\varsigma)J(|t-s|;\varsigma)
=\delta_{ij}\Xi(0;\xi)J(0;\varsigma)<\infty
\end{align}
\item The derivative exists such that
\begin{align}
&\partial^{j}\mathcal{R}_{ij}(x,y,t,s)=\partial^{j}\bm{\mathrm{I\!E}}
\bigg\lbrace\widehat{\mathscr{U}}_{i}(x,t)\widehat{\mathscr{U}}_{j}(y,s)\bigg\rbrace
\nonumber\\&=(\partial^{j}\Xi_{ij}(|x-y|;\xi))J(|t-s|;\varsigma)\equiv\delta_{ij}
\partial^{j}\Xi(|x-y|;\xi)J(|t-s|;\varsigma)
\end{align}
and the limit of the derivative is regulated such that
\begin{align}
&\partial^{j}\mathcal{R}_{ij}(x,x,t,t)=\lim_{x\uparrow y}\lim_{t\uparrow s}\partial{j}\bm{\mathrm{I\!E}}\bigg\lbrace\widehat{\mathscr{U}}_{i}(x,t)\widehat{\mathscr{U}}_{j}(y,s)\bigg\rbrace
\nonumber\\&=\lim_{x\uparrow y}\lim_{t\uparrow s}\delta_{ij}\partial^{j}\Xi(|x-y|;\xi)J(|t-s|;\varsigma)=
\delta_{ij}\partial^{j}\Xi(0;\xi)J(0;\varsigma)<\infty
\end{align}
\end{enumerate}
Then the stochastically averaged NS equations are
\begin{equation}
\bm{\mathrm{I\!E}}\bigg\lbrace\mathbf{N}_{n}\mathscr{U}_{i}(x,t)\bigg\rbrace
=\mathbf{N}_{n}u_{i}(x,t)+\partial^{j}\mathcal{R}_{ij}(x,x;t,t)=
\mathbf{N}_{n}u_{i}(x,t)+\delta_{ij}\partial^{j}\Xi(0;\xi)J(0;\varsigma)
\end{equation}
\end{thm}
and the averaged incompressibility condition still holds such that $\bm{\mathrm{I\!E}}
\lbrace\partial_{i}\widehat{u}_{i}(x,t)\rbrace=0 $
\begin{proof}
The randomly perturbed or turbulent flow is
$\widehat{u}_{i}(x,t)=u_{i}(x,t)+\widehat{\mathscr{U}}_{i}(x,t)$ so that perturbed PDEs are
\begin{align}
\mathbf{N}_{n}\widehat{u}_{i}(x,t)&=\sum_{i=1}^{n}\partial_{t}\widehat{u}_{i}(x,t)+
\sum_{i=1}^{n}\sum_{j=1}^{n}\partial^{j}[\widehat{u}_{j}(x,t)\widehat{u}_{i}(x,t)]\nonumber\\&
-\nu\sum_{i=1}^{n}\triangle\widehat{u}_{i}(x,t)+\sum_{i=1}^{n}\partial_{i}p(x,t)-\sum_{i=1}^{n}f_{i}(x,t)\nonumber\\&
=\sum_{i=1}^{n}\partial_{t}u_{i}(x,t)+\sum_{i=1}^{n}\partial_{t}\mathscr{U}_{i}(x,t)+
\sum_{i=1}^{n}\sum_{j=1}^{n}u_{i}(x,t)u_{j}(x,t)\nonumber\\&+\sum_{i=1}^{n}\sum_{j=1}^{n}u_{i}(x,t)
\mathscr{U}_{j}(x,t)+\sum_{i=1}^{n}
\sum_{j=1}^{n}u_{j}(x,t)\mathscr{U}_{i}(x,t)+
\sum_{i=1}^{n}\sum_{j=1}^{n}\mathscr{U}_{i}(x,t)\mathscr{U}_{j}(x,t)\nonumber\\&
-\nu\sum_{i=1}^{n}\triangle \widehat{u}_{i}(x,t)-\nu\sum_{i=1}^{n}\triangle\mathscr{U}_{i}(x,t)+
\sum_{i=1}^{n}\partial_{i}p(x,t)-\sum_{i=1}^{n}f_{i}(x,t)
\end{align}
Taking the stochastic expectation or average,
\begin{align}
\bm{\mathrm{I\!E}}\bigg\lbrace\mathbf{N}_{n}\widehat{u}_{i}(x,t)\bigg\rbrace&=\sum_{i=1}^{n}\partial_{t}u_{i}(x,t)
+\sum_{i=1}^{n}\sum_{j=1}^{n}u_{i}(x,t)u_{j}(x,t)\nonumber\\&
\equiv\sum_{i=1}^{n}\partial_{t}u_{i}(x,t)+\underbrace{\sum_{i=1}^{n}\bm{\mathrm{I\!E}}\bigg\lbrace\partial_{t}
\widehat{{\mathscr{U}}}_{i}(x,t)\bigg\rbrace}_{linear}+
\sum_{i=1}^{n}\sum_{j=1}^{n}u_{i}(x,t)u_{j}(x,t)\nonumber\\&+
\underbrace{\sum_{i=1}^{n}\sum_{j=1}^{n}u_{i}(x,t)\bm{\mathrm{I\!E}}\bigg\lbrace\mathscr{U}_{j}(x,t)\bigg\rbrace}_{linear}
+\underbrace{\sum_{i=1}^{n}\sum_{j=1}^{n}u_{j}(x,t)\bm{\mathrm{I\!E}}\bigg\lbrace\mathscr{U}_{i}(x,t)\bigg\rbrace}_{linear}\nonumber\\&+
\underbrace{\sum_{i=1}^{n}\sum_{j=1}^{n}\bm{\mathrm{I\!E}}\bigg\lbrace\mathscr{U}_{i}(x,t)
\mathscr{U}_{j}(x,t)\bigg\rbrace}_{nonlinear}-\nu\sum_{i=1}^{n}\triangle \widehat{u}_{i}(x,t)-\nu\underbrace{\sum_{i=1}^{n}\bm{\mathrm{I\!E}}\bigg\lbrace\triangle \mathscr{U}_{i}(x,t)\bigg\rbrace}_{linear}\nonumber\\&+
\sum_{i=1}^{n}\partial_{i}p(x,t)-\sum_{i=1}^{n}f_{i}(x,t)
\end{align}
The linear terms containing $\widehat{\mathscr{U}}_{i}(x,t)$ and $\partial_{t}\widehat{\mathscr{U}}_{i}(x,t)$ vanish, however, averaging over the nonlinear convective term induces an additional stress tensor term
\begin{align}
\bm{\mathrm{I\!E}}\bigg\lbrace\mathbf{N}_{n}\widehat{u}_{i}(x,t)\bigg\rbrace&=\mathbf{N}_{n}u_{i}(x,t)+
\sum_{i=1}^{n}\sum_{j=1}^{n}\partial^{j}\bm{\mathrm{I\!E}}\bigg\lbrace
\widehat{\mathscr{U}}_{i}(x,t)\widehat{\mathscr{U}}_{j}(x,t)\bigg\rbrace\nonumber\\
&=\mathbf{N}_{n}u_{i}(x,t)+\sum_{i=1}^{n}\sum_{j=1}^{n}\partial^{j}
\bm{\mathrm{I\!E}}\bigg\lbrace\widehat{\mathscr{U}}_{i}(x,t)\widehat{\mathscr{U}}_{j}(x,t)\bigg\rbrace
\nonumber\\&\equiv
\partial^{j}\mathcal{R}_{ij}(x,x;t,t)=\delta_{ij}\partial^{j}\Xi(0;\xi J(0;\varsigma)
\end{align}
The averaged incompressibility condition is
\begin{equation}
\bm{\mathrm{I\!E}}\bigg\lbrace\partial_{i}\widehat{u}_{i}(x,t)\bigg\rbrace=
\partial_{i}u_{i}(x,t)+\bm{\mathrm{I\!E}}\bigg\lbrace\partial_{i}
\widehat{\mathscr{U}}(x,t)\bigg\rbrace=\partial_{i}\bm{\mathrm{I\!E}}\bigg\lbrace \widehat{\mathscr{U}}_{i}(x,t)\bigg\rbrace=0
\end{equation}
\end{proof}
\subsection{Random perturbations of the dynamical solutions}
Using Lemma (3.8) there is also a dynamical solution of the stochastic Einstein equations.
\begin{thm}
Given the set of dynamical power-law solutions (3.44) and (3.45) whereby $\psi_{i}(t)=\psi_{i}^{E}+p_{i}\ln|t|$ is a solution of $\mathbf{H}_{n}\psi_{i}(t)=0$ and $a_{i}(t)=a_{i}^{E}|t|^{p_{i}}$ is a solution of $\mathbf{D}_{n}a(t)=0$ where $p_{i}$ satisfy the Kasner constraints, then the following stochastic solutions
\begin{align}
&\widehat{\psi}(t)=\psi_{i}^{E}+p_{i}\ln|t|+
\int_{0}^{t}\widehat{\mathscr{U}}_{i}(\tau)d\tau\equiv \psi_{i}(t)+\int_{0}^{t}\widehat{\mathscr{U}}_{i}(\tau)d\tau
\\&
\widehat{a}_{i}(t)=a_{i}^{E}p_{i}\exp\left(\int_{0}^{t}
\widehat{\mathscr{U}}_{i}(\tau)d\tau\right)\equiv
a_{i}^{E}|t|^{p_{i}}\widehat{\mathscr{J}}_{i}(t)
\end{align}
are dynamical solutions of the stochastically averaged Einstein equations
\begin{align}
&\bm{\mathrm{I\!E}}\bigg\lbrace\mathbf{H}_{n}\widehat{\psi}_{i}(t)\bigg\rbrace
=\bm{\mathrm{I\!E}}\left\lbrace\sum_{i=1}^{n}\partial_{tt}\widehat{\psi}_{i}(t)+\frac{1}{2}
\sum_{i\ne j}\partial_{t}\widehat{\psi}_{i}(t)\partial_{t}\widehat{\psi}_{j}(t)+\frac{1}{2}
\sum_{i=1}^{n}\partial_{t}\widehat{\psi}_{i}(t)\partial_{t}\widehat{\psi}_{i}(t)\right\rbrace=\lambda
\\&\bm{\mathrm{I\!E}}\bigg\lbrace
\mathbf{D}_{n}\widehat{a}_{i}(t)\bigg\rbrace=\bm{\mathrm{I\!E}}\left\lbrace\sum_{i=1}^{n}\frac{\partial_{tt}
\widehat{a}_{i}(t)}{\widehat{a}_{i}(t)}-\frac{1}{2}\sum_{i=1}^{n}\frac{\partial_{t}\widehat{a}_{i}(t)\partial_{t}\widehat{a}_{i}(t)}{\widehat{a}_{i}(t)
\widehat{a}_{j}(t)}+\frac{1}{2}\sum_{i=1}^{n}\sum_{j=1}^{n}\frac{\partial_{t}\widehat{a}_{i}(t)
\partial_{t}\widehat{a}_{i}(t)}{\widehat{a}_{i}(t)\widehat{a}_{j}(t)}\right\rbrace=\lambda
\end{align}
or
\begin{align}
&\bm{\mathrm{I\!E}}\bigg\lbrace \mathbf{H}_{n}\widehat{\psi}_{i}(t)\bigg\rbrace=
\mathbf{H}_{n}\psi_{i}(t)+\frac{1}{2}\zeta^{2}
\bigg\|\sqrt{\bm{\mathrm{I\!E}}\big\lbrace|\mathscr{U}_{i}(t)|^{2}\big\rbrace}
\bigg\|_{L_{2}}^{2}+\frac{1}{2}\zeta^{2}\bigg\|\sqrt{\bm{\mathrm{I\!E}}
\big\lbrace\mathscr{C}_{ij}(t)
\big\rbrace}\bigg\|_{F}^{2}\nonumber\\&=\mathbf{H}_{n}\psi_{i}(t)+\lambda_{1}+\lambda_{2}
=\lambda_{1}+\lambda_{2}=\lambda
\\&\bm{\mathrm{I\!E}}\bigg\lbrace \mathbf{D}_{n}\widehat{a}_{i}(t)\bigg\rbrace=\mathbf{D}_{n}\psi_{i}(t)+
\frac{1}{2}\zeta^{2}
\bigg\|\sqrt{\bm{\mathrm{I\!E}}\big\lbrace|\mathscr{U}_{i}(t)|^{2}\big\rbrace}
\bigg\|_{L_{2}}^{2}+\frac{1}{2}\zeta^{2}\bigg\|\sqrt{\bm{\mathrm{I\!E}}
\big\lbrace\mathscr{C}_{ij}(t)
\big\rbrace}\bigg\|_{F}^{2}\nonumber\\&=\mathbf{D}_{n}a_{i}(t)+\lambda_{1}+\lambda_{2}=\lambda_{1}+\lambda_{2}=\lambda
\end{align}
where $\lambda=\frac{1}{2}n^{2}C$ as before
\end{thm}
\begin{proof}
The derivatives are $\partial_{t}\widehat{\psi}(t)=\partial_{t}\psi_{i}(t)+\widehat{\mathscr{U}}_{i}(t)$ and $\partial_{tt}\psi(t)=\partial_{tt}\psi_{i}(t)+\partial_{t}\widehat{\mathscr{U}}_{i}(t)$. The stochastically perturbed Einstein system become
\begin{align}
\mathbf{H}_{n}\widehat{\psi}_{i}(t)&=\sum_{i=1}^{n}\partial_{tt}
\widehat{\psi}_{i}(t)+\frac{1}{2}\sum_{i=1}^{n}\sum_{j=1}^{n}\partial_{t}\widehat{\psi}_{i}(t)\partial_{t}\widehat{\psi}_{j}(t)+\frac{1}{2}
\sum_{i=1}^{n}\partial_{t}\widehat{\psi}_{i}(t)\partial_{t}\widehat{\psi}_{i}(t)
\nonumber\\&=\sum_{i=1}^{n}\partial_{tt}{\psi}_{i}(t)+\frac{1}{2}
\sum_{i=1}^{n}\sum_{j=1}^{n}\partial_{t}{\psi}_{i}(t)\partial_{t}{\psi}_{j}(t)+\frac{1}{2}
\sum_{i=1}^{n}\partial_{t}{\psi}_{i}(t)\partial_{t}{\psi}_{i}(t)\nonumber\\&
+\sum_{i=1}^{n}\partial_{t}\widehat{\mathscr{U}}_{i}(t)+\frac{1}{2}\sum_{i=1}^{n}\partial_{t}
\psi_{i}(t)\widehat{\mathscr{U}}_{i}(t)
+\frac{1}{2}\sum_{i=1}^{n}\widehat{\mathscr{U}}_{i}(t)\widehat{\mathscr{U}}_{i}(t)\nonumber\\&
+\frac{1}{2}\sum_{i=1}^{n}\partial_{t}\psi_{i}(t) \widehat{\mathscr{U}}_{j}(t)+\frac{1}{2}\sum_{i=1}^{n}\sum_{j=1}^{n}
\widehat{\mathscr{U}}_{i}(t)\widehat{\mathscr{U}}_{j}(t)
\end{align}
which is upon taking the stochastic expectation
{\allowdisplaybreaks
\begin{align}
\bm{\mathrm{I\!E}}\bigg\lbrace \mathbf{H}_{n}\widehat{\psi}_{i}(t)\bigg\rbrace&=\bm{\mathrm{I\!E}}\left\lbrace\sum_{i=1}^{n}
\partial_{tt}\widehat{\psi}_{i}(t)+\frac{1}{2}\sum_{i=1}^{n}\sum_{j=1}^{n}\partial_{t}
\widehat{\psi}_{i}(t)\partial_{t}\widehat{\psi}_{j}(t)+\frac{1}{2}
\sum_{i=1}^{n}\partial_{t}\widehat{\psi}_{i}(t)\partial_{t}\widehat{\psi}_{i}(t)
\right\rbrace\nonumber\\&=\sum_{i=1}^{n}\left(-\frac{p_{i}}{t^{2}}\right)+\frac{1}{2}\sum_{i=1}^{n}p_{i}p_{i}\frac{1}{t^{2}}+
\frac{1}{2}\sum_{i=1}^{n}\sum_{j=1}^{n}p_{i}p_{j}\frac{1}{t^{2}}\nonumber\\&
+\underbrace{\sum_{i=1}^{n}
\partial_{t}\bm{\mathrm{I\!E}}\bigg\lbrace\widehat{\mathscr{U}}_{i}(t)\bigg
\rbrace}_{linear}+\frac{1}{2}\underbrace{\sum_{i=1}^{n}\bm{\mathrm{I\!E}}\bigg\lbrace\partial_{t}\mathscr{U}_{i}(t) \widehat{\mathscr{U}}_{i}(t)\bigg\rbrace}_{nonlinear}+\frac{1}{2}\underbrace{\sum_{i=1}^{n}\sum_{j=1}^{n}
\bm{\mathrm{I\!E}}\bigg\lbrace\mathscr{U}_{i}(t)\mathscr{U}_{j}(t)\bigg\rbrace}_{nonlinear}\nonumber\\&
=\frac{1}{2}\sum_{i=1}^{n}\sum_{j=1}^{n}\bm{\mathrm{I\!E}}\bigg\lbrace \partial_{t}\mathscr{U}_{i}(t)\widehat{\mathscr{U}}_{j}(t)\bigg\rbrace+\frac{1}{2}
\sum_{i=1}^{n}\sum_{j=1}^{n}\bm{\mathrm{I\!E}}\bigg\lbrace\widehat{\mathscr{U}}_{i}(t)
\widehat{\mathscr{U}}_{j}(t)\bigg\rbrace
\end{align}}
Due to the nonlinearity, and since $\bm{\mathrm{I\!E}}\lbrace
\widehat{\mathscr{U}}_{i}(t)\rbrace=0$ terms involving the non-vanishing correlations $\bm{\mathrm{I\!E}}\lbrace\widehat{\mathscr{U}}_{i}(t)\widehat{\mathscr{U}}_{j}(t)\rbrace$ are retained so that equation (5.71) reduces to
\begin{align}
\bm{\mathrm{I\!E}}\bigg\lbrace\mathbf{H}_{n}\widehat{\psi}_{i}(t)\bigg\rbrace&=
\bm{\mathrm{I\!E}}\left\lbrace\sum_{i=1}^{n}\partial_{tt}\widehat{\psi}_{i}(t)+\frac{1}{2}
\sum_{i\ne j}\partial_{t}\widehat{\psi}_{i}(t)\partial_{t}\widehat{\psi}_{j}(t)+\frac{1}{2}
\sum_{i=1}^{n}\partial_{t}\widehat{\psi}_{i}(t)\partial_{t}\widehat{\psi}_{i}(t)\right\rbrace
\nonumber\\&\sum_{i=1}^{n}\left(-\frac{p_{i}}{t^{2}}\right)+\frac{1}{2}
\sum_{i=1}^{n}p_{i}p_{i}\frac{1}{t^{2}}+\frac{1}{2}\sum_{i=1}^{n}\sum_{j=1}^{n}p_{i}p_{j}\frac{1}{t^{2}}
\nonumber\\&+\frac{1}{2}\sum_{i=1}^{n}\bm{\mathrm{I\!E}}\bigg\lbrace\widehat{\mathscr{U}}_{i}(t)
\widehat{\mathscr{U}}_{i}(t)\bigg\rbrace+\frac{1}{2}\sum_{i=1}^{n}
\sum_{j=1}^{n}\bm{\mathrm{I\!E}}\bigg\lbrace
\widehat{\mathscr{U}}_{i}(t)\widehat{\mathscr{U}}_{j}(t)\bigg\rbrace
\end{align}
Taking $p_{i}=p_{j}=1$ for all $i,j=1...n$,then
\begin{equation}
\sum_{i=1}^{n}\left(-\frac{p_{i}}{t^{2}}\right)+\sum_{i=1}^{n}p_{i}p_{i}\frac{1}{t^{2}}+
\sum_{i=1}^{n}\sum_{j=1}^{n}p_{i}p_{j}\frac{1}{t^{2}}=-n(\frac{1}{t^{2}})+\frac{1}{2}n\frac{1}{t^{2}}+
\frac{1}{2}n\frac{1}{t^{2}}=0
\end{equation}
so that
\begin{align}
\bm{\mathrm{I\!E}}\bigg\lbrace \mathbf{H}_{n}\widehat{\psi}_{i}(t)\bigg\rbrace&
= \frac{1}{2}\sum_{i=1}^{n}\bm{\mathrm{I\!E}}\bigg\lbrace
\widehat{\mathscr{U}}_{i}(t)\widehat{\mathscr{U}}_{i}(t)\bigg\rbrace+
\frac{1}{2}\sum_{i=1}^{n}\sum_{j=1}^{n}\bm{\mathrm{I\!E}}\bigg\lbrace
\widehat{\mathscr{U}}_{i}(t)\widehat{\mathscr{U}}_{j}(t)\bigg\rbrace
\nonumber\\&= \frac{1}{2}\sum_{i=1}^{n}\delta_{ii}\bm{\mathrm{I\!E}}\bigg\lbrace
\widehat{\mathscr{U}}_{i}(t)\widehat{\mathscr{U}}_{i}(t)\bigg\rbrace+
\frac{1}{2}\sum_{i=1}^{n}\sum_{j=1}^{n}\delta_{ij}\bm{\mathrm{I\!E}}\bigg\lbrace
\widehat{\mathscr{U}}_{i}(t)\widehat{\mathscr{U}}_{j}(t)\bigg\rbrace
\nonumber\\&=\frac{1}{2}\zeta^{2}\sum_{i=1}^{n}\bigg|\sqrt{\bm{\mathrm{I\!E}}\big\lbrace
\big|\mathscr{U}_{i}|^{2}\big\rbrace}\bigg|^{2}+
\frac{1}{2}\zeta^{2}\sum_{i=1}^{n}\sum_{j=1}^{n}\bigg|\sqrt{\bm{\mathrm{I\!E}}\big\lbrace
\mathscr{C}_{ij}(t)\big\rbrace}\bigg|^{2}\nonumber\\&=\frac{1}{2}\zeta^{2}
\bigg\|\sqrt{\bm{\mathrm{I\!E}}\big\lbrace|\mathscr{U}_{i}(t)|^{2}\big\rbrace}
\bigg\|_{L_{2}}^{2}+\frac{1}{2}\zeta^{2}\bigg\|\sqrt{\bm{\mathrm{I\!E}}
\big\lbrace\mathscr{C}_{ij}(t)\big\rbrace}\bigg\|_{F}^{2}\nonumber\\&
=\frac{1}{2}\sum_{i=1}^{n}\delta_{ii}J(0;\varsigma)+\frac{1}{2}\sum_{i=1}^{n}\sum_{j=1}^{n}
\delta_{ij}J(0;\varsigma)\nonumber\\&=\frac{1}{2}n J(0;\varsigma)
+\frac{1}{2}n J(0;\varsigma)=n J(0;\varsigma)\equiv\lambda_{1}+\lambda_{2}=\lambda
\end{align}
and the result is proved. To prove (5.69), first establish the derivatives
$\partial_{t}\widehat{\mathscr{J}}_{i}(t)=\widehat{\mathscr{U}}_{i}(t)\widehat{\mathscr{J}}_{i}(t)$ and $\partial_{tt}\widehat{\mathscr{J}}_{i}(t)=\widehat{\mathscr{U}}_{i}(t)
+\partial_{t}\widehat{\eta}_{i}(t)\widehat{\mathscr{J}}_{i}(t)$ and then $\partial_{t}a_{i}^{(+)}(t)=a_{i}^{E}|t|^{p_{i}}\partial_{t}\widehat{\mathscr{J}}_{i}(t)
+a_{i}^{E}p_{i}|t|^{p_{i}-1}\beta_{i} $ and $\partial_{tt}a_{i}^{+}(t)=a_{i}^{E}|t|^{p_{i}}
\partial_{tt}\widehat{\mathscr{J}}_{i}(t)+a_{i}^{E}p_{i}|t|^{p_{i}-1}\partial_{t}\mathscr{J}_{i}(t)
+a_{i}^{E}p_{i}||t|^{p_{i}-1}\partial_{t}\widehat{\mathscr{J}}_{i}(t)$. The stochastically perturbed Einstein system of nonlinear ODEs becomes
{\allowdisplaybreaks
\begin{align}
\mathbf{D}_{n}\widehat{a}_{i}^{(+)}(t)&=\frac{\sum_{i=1}^{n}a_{i}^{E}|t|^{p_{i}}
\widehat{\mathscr{U}}_{i}(t)\widehat{\mathscr{U}}_{i}(t)\widehat{\mathscr{J}}_{i}(t)}{a_{i}^{E}|t|^{p_{i}}\widehat{\mathscr{J}}_{i}(t)}
+\sum_{i=1}^{n}\frac{a_{i}^{E}\left|t^{p_{i}}\right|\partial_{t}\widehat{\mathscr{U}}_{i}(t)
\widehat{\mathscr{J}}}{a_{i}^{E}|t|^{p_{i}}\mathscr{J}_{i}(t)}\nonumber\\&
+ \sum_{i=1}^{n}\frac{a_{i}^{E}\left|p_{i}\right||t|^{p_{i}-1}
\widehat{\mathscr{U}}_{i}(t)\widehat{\mathscr{J}}_{i}(t)}{
a_{i}^{E}|t|^{p_{i}}\widehat{\mathscr{J}}_{i}(t)} +\sum_{i=1}^{n}\frac{a_{i}^{E}\left|p_{i}\right||t^{p_{i}}
\widehat{\mathscr{U}}_{i}(t)\widehat{\mathscr{J}}_{i}(t)}{a_{i}^{E}|t|^{p_{i}}
\widehat{\mathscr{J}}_{i}(t)}\nonumber\\&
-\frac{1}{2}\sum_{i=1}^{n}\frac{a_{i}^{E}a_{i}^{E}|t|^{\frac{2}{n}}|t|^{p_{i}}
\widehat{\mathscr{U}}_{i}(t)\mathscr{U}_{i}\widehat{\mathscr{L}}_{i}(t)\widehat{\mathscr{J}}_{i}(t)}
{a_{i}^{E}|t|^{p_{i}}a_{i}^{E}|t|^{p_{i}}\widehat{\mathscr{L}}_{i}(t))
\widehat{\mathscr{J}}_{i}(t))}\nonumber\\&-\frac{1}{2}\sum_{i=1}^{n}\frac{2a_{i}^{E}a_{i}^{E}|t|^{p_{i}}|t|^{p_{i}-1}
\beta_{i}\beta_{i}\left|p_{i}\right|\widehat{\mathscr{J}}_{i}(t)) \widehat{\mathscr{J}}_{i}(t))\widehat{\mathscr{U}}_{i}(t))}{a_{i}^{E}
|t|^{p_{i}}a_{i}^{E}|t|^{p_{i}}\widehat{\mathscr{J}}_{i}(t))\widehat{\mathscr{J}}_{i}(t))}
\nonumber\\&-\frac{1}{2}\sum_{i=1}^{n}\frac{a_{i}^{E}a_{i}^{E}\left|\frac{2}{n}\right|p_{i}|t|^{p_{i}-1}
|t|^{p_{i}-1}\widehat{\mathscr{J}}_{i}(t))\widehat{\mathscr{J}}_{i}(t))}{a_{i}^{E}|t|^{p_{i}}
a_{i}^{E}|t|^{p_{i}}\widehat{\mathscr{J}}_{i}(t))\widehat{\mathscr{J}}_{i}(t))}
\nonumber\\&+\frac{1}{2}\sum_{i\ne j}^{n}\frac{a_{i}^{E}a_{j}^{E}|t|^{p_{i}}|t|^{p_{i}}
\widehat{\mathscr{U}}_{i}(t)\mathscr{U}{j}\widehat{\mathscr{J}}_{i}(t)\widehat{\mathscr{J}}_{j}(t)}
{a_{i}^{E}|t|^{p_{i}}a_{j}^{E}|t|^{p_{i}}\widehat{\mathscr{J}}_{i}(t))
\widehat{\mathscr{J}}_{j}(t))}\nonumber\\&+\frac{1}{2}\sum_{i\ne j}^{n}\frac{2a_{i}^{E}a_{j}^{E}|t|^{p_{i}}|t|^{p_{i}-1}
|p_{i}|\widehat{\mathscr{J}}_{i}(t))\widehat{\mathscr{J}}_{j}(t))\widehat{\mathscr{U}}_{i}(t))}{a_{i}^{E}
|t|^{p_{i}}a_{j}^{E}|t|^{p_{i}}\widehat{\mathscr{J}}_{i}(t))
\widehat{\mathscr{J}}_{j}(t))}\nonumber\\& +\frac{1}{2}\sum_{i\ne j}^{n}\frac{a_{i}^{E}a_{j}^{E}|p_{i}|t|^{p_{i}-1}|t|^{p_{i}-1}
\widehat{\mathscr{J}}_{i}(t))\widehat{\mathscr{J}}_{ij}(t))}{a_{i}^{E}|t|^{p_{i}}a_{j}^{E}|t|^{p_{i}}
\widehat{\mathscr{J}}_{i}(t))\widehat{\mathscr{J}}_{j}(t))}
\end{align}}
Cancelling terms
\begin{align}
\mathbf{D}_{n}a_{i}^{(+)}(t)&=\sum_{i=1}^{n}\widehat{\mathscr{U}}_{i}(t)
\widehat{\mathscr{U}}_{i}(t)+\sum_{i=1}^{n}\partial_{t}\widehat{\mathscr{U}}_{i}(t)+\sum_{i=1}^{n}
\left|p_{i}\right|t^{-1}\widehat{\mathscr{U}}_{i}(t)+\sum_{i=1}^{n}
\left|p_{i}\right|t^{-1}\widehat{\mathscr{U}}_{i}(t)\nonumber\\&
-\frac{1}{2}\sum_{i=1}^{n}\widehat{\mathscr{U}}_{i}(t)
\widehat{\mathscr{U}}_{i}(t)-\frac{1}{2}\sum_{i=1}^{n}p_{i}\widehat{\mathscr{U}}_{i}(t)
-\frac{1}{2}\sum_{i=1}^{n}p_{i}^{2}t^{-2}\nonumber\\&
+\frac{1}{2}\sum_{i\ne j}^{n}\widehat{\mathscr{U}}_{i}(t)
\widehat{\mathscr{U}}_{j}(t)+\frac{1}{2}\sum_{i\ne j}^{n}p_{i}\widehat{\mathscr{U}}_{i}(t)
+\frac{1}{2}\sum_{i\ne j}^{n}p_{i}p_{i}t^{-2}
\end{align}
Taking the stochastic expectation, only the nonlinear terms are nonvanishing so that
\begin{align}
\bm{\mathrm{I\!E}}\bigg\lbrace \mathbf{D}_{n}a_{i}^{(+)}(t)\bigg\rbrace&=
\sum_{i=1}^{n}\mathlarger{\mathsf{U}}\bigg\lbrace\widehat{\mathscr{U}}_{i}(t) \widehat{\mathscr{U}}_{i}(t)\bigg\rbrace\nonumber\\&-\frac{1}{2}\sum_{i=1}^{n}
\bm{\mathrm{I\!E}}\bigg\lbrace\widehat{\mathscr{U}}_{i}(t)\widehat{\mathscr{U}}_{i}(t)\bigg\rbrace
-\frac{1}{2}\sum_{i=1}^{n}\left|p_{i}\right|^{2}t^{-2}+
\frac{1}{2}\sum_{i=1}^{n}\sum_{j=1}^{n}\bm{\mathrm{I\!E}}\bigg\lbrace\widehat{\mathscr{U}}_{i}(t)
\widehat{\mathscr{U}}_{j}(t)\bigg\rbrace
\nonumber\\&=\frac{1}{2}\sum_{i=1}^{n}\delta_{ii}
\bm{\mathrm{I\!E}}\bigg\lbrace\widehat{\mathscr{U}}_{i}(t)\widehat{
\mathscr{U}}_{i}t)\bigg\rbrace-\frac{1}{2}\sum_{i=1}^{n}\left|p_{i}\right|^{2}t^{-2}\\&+
\frac{1}{2}\sum_{i=1}^{n}\sum_{j=1}^{n}\delta_{ij}\bm{\mathrm{I\!E}}\bigg\lbrace \widehat{\mathscr{U}}_{i}(t)\widehat{\mathscr{U}}_{j}(t)\bigg\rbrace+
\frac{1}{2}\sum_{i \ne j}^{n}p_{i}p_{j}t^{-2}
\nonumber\\&=\frac{1}{2}\sum_{i=1}^{n}\delta_{ii} J(0;\varsigma)-
\frac{1}{2}\sum_{i=1}^{n}\left|p_{i}\right|^{2}t^{-2}\nonumber\\&+
\frac{1}{2}\sum_{i=1}^{n}\sum_{j=1}^{n}\delta_{ij}J(0;\varsigma)
+\frac{1}{2}\sum_{i \ne j}^{n}p_{i}p_{j}t^{-2}
\end{align}
Now setting $p_{i}=1$ and $p_{j}=1$ for $i=1...n$ so that $p_{i}=p_{j}$ for $i\ne j$ and $\widehat{\mathscr{U}}_{i}(t)=\widehat{\mathscr{U}}(t)$ for $i=1...n$ gives
\begin{align}
\bm{\mathrm{I\!E}}\lbrace \mathbf{D}_{n}a_{i}(t)\rbrace&=\frac{1}{2}\sum_{i=1}^{n}\delta_{ii}J(0;\varsigma)-
\frac{1}{2}\sum_{i=1}^{n}\left|p_{i}\right|^{2}t^{-2}\nonumber\\&+
\frac{1}{2}\sum_{i=1}^{n}\sum_{j=1}^{n}\delta_{ij}J(0;\varsigma)
+\frac{1}{2}\sum_{i=1}^{n}\sum_{j=1}^{n}p_{i}p_{j}t^{-2}\nonumber\\&
=\frac{1}{2}n J(0;\varsigma)-\frac{1}{2}nt^{-2}+
\frac{1}{2}n J(0;\varsigma)
+\frac{1}{2}nt^{-2}=n J(0,\varsigma)\equiv \lambda
\end{align}
and the result is proved.
\end{proof}
\subsection{Averaged Einstein equations with a pre-existing cosmological constant}
We can also consider random perturbations of the dynamical cosmological solutions. The Einstein equations on a toroidal spacetime geometry  with a pre-existing cosmological constant $\bar{\lambda}$ are given by (3.70) and (3.71)
\begin{align}
&{\bm{\mathrm{H}}}_{n}\psi(t)=\bar{\lambda}
\\&
\mathlarger{\bm{\mathrm{D}}}_{n}a_{i}(t)=\bar{\lambda}
\end{align}
where $\bar{\lambda}=\Lambda(1+n)/(1-n)$, and where $bm{\mathrm{Ric}}_{AB}=\bm{g}_{AB}\Lambda$ are the underlying Einstein vacuum equations with cosmological constant. These have expanding and collapsing solutions such that $a_{i}^{(\pm)}(t)=a_{i}(0)\exp(\pm(\bar{\lambda}/n)^{1/2}t)$. Randomly perturbing the fields $\psi_{i}(t)$ and taking the stochastic average of the perturbed equations will then add a new contribution to $\bar{\lambda}$. This additional contribution is small if the fluctuations are weak.
\begin{lem}
Given the conditions of Thm (5.1) but with EVEs $\bm{\mathrm{Ric}}_{AB}=\bm{g}_{AB}\Lambda$, then on $\mathbb{T}^{n}\times\mathbb{R}^{+}$ the Einstein equations are $\mathbf{H}_{n}\psi_{i}(t)=\bar{\lambda}$ or $\mathbf{D}_{n}a_{i}(t)=\bar{\lambda}$, where $\bar{\lambda}=\Lambda(1+n)/(1-n)$ with 'inflating' solutions $a_{i}(t)=\exp((\bar{\lambda}/n)^{1/2}t)$,a nd $\psi_{i}(t)=\psi_{i}(0)+(\bar{\lambda}/n)^{1/2}$. If $\bm{\mathrm{I\!E}}\lbrace\widehat{\mathscr{U}}_{i}(t)\rbrace=0$ and the 2-point function is regulated as $\bm{\mathrm{I\!E}}\lbrace\widehat{\mathscr{U}}_{i}(t)
\widehat{\mathscr{U}}_{j}(t)\rbrace=\delta_{ij}J(0;\vartheta)=\delta_{ij}C$ for Gaussian random perturbations or noise, then the randomly perturbed solution (with $\zeta=1$)are $ \widehat{a}_{i}(t)=a_{i}(0)\exp((\bar{\lambda}/n)^{1/2}t)\exp\left(\int_{0}^{t}\widehat{\mathscr{U}}(\tau)d\tau\right)$
is a solution of the stochastically averaged Einstein systems of differential equations
\begin{equation}
\bm{\mathrm{I\!E}}\bigg\lbrace \mathbf{H}_{n}\widehat{\psi}_{i}(t)\bigg\rbrace=
\overline{\bar{\lambda}}+\frac{1}{2}\zeta^{2}\bigg\|\sqrt{\bm{\mathrm{I\!E}}\big\lbrace|\mathscr{U}_{i}(t)|^{2}\big\rbrace}
\bigg\|_{L_{2}}^{2}+\frac{1}{2}\zeta^{2}\bigg\|\sqrt{\bm{\mathrm{I\!E}}
\big\lbrace\mathscr{C}_{ij}(t)
\big\rbrace}\bigg\|_{F}^{2}
\end{equation}
\begin{equation}
\bm{\mathrm{I\!E}}\bigg\lbrace \mathbf{D}_{n}\widehat{a}_{i}(t)\bigg\rbrace=
\overline{\lambda}+\frac{1}{2}\zeta^{2}\bigg\|\sqrt{\bm{\mathrm{I\!E}}\big\lbrace|\mathscr{U}_{i}(t)|^{2}\big\rbrace}
\bigg\|_{L_{2}}^{2}+\frac{1}{2}\zeta^{2}\bigg\|\sqrt{\bm{\mathrm{I\!E}}
\big\lbrace\mathscr{C}_{ij}(t)
\big\rbrace}\bigg\|_{F}^{2}
\end{equation}
or
\begin{equation}
\bm{\mathrm{I\!E}}\bigg\lbrace \mathbf{H}_{n}\widehat{\psi}_{i}(t)\bigg\rbrace =\overline{\lambda}+(\lambda_{1}+\lambda_{2})=\overline{\lambda}+\lambda
\end{equation}
\begin{equation}
\bm{\mathrm{I\!E}}\bigg\lbrace \mathbf{D}_{n}\widehat{a}_{i}(t)\bigg\rbrace =\overline{\lambda}+(\lambda_{1}+\lambda_{2})=\overline{\lambda}+\lambda
\end{equation}
where $\lambda$ is an induced cosmological constant contribution arising from the nonlinearity, and if $\widehat{\mathscr{U}}_{i}(t)=\widehat{\mathscr{U}}(t)$ for all $i=1...n$ then $\lambda=nC$. The averaged effect of the random perturbation is to boost the expansion rate.
\end{lem}
\begin{proof}
Writing
\begin{equation}
\widehat{a}_{i}(t)=a_{i}(0)\exp((\bar{\lambda}/n)^{1/2}t)\exp\left(\int_{0}^{t}
\widehat{\mathscr{U}}_{i}(\tau)d\tau\right)= a_{i}(t)\widehat{\mathscr{J}}_{i}(t)
\end{equation}
where $a_{i}(t)=a_{i}(0)\exp((\bar{\lambda}/n)^{1/2}t)$. The derivatives are $\partial_{t}\widehat{\mathscr{J}}_{i}(t)=\widehat{\mathscr{U}}_{i}(t)\widehat{\mathscr{U}}_{i}(t)$ and $\partial_{t}\widehat{a}_{i}(t)=a_{i}(t)\widehat{\mathscr{J}}_{i}(t)
\widehat{\mathscr{J}}_{i}(t)+\widehat{\mathscr{j}}_{i}(t)\partial_{t}a_{i}(t)$. The second derivative is $\partial_{tt}a_{i}(t)=a_{i}(t)\widehat{\mathscr{U}}_{i}(t)
\widehat{\mathscr{U}}_{i}(t)+a_{i}(t)\partial_{t}\widehat{\mathscr{U}}_{i}\widehat{\mathscr{J}}_{i}(t)+
\partial_{t}a_{i}(t)\widehat{\mathscr{U}}_{i}(t)\widehat{\mathscr{J}}_{i}(t)
\widehat{\mathscr{J}}_{i}(t)+\partial_{t}a_{i}(t)\widehat{\mathscr{U}}_{i}(t)
\widehat{\mathscr{J}}_{i}(t)$. The randomly perturbed Einstein equations are then
{\allowdisplaybreaks
\begin{align}
\mathbf{D}_{n}\widehat{a}_{i}(t)&=\sum_{i=1}^{n}\frac{\partial_{tt}
\widehat{a}_{i}(t)}{\widehat{a}_{i}(t)}-\frac{1}{2}\sum_{i=1}^{n}\frac{\partial_{t}\widehat{a}_{i}(t)\partial_{t}\widehat{a}_{i}(t)}{\widehat{a}_{i}(t)
\widehat{a}_{j}(t)}+\frac{1}{2}\sum_{i=1}^{n}\sum_{j=1}^{n}\frac{\partial_{t}\widehat{a}_{i}(t)
\partial_{t}\widehat{a}_{i}(t)}{\widehat{a}_{i}(t)\widehat{a}_{j}(t)}\nonumber\\&
=\sum_{i=1}^{n}\frac{a_{i}(t)\widehat{\mathscr{U}}_{i}(t)\widehat{\mathscr{U}}_{i}(t)\widehat{\mathscr{j}}_{i}(t)}{  a_{i}(t)\widehat{\mathscr{J}}_{i}(t)} +\sum_{i=1}^{n}\frac{a_{i}(t)
\partial_{t}\widehat{\mathscr{U}}_{i}(t)\widehat{\mathscr{J                         }}_{i}(t)} {a_{i}(t)\widehat{\mathscr{J}}_{i}(t)}\nonumber\\&+\sum_{i=1}^{n}\frac{(\partial_{t}a_{i}(t))\widehat{\mathscr{U}}_{i}(t)\widehat{\mathscr{j}}_{i}(t)}{ a_{i}(t)\widehat{\mathscr{J}}_{i}(t)}+\sum_{i=1}^{n}\frac{\partial_{tt}a_{i}(t)
\widehat{\mathscr{J}}_{i}}{a_{i}(t)\widehat{\mathscr{J}}_{i}(t)}+
\sum_{i=1}^{n}\frac{\partial_{t}a_{i}(t)\widehat{\mathscr{U}}_{i}(t)\widehat{\mathscr{J}}_{i}}
{a_{i}(t)\widehat{\mathscr{J}}_{i}(t)}\nonumber\\&-\frac{1}{2}\sum_{i=1}^{n}\frac{a_{i}(t)a_{i}(t)\widehat{\mathscr{U}}_{i}(t)\widehat{\mathscr{U}}_{i}(t)
\widehat{\mathscr{J}}_{i}(t)\widehat{\mathscr{J}}_{i}(t)}{a_{i}(t)a_{i}(t)
\widehat{\mathscr{J}}_{i}(t)\widehat{\mathscr{J}}_{i}(t)}-\sum_{i=1}^{n}\frac{a_{i}(t)a_{i}(t)
\widehat{\mathscr{U}}_{i}(t)\widehat{\mathscr{J}}_{i}(t)\widehat{\mathscr{J}}_{i}(t)}
{a_{i}(t)a_{i}(t)\widehat{\mathscr{J}}_{i}(t)\widehat{\mathscr{J}}_{i}(t)}\nonumber\\&
-\frac{1}{2}\sum_{i=1}^{n}\frac{\partial_{t}a_{i}(t)\partial_{t}a_{i}(t)
\widehat{\mathscr{J}}_{i}(t)\widehat{\mathscr{J}}_{i}(t)}
{a_{i}(t)a_{i}(t)\widehat{\mathscr{J}}_{i}(t)\widehat{\mathscr{J}}_{i}(t)}
+\frac{1}{2}\sum_{i=1}^{n}\sum_{j=1}^{n}\frac{a_{i}(t)a_{j}(t)\widehat{\mathscr{U}}_{i}(t)
\widehat{\mathscr{J}}_{j}(t)\widehat{\mathscr{J}}_{i}(t)\widehat{\mathscr{J}}_{j}(t)}{a_{i}(t)a_{j}(t)
\widehat{\mathscr{J}}_{i}(t)\widehat{\mathscr{J}}_{j}(t)}\nonumber\\&
+\sum_{i=1}^{n}\sum_{j=1}^{n}\frac{a_{i}(t)a_{j}(t)\mathscr{U}_{i}(t)\widehat{\mathscr{J}}_{i}(t)
\widehat{\mathscr{J}}_{j}(t)}{a_{i}(t)a_{j}(t)\widehat{\mathscr{J}}_{i}(t)\widehat{\mathscr{J}}_{j}(t)}
+\frac{1}{2}\sum_{i=1}^{n}\sum_{j=1}^{n}\frac{\partial_{t}a_{i}(t)\partial_{t}a_{j}(t)
\widehat{\mathscr{J}}_{i}(t)\widehat{\mathscr{J}}_{j}(t)}
{a_{i}(t)a_{j}(t)\widehat{\mathscr{J}}_{i}(t)\widehat{\mathscr{J}}_{j}(t)}
\end{align}}
Cancelling terms and taking the stochastic average $\bm{\mathrm{I\!E}}\lbrace...\rbrace$ with $\bm{\mathrm{I\!E}}\lbrace\widehat{\mathscr{U}}(t)\rbrace=0$ and $\bm{\mathrm{I\!E}}\lbrace\partial_{t}\widehat{\mathscr{U}}(t)\rbrace=0$
{\allowdisplaybreaks
\begin{align}
\bm{\mathrm{I\!E}}\bigg\lbrace\mathbf{D}_{n}a_{i}
(t)\bigg\rbrace&=\sum_{i=1}^{n}\frac{\partial_{tt}
\widehat{a}_{i}(t)}{{a}_{i}(t)}-\frac{1}{2}
\sum_{i=1}^{n}\frac{\partial_{t}{a}_{i}(t)\partial_{t}{a}_{i}(t)}{{a}_{i}(t)
{a}_{j}(t)}+\frac{1}{2}\sum_{i=1}^{n}\sum_{j=1}^{n}\frac{\partial_{t}{a}_{i}(t)
\partial_{t}{a}_{i}(t)}{{a}_{i}(t){a}_{j}(t)}\nonumber\\&
+\sum_{i=1}^{n}\bm{\mathrm{I\!E}}\bigg\lbrace
\widehat{\mathscr{U}}_{i}(t)\widehat{\mathscr{U}}_{i}(t)\bigg\rbrace-\frac{1}{2}\sum_{i=1}^{n}
\bm{\mathrm{I\!E}}\bigg\lbrace\widehat{\mathscr{U}}_{i}(t)\widehat{\mathscr{U}}_{i}(t)\bigg\rbrace+
\frac{1}{2}\sum_{i=1}^{n}\sum_{j=1}^{n}\bm{\mathrm{I\!E}}
\bigg\lbrace\widehat{\mathscr{U}}_{i}(t)\widehat{\mathscr{U}}_{j}(t)\bigg\rbrace\nonumber\\&
\equiv \mathbf{D}_{n}a_{i}(t)+\sum_{i-1}^{n}\bm{\mathrm{I\!E}}\bigg\lbrace
\widehat{\mathscr{U}}_{i}(t)\widehat{\mathscr{U}}_{i}(t)\bigg\rbrace-\frac{1}{2}\sum_{i=1}^{n}
\bm{\mathrm{I\!E}}\bigg\lbrace\widehat{\mathscr{U}}_{i}(t)\widehat{\mathscr{U}}_{i}(t)\bigg\rbrace+
\frac{1}{2}\sum_{i=1}^{n}\sum_{j=1}^{n}\bm{\mathrm{I\!E}}
\bigg\lbrace\widehat{\mathscr{U}}_{i}(t)\widehat{\mathscr{U}}_{j}(t)\bigg\rbrace\nonumber\\&
\equiv\bar{\lambda}+\underbrace{\sum_{i=1}^{n}\bm{\mathrm{I\!E}}\bigg\lbrace
\widehat{\mathscr{U}}_{i}(t)\widehat{\mathscr{U}}_{i}(t)\bigg\rbrace-\frac{1}{2}\sum_{i=1}^{n}
\bm{\mathrm{I\!E}}\bigg\lbrace\widehat{\mathscr{U}}_{i}(t)\widehat{\mathscr{U}}_{i}(t)\bigg\rbrace}
+\frac{1}{2}\sum_{i=1}^{n}\sum_{j=1}^{n}\bm{\mathrm{I\!E}}
\bigg\lbrace\widehat{\mathscr{U}}_{i}(t)\widehat{\mathscr{U}}_{j}(t)\bigg\rbrace
\nonumber\\&\equiv\bar{\lambda}+\frac{1}{2}\sum_{i=1}^{n}\delta_{ii}
\bm{\mathrm{I\!E}}\bigg\lbrace\widehat{\mathscr{U}}_{i}(t)\widehat{\mathscr{U}}_{i}(t)\bigg\rbrace+
\frac{1}{2}\sum_{i=1}^{n}\sum_{j=1}^{n}\delta_{ij}\bm{\mathrm{I\!E}}\bigg\lbrace
\widehat{\mathscr{U}}_{i}(t)\widehat{\mathscr{U}}_{j}(t)\bigg\rbrace
\nonumber\\&=\frac{1}{2}\zeta^{2}\sum_{i=1}^{n}\bigg|\sqrt{\bm{\mathrm{I\!E}}\big\lbrace\big|\mathscr{U}_{i}|^{2}\big\rbrace}\bigg|^{2}+
\frac{1}{2}\sum_{i=1}^{n}\sum_{j=1}^{n}\bigg|\sqrt{\bm{\mathrm{I\!E}}\big\lbrace
\mathscr{S}_{ij}(t)\big\rbrace}\bigg|^{2}\nonumber\\&=\frac{1}{2}
\bigg\|\sqrt{\bm{\mathrm{I\!E}}\big\lbrace|\mathscr{U}_{i}(t)|^{2}\big\rbrace}
\bigg\|_{L_{2}}^{2}+\frac{1}{2}\bigg\|\sqrt{\bm{\mathrm{I\!E}}
\big\lbrace\mathscr{S}_{ij}(t)\big\rbrace}\bigg\|_{F}^{2}\nonumber\\&=\bar{\lambda}+\frac{1}{2}\sum_{i-1}^{n}\delta_{ii}J(0;\varsigma)
+\frac{1}{2}\sum_{i=1}^{n}\sum_{j=1}^{n}\delta_{ij}J(0;\varsigma)\nonumber\\&
=\bar{\lambda}+\frac{1}{2}n J(0;\varsigma)+\frac{1}{2}J(0;\varsigma)=\bar{\lambda}+\lambda_{1}+\lambda_{2}
\end{align}}
where we have taken $\widehat{\mathscr{U}}_{i}(t)=\widehat{\mathscr{U}}(t)$ for $i=1...n$ as before. This then reduces to
\begin{equation}
\bm{\mathrm{I\!E}}\bigg\lbrace \mathbf{D}_{n}\widehat{a}_{i}(t)\bigg\rbrace=\lambda+n\bm{\mathrm{I\!E}}\bigg\lbrace\widehat{\mathscr{U}}(t)\widehat{\mathscr{U}}(t)\bigg\rbrace=
\lambda+n J(t,t;\varsigma)\equiv\bar{\lambda}+\lambda
\end{equation}
\end{proof}
\subsection{Kretschmann invariant, shear and expansion}
The final lemma of this section considers the averaged Kretschmann invariant, the averaged expansion and the averaged shear for a dynamical solution $\psi_{i}(t)$ which is randomly perturbed.
\begin{lem}
If $\widehat{\mathscr{U}}_{i}(t)=\widehat{\mathscr{U}}(t)$, the averaged Kretschmann invariant is shifted as
\begin{equation}
\bm{\mathrm{I\!E}}\bigg\lbrace\mathbf{K}(t)\bigg\rbrace=\mathbf{K}(t)+6n\zeta^{2}J(0;\varsigma)
\end{equation}
The averaged expansion remains invariant so that
\begin{equation}
\bm{\mathrm{I\!E}}\bigg\lbrace\widehat{\bm{\chi}}(t)\bigg\rbrace=\bm{\chi}(t)
\end{equation}
The averaged shear is shifted as
\begin{equation}
\bm{\mathrm{I\!E}}\bigg\lbrace\widehat{\bm{\mathfrak{S}}}^{2}(t)\bigg\rbrace=\bm{\mathfrak{S}}^{2}(t)+4n\zeta^{2}J(0;\varsigma) \end{equation}
\end{lem}
\begin{proof}
The randomly perturbed Kretschmann invariant is
\begin{align}
&\widehat{\mathbf{K}}(t)=4\sum_{i=1}^{n}\partial_{tt}\widehat{\psi}_{i}(t)+4\sum_{i=1}^{n}
\partial_{t}\widehat{\psi}_{i}(t)\partial_{t}\widehat{\psi}_{i}(t)+
2\sum_{i=1}^{n}\sum_{j=1}^{n}\partial_{t}\widehat{\psi}_{i}(t)
\partial_{t}\widehat{\psi}_{j}(t)\nonumber\\&
=4\sum_{i=1}^{n}\partial_{tt}\psi_{i}(t)+\zeta\widehat{\mathscr{U}}_{i}(t))
+\sum_{i=1}^{n}\left(\partial_{t}\psi_{i}(t)\partial_{t}\psi_{i}(t)
+4\partial_{t}\psi_{i}(t)\zeta\widehat{\mathscr{U}}_{i}(t)+
\zeta^{2}\widehat{\mathscr{U}}_{i}(t)\widehat{\mathscr{U}}_{i}(t)\right)\nonumber\\&
+2\sum_{i=1}^{n}\sum_{j=1}^{n}\left(\partial_{t}\psi_{i}(t)\partial_{t}\psi_{j}(t)
+\zeta\partial_{t}\psi_{i}(t)\widehat{\mathscr{U}}_{i}(t)
+\zeta\partial_{t}\psi_{j}(t)\widehat{\mathscr{U}}_{j}(t)+\zeta^{2}
\widehat{\mathscr{U}}_{i}(t)\widehat{\mathscr{U}}_{j}(t)\right)
\end{align}
Taking the stochastic expectation
\begin{align}
&\bm{\mathrm{I\!E}}\bigg\lbrace\widehat{\mathbf{K}}(t)\bigg\rbrace=
4\sum_{i=1}^{n}\partial_{tt}{\psi}_{i}(t)+4\sum_{i=1}^{n}
\partial_{t}{\psi}_{i}(t)\partial_{t}{\psi}_{i}(t)+
2\sum_{i=1}^{n}\sum_{j=1}^{n}\partial_{t}{\psi}_{i}(t)
\partial_{t}{\psi}_{j}(t)\nonumber\\&+4\zeta^{2}\sum_{i=1}^{n}
\bm{\mathrm{I\!E}}\bigg\lbrace\widehat{\mathscr{U}}_{i}(t)\widehat{\mathscr{U}}_{i}(t)\bigg\rbrace
+2\zeta^{2}\sum_{i=1}^{n}\sum_{j=1}^{n}\bm{\mathrm{I\!E}}\bigg\lbrace\widehat{\mathscr{U}}_{i}(t)
\widehat{\mathscr{U}}_{j}(t)\bigg\rbrace\nonumber\\&=\mathbf{K}(t)+4\zeta^{2}\sum_{i=1}^{n}\delta_{ii}J(0;\varsigma)
+2\zeta^{2}\sum_{i=1}^{n}\sum_{i=1}^{n}\delta_{ij}J(0;\varsigma)\nonumber\\&
=\bm{K}(t)+4\zeta^{2}n J(0;\varsigma)+2\zeta^{2}n J(0;\varsigma)
=\bm{K}(t)+6\zeta^{2}J(0;\varsigma)
\end{align}
The randomly perturbed expansion is
\begin{equation}
\widehat{\bm{\chi}}(t)=\sum_{i=1}^{n}\partial_{t}\psi_{i}(t)+\sum_{i=1}^{n}
\zeta\widehat{\mathscr{U}}_{i}(t)
\end{equation}
so that
\begin{equation}
\bm{\mathrm{I\!E}}\bigg\lbrace\widehat{\bm{\chi}}(t)\bigg\rbrace=\bm{\chi}(t)+\sum_{i=1}^{n}
\bm{\mathrm{I\!E}}\bigg\lbrace\widehat{\mathscr{U}}_{i}(t)\bigg\rbrace=\bm{\chi}(t)
\end{equation}
Finally, the stochastically averaged shear is
\begin{align}
\widehat{\bm{E}}^{2}(t)=&\sum_{i=1}^{n}\sum_{j=1}^{n}[
\partial_{t}\widehat{\psi}_{i}(t)\partial_{t}\widehat{\psi}_{i}-2\partial_{t}\widehat{\psi}_{i}(t)
\widehat{\partial}_{t}\widehat{\psi}_{j}(t)(t)+\partial_{t}\widehat{\psi}_{j}(t)
\partial_{t}\widehat{\psi}_{j}(t)\nonumber\\&
=\partial_{t}\psi_{i}(t)\partial_{t}\psi_{i}(t)+2\zeta\partial_{t}\psi_{i}(t)\widehat{\mathscr{U}}(t)
+\zeta^{2}\widehat{\mathscr{U}}_{i}(t)\widehat{\mathscr{U}}_{j}(t)+\partial_{t}\psi_{i}(t)\partial_{t}\psi_{i}(t)\nonumber\\
&+2\zeta\partial_{t}\psi_{i}(t)\widehat{\mathscr{U}}(t)
+\zeta^{2}\widehat{\mathscr{U}}_{i}(t)\widehat{\mathscr{U}}_{j}(t)-\partial_{t}\psi_{i}(t)\partial_{t}\psi_{j}(t)
-\zeta\partial_{t}\psi_{i}(t)\widehat{\mathscr{U}}_{i}(t)-\zeta\partial_{t}\psi_{j}(t)\widehat{\mathscr{U}}_{j}(t)
\end{align}
Again, taking the stochastic average this reduces to
\begin{align}
&\bm{\mathrm{I\!E}}\bigg\lbrace \widehat{\mathfrak{S}}^{2}(t)\bigg\rbrace =\mathfrak{S}^{2}(t)+\zeta^{2}\sum_{i}^{n}\sum_{j=1}^{n}\bm{\mathrm{I\!E}}\bigg\lbrace \widehat{\mathscr{U}}_{i}(t)\widehat{\mathscr{U}}_{j}(t)\bigg\rbrace\nonumber\\&
-2\zeta^{2}\bm{\mathrm{I\!E}}\bigg\lbrace\mathscr{U}_{i}(t)\mathscr{U}_{j}(t)\bigg\rbrace +\zeta^{2}\sum_{i=1}^{n}\bm{\mathrm{I\!E}}\bigg\lbrace\widehat{\mathscr{U}}_{i}(t)
\widehat{\mathscr{U}}_{i}(t)\bigg\rbrace\nonumber\\
&=\bm{\mathscr{E}}^{2}(t)+\zeta^{2}\sum_{i=1}^{n}\delta_{ii}J(0;\zeta)-\zeta^{2}
\sum_{i=1}^{n}\sum_{j=1}^{n}\delta_{ij}J(0;\zeta)=\bm{\mathfrak{S}}^{2}(t)
\end{align}
\end{proof}
\section{Cumulant cluster integral expansion}
We will use the Euclidean or $\mathcal{L}_{2}$-norms of the stochastically perturbed fields $\widehat{\psi}_{i}(t)$ and the radii $\widehat{a}_{i}(t)$ and we wish to evaluate the estimates $\bm{\mathrm{I\!E}}\|\widehat{\bm{\psi}}_{i}(t)-\bm{\psi}_{i}^{E}\|$ and
$\bm{\mathrm{I\!E}}\|\widehat{\bm{a}}(t)-\bm{a}^{E}\|$, and then the asymptotic estimate $\lim_{t\uparrow\infty}\bm{\mathrm{I\!E}}\|\widehat{\bm{a}}(t)-\bm{a}^{E}\|$. The expectation of the perturbed norm for the moduli is zero
\begin{align}
&\bm{\mathrm{I\!E}}\bigg\lbrace\bigg\|\widehat{\bm{\psi}}_{i}(t)-\bm        {\psi}_{i}\bigg\|\bigg\rbrace=\bm{\mathrm{I\!E}}\left\lbrace\left(\sum_{i=1}^{n}\big|\widehat{\psi}_{i}(t)
-\psi_{i}\bigg|^{2}\right)^{1/2}\right\rbrace\nonumber\\&
=\bm{\mathrm{I\!E}}\left\lbrace\left(\sum_{i=1}^{n}\left|\int_{0}^{t}
\widehat{\mathscr{U}}_{i}(\tau)d\tau\right|^{2}\right)^{1/2}\right\rbrace
=n^{1/2}\zeta\bm{\mathrm{I\!E}}\left\lbrace\left|\int_{0}^{t}\widehat{\mathscr{U}}(\tau)d\tau\right|\right\rbrace=0
\end{align}
if $\widehat{\mathscr{U}}_{i}(t)=\widehat{\mathscr{U}}(t)$ for
$i=1...n$. The asymptotic behavior $\lim_{t\uparrow\infty}\|\widehat{\bm{a}}(t)
-\bm{a}^{E}\| $ essentially determines whether the initially static Kasner universe with $\widehat{a}_{i}(0)=a_{i}^{E}$ is stable or unstable to stochastic perturbations of the static moduli fields $\psi_{i}^{E}$.
\begin{lem}
Given the initially static spatial volume of the hyper-toroidal geometry $\mathbb{T}^{n}$
\begin{equation}
\mathbf{V}_{\mathbf{g}}^{E}=\prod_{i=1}^{n}\exp(\psi_{i}^{E})=\exp\left(\sum_{i=1}^{n}\psi_{i}^{E}\right)
\equiv\prod_{i=1}^{n}a_{i}^{E}
\end{equation}
The randomly perturbed spatial volume is
\begin{align}
&\widehat{\mathbf{V}}_{\mathbf{g}}(t)
=\prod_{i=1}^{n}\exp\left(\psi^{E}+\zeta\int_{0}^{t}
\widehat{\mathscr{U}}_{i}(\tau)d\tau\right)\nonumber\\& \equiv\exp\left(\sum_{i=1}^{n}\psi_{i}^{E}\right)\exp\left(\zeta\sum_{i=1}^{n}\int_{0}^{t}
\widehat{\mathscr{U}}_{i}(\tau)d\tau\right)\nonumber\\&\equiv|\mathbf{V}_{\mathbf{g}}^{E}|\exp\left(\zeta\sum_{i=1}^{n}\int_{0}^{t}
\widehat{\mathscr{U}|\mathbf{V}_{\mathbf{g}}^{E}|}_{i}(\tau)d\tau\right)\nonumber\\&
\equiv\exp\left(\zeta\sum_{i=1}^{n}\int_{0}^{t}\widehat{\mathscr{U}}_{i}(\tau)d\tau\right)
\end{align}
The stochastic expectation is then
\begin{align}
&\mathbf{V}(t)=\bm{\mathrm{I\!E}}\bigg\lbrace\widehat{\mathbf{V}}_{\mathbf{g}}(t)\bigg\rbrace=
|\mathbf{V}_{\mathbf{g}}^{E}|\bm{\mathrm{I\!E}}\left\lbrace\exp\left(\zeta\sum_{i=1}^{n}\int_{0}^{t}\widehat{\mathscr{U}}_{i}(\tau)d\tau\right)
\right\rbrace\nonumber\\&=|\mathbf{V}_{\mathbf{g}}^{E}|\bm{\mathrm{I\!E}}\left\lbrace\exp\left(n\zeta\int_{0}^{t}\widehat{\mathscr{U}}(\tau)d\tau\right)
 \right\rbrace
\end{align}
if $\psi_{i}^{E}=\psi^{E}$. Asymptotic stability of the spatial volume then occurs if
\begin{equation}
\lim_{t\uparrow\infty}\bm{\mathrm{I\!E}}\lbrace\widehat{\mathbf{V}}_{\mathbf{g}}(t)\bigg\rbrace=\lim_{t\uparrow\infty}
|\mathbf{V}_{\mathbf{g}}^{E}|\bm{\mathrm{I\!E}}\left\lbrace\exp\left(n\zeta\int_{0}^{t}\widehat{\mathscr{U}}(\tau)d\tau\right)\right\rbrace<\infty
\end{equation}
while instability occurs if
\begin{eqnarray}
\bm{\mathrm{I\!E}}\bigg\lbrace\widehat{\mathbf{V}}_{\mathbf{g}}(t)\bigg\rbrace
=\lim_{t\uparrow\infty}|\mathbf{V}_{\mathbf{g}}^{E}|
\bm{\mathrm{I\!E}}\left\lbrace\exp\left(n\zeta\int_{0}^{t}\mathscr{U}(\tau)d\tau\right)\right\rbrace=\infty
\end{eqnarray}
The norm of the randomly perturbed metric is
\begin{align}
&\|\widehat{\bm{g}}(t)\|_{(2,1)}=\sum_{j=1}^{n}\left(\sum_{i=1}^{n}|\widehat{g}_{ij}(t)|^{2}\right)^{1/2}\equiv
\sum_{i=1}^{n}\left(\sum_{i=1}^{n}|\widehat{g}_{ii}(t)|^{2}\right)^{1/2}\nonumber\\&
\equiv\sum_{i=1}^{n}\left(\sum_{i=1}^{n}|\exp(2\widehat{\psi}_{i}(t))|^{2}\right)^{1/2}\nonumber\\
&=\sum_{i=1}^{n}\left(\sum_{i=1}^{n}\bigg|\delta^{ii}\exp(2\psi_{i}^{E})
\exp\left(2\zeta\int_{0}^{t}\mathscr{U}_{i}(\tau)d\tau\right)\bigg|\right)^{1/2}
\end{align}
The stochastic average is then
\begin{align}
&\bm{\mathrm{I\!E}}\bigg\lbrace\bigg\|\widehat{\bm{g}}(t)\bigg\|_{(2,1)}\bigg\rbrace
=\sum_{i=1}^{n}\bm{\mathrm{I\!E}}\left\lbrace \left(\sum_{i=1}^{n}|\delta^{ii}\exp(2\psi_{i}^{E})
\exp\left(2\zeta\int_{0}^{t}\widehat{\mathscr{U}}_{i}(\tau)d\tau\right)\right)^{1/2}\right\rbrace\nonumber\\&
< \sum_{i=1}^{n}\sum_{i=1}^{n}|\delta^{ii}\exp(2\psi_{i}^{E})
\bm{\mathrm{I\!E}}\bigg\lbrace\exp\left(2\zeta\int_{0}^{t}\mathscr{U}_{i}(\tau)d\tau\right)\bigg\rbrace
\nonumber\\&=n^{2}|\exp(2\psi^{E})\bm{\mathrm{I\!E}}\bigg\lbrace\exp\left(2\zeta\int_{0}^{t}\widehat{\mathscr{U}}(\tau)d\tau\right)\bigg\rbrace
\end{align}
if $\widehat{\mathscr{U}}_{i}(t,\vartheta)=\widehat{\mathscr{U}}(t,\vartheta)$.
\end{lem}
\begin{lem}
For the stochastically perturbed radii $\widehat{a}_{i}(t)$, the $L_{2}$ norm is estimated as
\begin{align}
&\|\widehat{\bm{a}}(t)-\bm{a}^{E}\|\le\left(n|a^{E}|^{2}\exp\left(\zeta\int_{0}^{t}\widehat{\mathscr{U}}(\tau)d\tau\right)\right)^{1/2}\nonumber\\&
=n^{1/2}a^{E}\exp\left(\zeta\int_{0}^{t}\mathscr{U}(\tau)d\tau)\right)
\end{align}
where $\widehat{\mathscr{U}}_{i}(t)=\mathscr{U}(t)$ for $i=1...n$ and $a_{i}(t)$ is a solution of $\mathbf{D}_{n}a_{i}(t)=0$. The expectation is then estimated as
\begin{equation}
\bm{\mathrm{I\!E}}\left\lbrace\bigg\|\widehat{\bm{a}}(t)-\bm{a}^{E}\bigg\|\right\rbrace\le n^{1/2}a^{E}\bm{\mathrm{I\!E}}\left\lbrace\exp\left(\zeta\int_{0}^{t}\mathscr{U}(\tau)d\tau\right)\right\rbrace
\end{equation}
Then
\begin{enumerate}
\item If $\lim_{t\rightarrow\infty}\bm{\mathrm{I\!E}}\big\lbrace\big\|\widehat{\bm{a}}_{i}(t)-\bm{a}^{E}\big\|^{\ell}    \big\rbrace=0$, then the initially static toroidal universe is asymptotically stable to the random perturbations.
\item If $\lim_{t\uparrow\infty}\bm{\mathrm{I\!E}}\lbrace\big\|\widehat{\bm{a}}(t)-\bm{a}^{E}\big\|^{\ell}=0$, then the initially static Kasner universe is 'Lyapunov stable' to the random perturbations.
\item If $\lim_{t\uparrow\infty}\bm{\mathrm{I\!E}}\lbrace \|\widehat{\bm{a}}(t)-\bm{a}^{E}\|^{\ell}\rbrace=0$ then the initially static Kasner universe is unstable to the random perturbations and will undergo a stochastically induced expansion to infinity.
\end{enumerate}
\end{lem}
The asymptotic behavior of the norms then requires the estimation of the stochastic integral. In particular, it will be shown that the stochastically induced expansion is exponential in nature so that the expanding universe essentially inflates from a static (non-singular) Kasner state.
\begin{thm}
Setting $\widehat{\mathscr{U}}_{i}(t)=\widehat{\mathscr{U}}(t)$ for $i=1...n$ for Gaussian random fields with $\bm{\mathrm{I\!E}}\lbrace\widehat{\mathscr{U}}(t)\rbrace=0$ and defined by the regulated 2-point function, the stochastic integral in (6.8) can be estimated as
\begin{equation}
\mathlarger{\mathsf{Y}}(t)=\bm{\mathrm{I\!E}}\left\lbrace\exp\left(\zeta\int_{0}^{t}
\widehat{\mathscr{U}}(\tau)d\tau\right)\right\rbrace \sim
\exp\left(\frac{1}{2}\zeta^{2}\int_{0}^{t}\int_{0}^{\tau_{1}}
d\tau_{1}d\tau_{2}\bigg\lbrace\widehat{\mathscr{U}}(\tau_{1})\widehat{\mathscr{U}}(\tau_{2})\bigg\rbrace\right)
\end{equation}
\end{thm}
\begin{proof}
The proof depends on evaluating the stochastic integral
\begin{equation}
\mathlarger{\mathbb{Y}}(t)=\bm{\mathrm{I\!E}}\left\lbrace\exp\left(\zeta \int_{0}^{t}
\widehat{\mathscr{U}}(\tau)d\tau\right)\right\rbrace\equiv\bm{\mathrm{I\!E}}\left\lbrace
\exp\left(\zeta\int_{0}^{t}\partial\widehat{\mathscr{U}}(\tau)\right)\right\rbrace
\end{equation}
by a cluster expansion method or Van Kampen expansion, similar to cluster integral techniques used in stochastic analysis and statistical mechanics [13,66,67]. First, the m-point correlations or moments $\bm{\mathrm{I\!E}}[t;m]$ and the m-point cumulants $\bm{\mathrm{K}}(t;m)$ are
\begin{align}
&\bm{\mathrm{I\!E}}(t;m)=\bm{\mathrm{I\!E}}\left\lbrace\widehat{\mathscr{U}}(t_{1})\times...
\times\widehat{\mathscr{U}}(t_{m})\right\rbrace
=\bm{\mathrm{I\!E}}\left\lbrace\prod_{\gamma=1}^{m}\widehat{\mathscr{U}}(t_{\xi})\right\rbrace
\\&\bm{\mathrm{I\!K}}(t;m)=\bm{\mathrm{I\!K}}\left\lbrace\widehat{\mathscr{U}}(t_{1})\times...
\times\widehat{\mathscr{U}}(t_{m})\right\rbrace=\bm{\mathrm{I\!K}}\left\lbrace\prod_{\xi=1}^{m}
\widehat{\mathscr{U}}(t_{\xi})\right\rbrace
\end{align}
The second-order cumulants are for example
\begin{equation}
\bm{\mathrm{I\!K}}\bigg\lbrace\mathscr{U}(t_{1})\mathscr{U}(t_{2})\bigg\rbrace=
\bm{\mathrm{I\!E}}\bigg\lbrace\mathscr{U}(t_{1})\mathscr{U}(t_{2})\bigg\rbrace+
\bm{\mathrm{I\!E}}\bigg\lbrace\mathscr{U}(t_{1})\bigg\rbrace
\bm{\mathrm{I\!E}}\bigg\lbrace\mathscr{U}(t_{2})\bigg\rbrace
\end{equation}
so that $\bm{\mathrm{I\!K}}\lbrace\mathscr{U}(t_{1})\mathscr{U}(t_{2})\rbrace=
\bm{\mathrm{I\!E}}\lbrace\mathscr{U}(t_{1})\mathscr{U}(t_{2})\rbrace$ if $\lbrace\mathscr{U}(t_{1})\rbrace=0$.
The moment and cumulant m-point correlations can be related to the generating functions for the stochastic process $\widehat{\mathscr{U}}(t)$
\begin{align}
&\bm{{\Psi}}[\mathscr{U}(t)]=\sum_{m=0}^{\infty} \frac{\xi^{m}}{m!}\int_{o}^{t}...\int_{o}^{t_{m-1}} d\tau_{1}...\tau_{m} \bm{\mathrm{I\!K}}\left\lbrace\prod_{\xi=0}^{m}\mathscr{U}(\tau_{\xi})\right
\rbrace\nonumber\\&\equiv\sum_{m=0}^{\infty}\frac{\xi^{m}}{m!}\prod_{\xi=1}^{m}\int d\tau_{\xi}\bm{\mathrm{I\!K}}\left\lbrace\prod_{\xi=0}^{m}\mathscr{U}(\tau_{\xi})\right\rbrace
\end{align}
while the cumulant-generating functional is
\begin{align}
&\bm{{\Phi}}[\mathscr{U}(t)]=\sum_{m=1}^{\infty} \frac{\xi^{m}}{m!}\int_{o}^{t}...\int_{o}^{t_{m-1}} d\tau_{1}...\tau_{m} \bm{\mathrm{I\!K}}\left\lbrace\prod_{\xi=1}^{m}\mathscr{U}(\tau_{\xi})\right\rbrace
\nonumber\\&\equiv\sum_{m=1}^{\infty} \frac{\xi^{m}}{m!}\prod_{\xi=1}^{m}\int d\tau_{\xi} \bm{\mathrm{I\!K}}\left\lbrace\prod_{\xi=1}^{m}\mathscr{U}(\tau_{\xi})\right\rbrace
\end{align}
Note that the first summation begins form $m=0$ whereas the second begins from $m=1$. These can be written more succinctly as
\begin{equation}
\bm{\Psi}[\widehat{\mathscr{U}}(t)]= \sum_{m=0}^{\infty} \frac{\xi^{m}}{m!}\int \mathlarger{\mathfrak{D}}_{m}[\tau]\bm{\mathrm{I\!E}}\left\lbrace\prod_{\xi=1}^{m}\mathscr{U}
(\tau_{\xi})\right\rbrace
\end{equation}
\begin{equation}
\bm{\Phi}[\widehat{\mathscr{U}}(t)]=\sum_{m=1}^{\infty} \frac{\xi^{m}}{m!}\int \mathlarger{\mathfrak{D}}_{m}[\tau]\bm{\mathrm{I\!K}}\left\lbrace\prod_{\xi=1}^{m}
\mathscr{U}(\tau_{\xi})\right\rbrace
\end{equation}
where $\int\mathfrak{D}_{m}[\tau]$ is a 'path integral'. But the the moment-generating functional is equal to the integral
\begin{equation}
\bm{\Psi}[\mathscr{U}(t)]=\bm{\mathrm{I\!K}}
\left\lbrace\exp\bigg(\zeta\int_{0}^{t}\mathscr{U}(\tau)d\tau\bigg)\right\rbrace
\end{equation}
The relation between the generating functionals is
\begin{equation}
\bm{\Phi}=\ln \bm{\Psi}
\end{equation}
so that $\bm{\Psi}=\exp(\bm{\Phi})$. Hence
\begin{align}
&\bm{\Psi}[\widehat{\mathscr{U}}(t)]=
\sum_{m=0}^{\infty} \frac{\xi^{m}}{m!}\int \mathfrak{D}_{m}[\tau]\bm{\mathrm{I\!E}}\left\lbrace\prod_{\xi=1}^{m}\mathscr{U}
(\tau_{\xi})\right\rbrace\nonumber\\&=\exp\left(\sum_{m=1}^{\infty} \frac{\xi^{m}}
{m!}\int \mathfrak{D}_{m}[\tau]\bm{\mathrm{I\!K}}\left\lbrace\prod_{\xi=1}^{m}
\mathscr{U}(\tau_{\xi})\right\rbrace\right)
\end{align}
or equivalently using (6.18)
\begin{align}
\mathlarger{\mathbf{Y}}(t)&=\bm{\mathrm{I\!E}}\left\lbrace\exp\bigg(\xi\int_{0}^{t}
\mathscr{U}(\tau)d\tau\bigg)\right\rbrace\nonumber\\&\equiv\sum_{m=0}^{\infty} \frac{\xi^{m}}{m!}\int\mathfrak{D}_{m}[\tau]\bm{\mathrm{I\!K}}\left\lbrace\prod_{\xi=1}^{m}\widehat{\mathscr{U}}
(\tau_{\xi})\right\rbrace\nonumber\\&=\exp\left(\sum_{m=1}^{\infty} \frac{\xi^{m}}{m!}\int \mathfrak{D}_{m}[\tau]\bm{\mathrm{I\!K}}\left\lbrace\prod_{\xi=1}^{m}\mathscr{U}(\tau_{\xi})\right\rbrace\right)
\end{align}
Expanding (6.23)
\begin{align}
\bm{\mathrm{I\!E}}\left\lbrace\exp\left(\zeta\int_{0}^{t}
\widehat{\mathscr{U}}(\tau)d\tau\right)\right\rbrace&=
\exp\bigg(\int\mathfrak{D}_{1}[\tau]
\bm{\mathrm{I\!K}}\bigg\lbrace\widehat{\mathscr{U}}(\tau)\bigg\rbrace\nonumber\\&+\frac{\zeta{2}}{2}
\int\mathlarger{\mathfrak{D}}_{2}[\tau]\bm{\mathrm{I\!E}}\bigg\lbrace\mathscr{U}(\tau_{1})
\mathscr{U}(\tau_{2})\bigg\rbrace\nonumber\\&
+...+\int\mathlarger{\mathfrak{D}}_{m}[\tau]\bm{\mathrm{I\!E}}\bigg\lbrace\widehat{\mathscr{U}}(\tau_{1})\times...\times
\mathscr{U}(\tau_{m})]\bigg)
\end{align}
Although there are some subtle technical issues regarding temporal ordering [68,69] the series can be truncated at second order for a Gaussian process with $\bm{\mathrm{I\!E}}\lbrace\widehat{\mathscr{U}}(t)\rbrace=0$ so that
\begin{align}
\bm{\mathrm{I\!E}}\left\lbrace\exp\bigg(\zeta\int_{0}^{t}
\widehat{\mathscr{U}}(\tau)d\tau\bigg)\right\rbrace&=\exp\left(\frac{1}{2}\zeta^{2}\int \mathlarger{\mathfrak{D}}_{2}[\tau]\bm{\mathrm{I\!K}}\bigg\lbrace\mathscr{U}(\tau_{1})
\mathscr{U}(\tau_{2})\bigg\rbrace\right)\nonumber\\&
\equiv\exp\left(\frac{1}{2}\zeta^{2}\int\mathlarger{\mathfrak{D}}_{2}[\tau]
\bm{\mathrm{I\!K}}\bigg\lbrace\widehat{\mathscr{U}}(\tau_{1})
\mathscr{U}(\tau_{2})\bigg\rbrace\right)
\end{align}
which is
\begin{equation}
\bm{\mathrm{I\!E}}\left\lbrace\exp\bigg(\zeta\int_{0}^{t}\widehat{\mathscr{U}}(\tau)d\tau\bigg)\right\rbrace
=\exp\left(\frac{1}{2}\zeta^{2}\int_{0}^{t}\int_{0}^{\tau_{1}}d\tau_{1}d\tau_{2}
\bm{\mathrm{I\!E}}\bigg\lbrace\mathscr{U}(\tau_{1})
\widehat{\mathscr{U}}(\tau_{2})\bigg\rbrace\right)
\end{equation}
Choosing $\zeta=1$ gives
\begin{equation}
\bm{\mathrm{I\!E}}\bm{\mathrm{I\!E}}\left\lbrace\exp\bigg(\int_{0}^{t}
\widehat{\mathscr{U}}(\tau)d\tau\bigg)\right\rbrace=\exp\left(\frac{1}{2}
\int_{0}^{t}\int_{0}^{\tau_{1}}d\tau_{1}d\tau_{2}\bm{\mathrm{I\!E}}
\bigg\lbrace\widehat{\mathscr{U}}(\tau_{1})\mathscr{U}(\tau_{2})
\bigg\rbrace\right)
\end{equation}
and so the result follows.
\end{proof}
\subsection{Results for some classical regulated 2-point functions}
The norm $\bm{\mathrm{I\!E}}\lbrace\|\widehat{\bm{a}}(t)-\bm{a}^{E}\|\rbrace$ and its asymptotic behavior can be computed for viable regulated 2-point functions, such as that for 'colored noise' or for Gaussian-correlated noise.
\begin{thm}
Let the conditions of Theorem (5.1) hold for an initially static or stationary and isotropic hypertoroidal spacetime such that $a_{i}(0)=a_{i}^{E}=a^{E}$ and $\psi_{i}(t)=\psi_{i}^{E}=\psi^{E}$ so that the Einstein equations are $\mathbf{D}_{n}a^{E}=\mathbf{H}_{n}\psi^{E}0$. Introducing Gaussian stochastic perturbations $\widehat{\mathscr{U}}(t)$ of the moduli then $\widehat{\psi}_{i}(t)=\psi_{i}^{E}+\int_{0}^{t}\widehat{\mathscr{U}}(\tau)
d\tau$. The radii are
$\widehat{a}_{i}(t)=a^{E}\exp(\int_{0}^{t}\widehat{\mathscr{U}}(\tau,\varsigma)d\tau)$.
\begin{enumerate}
\item If the regulated 2-point function of the random perturbations is of the colored noise or Ornstein-Uhlenbeck form with correlation $\varsigma$ then set $\widehat{\mathscr{U}}(t)=\widehat{\mathscr{U}}(t)$ so that the regulated 2-point function is
\begin{align}
&J(\Delta;\varsigma)=\bm{\mathrm{I\!E}}\bigg\lbrace\widehat{\mathscr{U}}(t)\widehat{\mathscr{U}}(s)
\bigg\rbrace=\frac{C}{\varsigma}\exp\left(-\frac{|t-s|}{\varsigma}\right)
\\&J(0;\varsigma)=\bm{\mathrm{I\!E}}\bigg\lbrace\widehat{\mathscr{U}}(t)\widehat{\mathscr{U}}(t)
\bigg\rbrace=\frac{C}{\varsigma}
\end{align}
where $\varsigma$,is the correlation time and $\widehat{\mathscr{C}}(t)$ is a solution of the (Ornstein-Uhlenbeck) linear stochastic DE
\begin{equation}
\partial_{t}\widehat{\mathscr{U}}(t)=-\frac{1}{\varsigma}\widehat{\mathscr{U}}(t)
+\frac{1}{\varsigma}\widehat{\mathscr{W}}(t)=
\end{equation}
and $\widehat{\mathscr{W}}(t)$ is a white noise with $\bm{\mathrm{I\!E}}\lbrace \widehat{\mathscr{W}}(t)\widehat{\mathscr{W}}(t)\rbrace=\alpha\delta(t-s)$. The expectation of the randomly perturbed norms then evolve as
\begin{align}
&\delta\bm{\mathrm{I\!E}}\bigg\lbrace\bigg\|\widehat{a}_{i}(t)\bigg\|\bigg\rbrace
=\bm{\mathrm{I\!E}}\bigg\lbrace\bigg\|\widehat{a}_{i}(t)-a_{i}^{E}\bigg\|\bigg\rbrace\nonumber\\& \le a^{E}n^{1/2}\bm{\mathrm{I\!E}}\bigg\lbrace\exp\left(\int_{0}^{t}\widehat{\mathscr{U}}
(\tau)d\tau\right)\bigg\rbrace\nonumber\\& a^{E}n^{1/2}\exp\left(C t-C\exp(-\beta t)+C\right)\sim a^{E}n^{1/2}\exp(Ct)
\end{align}
for large t.
\item If the random perturbations or noise is Gaussian correlated then $\bm{\mathrm{I\!E}}\lbrace\widehat{\mathscr{G}}(t)\rbrace=0 $ with 2-point function
    \begin{equation}
    J(\Delta;\varsigma)=\bm{\mathrm{I\!E}}\bigg\lbrace\widehat{\mathscr{G}}(t)\widehat{\mathscr{G}}(s)\bigg\rbrace
=\frac{C}{\varsigma^{2}}\exp\left(-\frac{|t-s|^{2}}{\varsigma^{2}}\right)
    \end{equation}
    with $J(0;\varsigma)=\bm{\mathrm{I\!E}}\lbrace\widehat{\mathscr{G}}(t)\widehat{\mathscr{G}}(s)\rbrace=\frac{C}{\varsigma^{2}}$
    Then similarly for $t\gg \varsigma$
    \begin{equation}
    \delta\bm{\mathrm{I\!E}}\bigg\lbrace\bigg\|\widehat{a}_{i}(t)\bigg\|\bigg\rbrace
=\bm{\mathrm{I\!E}}\bigg\lbrace\bigg\|\widehat{a}_{i}(t)-a_{i}^{E}\bigg\|\bigg\rbrace\sim a^{E}n^{1/2}\exp(\lambda t)
\end{equation}
\end{enumerate}
Then on average, in both cases, the randomly perturbed radii then evolve exponentially or 'inflate'.
\end{thm}
\begin{proof}
The averaged perturbed Einstein equations follow from
\begin{equation}
\bm{\mathrm{I\!E}}\left\lbrace \mathbf{D}_{n}\widehat{a}_{i}(t)\right\rbrace=\frac{1}{2}\xi^{2}
n\Sigma(0;\varsigma)\equiv\frac{1}{2}\xi^{2}n\frac{C}{\varsigma}\equiv\lambda
\end{equation}
The expectation of the stochastically perturbed norms is estimated as
\begin{align}
\bm{\mathrm{I\!E}}\bigg\lbrace\bigg\|\delta_{n}\widehat{a}_{i}(t)\bigg\|\bigg\rbrace&=\bm{\mathrm{I\!E}}\bigg\lbrace\bigg\|\bm{a}(t)-\bm{a}^{E}\bigg\|\bigg\rbrace\nonumber\\&\le
a^{E}n^{1/2}\bm{\mathrm{I\!E}}\left\lbrace\exp\left(\zeta\int_{0}^{t}\widehat{\mathscr{U}}(\tau)d\tau\right)\right\rbrace
\nonumber\\&=a^{E}n^{1/2}\exp\left(\frac{1}{2}\zeta^{2}\int_{0}^{t}\int_{0}^{\tau_{1}}d\tau_{1}d\tau_{2}
\bm{\mathrm{I\!E}}\bigg\lbrace\widehat{\mathscr{U}}(\tau_{1})
\widehat{\mathscr{U}}(\tau_{2})\bigg\rbrace\right)\nonumber\\&a^{E}n^{1/2}\exp\left(\frac{1}{2}\frac{C\zeta^{2}}{\varsigma}\int_{0}^{t}\int_{0}^{\tau_{1}}d\tau_{1}d\tau_{2}
\exp\left(-\frac{|\tau_{1}-\tau_{2}|}{\varsigma}\right)\right)\nonumber\\&
a^{E}n^{1/2}\exp\left(\frac{C\zeta^{2}}{\varsigma}\int_{0}^{t}d\tau_{1}(\varsigma-\varsigma
\exp(-\tau_{1}/\varsigma))\right)\nonumber\\&=a^{E}n^{1/2}\exp\left(\frac{C\zeta^{2}}{\varsigma}\left(\varsigma t-\varsigma\int_{0}^{t}d\tau_{1}\exp(-\tau_{1}/\varsigma)\right)\right)\nonumber\\&=a^{E}n^{1/2}\exp\left(\frac{1}{2}\frac{C\zeta^{2}}{\varsigma}\left(\varsigma t+\varsigma(1-\exp(-t/\varsigma))\right)\right)\nonumber\\&=a^{E}n^{1/2}\exp\left(\frac{1}{2}\zeta^{2}Ct+\frac{1}{2}\gamma^{2}C(1-\exp(-t/\varsigma))
\right)
\end{align}
For large $t\gg \varsigma$, the estimate is
\begin{align}
&\bm{\mathrm{I\!E}}\bigg\lbrace\delta\bigg\|\widehat{\bm{a}}(t)\bigg\|\bigg\rbrace \sim\bm{\mathrm{I\!E}}\bigg\lbrace\bigg\|\widehat{\bm{a}}(t)-\bm{a}^{E}\bigg\|\bigg\rbrace\nonumber\\&=
a^{E}n^{1/2}\exp\left(\frac{1}{2}\zeta^{2}Ct\right)\equiv a^{E}n^{1/2}\exp(Qt)
\end{align}
For the Gaussian-correlated 2-point function
\begin{align}
\bm{\mathrm{I\!E}}\bigg\lbrace\bigg\|\delta_{n}\widehat{a}_{i}(t)&\bigg\|\bigg\rbrace=
\bm{\mathrm{I\!E}}\bigg\lbrace\bigg\|\widehat{\bm{a}}(t)-\bm{a}^{E}\bigg\|\bigg\rbrace\nonumber\\&\le a^{E}n^{1/2}\bm{\mathrm{I\!E}}\bigg\lbrace\exp\big(\mu\int_{0}^{t}\widehat{\mathscr{U}}(\tau)d\tau)
\bigg\rbrace\nonumber\\&=a^{E}n^{1/2}\exp\big(\frac{1}{2}\mu^{2}\int_{0}^{t}\int_{0}^{\tau_{1}}d\tau_{1}d\tau_{2}
\bm{\mathrm{I\!E}}\bigg\lbrace\widehat{\mathscr{U}}(\tau_{1})\widehat{\mathscr{U}}(\tau_{2})\bigg\rbrace\big)\nonumber\\&
=a^{E}n^{1/2}\exp(\frac{1}{2}C\frac{\mu^{2}}{\varsigma}\int_{0}^{t}d\tau_{1}(\int_{0}^{\tau_{1}}d\tau_{2}\exp\big(-\frac{|\tau_{1}-\tau_{2}|^{2}}{\varsigma^{2}})\big)\nonumber\\&
=a^{E}n^{1/2}\exp((\frac{1}{2}C\frac{\mu^{2}}{\varsigma}\int_{0}^{t}d\tau_{1}(-\frac{1}{2}\pi^{1/2}erf(0)+\frac{1}{2}\pi^{1/2}\varsigma erf(\tau_{1}/\varsigma))))\nonumber\\&=a^{E}n^{1/2}\exp((\frac{1}{4}C\frac{\mu^{2}}{\varsigma}\int_{0}^{t}
\pi^{1/2}\varsigma erf(\tau_{1}/\varsigma)))\nonumber\\&
=a^{E}n^{1/2}\exp(\frac{C\mu^{2}}{2\varsigma}(\varsigma
t erf(t/\varsigma)+\frac{\varsigma}{\pi^{1/2}}\exp(-t^{2}/\varsigma^{2})0)\nonumber\\&
=a^{E}n^{1/2}\exp (\frac{1}{2}C\mu^{2}t erf(t/\varsigma)+\frac{1}{2}C\mu^{2}\pi^{-1/2}
\exp(-t^{2}/\varsigma^{2})))
\end{align}
and for large $t\gg \varsigma$ one has $erf(t/\varsigma)=1$ and $\exp(-t^{2}/\varsigma^{2})=0$ so that
\begin{equation}
\bm{\mathrm{I\!E}}\bigg\lbrace\bigg\|\delta_{n}\widehat{a}_{i}(t)\bigg\|\bigg\rbrace=
\bm{\mathrm{I\!E}}\bigg\lbrace\bigg\|\widehat{\bm{a}}(t)-\bm{a}^{E}\bigg\|\bigg\rbrace\sim a^{E}n^{1/2}\exp(\tfrac{1}{2}C\mu^{2}t)=a^{E}n^{1/2}\exp(Qt)
\end{equation}
\end{proof}
\begin{cor}
Asymptotically, the system is then unstable to the stochastic perturbations since
\begin{equation}
\lim_{t\uparrow\infty}\bm{\mathrm{I\!E}}\bigg\lbrace\bigg\|\delta \widehat{\bm{a}}(t)\bigg\|\bigg\rbrace=\lim_{t\uparrow\infty}\bm{\mathrm{I\!E}}\bigg\lbrace\bigg\|
\widehat{\bm{a}}(t)-\bm{a}^{E}\bigg\|\bigg\rbrace \sim \lim_{t\uparrow\infty}a^{E}n^{1/2}\exp(\tfrac{1}{2}\zeta^{2}C t)=\infty
\end{equation}
so the system grows exponentially or undergoes a noise-induced inflation for eternity. Also
\begin{equation}
\bm{\mathrm{I\!E}}\bigg\lbrace\delta\bigg\|\widehat{a}_{i}(t)\bigg\|\bigg\rbrace\bigg\|\bm{a}^{E}\bigg\|^{-1}
\equiv\bm{\mathrm{I\!E}}\bigg\lbrace\bigg\|\widehat{a}_{i}(t)-a_{i}^{E}\bigg\|\bigg\rbrace\bigg\|\bm{a}^{E}\bigg\|^{-1} \sim \exp(Q t)
\end{equation}
so that $Q$ plays the role of a positive Lyupunov exponent.
\end{cor}
Since the stochastic integral has been estimated, one can reprise Theorem 2.18. for the probability of instability in terms of a Chernoff bound
\begin{cor}
The probability that $\lim_{t\uparrow\infty}\|\widehat{\bm{a}}(t)-\bm{a}^{E}\|$ is bounded by any finite $|L|$ is zero--hence the randomly perturbed static Kasner universe is unstable and expands forever. From (2.79) $\exists$$\xi>0$ such that
\begin{align}
\lim_{t\uparrow\infty}\bm{\mathrm{I\!P}}&[\|\widehat{\bm{a}}(t)-\bm{a}^{E}\|\le |L|]\nonumber\\&\le \lim_{t\uparrow\infty}\inf_{\beta>0}\exp(\beta
|L|)\bm{\mathrm{I\!E}}\lbrace\exp\left(-\beta\bigg \|a(t)-a^{E}\bigg \|\right)\bigg\rbrace\nonumber\\& \equiv
\lim_{t\uparrow\infty}\inf_{\beta>0}\exp(\beta|L|)\bm{\mathrm{I\!E}}\left\lbrace
a^{E}n^{1/2}\exp\left(\beta\int_{0}^{t}\widehat{\mathscr{U}}(\tau)d\tau\right)\right\rbrace
\nonumber\\&\equiv\lim_{t\uparrow\infty}\inf_{\beta>0}\exp(\beta|L|)\exp\left(-\beta|a^{E}|n^{1/2}
\left|\bm{\mathrm{I\!E}}\left\lbrace\exp\left(\gamma\ell\int_{0}^{t}\widehat{\mathscr{U}}(\tau)d\tau\right)\right\rbrace\right|^{1/\ell}\right)
\nonumber\\&=\lim_{t\uparrow\infty}\inf_{\beta>0}\exp(\beta|L|)\exp(-\beta|a^{E}|n^{1/2}\exp(Qt))=0
\end{align}
Since
\begin{equation}
\lim_{t\uparrow\infty}\bm{\mathrm{I\!P}}[\|\widehat{a}_{i}(t)-a_{i}^{E}\|\le L]+\lim_{t\uparrow\infty}\bm{\mathrm{I\!P}}[\|\widehat{a}_{i}(t)-a_{i}^{E}\|> L]=1
\end{equation}
then
\begin{equation}
\lim_{t\uparrow\infty}\bm{\mathrm{I\!P}}[\|\widehat{a}_{i}(t)-a_{i}^{E}\|> L]=1
\end{equation}
\begin{equation}
\lim_{t\uparrow\infty}\bm{\mathrm{I\!P}}[\|\widehat{a}_{i}(t)-a_{i}^{E}\|\le L]=0
\end{equation}
so the randomly perturbed system is asymptotically unstable in probability and never reaches equilibrium since the perturbed norm can never be contained with any ball of radius $L$.
\end{cor}
The same conclusion follows form the Hoeffding Lemma of Proposition 2.22.
\begin{lem}
As before, let $(\widehat{a}_{i}(t))_{i=1}^{n}=(\widehat{a}_{1}(t)...\widehat{u}_{n}(t))$ be the set of randomly perturbed radii due to random perturbations of the initially static equilibria $a_{i}^{E}$ such that $\widehat{a}_{i}(t)=a_{i}^{E}\exp(\zeta
\int_{0}^{t}\mathscr{U}_{i}(s)ds$. Then $\widehat{a}_{i}(t)$ solves the averaged Einstein equation such that $\bm{\mathrm{I\!E}}\lbrace\mathbf{D}_{n}\widehat{a}{i}(t)\rbrace=\lambda$. If the the randomly perturbed actually radii converged to 'attractors' or new equilibria within a finite time such that $\widehat{a}_{i}(t)\rightarrow a_{i}^{E*}$, then $\widehat{a}_{i}^{E}\le \widehat{a}_{i}(t)\le a_{i}^{E*}$ for all finite $t>0$. If
\begin{align}
&\widehat{S}(t)=\frac{1}{n}\sum_{i=1}^{n}\widehat{a}_{i}(t)
=\frac{1}{n}(\widehat{a}_{1}(t)+...+\widehat{a}_{n}(t))\\&
\bm{\mathrm{I\!E}}\big\lbrace\widehat{S}(t)\big\rbrace
=\frac{1}{n}\sum_{i=1}^{n}\bm{\mathrm{I\!E}}\bigg\lbrace\widehat{a}_{i}(t)\bigg\rbrace
\end{align}
If however, $u_{i}^{E*}\rightarrow\infty$ the the random variables $\widehat{a}_{i}(t)$ are no longer bounded. The estimate is then
\begin{align}&
\bm{\mathrm{I\!P}}(\widehat{S}(t)-\bm{\mathrm{I\!E}}
\big\lbrace\widehat{S}(t)\big\rbrace\ge |L|)\nonumber\\&=
\bm{\mathrm{I\!P}}\bigg(\frac{1}{n}\bigg\|\sqrt{a_{i}^{E}}
\exp\bigg(\frac{\zeta}{2}\int_{0}^{t}\mathscr{U}_{i}(s)ds\bigg)\bigg\|_{L_{2}}^{2}
-\frac{1}{n}\bigg\|\sqrt{a_{i}^{E}\bm{\mathrm{I\!E}}\bigg\lbrace
\bigg(\exp\bigg(\zeta\int_{0}^{t}\mathscr{U}_{i}(s)ds\bigg)\bigg\rbrace}\bigg\|_{L_{2}}^{2}\ge |L|\bigg)\nonumber\\&=\lim_{a_{i}^{E*}\uparrow\infty}\exp(-2n^{2}L^{2}\big\|a_{i}^{E*}-a_{i}^{E}\big\|_{L_{2}}^{-2})=1
\end{align}
and there is then unit probability that the growth of the norm of perturbed solutions cannot be contained within any finite $L>0$. Hence, the system is asymptotically unstable in probability to the random perturbations.
\end{lem}
The averaged solution also satisfies the Einstein equations with a cosmological constant.
\begin{lem}
Let $\mathlarger{\mathbb{A}}_{i}(t)=\bm{\mathrm{I\!E}}\big\lbrace a_{i}^{+}(t)\big\rbrace=a_{i}^{E}\exp(Q^{1/2}t)$, which is the stochastic average as computed from(-)and Thm (-). Then $\mathlarger{\mathbb{A}}^{+}(t)$ is a solution of the system of deterministic nonlinear ODES
\begin{equation}
\mathbf{D}_{n}\mathlarger{\mathbb{A}}_{i}(t)=\sum_{i=1}^{n}\frac{\partial_{tt}\mathlarger{\mathbb{A}}_{i}(t)}{\mathlarger{\mathbb{A}}_{i}^{+}(t)}
-\frac{1}{2}\sum_{i=1}^{n} \frac{\partial_{t}\mathlarger{\mathbb{A}}_{i}(t)\partial_{t}\mathlarger{\mathbb{A}}_{i}(t)}{\mathlarger{\mathbb{A}}_{i}(t)
\mathlarger{\mathbb{A}}_{i}(t)}+\frac{1}{2}\sum_{i=1}^{n} \frac{\partial_{t}\mathlarger{\mathbb{A}}_{i}(t)\partial_{t}\mathlarger{\mathbb{A}}_{j}(t)}{\mathlarger{\mathbb{A}}_{i}(t)
\mathlarger{\mathbb{A}}_{j}(t)}
=\lambda=nV^{2}
\end{equation}
where $\lambda$ is a cosmological constant term and $Q_{i}=Q$. This is equivalent to the system of ODEs.
\end{lem}
\begin{proof}
From (6.34) or (6.36), $\mathlarger{\mathbb{A}}_{i}(t)=a^{E}_{i}\exp(Q_{i}t)$. The derivatives are: $\partial_{t}\mathlarger{\mathbb{A}}_{i}^{+}(t)=a_{i}^{E}Q_{i}\exp(Q_{i}t)$ and $\partial_{tt}\mathlarger{\mathbb{A}}^{+}_{i}(t)=a_{i}^{E}Q_{i}Q_{i}\exp(Q_{i}t)$.
\begin{align}
\mathbf{D}_{n}\mathlarger{\mathbb{A}}_{i}(t)&=\sum_{i=1}^{n}\frac{a_{i}^{E}\exp(Q_{i} t)}{a_{i}^{E}\exp(Q_{i}t)}-\frac{1}{2}\sum_{i=1}^{n}\frac{a_{i}^{E}a_{i}^{E}Q_{i}Q_{i}\exp(Q_{i} t)\exp(Q_{i} t)}{a_{i}^{E}a_{i}^{E}\exp(Q_{i} t)\exp(Q_{i}t)}\nonumber\\&
+\frac{1}{2}\sum_{i=1}^{n}\sum_{j=1}^{n}\frac{a_{i}^{E}a_{j}^{E}Q_{i}Q_{j}\exp(Q_{i} t)
\exp(Q_{i} t)}{a_{i}^{E}a_{j}^{E}\exp(Q_{i} t)\exp(Q_{i} t)}\nonumber\\&
\frac{1}{2}\sum_{i=1}^{n}Q_{i}^{2}+\frac{1}{2}\sum_{i=1}^{n}V_{i}^{2}=\sum_{i=1}^{n}Q_{i}^{2}
=nQ^{2}=\lambda=\lambda
\end{align}
setting $a_{i}^{E}=a^{E}$ and $Q_{i}=V$ for $i=1...n$ and cancelling terms. Then $Q=(\lambda/n)^{1/2}$ so that $\mathlarger{\mathbb{A}}_{i}(t)=a^{E}\exp((Q/n)^{1/2}t)$. This then
agrees with equation (3.72) of Lemma (3.12)
\end{proof}
\subsection{Lyapunov exponents and relation to induced cosmological constant terms}
The final result establishes that the perturbed Einstein vacuum equations of (4.97)and (4.98) for a constant perturbation, namely $\mathbf{D}_{n}\overline{a_{i}(t)}$ and the stochastically averaged Einstein equations $\bm{\mathrm{I\!E}} \mathbf{D}_{n}\widehat{a}_{i}(t)\rbrace $ are equivalent in that an initially static toroidal universe expands exponentially or inflates under either of these perturbations.
\begin{lem}
Given the Einstein equations $\mathbf{D}_{n}{a}^{E}(t)\equiv \mathbf{H}_{n}\psi^{E}(t)=0$ describing a static hypertoroidal micro-universe, the continuous deterministic and stochastic perturbations are
\begin{align}
&\overline{a_{i}(t)}=a^{E}\exp\left(\int_{0}^{t}\mathcal{A} d\tau\right)
\\&
\widehat{a}_{i}(t)=a^{E}\exp\left(\int_{0}^{t}\widehat{\mathscr{U}}(\tau)d\tau\right)
\end{align}
arising from $\overline{\psi}_{i}(t)=\psi^{E}+\int_{0}^{t}\mathcal{A} d\tau=\psi^{E}+\mathcal{A} t$ and $\widehat{\psi}_{i}(t)
=\psi^{E}+\int_{0}^{t}\widehat{\mathscr{U}}(\tau)d\tau$. Using previous results it follows that we have the equivalent estimates
\begin{align}
&\frac{\big\|\delta\overline{\bm{a}(t)}\|}{\big\|a^{E}\big\|}
\equiv\big \|\frac{\overline{\bm{a}(t)}-\bm{a}^{E}\big\|}{\big\|\bm{a}^{E}\big\|}
\sim \exp(Q t)\equiv \exp(\mathbf{Ly}_{1}t)\\&\bm{\mathrm{I\!E}}\bigg\lbrace\delta\bigg\|\widehat{\bm{a}}(t)\bigg\|\bigg\rbrace\bigg\|
\bm{a}^{E}\bigg\|^{-1}\equiv\bm{\mathrm{I\!E}}\bigg\lbrace\bigg\|\widehat{\bm{a}}(t)-\bm{a}^{E}\bigg\|\bigg\rbrace\bigg\|a_{i}^{E}\bigg\|^{-1} \sim \exp(Qt)\equiv \exp(\mathbf{Ly}_{2} t)
\end{align}
It follows that $\mathbf{Ly}_{1},\mathbf{Ly}_{2}$ are essentially Lyupunov exponents so that
\begin{align}
&\mathbf{Ly}_{1}\sim\lim_{t\uparrow}\frac{1}{t}\ln\left(\frac{\|\overline{\bm{a}(t)}-\bm{a}^{E}\|}{\|\bm{a}^{E}\|}\right)
\\& \mathbf{Ly}_{2}\sim \lim_{t\uparrow\infty}\frac{1}{t}\log\bm{\mathrm{I\!E}}
\bigg\lbrace\bigg\|\widehat{\bm{a}}(t)-\bm{a}^{E}\bigg\|\bigg\rbrace
\bigg\|\bm{a}_{i}^{E}\bigg\|^{-1}=\lim_{t\uparrow\infty}\frac{1}{t}\ln\bm{\mathrm{I\!E}}
\bigg\lbrace\bigg\|\widehat{\bm{a}}(t)-\bm{a}^{E}\bigg\|\bigg\rbrace-\lim_{t\uparrow\infty}\frac{1}{t}
\bigg\|\bm{a}^{E}\bigg\|\nonumber\\&\sim\lim_{t\uparrow\infty}\frac{1}{t}\ln
\bm{\mathrm{I\!E}}\bigg\lbrace\bigg\|\widehat{\bm{a}}(t)-\bm{a}^{E}\bigg\|\bigg\rbrace
\end{align}
\end{lem}
The perturbed radii can be shown to be equivalent to the characteristic function, and the $\ell^{th}$-order Lyapunov exponent (LE) can be estimated as follows.
\begin{lem}
let $a^{E}$ be a set of equilibrium solutions of the n-dimensional Einstein system $\mathbf{D}_{n}a_{i}(t)=0$, and as before let $\widehat{a}_{i}(t)$ be the randomly perturbed solutions which are a solution of the averaged Einstein system
$\mathbf{D}_{n}a_{i}(t)=\lambda$ as in (3.72). Then
\begin{align}
&\bm{\mathrm{I\!E}}\bigg\lbrace\bigg\|\widehat{\bm{a}}(t)-\bm{a}^{E}\bigg\|\bigg\rbrace=\bm{\Psi}(z)=\bm{\Psi}(-i\ell)\nonumber\\&
\mathlarger{\mathbf{Ly}}(iz)\equiv\mathlarger{\mathbf{Ly}}(\ell)\sim n^{\ell/2}|u^{E}|^{\ell}Q>0
\end{align}
where $\bm{\Psi}(-i\ell)=\bm{\Psi}(z)$ is the characteristic function (-) and $\bm{\mathbf{Ly}}(iz)$ is the LE.
\end{lem}
\begin{proof}
\begin{align}
&\bm{\Psi}(z)=\bm{\mathrm{I\!E}}\bigg\lbrace\exp(iz\ln\|\widehat{\bm{a}}(t)-a^{E}\|)\bigg\rbrace\nonumber\\&
<\bm{\mathrm{I\!E}}\bigg\lbrace\exp(iz\ln(n^{1/2}|u^{E}|\exp(\zeta\int_{0}^{t}\mathscr{U}(s)ds))
\bigg\rbrace\nonumber\\&=\bm{\mathrm{I\!E}}\bigg\lbrace\exp((iz\ln(n^{1/2}|u^{E}|)+iz\zeta\int_{0}^{t}\mathscr{U}(s)ds)\bigg\rbrace\nonumber\\&
=\bm{\mathrm{I\!E}}\bigg\lbrace\exp(\ln(n^{1z/2}|u^{E}|^{iz}+iz\zeta\int_{0}^{t}\mathscr{U}(s)ds)\bigg\rbrace\nonumber\\&
=n^{\ell/2}|u^{E}|^{\ell}\bm{\mathrm{I\!E}}\bigg\lbrace\exp(\zeta\ell\int_{0}^{t}\mathscr{U}(s)ds)\bigg\rbrace\nonumber\\&
\equiv\bm{\mathrm{I\!E}}\bigg\lbrace\|\widehat{\bm{a}}(t)-\bm{a}^{E}\|^{\ell}\bigg\rbrace\sim n^{\ell/2}|u^{E}|^{\ell}\exp(Qt)
\end{align}
where $z=-i\ell$. The LE is estimated as
\begin{align}
&\bm{\mathsf{Ly}}(iz)=\lim_{t\uparrow\infty}\frac{1}{t}\ln\bigg\lbrace \exp(iz\ln\|\widehat{\bm{a}}(t)-\bm{a}^{E}\|\bigg\rbrace\nonumber\\&
=\lim_{t\uparrow\infty}\frac{1}{t}\ln\bigg\lbrace{\bm{\mathrm{I\!E}}}\bigg\lbrace\exp(iz\ln(n^{1/2}|u^{E}|\exp(\zeta
\int_{0}^{t}\mathscr{U}(s)ds))\bigg\rbrace\nonumber\\&=\lim_{t\uparrow\infty}\frac{1}{t}
\ln{\bm{\mathrm{I\!E}}}\bigg\lbrace\exp(iz\ln(n^{1/2}|u^{E}|)+iz\zeta\int_{0}^{t}\mathscr{U}(s)ds\bigg\rbrace\nonumber\\&
=\lim_{t\uparrow\infty}\frac{1}{t}\ln\bigg\lbrace\bm{\mathrm{I\!E}}\bigg\lbrace
\exp(\ln(n^{iz/2}|u^{E}|^{iz})+iz\zeta\int_{0}^{t}\mathscr{U}(s)ds\bigg\rbrace\nonumber\\&
=\lim_{t\uparrow\infty}\frac{1}{t}\ln n^{\ell/2}|u^{E}|^{\ell}{\bm{\mathrm{I\!E}}}\bigg\lbrace\exp(\zeta\ell\int_{0}^{t}\mathscr{U}(s)ds)\bigg\rbrace\nonumber\\&
\sim \lim_{t\uparrow\infty}\frac{1}{t}\ln(n^{\ell/2}|u^{E}|^{\ell})
+\lim_{t\uparrow\infty}\frac{1}{t}\ln{\bm{\mathrm{I\!E}}}\bigg\lbrace\exp(\zeta\ell\int_{0}^{t}\mathscr{U}(s)ds)\bigg\rbrace\nonumber\\&
\sim\lim_{t\uparrow\infty}\frac{1}{t}\ln{\bm{\mathrm{I\!E}}}\bigg\lbrace\exp(\zeta\ell\int_{0}^{t}\mathscr{U}(s)ds)
\bigg\rbrace\sim\lim_{t\uparrow\infty}\frac{1}{t}\ln(\exp(Qt))=Q>0
\end{align}
where $z=-il$, and $Q>0$ implies instability.
\end{proof}
\subsection{Alternative Estimate}
As a consistency check, it can also be shown that the estimate
\begin{equation}
\bm{\mathrm{I\!E}}\lbrace\bigg\|\widehat{\bm{a}}(t)-\bm{a}^{E}\bigg\|\rbrace\sim
n^{1/2}a^{E}\exp(Q|t-t_{o}|)
\end{equation}
also arises via a different method, when only certain conditions are imposed on the covariance of the field $\widehat{\mathscr{U}}_{i}(t)$.
\begin{thm}
Let $\widehat{\mathscr{U}}_{i}(t)=\widehat{\mathscr{U}}(t)$ be a Gaussian random field such that $\bm{\mathrm{I\!E}}\lbrace\widehat{\mathscr{U}}_{i}(t)\rbrace =0$ and $\bm{\mathrm{I\!E}}\lbrace
\widehat{\mathscr{U}}(t)\widehat{\mathscr{U}}(s)\rbrace=J(\Delta;\varsigma)\equiv J(s-t;\varsigma)$. The following hold:
\begin{enumerate}
\item $\exists C_{1}>0$ and $C_{2}>0$ such that
\begin{align}
&\bm{\mathrm{I\!E}}\lbrace\bigg\|\widehat{\mathscr{U}}(t)\bigg\|^{2}\rbrace \le\nonumber\\& C_{1}\int_{t_{o}}^{\infty}\bm{\mathrm{I\!E}}\bigg\lbrace\|\widehat{\mathscr{U}}(t)\widehat{\mathscr{U}}(s)ds
\bigg\|\rbrace =\int_{t_{o}}^{\infty}J(t-s;\varsigma)ds\le C_{2}
\end{align}
\item The random field $\widehat{\mathscr{U}}(t)$ has a Fourier expansion
\begin{equation}
\widehat{\mathscr{U}}(t)=\sum_{\xi=1}^{\infty}(\mathcal{E}_{\xi})^{1/2}\varphi(t)
\widehat{\mathscr{U}}(\xi)
\end{equation}
where $\bm{\mathrm{I\!E}}\lbrace\widehat{\mathscr{U}}(\xi)\widehat{\mathscr{U}}(\xi)\rbrace=1$ and where $\mathcal{E}_{\xi}$ and $\varphi(t)$ are the normalised eigenfunctions and eigenvalues such that $
|\int_{t_{o}}^{t}J(t-s;\varsigma)|\varphi(s)ds=\mathcal{E}\varphi(t) $ and where
\begin{equation}
\int_{t_{o}}^{t}|\widehat{\mathscr{U}}(s)|^{2}ds=\sum_{\xi=1}^{\infty}
\mathcal{E}_{\xi}|\widehat{\mathscr{U}}(\xi)|^{2}
\end{equation}
is the Parsaval identity.
\item
\begin{equation}
\bm{\mathrm{I\!E}}\bigg\lbrace\exp\bigg(\beta\mathcal{E}_{\xi}|\widehat{\mathscr{U}}(\xi)|^{2}\bigg\rbrace=(1-2\beta\mathcal{E}_{\xi})^{-1/2}
\end{equation}
for $\beta<(2\mathcal{E}_{\xi})^{-1}$
\end{enumerate}
Then the following estimates can be made
\begin{align}
&\bm{\mathrm{I\!E}}\bigg\lbrace\exp\bigg(\beta\int_{t_{o}}^{t}|\widehat{\mathscr{U}}(t)|^{2}dt\bigg\rbrace
=\prod_{\xi=1}^{\infty}(1-2\beta\mathcal{E}_{\xi})^{-1}
|nonumber\\&\bm{\mathrm{I\!E}}\bigg\lbrace\bigg\|a(t)-a^{E}\bigg\|\bigg\rbrace
=n^{1/2}a^{E}\exp\bigg(\zeta \int_{t_{o}}^{t}\widehat{\mathscr{U}}(s)ds\bigg)
\nonumber\\& <n^{1/2}a^{E}\exp\left(\zeta[C_{1}^{1/2}+\frac{1}{2}\zeta C_{2}]|t-t_{o}|\right)\nonumber\\&
=n^{1/2}a^{E}\exp(Q|t-t_{o}|)
\end{align}
which agrees with the estimate (6.38).
\end{thm}
\begin{proof}
Using(6.60) and (6.61)it follows that
\begin{equation}
\sum_{\xi=1}^{\infty}\mathcal{E}_{\xi}=\bm{\mathrm{I\!E}}
\bigg\lbrace \int_{t_{o}}^{t}\bigg|\widehat{\mathscr{U}}(s)\bigg|^{2}\bigg\rbrace\le C_{1}|t-t_{o}|
\end{equation}
Now $\mathcal{E}_{max}=\sup(\mathcal{E}_{\xi})_{\xi=1}^{\infty}$ and we can choose
$\mathcal{E}=\mathcal{E}_{1}$. Using the Parsaval identity
\begin{align}
&\bm{\mathrm{I\!E}}\bigg\lbrace
\exp\bigg(\beta\int_{t_{o}}^{t}|\widehat{\mathscr{U}}(t)|^{2}dt\bigg\rbrace
=\bm{\mathrm{I\!E}}\bigg\lbrace\exp\bigg(\sum_{\xi=1}^{\infty}\beta\mathcal{E}_{\xi}|\widehat{\mathscr{U}}(\xi)|^{2}\bigg)
\bigg\rbrace\nonumber\\&
\prod_{\xi=1}^{\infty}\bm{\mathrm{I\!E}}
\bigg\lbrace\exp\bigg(\beta\mathcal{E}_{\xi}|\widehat{\mathscr{U}}(\xi)|^{2}\bigg)=\prod_{\xi=1}^{\infty}(1-2\beta\mathcal{E}_{\xi})^{-1/2}
\end{align}
Next
\begin{align}
&\bm{\mathcal{E}}_{1}\le\int_{t_{o}}^{t}\int_{t_{o}}^{t}\|J(s,\tau;\varsigma)\|\varphi_{1}(t)\varphi_{1}(s)dsd\tau\nonumber\\&
\le\frac{1}{2}\int_{t_{o}}^{t}\int_{t_{o}}^{t}\|J(s,\tau;\varsigma)\|(|\varphi_{1}(t)|^{2}+\varphi_{2}(t)|^{2})dsdt\le C_{2}
\end{align}
so that $\bm{\mathcal{E}}_{1}\le C_{2}$. Using the basic inequality $(1+x)<\exp(x)$ for $x>0$
\begin{align}
&\bm{\mathrm{I\!E}}\bigg\lbrace
\exp\bigg(\beta\int_{t_{o}}^{t}|\widehat{\mathscr{U}}(t)|^{2}dt\bigg\rbrace=
\prod_{\xi=1}^{\infty}(1-2\beta\mathcal{E}_{\xi})^{-1/2}\nonumber\\&
=\prod_{\xi=1}^{\infty}(1+2\beta\mathcal{E}_{\xi}+4\beta^{2}\mathcal{E}_{\xi}^{2}(1-2\beta\mathcal{E}_{\xi})^{-1})^{1/2}\nonumber\\&
=\exp\bigg(\beta(1-2\beta\mathcal{E}_{1})^{-1})\sum_{\xi=1}^{\infty}\mathcal{E}_{\xi}\bigg)\nonumber\\&
\le\exp(\beta(1-2\beta\mathcal{E}_{1})^{-1})C_{1}|t-t_{o}|)
\end{align}
Finally, applying the basic inequality $x<\tfrac{1}{2}\beta x^{2}+(2\beta)^{-1}$ gives
\begin{align}
&\bm{\mathrm{I\!E}}\bigg\lbrace\bigg\|\widehat{\bm{a}}(t)-\bm{a}^{E}\bigg\|
\bigg\rbrace=n^{1/2}a^{E}\bm{\mathrm{I\!E}}\bigg\lbrace\exp\bigg(\zeta\int_{t_{o}}^{t}
\widehat{\mathscr{U}}(s)ds\bigg)\bigg\rbrace\nonumber\\&
\le n^{1/2}a^{E}\exp\bigg(\frac{\zeta}{2\beta}|t-t_{o}|\bigg)\bm{\mathrm{I\!E}}\bigg\lbrace
\exp\bigg(\frac{1}{2}{E}_{1}\beta\int_{t_{o}}^{t}|\widehat{\mathscr{U}}(s)|^{2}ds
\bigg\rbrace\nonumber\\&=n^{1/2}a^{E}\exp\left(\left[\frac{{E}_{1}}{2\beta}+\frac{1}{2}{E}_{1}\beta C_{1}(1-\beta{E}_{1}C_{2})^{-1}\right]|t-t_{o}|\right)
\end{align}
Setting $\beta=\mathcal{E}_{1}C_{2}+C_{1}^{1/2})^{-1}$ then gives (-) so that
\begin{equation}
\bm{\mathrm{I\!E}}\bigg\lbrace\bigg\|\widehat{\bm{a}}(t)-\bm{a}^{E}\bigg\|
\bigg\rbrace=n^{1/2}a^{E}\bm{\mathrm{I\!E}}\bigg\lbrace\exp\bigg(\zeta\int_{t_{o}}^{t}
\widehat{\mathscr{U}}(s)ds\bigg)\bigg\rbrace \sim \exp(Q|t-t_{0}|)
\end{equation}
which agrees with previous estimates.
\end{proof}
\subsection{Random perturbations of two Bianchi-I type cosmologioes}
For dynamic solutions describing a universe that has some dimensions expanding and some collapsing, the average effect of the random perturbations is to 'boost' both the expansion or collapse of those toroidal radii which are already expanding or collapsing. The following lemmas illustrate the averaged asymptotic expansion behavior of some Bianchi-I cosmologies subject to the random perturbations or noise
\begin{lem}
Given the Bianchi-I triplet $\mathfrak{B}=(p_{1},p_{2},p_{3})=(-1/3,2/3,2/3)$ for n=3 or $\mathbb{T}^{3}$ the radii collapse and expand as for the Bianchi-I triplet $\mathfrak{B}=(-1/3,2/3,2/3)$, the rolling radii are
\begin{align}
&a_{1}(t)=a_{1}^{E}|t|^{-1/3}\equiv a_{1}(0)|t|^{-1/3}\\&
a_{2}(t)=a_{1}^{E}|t|^{2/3}\equiv a_{2}(0)|t|^{2/3}\\&
a_{3}(t)=a_{1}^{E}|t|^{2/3}\equiv a_{3}(0)|t|^{2/3}
\end{align}
where $a_{1}(0)=a_{2}(0)=a_{3}(0)$. The randomly perturbed radii are then
\begin{align}
&\widehat{a}_{1}(t)=a_{1}^{E}|t|^{-1/3}\exp\left(\zeta\int_{0}^{t}
\mathlarger{\mathscr{U}}(\tau)d\tau\right)\\&
\widehat{a}_{2}(t)=a_{2}^{E}|t|^{-1/3}\exp\left(\zeta\int_{0}^{t}
{\mathscr{U}}(\tau)d\tau\right)\\&
\widehat{a}_{3}(t)=a_{3}^{E}|t|^{-1/3}\exp\left(\zeta\int_{0}^{t}
{\mathscr{U}}(\tau)d\tau\right)
\end{align}
with stochastic expectations
\begin{align}
&\mathlarger{\mathbf{A}}_{1}(t)=\bm{\mathrm{I\!E}}\bigg\lbrace\widehat{a}_{1}(t)\bigg\rbrace
=a_{1}^{E}|t|^{-1/3}\bm{\mathrm{I\!E}}\left\lbrace\exp\left(\zeta\int_{0}^{t}\widehat{\mathscr{U}}(\tau)d\tau)\right)\right\rbrace \\&\sim a_{1}^{E}|t|^{-1/3}\exp(Q t)\\&
\mathlarger{\mathbf{A}}_{2}(t)=\bm{\mathrm{I\!E}}\bigg\lbrace\widehat{a}_{1}(t)\bigg\rbrace=a_{1}^{E}|t|^{-1/3}
\bm{\mathrm{I\!E}}\left\lbrace\exp\left(\zeta\int_{0}^{t}\widehat{\mathscr{U}}(\tau)d\tau)\right)\right\rbrace\\&\sim a_{1}^{E}|t|^{2/3}\exp(Q t)\\&
\mathlarger{\mathbf{A}}_{3}(t)=\bm{\mathrm{I\!E}}\bigg\lbrace\widehat{a}_{1}(t)\bigg\rbrace=a_{1}^{E}|t|^{-1/3}
\bm{\mathrm{I\!E}}\left\lbrace\exp\left(\zeta\int_{0}^{t}\widehat{\mathscr{U}}(\tau)d\tau)\right)\right\rbrace\\&\sim a_{1}^{E}|t|^{2/3}\exp(Q t)
\end{align}
Hence, the collapse and expansions are boosted by an exponential factor $\exp(Qt)$. Given the Bianchi-I triplet $\mathcal{B}=(p_{1},p_{2},p_{3})=(0,0,1)$ for n=3 or $\mathbb{T}^{3}$ one dimension expands and two remain static
\begin{align}
&a_{1}(t)=a_{1}^{E}|t|^{-1/3}\equiv a_{1}(0)\equiv a_{1}^{E}\\&
a_{2}(t)=a_{1}^{E}|t|^{2/3}\equiv a_{2}(0)\equiv a_{2}^{E}\\&
a_{3}(t)=a_{1}^{E}|t|^{2/3}\equiv a_{3}(0)|t|\equiv a_{3}^{E}|t|
\end{align}
where $a_{1}(0)=a_{2}(0)=a_{3}(0)$. The randomly perturbed radii are then
\begin{align}
&\widehat{a}_{1}(t)=a_{1}^{E}\exp\left(\zeta\int_{0}^{t}
\widehat{\mathscr{U}}(\tau)d\tau\right)\\&
\widehat{a}_{2}(t)=a_{2}^{E}\exp\left(\zeta\int_{0}^{t}
\widehat{\mathscr{U}}(\tau)d\tau\right)\\&
\widehat{a}_{3}(t)=a_{3}^{E}|t|\exp\left(\zeta\int_{0}^{t}
\widehat{\mathscr{U}}(\tau)d\tau\right)
\end{align}
with stochastic expectations
\begin{align}
&\mathlarger{\mathbf{A}}_{1}(t)=\bm{\mathrm{I\!E}}\bigg\lbrace\widehat{a}_{1}(t)\bigg\rbrace
=a_{1}\bm{\mathrm{I\!E}}\left\lbrace\exp\left(\zeta\int_{0}^{t}\widehat{\mathscr{U}}(\tau)d\tau)\right)\right\rbrace\sim a_{1}^{E}\exp(Q t)\\&
\mathlarger{\mathbf{A}}_{2}(t)=\bm{\mathrm{I\!E}}\bigg\lbrace\widehat{a}_{1}(t)\bigg\rbrace=a_{1}^{E}
\bm{\mathrm{I\!E}}\left\lbrace\exp\left(\zeta\int_{0}^{t}\widehat{\mathscr{U}}(\tau)d\tau)\right)\right\rbrace\sim a_{2}^{E}\exp(Q t)\\&
\mathlarger{\mathbf{A}}_{3}(t)=\bm{\mathrm{I\!E}}\bigg\lbrace\widehat{a}_{1}(t)\bigg\rbrace
=a_{3}^{E}|t|\bm{\mathrm{I\!E}}\left\lbrace\exp\left(\zeta\int_{0}^{t}
\widehat{\mathscr{U}}(\tau)d\tau)\right)\right\rbrace\sim a_{1}^{E}|t|^{2/3}\exp(Q t)
\end{align}
The static dimensions are then exponentially boosted and 'inflate', while the expanding dimension now expands faster.
\end{lem}
\begin{lem}
The randomly perturbed radii for the Bianchi-I triplet $(p_{1},p_{2},p_{3})=(-1/3,2/3,2/3)$
\begin{align}
&\widehat{a}_{1}(t)=a_{1}^{E}|t|^{-1/3}\exp\left(\zeta\int_{0}^{t}
\widehat{\mathscr{U}}_{i}(\tau)d\tau\right)\\&
\widehat{a}_{2}(t)=a_{2}^{E}|t|^{2/3}\exp\left(\zeta\int_{0}^{t}
\widehat{\mathscr{U}}_{i}(\tau)d\tau\right)\\&
\widehat{a}_{3}(t)=a_{3}^{E}|t|^{2/3}\exp\left(\zeta\int_{0}^{t}
\widehat{\mathscr{U}}_{i}(\tau)d\tau\right)
\end{align}
and the Bianchi-I triplet $(p_{1},p_{2},p_{3})=(0, 0, 1)$.
\begin{align}
&\widehat{a}_{1}(t)=a_{1}^{E}\exp\left(\zeta\int_{0}^{t}
\widehat{\mathscr{U}}_{i}(\tau)d\tau\right)
\\&
\widehat{a}_{2}(t)=a_{2}^{E}\exp\left(\zeta\int_{0}^{t}
\widehat{\mathscr{U}}_{i}(\tau)d\tau\right)
\\&
\widehat{a}_{3}(t)=a_{3}^{E}|t|\exp\left(\zeta\int_{0}^{t}
\widehat{\mathscr{U}}_{i}(\tau)d\tau\right)
\end{align}
are solutions of the stochastically averaged Einstein vacuum equations such that
\begin{equation}
\bm{\mathrm{I\!E}}\bigg\lbrace \mathbf{D}_{n}\widehat{a}_{i}(t)\bigg\rbrace=nJ(0;\varsigma)\equiv \lambda
\end{equation}
\end{lem}
\begin{proof}
Taking the stochastic expectation, only the nonlinear terms are nonvanishing so that
\begin{align}
\bm{\mathrm{I\!E}}\bigg\lbrace\mathbf{D}_{n}a_{i}(t)\bigg\rbrace&=
\sum_{i=1}^{n}\bm{\mathrm{I\!E}}\bigg\lbrace\widehat{\mathscr{U}}_{i}(t) \widehat{\mathscr{U}}_{i}(t)\bigg\rbrace\nonumber\\&-\frac{1}{2}\sum_{i=1}^{3}
\bm{\mathrm{I\!E}}\bigg\lbrace\widehat{\mathscr{U}}_{i}(t)\widehat{\mathscr{U}}_{i}(t)\bigg\rbrace
-\frac{1}{2}\sum_{i=1}^{3}\left|p_{i}\right|^{2}t^{-2}+\frac{1}{2}\sum_{i=1}^{3}\sum_{j=1}^{3}
\bm{\mathrm{I\!E}}\bigg\lbrace\widehat{\mathscr{U}}_{i}(t)\widehat{\mathscr{U}}_{j}(t)\bigg\rbrace\nonumber\\&=\frac{1}{2}\sum_{i=1}^{3}\delta_{ii}
\bm{\mathrm{I\!E}}\bigg\lbrace\widehat{\mathscr{U}}(t)\widehat{\mathscr{U}}(t)\bigg\rbrace-\frac{1}{2}
\sum_{i=1}^{3}\left|p_{i}\right|^{2}t^{-2}\nonumber\\&+\frac{1}{2}\sum_{i=1}^{3}\sum_{j=1}^{3}\delta_{ij}\bm{\mathrm{I\!E}}\bigg\lbrace \widehat{\mathscr{U}}(t)\widehat{\mathscr{U}}(t)\bigg\rbrace+\frac{1}{2}\sum_{i=1}^{3}\sum_{j=1}p_{i}p_{j}t^{-2}\nonumber\\&
=\frac{1}{2}\sum_{i=1}^{3}\delta_{ii}J(0;\varsigma)-\frac{1}{2}\sum_{i=1}^{3}\left|p_{i}\right|^{2}t^{-2}+\frac{1}{2}
\sum_{i=1}^{3}\sum_{j=1}^{3}\delta_{ij}J(0;\varsigma)+\frac{1}{2}\sum_{i=2}^{3}\sum_{j=1}^{3}p_{i}p_{j}t^{-2}
\end{align}
If $m_{ij}=p_{i}p_{j}$ set $m_{ij}=p_{i}p_{i}=p_{i}^{2}$ if $i=j$ and $m_{ij}=0$ if $i\ne j$. Then $M_{ij}$ is a diagonal matrix.
\begin{align}
&\bm{\mathrm{I\!E}}\bigg\lbrace D_{n}a_{i}(t)\bigg\rbrace=
\frac{1}{2}\sum_{i=1}^{3}\delta_{ii}J(0;\varsigma)-
\frac{1}{2}\sum_{i=1}^{3}|p_{i}|^{2}t^{-2}\nonumber\\&
+\frac{1}{2}\sum_{i=1}^{3}\sum_{j=1}^{3}\delta_{ij}J(0;\varsigma)
+\frac{1}{2}\sum_{i=2}^{3}|p_{i}|^{2}t^{-2}\nonumber\\&
\equiv\frac{1}{2}\sum_{i=1}^{3}\delta_{ii}J(0;\varsigma)-
\frac{1}{2}t^{-2}+\frac{1}{2}\sum_{i=1}^{3}\sum_{j=1}^{3}\delta_{ij}J(0;\varsigma)
+\frac{1}{2}t^{-2}=n J(0;\varsigma)\equiv\lambda
\end{align}
and since
\begin{equation}
\sum_{i=1}^{n}|p_{i}|^{2}=|\tfrac{-1}{3}|^{2}+|\tfrac{2}{3}|^{2}+|\tfrac{2}{3}|^{2}=1
\end{equation}
for $(p_{1},p_{2},p_{3})=(\tfrac{-1}{3},\tfrac{2}{3},\tfrac{2}{3})$
and
\begin{equation}
\sum_{i=1}^{n}|p_{i}|^{2}=|1|^{2}=1
\end{equation}
for $(p_{1},p_{2},p_{3})=(0,0,1)$.
\end{proof}
\section{Stability of the Einstein system to a specific class of random perturbations}
In this final section, it is shown that a class of random perturbations of the radial moduli can lead to stability of the Einstein system. We can consider random perturbations of the radial moduli fields of the form $\widehat{\psi}_{i}(t)=\psi_{i}^{E}+\widehat{\mathscr{U}}(t)$, rather than taking the integral over the field as in (5.2).
\begin{prop}
Let the conditions of Theorem (5.1) hold and also the following:
\begin{enumerate}
\item The moduli fields are randomly perturbed as $\widehat{\psi}(t)=\psi_{i}^{E}+\zeta\mathscr{U}(t)$ so that the perturbed toroidal radii are
\begin{equation}
\widehat{a}_{i}(t)=a_{i}^{E}+\zeta\widehat{\mathscr{U}}_{i}(t)
\end{equation}
\item The first and second derivatives of the random field exist so that $\widehat{\mathscr{Y}}(t)=\partial_{t}\widehat{\mathscr{U}}(t)$ and $\partial_{t}\widehat{\mathscr{Y}}(t)=\partial_{tt}\widehat{\mathscr{U}}(t)$ with $\bm{\mathrm{I\!E}}\lbrace\widehat{\mathscr{U}}(t)\rbrace=0$
\item The 2-point function of the fields is of the form
\begin{equation}
\bm{\mathrm{I\!E}}\left\lbrace\widehat{\mathscr{Y}}_{i}(t)\widehat{\mathscr{Y}}_{j}(s)\right\rbrace=
\delta_{ij}\Xi(\Delta;\varsigma)
\end{equation}
\begin{equation}
\bm{\mathrm{I\!E}}\left\lbrace\widehat{\mathscr{U}}_{i}(t)\widehat{\mathscr{U}}_{j}(s)\right\rbrace=
\delta_{ij}J(\Delta;\varsigma)
\end{equation}
where $\Delta=|t-s|$ and is regulated so that $\bm{\mathrm{I\!E}}\lbrace\widehat{\mathscr{Y}}(t)\widehat{\mathscr{Y}}(s)\rbrace=
\delta_{ij}\Xi(0;\varsigma)<\infty $.
\end{enumerate}
Then the stochastically averaged Einstein vacuum equations are
\begin{equation}
\bm{\mathrm{I\!E}}\bigg\lbrace \mathbf{H}_{n}\widehat{\psi}(t)\bigg\rbrace = n\zeta^{2}\Xi(0;\varsigma)\equiv \lambda
\end{equation}
\begin{equation}
\bm{\mathrm{I\!E}}\bigg\lbrace \mathbf{D}_{n}\widehat{a}_{i}(t)\bigg\rbrace = n\zeta^{2}\Xi(0;\varsigma)\equiv \lambda
\end{equation}
and the perturbed norm has the asymptotic behavior
\begin{equation}
\lim_{t\uparrow\infty}\bm{\mathrm{I\!E}}\bigg\lbrace\bigg\|\bm{a}(t)-\bm{a}^{E}\bigg\|\bigg\rbrace \le\lim_{t\uparrow\infty}a^{E}n^{1/2}
\bm{\mathrm{I\!E}}\bigg\lbrace\exp(\zeta\widehat{\mathscr{U}}(t))\bigg\rbrace
\end{equation}
\end{prop}
\begin{proof}
The derivatives are $\partial_{t}\widehat{\psi}(t)=\zeta\partial_{t}\widehat{\mathscr{U}}(t)
=\zeta\widehat{\mathscr{Y}}(t)$ and $\partial_{tt}\widehat{\psi}(t)=\zeta\partial_{t}
\widehat{\mathscr{U}}(t)=\zeta\partial_{t}\widehat{\mathscr{Y}}(t)$. The randomly perturbed Einstein equations are
\begin{align}
&\mathbf{H}_{n}\widehat{\psi}(t)=\sum_{i=1}^{n}\partial_{tt}\psi_{i}(t)+\frac {1}{2}\sum_{i=1}^{n}\partial_{t}\widehat{\psi}_{i}(t)\partial_{t}\widehat{\psi}_{i}(t)+\frac {1}{2}\sum_{i=1}^{n}\partial_{t}\widehat{\psi}_{i}(t)\partial_{t}\widehat{\psi}_{i}(t)
\nonumber\\&=\sum_{i=1}^{n}\partial_{t}\widehat{\mathscr{Y}}_{i}(t)+\frac{1}{2}\sum_{i=1}^{n}
\widehat{\mathscr{Y}}_{i}(t)\widehat{\mathscr{Y}}_{i}(t)+\frac{1}{2}\sum_{i=1}^{n}
\sum_{j=1}^{n}\widehat{\mathscr{Y}}_{i}(t)\widehat{\mathscr{Y}}_{j}(t)
\end{align}
Taking the expectation
\begin{align}
&\bm{\mathrm{I\!E}}\bigg\lbrace\mathbf{H}_{n}\widehat{\psi}(t)\bigg\rbrace=\frac{1}{2}\zeta^{2}
\sum_{i=1}^{n} \bm{\mathrm{I\!E}}\left\lbrace\widehat{\mathscr{Y}}_{i}(t)\widehat{\mathscr{Y}}_{i}(t)
\right\rbrace+\frac{1}{2}\zeta^{2}\sum_{i=1}^{n}\sum_{j=1}^{n}\bm{\mathrm{I\!E}}\left\lbrace \widehat{\mathscr{Y}}_{i}(t)\widehat{\mathscr{Y}}_{j}(t)\right\rbrace\nonumber\\&
=\frac{1}{2}\zeta^{2}\sum_{i=1}^{n}\delta_{ii}\Xi(0;\varsigma)+\frac{1}{2}\zeta^{2}\sum_{i=1}^{n}\sum_{j=1}^{n}\delta_{ij}\Xi(0;\varsigma)\nonumber\\&
=\frac{1}{2}n\zeta^{2}\Xi(0;\varsigma)+\frac{1}{2}n\zeta^{2}\Xi(0;\varsigma)
=n\zeta^{2}\Xi(0;\varsigma)\equiv \lambda
\end{align}
since $\lbrace\widehat{\mathscr{Y}}(t)\rbrace=0$ and assuming $\widehat{\mathscr{Y}}_{i}(t)=\widehat{\mathscr{Y}}_{j}(t)=\widehat{\mathscr{K}}(t)$.
\end{proof}
The norm estimate is
\begin{align}
&\lim_{t\uparrow\infty}\bm{\mathrm{I\!E}}\bigg
\lbrace\bigg\|\bm{\widehat{a}}(t)-\bm{a}^{E}\bigg\|\bigg\rbrace \le
\lim_{t\uparrow\infty}\bm{\mathrm{I\!E}}\bigg
\lbrace\bigg\|\bm{\widehat{a}}(t)\bigg\|\bigg\rbrace\nonumber\\&<\lim_{t\uparrow\infty}
\bm{\mathrm{I\!E}}\left\lbrace\left(\sum_{i=1}^{n}\left|a_{i}^{E}\exp(\zeta\widehat{\mathscr{U}}_{i}(t))\right|^{2}\right)^{1/2}\right\rbrace\nonumber\\&
=\lim_{t\uparrow\infty}\bm{\mathrm{I\!E}}\left\lbrace\left(n|a_{i}^{E}|^{2}\exp
(2\zeta\widehat{\mathscr{U}}_{i}(t))\right)^{1/2}\right\rbrace\nonumber\\&
=\lim_{t\uparrow\infty}a^{E}n^{1/2}\bm{\mathrm{I\!E}}
\bigg\lbrace\exp(\zeta\widehat{\mathscr{U}}(t))\bigg\rbrace
\end{align}
so that for stability
\begin{equation}
\lim_{t\uparrow\infty}\bm{\mathrm{I\!E}}\bigg
\lbrace\bigg\|\bm{\widehat{a}}(t)-\bm{a}^{E}\bigg\|\bigg\rbrace<\infty
\end{equation}
and for instability
\begin{equation}
\lim_{t\uparrow\infty}\bm{\mathrm{I\!E}}\bigg
\lbrace\bigg\|\bm{\widehat{a}}(t)-\bm{a}^{E}\bigg\|\bigg\rbrace=\infty
\end{equation}
This involves an estimation of
$\bm{\mathrm{I\!E}}\big\lbrace\exp(\zeta\widehat{\mathscr{U}}(t))\big\rbrace$ using a cluster expansion as before. In particular, if the set is sub-Gaussian then it is bounded and the Hoeffding inequality and the Chernoff bound inequality apply (ref). This suggests that if a randomly perturbed set $\widehat{u}_{i}(t)$ is sub-Gaussian then it is bounded and therefore the system is stable to the random perturbations.
\begin{lem}
If the random perturbations are of the form $\widehat{a}_{i}(t)=a_{i}^{E}\exp(\zeta\mathscr{U}_{i}(t))$  for a Gaussian random field with $\bm{\mathrm{I\!E}}\lbrace \mathscr{U}_{i}(t)\rbrace=0$ and $\bm{\mathrm{I\!E}}\lbrace \mathscr{U}_{i}(t)\mathscr{U}^{i}(s)\rbrace=J(\Delta;\sigma)$ and regulated $\bm{\mathrm{I\!E}}\lbrace
\mathscr{U}_{i}(t)\mathscr{U}^{i}(t)\rbrace=J(0;\sigma)<\infty$ then for all $t>0$
\begin{equation}
\bm{\mathrm{I\!E}}\bigg\lbrace \widehat{a}_{i}(t)\bigg\rbrace \sim a_{i}^{E}\exp(\tfrac{1}{2}\zeta^{2}J(0;\varsigma))<\infty
\end{equation}
\end{lem}
\begin{proof}
Interpreting $\mathbf{\Phi}(t)=\bm{\mathrm{I\!E}}\lbrace\exp(\zeta\mathscr{U}_{i}(t))$ as a moment-generating function (MGF) then the corresponding cumulant generating function (CGF)is
\begin{equation}
\mathbf{\Phi}(t)=\log\mathbf{\Psi}(t)=\log\bm{\mathrm{I\!E}}\bigg\lbrace\exp(\zeta\mathscr{U}_{i}(t))\bigg\rbrace
\end{equation}
The CGF has the McLauren power-series representation which can be truncated at second order for Gaussian random fields so that
\begin{align}
\mathbf{\Phi}(t)&=\sum_{\alpha=1}^{\infty}\frac{\zeta^{\alpha}}{\alpha!}\bm{\mathrm{I\!K}}\bigg\lbrace |\mathscr{U}_{i}(t)|^{2}\bigg\rbrace=\zeta\bm{\mathrm{I\!K}}\lbrace\mathscr{U}_{i}(t)\bigg\rbrace
+\frac{1}{2}\zeta^{2}\bm{\mathrm{I\!K}}\bigg\lbrace\mathscr{U}_{i}(t)\mathscr{U}^{i}(t)\bigg\rbrace+...\nonumber\\&
\equiv\zeta\bm{\mathrm{I\!E}}\bigg\lbrace\mathscr{U}_{i}(t)\bigg\rbrace
+\frac{1}{2}\zeta^{2}\bm{\mathrm{I\!E}}\bigg\lbrace\mathscr{U}_{i}(t)\mathscr{U}^{i}(t)\bigg\rbrace+....=
\frac{1}{2}\zeta^{2}J(0;\varsigma)
\end{align}
Hence
\begin{equation}
\mathbf{\Psi}(t)=\bm{\mathrm{I\!E}}\bigg\lbrace\exp(\zeta\mathscr{U}_{i}(t))\bigg\rbrace
=\exp\big(\tfrac{1}{2}\zeta^{2}J(0;\varsigma)\big)
\end{equation}
so that
\begin{equation}
\bm{\mathrm{I\!E}}\lbrace \widehat{a}_{i}(t)\rbrace \sim a_{i}^{E}\exp(\tfrac{1}{2}\zeta^{2}J(0;\varsigma))<\infty
\end{equation}
\end{proof}
\begin{cor}
The norm estimate is
\begin{align}
&\lim_{t\uparrow\infty}\bm{\mathrm{I\!E}}\bigg
\lbrace\bigg\|\bm{\widehat{a}}(t)-\bm{a}^{E}\bigg\|\bigg\rbrace \le \lim_{t\uparrow\infty}\bm{\mathrm{I\!E}}\bigg
\lbrace\bigg\|\bm{\widehat{a}}(t)\bigg\|\bigg\rbrace\nonumber\\&=\lim_{t\uparrow\infty}a^{E}n^{1/2}
\bm{\mathrm{I\!E}}\bigg\lbrace\exp(\zeta\widehat{\mathscr{U}}(t))\bigg\rbrace
=n^{1/2}a_{i}^{E}\exp(\tfrac{1}{2}\zeta{^2}J(0;\varsigma))<\infty
\end{align}
\end{cor}
The Hoeffding and maximal estimates from Section 1 can now be applied.
\begin{prop}
Let $a_{i}^{E}$ be a set of static equilibrium solutions of the Einstein system of nonlinear ODEs $\mathbf{D}_{n}a_{i}^{E}=0$. Let the randomly perturbed set of solutions be $\widehat{a}_{i}(t)=a_{i}^{E}\exp\zeta\mathscr{U}_{i}(t))$ with
\begin{equation}
\bm{\mathrm{I\!E}}\bigg\lbrace\widehat{a}_{i}(t)\bigg\rbrace=a_{i}^{E}
\exp(\frac{1}{2}\zeta^{2}J(0;\varsigma))<\infty
\end{equation}
Let $\widehat{a}_{i}^{E*}$ be 'attractors' or new stable equilibrium fixed points such that the perturbed system converges as $\widehat{a}_{i}(t)\rightarrow a_{i}^{E*}$ for some finite $t\gg 0$ or $t\rightarrow\infty$. Then for all finite $t>0$ the set is bounded in that $\exists B>0$ such that
\begin{equation}
a_{i}^{E}\le \widehat{a}_{i}(t)\le a_{i}^{E*}<\bm{\mathrm{I\!E}}\bigg\lbrace\widehat{a}_{i}(t)\bigg\rbrace\le B\nonumber
\end{equation}
\begin{enumerate}
\item The Hoeffding inequality applies and for any $L>0$ is
\begin{align}
\bm{\mathrm{I\!P}}\big(\widehat{S}(t)-&\bm{\mathrm{I\!E}}\lbrace\widehat{ S}(t)\rbrace\ge L\big)\equiv \bm{\mathrm{I\!P}}\bigg(\frac{1}{n}
\sum_{i=1}^{n}\widehat{a}_{i}(t)-\frac{1}{n}\sum_{i=1}^{n}\bm{\mathrm{I\!E}}\bigg\lbrace\widehat{a}_{i}(t)\bigg\rbrace\ge L\bigg)\nonumber\\&\equiv\bm{\mathrm{I\!P}}\bigg(\frac{1}{n}\bigg\|\widehat{a}_{i}(t)\bigg\|
-\frac{1}{n}\bigg\|\bm{\mathrm{I\!E}}\bigg\lbrace\widehat{a}_{i}(t)z\bigg\rbrace\bigg\|\ge L\bigg)\nonumber\\&\le\exp\bigg(-\frac{2n^{2}|L|^{2}}{\sum_{i=1}^{n}\big\|B-a_{i}^{E}\big\|^{2}}
\bigg)\nonumber\\&\equiv\exp\bigg(-\frac{2n^{2}|L|^{2}}{\big\|B-a_{i}^{E}\big\|^{2}}\bigg)
\end{align}
\end{enumerate}
The randomly perturbed Einstein system is then stable in probability
\begin{align}
\bm{\mathrm{I\!P}}\big(\widehat{S}(t)&-\bm{\mathrm{I\!E}}\lbrace\widehat{ S}(t)\rbrace=\infty\big)\equiv \bm{\mathrm{I\!P}}\bigg(\frac{1}{n}
\sum_{i=1}^{n}\widehat{a}_{i}(t)-\frac{1}{n}\sum_{i=1}^{n}\bm{\mathrm{I\!E}}\bigg\lbrace\widehat{a}_{i}(t)\bigg\rbrace= \infty\bigg)\nonumber\\&\equiv\bm{\mathrm{I\!P}}\bigg(\frac{1}{n}\bigg\|\widehat{a}_{i}(t)\bigg\|
-\frac{1}{n}\bigg\|\bm{\mathrm{I\!E}}\bigg\lbrace\widehat{u}_{i}(t)\bigg\rbrace\bigg\|=\infty\bigg)=1
\end{align}
The Chernoff bound can also be expressed as
\begin{equation}
\bm{\mathrm{I\!P}}\big(\big\|\widehat{a}_{i}(t)-a_{i}^{E}\big\|\le |L|
\big)\le\exp(\beta|L|)\bm{\mathrm{I\!E}}\bigg\lbrace\exp\bigg(-\beta\bigg(\bigg\|\widehat{a}_{i}(t)-a_{i}^{E}
\bigg\|\bigg)\bigg\rbrace
\end{equation}
then
\begin{equation}
\bm{\mathrm{I\!P}}\big(\big\|\widehat{a}_{i}(t)-a_{i}^{E}\big\|\le |L|\big)\ne 0
\end{equation}
if $\bm{\mathrm{I\!E}}\big\lbrace\exp(-\beta(\big\|\widehat{a}_{i}(t)-a_{i}^{E}\big\|)
\big\rbrace<\infty $ and the randomly perturbed radii are bounded.
\end{prop}
These exponential inequalities are valid for linear combinations of bounded independent random variables, and in particular for the average. But one is often more interested in controlling the maximum or supremum of the set in terms of the maximal estimates.
\begin{lem}
Let $u_{i}^{E}$ be a set of n equilibrium solutions of the Einsetin system  $D_{n}a_{i}^{E}=0$ and let $\widehat{a}_{i}(t)$ be the set of n randomly perturbed solutions. Let $\sup_{1\le i\le n}\widehat{a}_{i}(t)$ be the supremum or maximum of the set. If the set if bounded it is sub-Gaussian and vice-versa so that
\begin{equation}
\bm{\mathrm{I\!P}}\big(\sup_{1\le i\le n} \widehat{a}_{i}(t)\ge |L|\big)\le\exp\bigg(-\frac{L^{2}}{2C^{2}}\bigg)
\end{equation}
and
\begin{equation}
\widehat{a}_{i}^{E}\le \sup_{1\le i\le n}\widehat{a}_{i}(t)\le\bm{\mathrm{I\!E}}\big\lbrace\widehat{a}_{i}(t)\big\rbrace\le B
\end{equation}
Then $\exists(L,C,B,D)>0$ such that the maximal inequalities hold and the system is stable so that for all $t\in\mathbb{R}^{+}\cup\infty$
\begin{align}
&\bm{\mathrm{I\!E}}\bigg\lbrace\sup_{1\le i\le n}\widehat{a}_{i}(t)
\bigg\rbrace\le C\sqrt{2\log(n)}\le B<\infty
\\&\lim_{t\uparrow\infty}\bm{\mathrm{I\!E}}\bigg\lbrace\sup_{1\le i\le n}\widehat{a}_{i}(t)\bigg\rbrace\le C\sqrt{2\log(n)}\le B < \infty
\end{align}
\begin{align}
&\bm{\mathrm{I\!P}}\big(\sup_{1\le i\le n}\widehat{a}_{i}(t)\ge|L|\big)\le n\exp\bigg(-\frac{L^{2}}{2C^{2}}\bigg)\le D < \infty
\\&\lim_{t\uparrow\infty}\bm{\mathrm{I\!P}}\big(\sup_{1\le i\le n}\widehat{a}_{i}(t)\ge|L|\big)\le n\exp\bigg(-\frac{L^{2}}{2C^{2}}\bigg)\le D < \infty
\end{align}
The probability of blowup or asymptotic instability is zero so that
\begin{align}
&\bm{\mathrm{I\!P}}\big(\sup_{1\le i\le n}\widehat{a}_{i}(t)=\infty\big)= 0
\\& \lim_{t\uparrow\infty}\bm{\mathrm{I\!P}}\big(\sup_{1\le i\le n}\widehat{a}_{i}(t)=\infty\big)= 0
\end{align}
\end{lem}
\section{Conclusion}
In this paper, it has been tentatively explored how one might incorporate classical randomness and stochasticity into general relativity within the context of specific solvable cosmological models, in order to incorporate the effects of fluctuations or 'noise'. In particular, cosmological constant terms arise when one stochastically averages the nonlinear Einstein equations formulated on a randomly perturbed n-torus, in analogy with induced Reynolds stresses and numbers within hydrodynamical turbulence theory when the Navier-Stokes PDEs are averaged.
\clearpage
\appendix
\renewcommand{\theequation}{\Alph{section}.\arabic{equation}}
\section{Appendix A: Gaussian random fields}
In this appendix, the definitions, existence, properties, correlations, statistics, derivatives and integrals are defined for random scalar vector fields (RVFS) on $\mathbb{R}^{n}$. Details can be found in a number of texts [19,20,21,22,23,24,68].
\subsection{Existence and statistical correlations of random fields}
\begin{defn}
let $(\Omega,\mathfrak{F},\mu$ be a probability space. Within the probability triplet, $(\Omega,\mathfrak{F})$ is a \emph{measurable space}, where $\mathfrak{F}$ is the $\sigma$-algebra (or Borel field) that should be interpreted as being comprised of all reasonable subsets of the state space $\Omega$. Then $\mu$ is a function such that $\mu:\mathscr{U}\rightarrow [0,1]$, so that for all $A\in\mathfrak{F}$, there is an associated probability $\mu(A)$. The measure is a probability measure when $\mu(\Omega)=1$. The probability space obeys the Kolmogorov axioms:
\begin{itemize}
\item $\mu(\Omega)=1.$
\item $0\le\mu(A_{i})\le 1$ for all sets $A_{i}\in\mathfrak{F}$
\item If $A_{i}\bigcap A_{j}=\empty$, then    $\mu(\bigcup_{i=1}^{\infty}A_{i})=\sum_{i=1}^{\infty}\bm{\mu}(A_{i})$.
\end{itemize}
\end{defn}
These are standard (and abstract) definitions within probability theory, stochastic functional analysis and ergodic theory.
\begin{defn}
Let $t\in\mathbb{R}^{+}$ and let $(\Omega,\mathfrak{F},\bm{\mu})$ be a probability space. Let $\mathscr{U}(x;\omega)$ be a random scalar function that depends only on $t\in\mathbb{R}$ and also $\omega\in\Omega$. Given any pair $(t,\omega)$ there is a mapping $M:\mathbb{R}\times\Omega\rightarrow\mathbb{R}$ such that
\begin{equation}
\mathfrak{M}:(\omega,x)\longrightarrow\widehat{\mathscr{U}}(t;\omega)\nonumber
\end{equation}
so that $\widehat{\mathscr{U}}(t,\omega)$ is a random Variable or field on $\mathbb{R}^{+}$ with respect to the probability space $(\Omega,\mathfrak{F},\bm{\mu})$. The stochastic field is then essentially a family of random variables $\lbrace\mathscr{U}(t;\omega)\rbrace$ defined with respect to $(\Omega,\mathfrak{F},\bm{\mu})$ and $\mathbb{R}^{+}$.
\end{defn}
The scalar random field can also include a spatial variable $x\in\mathbb{R}^{3}$ so that given any triplet $(x,t,\omega)$ there is a mapping $\mathfrak{M}:\mathbb{R}^{+}\times\mathbb{R}^{+}\times\Omega\rightarrow\mathbb{R}^{+}$ such that
\begin{equation}
\mathfrak{M}:(\omega,x,t)\longrightarrow\widehat{\mathscr{U}}(x,t;\omega)\nonumber
\end{equation}
However, it will be sufficient to consider fields that vary randomly in time only. The expected value of the random field with respect to $(\Omega,\mathfrak{F},\bm{\mu}$ is defined as follows
\begin{defn}
Given the random scalar field $\widehat{\mathscr{U}}(t;\omega)$, then if $
\int_{\Omega}\|\widehat{\mathscr{U}}(t;\omega)\|d\bm{\mu}(\omega)<\infty$, the expectation of $\widehat{\mathscr{U}}(t;\omega)$ is
\begin{equation}
\bm{\mathrm{I\!E}}\left\lbrace\widehat{\mathscr{Z|}}(t;\omega)\right\rbrace=
\int_{\Omega}\mathscr{U}(t;\omega)d\bm{\mu}(\omega)
\end{equation}
\end{defn}
\begin{defn}
Let $(\Omega,\mathfrak{F},\bm{\mu})$ be a probability space, then an $L_{p}(\Omega,\mathfrak{F},\bm{\mu})$ space or an $L_{p}$-space for $p\ge 1$ is a linear normed space of random scalar fields that satisfies the conditions
\begin{equation}
\bm{\mathrm{I\!E}}\left\lbrace|\mathscr{U}(t;\omega)|^{p}\right\rbrace =\int_{\Omega}|\widehat{\mathscr{U}}(t;\omega)|^{p}d\bm{\mu}(\omega)<\infty
\end{equation}
and the corresponding norm $L_{p}$ norm is
\begin{equation}
\|\widehat{\mathscr{U}}(x)\|=(\bm{\mathrm{I\!E}}\left\lbrace\big|\widehat{\mathscr{U}}(t;\omega)\big|^{p}
\right\rbrace)^{1/p}
\end{equation}
with the usual $\mathcal{L}_{2}$ Euclidean norm for $p=2$. When $p=2$ the fields are second-order random fields. Note that an $\mathcal{L}_{2}$-space equipped with the scalar product
\begin{equation}
\bm{\mathrm{I\!E}}\left\lbrace\widehat{\mathscr{U}}(x,\omega)\otimes\widehat{\mathscr{U}}'(t,\omega)
\right\rbrace = \int_{\Omega}\widehat{\mathscr{U}}(t,\omega)
\otimes\widehat{\mathscr{U}}'(t,\omega)d\bm{\mu}(\omega)
\end{equation}
is also a Hilbert space.
\end{defn}
The second-order correlations, moments and covariances are of the most interest.
\begin{defn}
Let $t,s\in\mathbb{R}^{+}$ and let $\omega,\xi\in\Omega$. The expectations of mean values of the fields $\widehat{\mathscr{U}}(t,\omega)$ and $\widehat{\mathscr{U}}(y,\xi)$ are
\begin{equation}
\bm{\mathrm{I\!E}}(t)=\bm{\mathrm{I\!E}}\lbrace\widehat{\mathscr{U}}(t)\rbrace = \int_{\Omega}\widehat{\mathscr{U}}(t,\omega)d\bm{\mu}(\omega)
\end{equation}
\begin{equation}
\bm{\mathrm{I\!E}}(s)=\bm{\mathrm{I\!E}}\lbrace\widehat{\mathscr{U}}(s)\rbrace = \int_{\Omega}\widehat{\mathscr{U}}(s,\xi)d\bm{\mu}(\omega)
\end{equation}
then the 2nd-order moment or stochastic expectation is
\begin{equation}
\bm{\mathrm{I\!E}}\lbrace\widehat{\mathscr{U}}(t)\otimes\widehat{\mathscr{U}}(s)\rbrace=
\int_{\Omega}\int_{\Omega}\widehat{\mathscr{U}}(t,\omega)\otimes\widehat{\mathscr{U}}(s,\xi)
d\bm{\mu}(\omega)d\bm{\mu}(\xi)
\end{equation}
The covariance is then
\begin{equation}
\mathsf{Cov}(t,s)=\bm{\mathrm{I\!E}}\left\lbrace(\widehat{\mathscr{U}}(t)
-\bm{\mathrm{I\!E}}(t))(\widehat{\mathscr{U}}(s)-\bm{\mathrm{I\!E}}(s))\right\rbrace
\end{equation}
or
\begin{equation}
\mathsf{Cov}(t,s)=
\int_{\Omega}\int_{\Omega}(\widehat{\mathscr{U}}(t;\omega)-\bm{\mathrm{I\!E}}(t))
\otimes(\widehat{\mathscr{U}}(s;\omega)-\bm{\mathrm{I\!E}}(s))d\bm{\mu}(\omega)
d\bm{\mu}(\xi)
\end{equation}
so that
\begin{equation}
\mathsf{COV}(t,s)=\bm{\mathrm{I\!E}}\lbrace\widehat{\mathscr{U}}(t)
\widehat{\mathscr{U}}(s)\rbrace-\bm{\mathrm{I\!E}}(t))\bm{\mathrm{I\!E}}(s))
\end{equation}
\end{defn}
\begin{defn}
Given a set of fields $\widehat{\mathscr{U}}(t)_{1}),...,\widehat{\mathscr{U}}(t_{n})$ at points $t_{1})...t_{n}\in\mathbb{R}^{+}$ then the $m^{th}$-order moments and cumulants are
\begin{equation}
\bm{\mathrm{I\!E}}\bigg\lbrace\widehat{\mathscr{U}}(t_{1})...\widehat{\mathscr{U}}(t_{m})\bigg\rbrace
\end{equation}
\begin{equation}
\mathlarger{\mathsf{COV}}\bigg\lbrace\widehat{\mathscr{U}}(t_{1})...\widehat{\mathscr{U}}(t_{m})\bigg\rbrace
\end{equation}
where at second order
\begin{equation}
\mathlarger{\mathsf{COV}}\bigg\lbrace\widehat{\mathscr{U}}(t)\widehat{\mathscr{U}}(s)\bigg\rbrace\equiv
\mathlarger{\mathsf{COV}}(t,s)=
\bm{\mathrm{I\!E}}\bigg\lbrace\widehat{\mathscr{U}}(t)\widehat{\mathscr{U}}(s)\bigg\rbrace-
\bm{\mathrm{I\!E}}(t))\bm{\mathrm{I\!E}}(s)
\end{equation}
\end{defn}
The covariance must have the following important properties
\begin{lem}
A function $\mathsf{COV}(t,s)$ is formally a covariance if the following are satisfied:
\begin{enumerate}
\item Let $t_{\alpha},s_{\beta}\in\mathbb{R}^{+} $ with $\alpha\beta\in\mathbb{Z}$. Then any covariance $\mathsf{COV}(t_{\alpha},s_{\beta})$ is always nonnegative semi-definite such that for any $q_{\alpha},q_{\beta}>0$
    \begin{equation}
    \sum_{\alpha}^{N}\sum_{\beta}^{N}q_{\alpha}q_{\beta}
    \mathlarger{\mathsf{COV}}(t_{\alpha},s_{\beta})\ge 0 \nonumber
    \end{equation}
\item Symmetry, $\mathsf{COV}(s,t_{\beta})=\mathsf{COV}(s,t)$
\item $\lim_{\|t-s\|\rightarrow \infty}\mathsf{COV}(t,s)=0$. If $\mathsf{COV}\lbrace\widehat{\mathscr{U}}(t)\widehat{\mathscr{U}}(s)\rbrace\equiv
\mathsf{COV}(t,s)=0$ then $\widehat{\mathscr{U}}(s)$ and $\widehat{\psi}(s)$ are uncorrelated.
\end{enumerate}
\end{lem}
\begin{proof}
To prove (1),
\begin{equation}
\bm{\mathrm{I\!E}}\left\lbrace\sum_{\alpha=1}^{N}q_{\alpha}[\mathscr{U}(t_{\alpha})
-\bm{\mathrm{I\!E}}(t_{\alpha})]^{2}\right\rbrace=\sum_{\alpha}\sum_{\beta}q_{\alpha}q_{\beta}
\bm{\mathrm{I\!E}}\bigg\lbrace[\mathscr{U}(t_{\alpha})-\bm{\mathrm{I\!E}}(t_{\alpha})]\otimes[\mathscr{U}(t_{\beta})-
\bm{\mathrm{I\!E}}(t_{\beta})]
\bigg\rbrace\ge 0
\end{equation}
\end{proof}
A \emph{lognormal} scalar random field is defined as follows
\begin{defn}
Let $\widehat{\mathscr{U}}(t)$ be a scalar random field, then there is a scalar random field $\widehat{\mathscr{I}}(t)$ such that
\begin{equation}
\widehat{\mathscr{I}}(t)=\exp(\widehat{\mathscr{U}}(t))
\end{equation}
with inverse $\widehat{\mathscr{U}}(t)=\ln(\widehat{\mathscr{X}(t)}$
\end{defn}
These definitions now extend naturally to random vector fields $\mathscr{U}_{i}(t)$  for all $i=1$ to $n$.
\begin{defn}
Let $x^{i}\subset\mathbb{D}\subset\mathbb{R}^{3}$ be Euclidean coordinates and let $(\Omega,\mathfrak{F},\bm{\mu})$ be a probability space. Let $\widehat{\mathscr{U}}_{i}(x;\omega)$ be a random vector function that depends on $t\subset\mathbb{R}^{+}$ and also $\omega\in\Omega$. Given any pair $(t,\omega)$ there is a mapping $\mathfrak{M}:\mathbb{R}^{+}\times\Omega\rightarrow\mathbb{R}^{n}$ such that
\begin{equation}
\mathfrak{M}:(\omega,t)\longrightarrow\widehat{\mathscr{U}}_{i}(t;\omega)\nonumber
\end{equation}
so that $\widehat{\mathscr{U}}_{i}(t,\omega)$ is a random vector field spanning $\mathbb{R}^{n}$ with respect to the probability space $(\Omega,\mathfrak{F},\bm{\mu})$.
\end{defn}
The expected value of the random vector field with respect to $(\Omega,\mathfrak{F},\bm{\mu}$ is defined as before
\begin{defn}
Given the random vector field $\widehat{\mathscr{U}}_{i}(t;\omega)$, then if $ \int_{\Omega}\|\widehat{\mathscr{U}}(x;\omega)\|d\bm{\mu}(\omega)<\infty$, the expectation of $\widehat{\mathscr{U}}(t;\omega)$ is
\begin{equation}
\bm{\mathrm{I\!E}}\bigg\lbrace\widehat{\mathscr{U}}_{i}(t;\omega)\bigg\rbrace=
\int_{\Omega}\mathscr{U}_{i}(t)d\bm{\mu}(\omega)
\end{equation}
\end{defn}
\begin{defn}
An $L_{p}(\Omega,\mathfrak{F},\bm{\mu})$ space or an $L_{p}$-space for $p\ge 1$ is a linear normed space of random fields that satisfies the conditions
\begin{equation}
\bm{\mathrm{I\!E}}\bigg\lbrace\bigg\|\mathscr{U}_{i}(t;\omega)\bigg\|^{p}\bigg\rbrace =\int_{\Omega}\left(\sum_{i=1}^{n}|\mathscr{U}_{i}(t;\omega)|^{p}\right)^{1/p}d\bm{\mu}(\omega)<\infty
\end{equation}
with the usual Euclidean or $L_2$ norm for $p=2$. The second-order correlations, moments and covariances are now
\end{defn}
\begin{defn}
Let $t,s\in\mathbb{R}^{+}$ and let $\omega,\xi\in\Omega$. The expectations of mean values of the fields $\widehat{\mathscr{U}}_{i}(t,\omega)$ and $\widehat{\mathscr{U}}_{j}(s,\xi)$ are
\begin{equation}
\bm{\mathrm{I\!E}}_{i}(t)=\bm{\mathrm{I\!E}}\bigg\lbrace\widehat{\mathscr{U}}_{i}(t,\omega)\bigg\rbrace = \int_{\Omega}\widehat{\mathscr{U}}_{i}(t,\omega)d\mu(\omega)
\end{equation}
\begin{equation}
\bm{\mathrm{I\!E}}_{j}(s)=\bm{\mathrm{I\!E}}\bigg\lbrace\widehat{\mathscr{U}}_{j}(s,\xi)\bigg\rbrace = \int_{\Omega}\widehat{\mathscr{U}}_{j}(s,\xi)d\mu(\xi)
\end{equation}
then the 2nd-order moment or expectation is
\begin{equation}
\bm{\mathrm{I\!E}}\left\lbrace\widehat{\mathscr{U}}_{i}(t)\widehat{\mathscr{U}}_{j}(t)
\right\rbrace=\int_{\Omega}\int_{\Omega}
\widehat{\mathscr{U}}_{i}(t,\omega)\otimes\widehat{\mathscr{U}}_{j}(s,\xi)
d\mu(\omega)d\mu(\xi)
\end{equation}
The covariance is then
\begin{equation}
\mathlarger{\mathsf{COV}}_{ij}(t,s)=\bm{\mathrm{I\!E}}\left\lbrace(\widehat{\mathscr{U}}(t)-
\bm{\mathrm{I\!E}}_{i}(t))(\widehat{\mathscr{U}}(s)-\bm{\mathrm{I\!E}}_{j}(s)\right\rbrace
\end{equation}
or
\begin{equation}
\mathlarger{\mathsf{COV}}_{ij}(t,s)=\int_{\Omega}\int_{\Omega}
(\widehat{\mathscr{U}}_{i}(t;\omega)-\bm{\mathrm{I\!E}}_{i}(t))\otimes
(\widehat{\mathscr{U}}_{i}(t;\omega)-\bm{\mathrm{I\!E}}_{j}(s))d\bm{\mu}(\omega)
d\bm{\mu}(\xi)
\end{equation}
so that
\begin{equation}
\mathsf{COV}_{ij}(t,s)=\bm{\mathrm{I\!E}}\bigg\lbrace\widehat{\mathscr{U}}_{i}(t)
\widehat{\mathscr{U}}_{j}(s)\bigg\rbrace-\bm{\mathrm{I\!E}}_{i}(t))\bm{\mathrm{I\!E}}_{j}(s))
\end{equation}
\end{defn}
\begin{defn}
Given a set of fields $\widehat{\mathscr{U}}(t_{1}),...,\widehat{\mathscr{U}}(t_{n})$ at points $t_{1})...t_{m}\in\mathbb{R}^{+}$ then the $m^{th}$-order moments and cumulants are
\begin{equation}
\bm{\mathrm{I\!E}}\bigg\lbrace\widehat{\mathscr{U}}_{i_{1}}(t_{1})...
\widehat{\mathscr{U}}_{i_{m}}(t_{m})\bigg\rbrace
\end{equation}
\begin{equation}
\bm{\mathrm{I\!K}}\bigg\lbrace\widehat{\mathscr{U}}_{j_{1}}(t_{1})...
\widehat{\mathscr{U}}_{j_{m}}(t_{m})\bigg\rbrace
\end{equation}
where at second order
\begin{equation}
\bm{\mathrm{I\!K}}\bigg\lbrace\widehat{\mathscr{U}}(t)\widehat{\mathscr{U}}(s)\bigg\rbrace\equiv
\bm{\mathsf{COV}}_{ij}(t,s)=\bm{\mathrm{I\!E}}\bigg\lbrace\widehat{\mathscr{U}}_{i}(t)
\widehat{\mathscr{U}}_{j}(s)\bigg\rbrace-\bm{\mathrm{I\!E}}_{i}(t))\bm{\mathrm{I\!E}}_{j}(s)
\end{equation}
\end{defn}
The covariance tensor is again nonnegative semi-definite so that
\begin{equation}
\sum_{\alpha}^{N}\sum_{\beta}^{N}q_{\alpha}q_{\beta}\bm{\mathrm{COV}}_{i_{\alpha}j_{\beta}}
(t_{\alpha},s_{\beta})\ge 0 \nonumber
\end{equation}
with symmetry $\mathsf{COV}_{ij}(t,s_{\beta})=\mathsf{COV}_{ij}(t,s)$ and $\lim_{\|t-s\|\rightarrow \infty}\mathsf{COV}(t,s)=0$. If
\begin{equation}
\bm{\mathrm{I\!K}}\bigg\lbrace \widehat{\mathscr{U}}_{i}(t)\widehat{\mathscr{U}}_{j}(t)\bigg\rbrace\equiv \mathsf{COV}_{ij}(t,s)=0
\end{equation}
then $\mathscr{U}_{i}(s)$ and $\mathscr{U}_{j}(s)$ are uncorrelated.

A very important class of random fields are the Gaussian random vector fields (GRVFS) which are characterized only by their first and second moments. The GRVFS can also be isotropic, homogenous and stationary. The details will be made more precise but the advantages of GRVFs are briefly enumerated.
\begin{enumerate}
\item GRVFS have convenient mathematical properties which generally simplify calculations; indeed, many results can only be evaluated using Gaussian fields.
\item A GRVF can be classified purely by its first and second moments and high-order moments and cumulants can be ignored.
\item Gaussian fields accurately describe many natural stochastic processes including Brownian motion.
\item A large superposition of non-Gaussian fields can approach a Gaussian field.
\end{enumerate}
For this paper, the following definitions are sufficient for isotropic GRVFS.
\begin{defn}
Any GRVF has normal probability density functions. The following always hold:
\begin{enumerate}
\item The first moment vanishes so that
\begin{equation}
\bm{\mathrm{I\!E}}_{i}(x)=\bm{\mathrm{I\!E}}\bigg\lbrace\widehat{\mathscr{U}}_{i}(t;\omega)\bigg\rbrace
=\int_{\Omega}\widehat{\mathscr{U}}_{i}(t;\omega)d\bm{\mu}(\omega)=0\nonumber
\end{equation}
\item The covariance then reduces to
\begin{equation}
\mathsf{COV}_{ij}(t,s)\equiv\bm{\mathrm{I\!K}}\bigg\lbrace \widehat{\mathscr{U}}_{i}(t)\widehat{\mathscr{U}}_{j}(s)\bigg\rbrace
\equiv\bm{\mathrm{I\!E}}\bigg\lbrace
\widehat{\mathscr{U}}_{i}(t)\widehat{\mathscr{U}}_{j}(s)
\bigg\rbrace=J_{ij}(\Delta;\varsigma)\nonumber
\end{equation}
\end{enumerate}
\end{defn}
\begin{defn}
The GRVF is isotropic if $\bm{\mathrm{Cov}}_{ij}(t,s)=J_{ij}(\Delta;\varsigma)$ depends only on the separation $\Delta=|t-s|$ and is stationary if $\bm{\mathrm{Cov}}_{ij}(t+\delta s,s+ \delta s)=J_{ij}(\Delta;\varsigma)$. Hence, the 2-point function or Greens function is translationary invariant.
\end{defn}
\begin{defn}
An important class of random fields are white noises which have the delta-function correlated 2-point function. But for white noises $\mathlarger{\mathsf{COV}}_{ij}(t,t)=\infty$ so the equal-time correlation diverges. The standard differential for Brownian motion is then $ d\widehat{\mathscr{B}}(t)=\widehat{\mathscr{W}}dt$. If the noise has a finite correlation time $\varsigma$ then a regulated covariance or 2-point function is possible such that
\begin{equation}
\bm{\mathrm{COV}}_{ij}(t,s)=\bm{\mathrm{I\!E}}\bigg\lbrace\widehat{\mathscr{U}}_{i}(t)
\widehat{\mathscr{U}}_{j}(s)\bigg\rbrace=\alpha\delta_{ij}J(|\Delta|;\varsigma)
\end{equation}
where now $\bm{\mathrm{COV}}_{ij}(t,t)=J(0;\varsigma)<\infty$. Examples are colored noise or the Orstein-Uhlenbeck process.
\end{defn}
\begin{defn}
A GRVF $\widehat{\mathscr{U}}_{i}(t)$ is almost surely continuous at $x\in\mathbb{R}^{+}$ if
\begin{equation}
\widehat{\mathscr{U}}_{i}(t+\varsigma)\longrightarrow\widehat{\mathscr{U}}_{i}(t)
\end{equation}
as $\varsigma \rightarrow 0$
\end{defn}
When this holds for all $t\in\mathbb{R}^{+}$ then this is known as 'sample function continuity'. The following result due to Adler [24,68],gives the sufficient condition for continuous sample paths
\begin{lem}
Let $\widehat{\mathscr{U}}_{i}(t)$ be a non-white GRVF. Then if for some $C>0$ and $\lambda>0$ with $\eta>\lambda$
\begin{equation}
\bm{\mathrm{I\!E}}\bigg\lbrace\bigg\|\widehat{\mathscr{U}}_{i}(t+\zeta)-
\widehat{\mathscr{U}}_{i}(t\bigg\|^{\lambda}\bigg\rbrace \le\frac{C|\zeta|2n}{|\ln|\zeta||^{1+\eta}}
\end{equation}
If $\widehat{\mathscr{U}}_{i}(t)$ is a Gaussian random field with continuous $\mathsf{COV}_{ij}(t,s)$ then given sone $C>0$ and some $\epsilon>0$
\begin{equation}
\bm{\mathrm{I\!E}}\bigg\lbrace |\widehat{\mathscr{U}}_{i}(t+\zeta)-\widehat{\mathscr{U}}_{i}(t)|^{2}\bigg\rbrace\le \frac{C}{\ln|\zeta|^{1+\epsilon}}
\end{equation}
\end{lem}
\begin{thm}
If $\phi$ is a convex function and $X$ is a random variable or field then Jensen's inequality is that statement
\begin{equation}
\phi\bm{\mathrm{I\!E}}\big\lbrace(X)\big\rbrace\big)\le \bm{\mathrm{I\!E}}\big\lbrace\phi(X)\big\rbrace
\end{equation}
\end{thm}
The differentiability of a RGVF is defined as follows
\begin{defn}
Let $\widehat{\mathscr{U}}_{i}(t)$ be a RVF and let $t\in\mathbb{R}$. Then
\begin{equation}
\partial_{t}\widehat{\mathscr{U}}_{i}(t)\equiv \frac{d}{dt}\widehat{\mathscr{U}}_{i}(t)=\lim_{\zeta\rightarrow 0}\left(
\frac{\widehat{\mathscr{U}}_{i}(t+\zeta)-\widehat{\mathscr{U}}_{i}(t)}{\zeta}\right)
\end{equation}
for all $t\in\mathbb{R}^{+}$. It follows that
\begin{equation}
\lim_{\zeta\rightarrow 0}\bm{\mathrm{I\!E}}\bigg\lbrace\left(\frac{\widehat{\mathscr{U}}_{i}(t+\zeta)-
\widehat{\mathscr{U}}_{i}(t)}{\zeta}-\partial_{t}\widehat{\mathscr{U}}_{i}(t)\right)^{2}
\bigg\rbrace=0
\end{equation}
The second-order derivative is
\begin{equation}
\partial_{tt}\widehat{\mathscr{U}}_{i}(t)
=\lim_{\zeta\rightarrow 0}\lim_{\xi\rightarrow 0}\frac{1}{\zeta\xi}[\widehat{\mathscr{U}}_{i}(t+\zeta+\xi)
-\widehat{\mathscr{U}}_{i}(t+\zeta)-\widehat{\mathscr{U}}_{i}
(t+\xi)+\widehat{\mathscr{U}}_{i}(t)]
\end{equation}
An alternative(and better) definition in the mean-square sense is as follows
\begin{equation}
\lim_{h\rightarrow 0}\lim_{g\rightarrow 0}\bm{\mathrm{I\!K}}\bigg\lbrace
\bigg(\frac{\widehat{\mathscr{U}}_{i}(t+\zeta)-\widehat{\mathscr{U}}_{i}(t)}{\zeta}-
\frac{\widehat{\mathscr{U}}_{i}(t+\xi)-\widehat{\mathscr{U}}_{i}(t)}{\xi}\bigg)^{2}\bigg\rbrace=0
\end{equation}
Clearly the existence of derivatives requires that the 2-point function or covariance is non-white and therefore finite or regulated at $s=t$.
\end{defn}
This leads to the following important lemma which establishes the correlations for a gradient of a RVF.
\begin{lem}
A SRF if differentiable iff:
\begin{enumerate}
\item The first moment $\bm{\mathrm{I\!E}}_{i}(x)=\bm{\mathrm{I\!E}}\lbrace\widehat{\mathscr{U}}_{i}(t)
\rbrace$ is differentiable.
\item The covariance $\mathlarger{\mathsf{COV}}_{ij}(t,s)$ exists and is finite at all points $t=s$ in $\mathbb{R}^{+}$.
\end{enumerate}
\end{lem}
Then
\begin{eqnarray}
\partial_{t}\bm{\mathrm{COV}}_{ij}(t,s)\equiv
\partial_{t}\bm{\mathrm{COV}}_{ij}(t,s)=
\bm{\mathrm{I\!E}}\lbrace\widehat{\mathscr{U}}_{i}(t)
\widehat{\mathscr{U}}_{j}(s)\rbrace
\end{eqnarray}
\subsection{Stochastic Integration}
Having established existence of the derivative of a SRVF, stochastic integrals can be defined as a mean-square Riemann integration.
\begin{defn}
Let $\widehat{\mathscr{U}}_{i}(x)$ be a RGVF spanning $\mathbb{R}^{n}$. Let $f(t)$ be a deterministic continuous and bounded function such that $f:\mathbb{R}^{+}\times\mathbb{R}^{+}\rightarrow\mathbb{R}^{+}$. The stochastic integral is the mean-square Riemann integral
\begin{equation}
\widehat{\mathscr{I}}_{i}(t)=\int_{\mathbb{R}^{(+)}}f(t)\widehat{\mathscr{U}}_{i}(t)dt
\equiv\int_{\mathscr{B}^{(+)}} f(t)\widehat{\mathscr{U}}_{i}(t,\zeta)dt
\end{equation}
The integral exists if the limit of the Riemann sum exists
\begin{equation}
\widehat{\mathscr{I}}_{i}^{(m)}(t)=\sum_{\eta=1}^{m}f(t_{\eta})
\otimes\widehat{\mathscr{U}}_{i_{\eta}}(x_{\eta})\Delta(t_{\eta})
\end{equation}
where $\Delta(t_{\eta})$ is a line element. The integral then exists if
\begin{equation}
\widehat{\mathscr{I}}_{i}(t)=\lim_{m\rightarrow \infty}\mathscr{I}^{(m)}_{i}(t)
\end{equation}
Since $\bm{\mathrm{I\!E}}\lbrace\mathscr{U}(x)\rbrace=0$ then $\bm{\mathrm{I\!E}}\lbrace\widehat{\mathbb{Y}}(y)\rbrace=0$. When $f(t)=1$ then
\begin{equation}
\widehat{\mathscr{I}}_{i}(t)=\int_{\mathbb{R}^{+}}\widehat{\mathscr{U}}_{i}(t)dt
\end{equation}
\end{defn}
This definition leads to the following corollary
\begin{cor}
A SRVF $\widehat{\mathscr{U}}_{i}(x)$ is mean-square Riemann integrable iff
\begin{align}
&\bm{\mathrm{I\!E}}\bigg\lbrace\widehat{\mathscr{I}}_{i}(t)\widehat{\mathscr{I}}_{j}(s)\bigg\rbrace=
\int_{\mathbb{R}}\int_{\mathbb{R}}f(t)f(s)\bm{\mathrm{Cov}}_{ij}(t,s)dtds<\infty\nonumber\\&
\lim_{m\rightarrow\infty}\sum_{\xi=1}^{m}\sum_{\eta=1}^{m}f(t_{\eta})f(s_{\eta}')
\bm{\mathrm{I\!E}}\bigg\lbrace\widehat{\mathscr{U}}_{i_{\xi}}(t_{\xi})\otimes
\widehat{\mathscr{U}}_{j_{\eta}}(s')\bigg\rbrace\Delta(t_{\xi})
\Delta(x_{\eta})
\end{align}
For a Gaussian RVF this is equivalent to
\begin{equation}
\bm{\mathrm{I\!E}}\bigg\lbrace\widehat{\mathscr{I}}_{i}(s)\widehat{\mathscr{I}}_{j}(y')\bigg\rbrace=
\int_{\mathbb{R}^{+}}\int_{\mathbb{R}^{+}}f(t)f(s)
\bm{\mathrm{I\!E}}\bigg\lbrace\widehat{\mathscr{U}}(t)\widehat{\mathscr{U}}(s')\bigg\rbrace d^{3}td^{3}s'<\infty
\end{equation}
\end{cor}
\begin{prop}
For a GRVF, the 2-point correlation is
\begin{equation}
\bm{\mathrm{I\!E}}\bigg\lbrace\widehat{\psi}(t)\widehat{\psi}(s)\bigg\rbrace
=\int_{\mathbb{R}^{+}}\int_{\bar{\mathbb{R}^{+}}}\bm{\mathrm{I\!E}}\bigg\lbrace\widehat{\mathscr{U}}_{i}(t)\otimes
\widehat{\mathscr{U}}_{i}(s)\bigg\rbrace dtds
\end{equation}
The n-point correlation is
\begin{equation}
\bm{\mathrm{I\!E}}\lbrace\widehat{\psi}(t_{1})\times...\times\widehat{\psi}(t_{n})\rbrace=
\int...\int dt_{1}...dt_{m}
\bm{\mathrm{I\!E}}\lbrace\widehat{\mathscr{U}}_{i_{1}}(t_{1})\times...
\times\widehat{\mathscr{U}}_{i_{1}}(t_{1})\rbrace
\end{equation}
which can be expressed in a path integral form as
\begin{equation}
\bm{\mathrm{I\!E}}\lbrace\widehat{\psi}(t_{1})\times...\times\widehat{\psi}(t_{n})\rbrace=
\int \mathfrak{D}_{n}t\bm{\mathrm{I\!E}}\lbrace\widehat{\mathscr{U}}_{i_{1}}(t_{1})\times...
\times\widehat{\mathscr{U}}_{i_{1}}(t_{1})\rbrace
\end{equation}
\end{prop}
\begin{thm}
let X be a random variable or field and let $X_{i}$ be s et of n random variables or fields, then Fubini's theorem states that
\begin{equation}
\bm{\mathrm{I\!E}}\bigg\lbrace\int X\bigg\rbrace =\int\bm{\mathrm{I\!E}}\bigg\lbrace X\bigg\rbrace
\end{equation}
or
\begin{equation}
\bm{\mathrm{I\!E}}\bigg\lbrace\sum_{i=1}^{n} X_{i}\bigg\rbrace
=\sum_{i=1}^{n}\bm{\mathrm{I\!E}}\bigg\lbrace X_{i}\bigg\rbrace
\end{equation}
\end{thm}
\subsection{Ito and stochastic integrals}
The Ito and and Stratanovich interpretations of stochastic integrals are defined as follows
\begin{defn}
Let $f(t)$ be a continuous function of t. If $\mathscr{Q}=[0,T]$ is partitioned so that $t_{\xi=0}=t_{0}$ then $t_{0}<t_{1}<t_{2}<...<t_{\xi}$ and $t_{*}=t_{N}$. For $\mathscr{U}=\mathscr{W}(t)$, a Weiner process then the following Riemann-Steiltjes
sum over $\mathscr{Q}$ defines an Ito integral
\begin{align}
&\widehat{\mathscr{I}}(t)=\int_{t_{\epsilon}}^{t}f(\tau)d\widehat{\mathscr{W}}(\tau)\equiv
\int_{t_{\epsilon}}^{t}f(\tau)\widehat{\mathscr{W}}(\tau)d\tau\nonumber\\&
\sum_{\xi=0}^{N-1}f(t_{\xi}^{n})\widehat{\mathscr{W}}(t_{\xi+1}^{N})-
\widehat{\mathscr{W}}(t_{i}^{N})]
\end{align}
in the limit that partitions $\lbrace t_{i}^{n}\rbrace\rightarrow 0$. For each $i=1$ to $(n-1)$. This interpretation essentially always takes the minimum value of the pair $[ f(t_{\xi+1}^{n}), f(t_{\xi}^{n}))]$. The alternative Stratanovich interpretation always takes the averaged value $ \frac{1}{2}|f((t_{\xi+1}^{n}))+f(t_{\xi}^{n})|$ so that the Stratanovich stochastic integral is
\begin{align}
&\widehat{\mathscr{I}}= \int_{t_{\epsilon}}^{t}f(\tau) d\widehat{\mathscr{W}}(\tau)=\int_{t_{\epsilon}}^{t}f(\tau) \widehat{\mathscr{W}}(\tau)d\tau\nonumber\\&
=\sum_{\xi=0}^{n-1}\frac{1}{2}[f(t_{\xi+1}^{n})-f(t_{i}^{N})][\widehat{\mathscr{W}}(t_{i+1}^{n})-
\widehat{\mathscr{W}}(t_{\xi}^{n})]
\end{align}
For the Gaussian random field $\widehat{\mathscr{U}}(t)$ with regulated 2-point funcntion
\begin{align}
&\widehat{\mathscr{I}}= \int_{t_{\epsilon}}^{t}f(\tau)d\widehat{\mathscr{U}}(\tau)=\int_{t_{\epsilon}}^{t}f(\tau) \widehat{\mathscr{U}}(\tau)d\tau\nonumber\\&
\sum_{\xi=0}^{n-1}\frac{1}{2}[f(t_{\xi+1}^{n})-f(t_{i}^{N})][\widehat{\mathscr{U}}(t_{i+1}^{n})-
\widehat{\mathscr{U}}(t_{\xi}^{n})]
\end{align}
\end{defn}
The Ito interpretation requires the Ito calculas and $d\widehat{\mathscr{W}}(t)d\widehat{\mathscr{W}}(t)=dt$ and $(dt^{2})=0$. However, within the Stratanovich interpretation the rules of ordinary calculas apply. Two important properties of Ito integrals are zero mean an Ito isometry such that
\begin{equation}
\bm{\mathrm{I\!E}}\bigg(\int_{0}^{T}f(t)d\widehat{\mathscr{W}}(t)\bigg)=0
\end{equation}
and
\begin{align}
&\bm{\mathrm{I\!E}}\bigg(\int_{0}^{T}f(t)d\widehat{\mathscr{W}}(t)\bigg)^{2}=\int_{0}^{T}
\bm{\mathrm{I\!E}}f(t))^{2}dt\\&
\bm{\mathrm{I\!E}}\bigg(\int_{0}^{T}f(t)d\widehat{W}(t)\int_{0}^{T}
Y(t)d\widehat{\mathscr{W}}(t)\bigg)^{2}=\int_{0}^{T}\bm{\mathrm{I\!E}}(f(t)Y(t))dt
\end{align}
\begin{prop}
Let $\widehat{\mathscr{B}}_{i}(t)$ be a random field which is the lognormal of a field $\mathscr{U}(x)$ so that $\widehat{\mathscr{B}}(t)=\ln(\widehat{\mathscr{U}}(t))$. Then
\begin{equation}
\widehat{\mathscr{B}}(t)=\exp(\widehat{\mathscr{I}}(t))\equiv \exp\left(
\int_{\mathfrak{S}}\widehat{\mathscr{U}}_{i}(t)dt
\right)
\end{equation}
The expectation is
\begin{equation}
\bm{\mathrm{I\!E}}\lbrace\widehat{\mathscr{Y}}(x)\rbrace
=\bm{\mathrm{I\!E}}\bigg\lbrace\exp(\widehat{\mathscr{I}}(t))\bigg\rbrace \equiv
\bm{\mathrm{I\!E}}\left\lbrace\exp\left(\int\widehat{\mathscr{U}}_{i}(t)dt^{i}\right)\right\rbrace
\end{equation}
\end{prop}
\section{Appendix B: Proof of Lemma}
\begin{lem}
Given white-noise perturbations of the static radial moduli set $(\psi_{i}(t))_{i=1}^{n}$ of the form
\begin{equation}
\widehat{\psi}(t)=\psi_{i}(t)+\zeta\int_{0}^{t}f(\psi_{i}(s))d\mathscr{W}(s)
\end{equation}
with the conditions
\begin{equation}
\|f_{i}(\psi_{i}(t))\|^{2} \le K\|\psi_{i}(t)\|^{2}
\end{equation}
and
\begin{equation}
\int_{t_{0}}^{t}\bigg\|f(\psi_{i}(s)\bigg\|^{2}ds < \infty
\end{equation}
then the $l^{th}$-order moments are finite and bounded for all finite $t>t_{o}$ and grow exponentially, with the estimates
\begin{align}
&\bm{\mathrm{I\!E}}\bigg\lbrace\bigg\|\sup_{t\le T}\widehat{\psi}_{i}(t)\bigg\|^{\ell}\bigg\rbrace\le \|\psi(t)\|^{\ell}\exp(\tfrac{1}{2}K\ell(\ell-1)|T-t_{0}|)
\\&\bm{\mathrm{I\!E}}\bigg\lbrace\|\sup_{t\le T}\widehat{\psi}_{i}(t)\bigg\|^{\ell}\bigg\rbrace\le \bigg\|\psi_{\epsilon}\|^{\ell}\exp(\tfrac{1}{2}K\ell(\ell-1)|T-t_{0}|)
\end{align}
\end{lem}
\begin{proof}
Let $X(\widehat{\psi}(t))$ be a $C^{2}$-differentiable functional of $\widehat{\psi}(t)$ then by Ito's Lemma
\begin{align}
dX(\widehat{\psi}_{i}(t))&=\bm{\nabla}X(\psi_{i}(t))d\widehat{\psi}_{i}(t)
+\frac{1}{2}\bm{\nabla}^{2}X(\psi_{i}(t))d[\widehat{\psi}_{i},\widehat{\psi}_{i}](t)\nonumber\\&\equiv
(\bm{\nabla}X(\psi_{i}(t)))[d\psi_{i}(t)+\zeta f(\psi_{i}(t))d\mathscr{W}(t)]+\frac{1}{2}(\bm{\nabla}^{2}X(\psi_{i}(t)))
\|f(\psi_{i}(t))\|^{2}dt
\end{align}
where $\bm{\nabla}=d/d\psi_{i}(t)$ and $[\widehat{\psi},\widehat{\psi}](t)$ is the quadratic variation. Integrating
\begin{align}
&X(\widehat{\psi}_{i}(t))=X(\psi_{i}^{E})+\int_{0} ^{t}\bm{\nabla}X(\psi_{i}(s))d\psi_{i}(s)\nonumber\\&+\zeta\int_{t_{0}}^{t}
\bm{\nabla}X(\psi_{i}(t))\|f(\psi_{i}(s))|^{2}d\mathlarger{\mathscr{W}}(s)+\zeta^{2}
\frac{1}{2}\int_{t_{0}}^{t}(\bm{\nabla}^{2}X(\psi_{i}(t)))\|f(\psi_{i}(s))\|^{2}ds
\end{align}
Averaging then gives
\begin{align}
&\bm{\mathrm{I\!E}}\bigg\lbrace X(\widehat{\psi}_{i}(t))\bigg\rbrace=X(\psi_{i}^{E})+\int_{0} ^{t}\bm{\nabla}\bm{\mathrm{I\!E}}\lbrace X(\psi_{i}(s))\rbrace d\psi_{i}(s)\nonumber\\&
+\zeta^{2}\frac{1}{2}\int_{t_{0}}^{t}(\bm{\nabla}^{2}\bm{\mathrm{I\!E}}\lbrace X(\psi_{i}(t)))\rbrace\|f(\psi_{i}(s))\|^{2}ds
\end{align}
Now letting $X(\psi_{i}(t))=\|\psi_{i}(t)\|^{\ell}$
\begin{align}
&\bm{\mathrm{I\!E}}\bigg\lbrace \|\widehat{\psi}_{i}(t)\|^{\ell}
\bigg\rbrace= \|\psi_{i}^{E}\|^{\ell}
+\int_{0} ^{t}\bm{\nabla}\bm{\mathrm{I\!E}}\lbrace \|\psi_{i}(t)\|^{\ell}
\rbrace d\psi_{i}(s)\nonumber \\&
+\zeta^{2}\frac{1}{2}\int_{t_{0}}^{t}(\bm{\nabla}^{2}\bm{\mathrm{I\!E}}\lbrace \|\psi_{i}(t)\|^{\ell})\rbrace\|f(\psi_{i}(s))\|^{2}ds\nonumber \\&
=\|\psi_{i}^{E}\|^{\ell}+\ell\int_{t_{0}}^{t}\bm{\mathrm{I\!E}}\lbrace\|\psi_{i}(s)\|^{\ell-1}\rbrace d\psi_{i}(t
+\frac{1}{2}\zeta^{2}\ell(\ell-1)\int_{t_{0}}^{t}\|f(\psi_{i}(t)\|^{2}\bm{\mathrm{I\!E}}\lbrace\|\psi_{i}(s)\|^{\ell-2}\rbrace ds\nonumber\\&\le
\|\psi_{i}^{E}\|^{\ell}+\ell\int_{t_{0}}^{t}\bm{\mathrm{I\!E}}\lbrace\|\psi_{i}(s)\|^{\ell-1}\rbrace d\psi_{i}(t)
+\frac{1}{2}\zeta^{2}\ell(\ell-1)\int_{t_{0}}^{t}\|f(\psi_{i}(t)\|^{2}\bm{\mathrm{I\!E}}
\lbrace\|\psi_{i}(s)\|^{\ell-2}\rbrace ds\nonumber\\&\le
\|\psi_{i}^{E}\|^{\ell}+\ell\int_{t_{0}}^{t}\bm{\mathrm{I\!E}}\lbrace\|\psi_{i}(s)\|^{\ell-1}\rbrace d\psi_{i}(t)+\frac{1}{2}\zeta^{2}\ell(\ell-1)\int_{t_{0}}^{t}K\|\psi_{i}(t)\|^{2}
\bm{\mathrm{I\!E}}\lbrace\|\psi_{i}(s)\|^{\ell-2}\rbrace ds\nonumber
\\&\le \|\psi_{i}^{E}\|^{\ell}+|\bm{\mathrm{I\!E}}\lbrace\|\psi_{i}(s)\|^{\ell}\rbrace
-\|\psi_{i}^{E}\|^{\ell} +\frac{1}{2}\zeta^{2}\ell(\ell-1)\int_{t_{0}}^{t}K
\bm{\mathrm{I\!E}}\lbrace\|\psi_{i}(s)\|^{\ell}\rbrace ds
\nonumber\\& f\le \bm{\mathrm{I\!E}}\lbrace\|\psi_{i}(s)\|^{\ell}\rbrace
+\frac{1}{2}\zeta^{2}\ell(\ell-1)\int_{t_{0}}^{t}K
\bm{\mathrm{I\!E}}\lbrace\|\psi_{i}(s)\|^{\ell}\rbrace ds
\end{align}
The Gronwall lemma then gives the estimates (B4) and (B5).
\end{proof}
}
\end{document}